\documentclass[aps,prb,twocolumn,superscriptaddress,floatfix,nofootinbib]{revtex4-2}
\usepackage{amssymb}
\usepackage{physics}
\usepackage{graphicx}
\graphicspath{ {./figures/} }
\usepackage{times}
\usepackage{amsmath}
\usepackage{amsthm}
\usepackage{amsfonts}
\usepackage{comment}
\usepackage{mathtools}
\usepackage[T1]{fontenc}
\usepackage[latin9]{inputenc}
\usepackage{array}
\usepackage{multirow}
\usepackage{color}
\usepackage{esint}
\usepackage{bm}
\usepackage{color}
\usepackage{xcolor}
\usepackage{bbm}
\usepackage{hyperref}
\usepackage{babel}
\usepackage[normalem]{ulem}
\newtheorem{thm}{Theorem}
\newtheorem{lemma}{Lemma}

\theoremstyle{definition}

\newtheorem{conjecture}{Conjecture}

\newcommand{\eqnref}[1]{Eq.~\eqref{#1}}

\usepackage{float}
\hypersetup
{	colorlinks,%
	citecolor=green,%
	linkcolor=red,%
	urlcolor=blue%
}

\usepackage{tikz}
\usetikzlibrary{cd}

\newcommand{\QCA}{\mathrm{QCA}}
\newcommand{\inv}{\mathrm{inv}}

\newcommand{\ZZ}{\mathbb{Z}}
\newcommand{\QQ}{\mathbb{Q}}
\newcommand{\RR}{\mathbb{R}}

\newcommand{\BG}{{\mathrm{B}G}}

\newcommand{\bdy}{{\mathrm{bdy}}}
\newcommand{\id}{{\mathrm{id}}}

\DeclareMathOperator{\LAut}{LAut}
\DeclareMathOperator{\im}{im}
\newcommand{\blend}{\sim}
\newcommand{\BlendClass}[1]{[#1]}

\newcommand{\TwoWayBlend}{\leftrightarrow}

\newcommand{\BlendVia}[1]{\xrightarrow{#1}}
\newcommand{\InverseBlendVia}[1]{\xrightarrow{\overline{#1}}}

\newcommand{\repU}{\mathcal{U}}
\newcommand{\resU}{\widetilde{\repU}}
\newcommand{\bdyU}{\mathcal{V}}
\newcommand{\diffU}{\mathcal{W}}
\newcommand{\repV}{\mathcal{V}}
\newcommand{\repW}{\mathcal{W}}
\newcommand{\onsiteU}{\repU^\bullet}
\newcommand{\onsiteV}{\repV^\bullet}
\newcommand{\onsiteW}{\repW^\bullet}
\newcommand{\grpH}[3]{\mathcal{H}^{#1}[#2,#3]}
\newcommand{\grpHtwisted}[4]{\mathcal{H}^{#1}_{#2}[#3,#4]}
\newcommand{\singH}[3]{H^{#1}[#2,#3]}

\newcommand{\QFTAnom}[2]{\mathrm{Anom}^{\mathrm{QFT}}_{#1}(#2)}
\newcommand{\LatAnom}[2]{\mathrm{Anom}^{\mathrm{lattice}}_{#1}(#2)}
\newcommand{\LatGravAnom}[1]{\mathrm{Anom}^{\mathrm{lattice}}_{#1}}
\newcommand{\QFTGravAnom}[1]{\mathrm{Anom}^{\mathrm{QFT}}_{#1}}
\newcommand{\CLatAnom}[2]{\mathrm{CAnom}^{\mathrm{lattice}}_{#1}(#2)}
\newcommand{\QCAblendclass}[2]{Q_{#1}^{#2}}
\newcommand{\InvModelClass}[2]{P_{#1}^{#2}}
\newcommand{\invblendclass}[2]{I_{#1}^{#2}}

\usepackage{amsthm}

\theoremstyle{definition}

\usepackage{soul}

\newcommand{\unit}{\mathbb{I}}

\begin{document}
\title{Anomalies of global symmetries on the lattice}
\author{Yi-Ting Tu}
\email{yttu@umd.edu}
\affiliation{Condensed Matter Theory Center and Joint Quantum Institute, Department of Physics, University of Maryland, College Park, Maryland 20742, USA}
\author{David M. Long}
\email{dmlong@stanford.edu}
\affiliation{Condensed Matter Theory Center and Joint Quantum Institute, Department of Physics, University of Maryland, College Park, Maryland 20742, USA}
\affiliation{Department of Physics, Stanford University, Stanford, California 94305, USA}
\author{Dominic V. Else}
\email{delse@perimeterinstitute.ca}
\affiliation{Perimeter Institute for Theoretical Physics, Waterloo, Ontario N2L 2Y5, Canada}

\begin{abstract}
    \emph{'t Hooft anomalies} of global symmetries play a fundamental role in quantum many-body systems and quantum field theory (QFT).
    In this paper, we make a systematic analysis of \emph{lattice anomalies}---the analog of 't Hooft anomalies in lattice systems---for which we give a precise definition.
    Crucially, a lattice anomaly is not a feature of a specific Hamiltonian, but rather is a topological invariant of the symmetry action. 
    The controlled setting of lattice systems allows for a systematic and rigorous treatment of lattice anomalies, shorn of the technical challenges of QFT. 
    We find that lattice anomalies  reproduce the expected properties of QFT anomalies in many ways, but also have crucial differences.
    In particular, lattice anomalies and QFT anomalies are not, contrary to a common expectation, in one-to-one correspondence, and there can be nontrivial anomalies on the lattice that are \emph{infrared (IR) trivial}: They admit symmetric trivial gapped ground states, and map to trivial QFT anomalies at low energies. 
    Nevertheless, we show that lattice anomalies (including IR-trivial ones) have a number of interesting consequences in their own right, including connections to commuting projector models, phases of many-body localized (MBL) systems, and quantum cellular automata (QCA). 
    We make substantial progress on the classification of lattice anomalies and develop several theoretical tools to characterize their consequences on symmetric Hamiltonians. 
    Our work places symmetries of quantum many-body lattice systems into a unified theoretical framework and may also suggest new perspectives on symmetries in QFT.
\end{abstract}

\maketitle

\section{Introduction}
\label{sec:introduction}

Symmetry plays a key role in the physics of both quantum field theories (QFTs) and quantum many-body lattice systems~\cite{Wen2007book}.
One of the most important properties that a global symmetry can have is a \emph{'t Hooft anomaly}, henceforth just \emph{anomaly} ~\cite{Hooft1980anomaly,Chen_1106,Kapustin_1403,Else2014,TongLecturesGauge,Arouca2022anomalylectures}. 
One way to think about an 
anomaly is as an obstruction to gauging the symmetry, such that it can act locally rather than globally. 
However, 
anomalies have many other wide-reaching consequences. 
In quantum field theory, 
anomalies can strongly constrain renormalization group (RG) flows~\cite{Hooft1980anomaly}. 
In quantum many-body systems, 
anomalies are responsible for protecting the surface states of symmetry-protected topological (SPT) phases~\cite{Senthil_1405}.

Despite their importance, 
anomalies remain somewhat mysterious. The key principle of an 
anomaly is that it should not viewed as a feature of a specific lattice Hamiltonian or QFT action. Rather, it is a universal property of the way the symmetry acts. 
The 
anomaly, as it were, is the stage in a theater, and the dynamics of the system is the specific play that gets performed on that stage.
However, existing field-theoretic treatments do not always make this fundamental property manifest.
Frequently, one starts with a specific QFT action and then derives the anomaly.

In this paper, we propose a precise definition for the 
anomaly of a global symmetry on the lattice.
This definition encapsulates the idea that an anomaly is a topological invariant of a symmetry action, and conversely that, upon fixing the symmetry group $G$, anomalies should be in one-to-one correspondence with topological classes of $G$-actions.
Our definition does not appeal to field theory nor background gauge fields, and is  explicitly phrased in terms of the symmetry action rather than any given Hamiltonian.
We also make considerable progress toward the general classification of such lattice anomalies, and the characterization of their physical consequences. Past works in this direction include Refs.~\cite{Chen_1106,Else2014,Kapustin_2401,Cheng2023LatticeAnomalyMatching,Zhang2023entanglers,Liu_2405,Seifnashri2025disentangling,Kawagoe_2507}.

One motivation for our work is as a step toward formulating a mathematically precise notion of 
anomaly for QFTs.
Indeed, there is a significant body of folklore on the nature and consequences of anomalies.
Lattice systems offer a well-controlled and mathematically precise setting to determine which conventional expectations regarding anomalies are actually logical consequences of a set definition, and which represent additional assumptions.

Aside from the motivation from QFT, lattice anomalies are of great interest in their own right.
For example, the Lieb-Schultz-Mattis (LSM) theorem~\cite{Lieb1961LSM,Tasaki2022LSMreview}---which forbids a topologically trivial symmetric ground state in a spin system with $\mathrm{SO}(3)$ and translation symmetry and half-integer spin per unit cell---can be viewed as a manifestation of the anomalous action of the symmetry group~\cite{Cheng2016Tranlsation,Cho2017LSMAnomaly,Jian2018LSMSPT,Cheng2023LatticeAnomalyMatching}. 

One of the main observations we make is that the classification of lattice anomalies does not exactly correspond to the classification of anomalies in QFTs. 
There is a map from lattice to QFT anomalies: For a lattice system with a global symmetry, one should be able to construct a Hamiltonian whose low-energy, long-wavelength effective theory is governed by a QFT which inherits some anomaly from the symmetry of the lattice system. 
One of the contributions of this paper is to give a more precise characterization of this map. 
However, we find that there are lattice anomalies that map to \emph{trivial} QFT anomalies, and there are also QFT anomalies that cannot be realized via lattice anomalies~\cite{Kapustin_2401}. 
Lattice anomalies also violate other commonly held assumptions about QFT anomalies---for example, in QFT, an obstruction to gauging is believed to be equivalent to inability to have a trivial gapped ground state, either of which could be used as a definition of nontrivial anomaly. 
On the lattice, it turns out that these are not equivalent.

Our work also has relevance to topological phases of many-body-localized (MBL) systems~\cite{Anderson1958,Gornyi2005,Basko2006,Oganesyan2007,Schreiber2015,Smith2016,Sierant2024review}, which have a different structure to zero-temperature topological phases of quantum many-body systems~\cite{Huse2013order,Pekker2014eigord,Parameswaran2018MBLsymm,Chan2020SPT,Wahl2020mblto,Long2024}. 
In particular, while the boundaries of invertible zero-temperature topological phases are characterized by QFT 
anomalies, we show that the boundaries of invertible MBL phases are characterized by lattice anomalies. 

We also address so-called ``gravitational'' anomalies~\cite{Alvarez1984GravAnomaly, Ryu_1010,Stone_1201,Wen_1303,Thorngren_1404}, which are 't Hooft anomalies that do not depend on any global symmetry. 
We put the word gravitational in quotation marks because it is something of a misnomer, especially in lattice systems, where there is no meaningful way to couple to gravity. 
Nevertheless, we show that the concept of a lattice gravitational anomaly---i.e.,\ a lattice anomaly that does not depend on a global symmetry---is perfectly well-defined.

\section{What is a symmetry? What is an anomaly? A lattice perspective}
\label{sec:what_is_a_symmetry}

In order to define anomalies of global symmetries in quantum many-body lattice systems, the first step is to define \emph{symmetry}. It is a fundamental principle in quantum mechanics that a symmetry group $G$ acts via a unitary (or antiunitary) representation of $G$ on the Hilbert space of the system (in general it is allowed to be a projective representation). However, this definition does not encode the concept of \emph{locality}.

Instead, the symmetry should be defined by its action on \emph{local operators} of the system~\cite{Kapustin_2401,Seifnashri2025disentangling}.
Consider, for example, the paradigmatic example of the spin-flip symmetry in the Ising model. In a finite system with \(N\) qubits, the \(\ZZ_2\)-symmetry generator can be expressed as a product \(u_x = \prod_{j=1}^N X_j\), where \(X_j\) is the Pauli-\(x\) operator.
However, we can just as well specify the symmetry by describing how it acts on local operators.
Consider the action of \(u_x\) on the Pauli operators \(X_j\) and \(Z_j\) when acting by conjugation. Denote \(\repU_x[a] =u_x a u_x^\dagger\). Then, we have
\begin{equation}\label{eqn:IsingAction}
    \repU_x[X_j] = X_j
    \quad\text{and}\quad
    \repU_x[Z_j] = -Z_j.
\end{equation}
The information of \eqnref{eqn:IsingAction} is enough to compute $\repU_x[a]$ for \emph{any} local operator $a$.
Indeed, any such operator can be expressed as a finite sum of finite products of Pauli \(X_j\) and \(Z_j\) operators.
Using this fact and the rules \(\repU_x[a+b] = \repU_x[a] + \repU_x[b]\) and \(\repU_x[ab] = \repU_x[a] \repU_x[b]\)---which follow from the formula \(\repU_x[a] =u_x a u_x^\dagger\)---we can use Eq.~\eqref{eqn:IsingAction} to work out \(\repU_x[a]\) for any local operator \(a\).

In this paper, we thus \emph{define} symmetries by their actions on local operators. An advantage of this perspective is that it does not require that $N$ be finite; instead, we can work directly with infinite systems, which simplifies the discussion since anomalies (like, for example, phases of matter) are strictly well defined only in the thermodynamic limit.

In order to make precise the definition of symmetries in terms of their action on local operators, we set up a few mathematical preliminaries, beginning with the definition of the algebra of local operators.
Consider a lattice system in which the sites are labeled by a set $\Lambda$ (which could be infinite), and each site $s \in \Lambda$ carries a finite-dimensional Hilbert space $\mathcal{H}_s$. 
Thus, the Hilbert space associated with a finite subset $S \subseteq \Lambda$ of sites can be written as a tensor product
\begin{equation}\label{eqn:TensorHilbertSpace}
    \mathcal{H}_S = \bigotimes_{s \in S} \mathcal{H}_s.
\end{equation}
A \emph{local operator} $a$ is an operator that acts on a finite set $S$ of lattice sites; thus, it can be represented as a linear operator on $\mathcal{H}_S$.

Local operators can be added together, multiplied with each other, and multiplied by complex numbers, and we can take their adjoint \(a \mapsto a^\dagger\). (Mathematically, this means that local operators form a \emph{$*$-algebra}.) The key feature of the Ising symmetry was that it mapped local operators to local operators while preserving all this structure: It is an \emph{automorphism} of the algebra of local operators. That is, a linear map from the algebra of local operators, \(A\), to itself:
\begin{equation}
    U : A \to A
\end{equation}
such that
\begin{equation}
    U[ab] = U[a] U[b]
    \quad\text{and}\quad
    U[a]^\dagger = U[a^\dagger]
\end{equation}
and such that $U$ has an inverse $U^{-1}$ satisfying the same conditions.

When the system is finite, this definition reduces to the more familiar one based on unitaries acting on Hilbert space. Indeed, when the set $\Lambda$ is finite, $A$ is just the algebra of all linear operators acting on the finite-dimensional Hilbert space of the system, $\mathcal{H}_\Lambda$. Then, one can show that any automorphism $U : A\to A$ can be expressed in terms of a unitary operator $u$ such that $U[a] = u a u^\dagger$, which agrees with our conventional understanding of symmetries and with the Ising example.

So far the only locality restriction we have placed on the symmetry is that it must map local operators to local operators. However, $U[a]$ could be supported on sites that are arbitrarily far away from the support of $a$. This would allow us to consider symmetries such as reflection symmetry, for example.
However, in the rest of this paper, we choose not to consider such symmetries and require that we map operators supported near a point $\mathbf{x}$ in space to operators supported near the same point $\mathbf{x}$ in space.

This restriction motivates a definition of \emph{local automorphism}. Let us assume that the set of lattice sites $\Lambda$ is embedded in $\mathbb{R}^d$ for some $d$. We say that an automorphism $U$ of the algebra of local operators is a local automorphism, also known as a \emph{quantum cellular automaton (QCA)} or \emph{locality-preserving unitary}~\cite{Gross2012QCA,Farrelly2020qca,Arrighi2019qcareview}, if there exists some finite $r$ such that for any finite set $S \subseteq \Lambda$, and for any $a$ that is supported on $S$, then $U[a]$ is supported on $S^{+r}$, where we have defined $S^{+r}$ to be all the lattice sites in $\Lambda$ whose Euclidean distance from $S$ is at most $r$. We call $r$ the \emph{range} of $U$.
In the Ising example, \(\repU_x\) has range \(r=0\), as it does not grow the support of any operator.

We call a symmetry action by local automorphisms a \emph{local representation of $G$}, henceforth known as a \emph{$G$-rep}.
This is a generalization of the familiar concept of a finite-dimensional unitary (projective) representation as commonly used in quantum mechanics, and, in fact, it reduces to this for $d=0$. To specify a $G$-rep, we first choose a set of lattice sites in $d$-dimensional space $\Lambda \subseteq \mathbb{R}^d$ (for technical reasons, we demand that the lattice does not have any accumulation points; i.e.,\ for each lattice site, $s \in \Lambda$, and any finite radius $r$, there are at most finitely many other lattice sites within distance $r$ of $s$) as well as an assignment of finite-dimensional site Hilbert spaces $\mathcal{H}_s$ to each site $s$. This defines an algebra $A$ of local operators~\footnote{Later, we will also consider algebras of local operators \(A\) that are more general than the tensor product algebras discussed so far, and continue to call homomorphisms as in \eqnref{eqn:GrepDef} \(G\)-reps.}. Then, for each $g \in G$, we choose a local automorphism of this algebra $\repU_g$, such that $\repU_{g_1} \repU_{g_2} = \repU_{g_1 g_2}$ for all $g_1, g_2 \in G$; formally, this defines a group homomorphism~\footnote{If $G$ is continuous (e.g.\ a Lie group) one should require this homomorphism to be continuous, but this would require defining a sensible topology on $\mathrm{LAut}(A)$, which we will not attempt to do in this paper.}
\begin{equation}\label{eqn:GrepDef}
    \repU : G \to \mathrm{LAut}(A),
\end{equation}
where $\mathrm{LAut}(A)$ is the group of local automorphisms of $A$. For example, the Ising symmetry discussed above is an example of a $\mathbb{Z}_2$-rep since $\repU_x^2$ is the identity. (We denote the nontrivial element of \(\ZZ_2\) by \(x\).)

Many of the symmetries that one is most used to considering in lattice systems (such as, for example, spin rotation symmetry and the Ising symmetry) are \emph{on-site symmetries}, in which $\repU_g$ has range zero for each $g \in G$. In that case, one can show that $\repU$ just corresponds to acting with a unitary projective representation of $G$ at each site. As we will see, such symmetries should be viewed as ``nonanomalous''. To obtain anomalous symmetries on the lattice, one must consider non-on-site symmetries.

Even if a $G$-rep is non-on-site, we nevertheless say that it is \emph{internal} if there exists a finite $r$ such that $\mathcal{U}_g$ has range $\leq r$ for all $g \in G$. (If $G$ is finite, this condition is automatically satisfied.) 
An example of a noninternal symmetry would be lattice translation symmetry, with $G = \mathbb{Z}$, since if the lattice spacing is 1, a translation by $n$ lattice sites has range $n$. For most of the results in this paper, it will not be necessary to restrict to internal symmetries (although as we mentioned earlier, crystalline symmetries other than translation symmetry, such as reflection or spatial rotation symmetry, are excluded as they do not act via local automorphisms).

Just as with finite-dimensional unitary representations, given two $G$-reps $\repU$ and $\repV$, there is a sensible notion of the tensor product $\repU \otimes \repV$. 
(Note that the tensor-product rep will have a set of lattice sites that is the union of the two sets of lattice sites of the individual reps.)
Physically, we can view this as a ``stacking''  operation.

We now introduce the notion of an \emph{anomalous} symmetry. We are inspired by the following observation that is believed to hold in QFTs: Two QFTs with global symmetry $G$ have the same 't Hooft anomaly if and only if there exists a spatial interface between the two QFTs that also retains the global symmetry $G$.
For lattice systems, the philosophy of this paper is that we should express the anomaly purely in terms of the action of the symmetry, i.e.,\ in terms of the $G$-reps defined above.

\begin{figure}
    \centering
    \includegraphics[width=\linewidth]{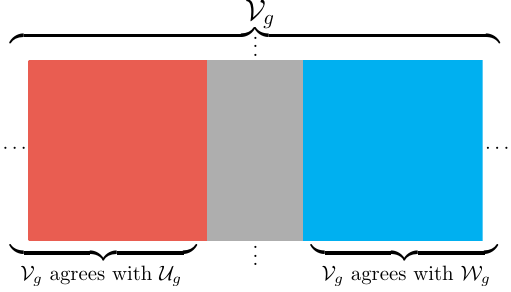}
    \caption{The \(G\)-rep \(\repV\) is a blend between \(\repU\) and \(\repW\) if \(\repV_g\) agrees with \(\repU_g\) in a left half-volume and with \(\repW_g\) in a right half-volume for all \(g \in G\). If a blend exists between \(\repU\) and \(\repW\), they are blend equivalent. A lattice anomaly is a blend equivalence class.}
    \label{fig:blend_illustration}
\end{figure}

We can formulate this by introducing the concept of \emph{blend equivalence} of two $G$-reps. 
Two \(G\)-reps \(\repU\) and \(\repW\) are said to be blend equivalent if there is a third \(G\)-rep \(\repV\) that interpolates between them in space. That is, \(\repV\) acts on local operators supported in the left half of the system (say) in the same way as \(\repU\) and acts on the right half of the system in the same way as \(\repW\). In between the left and right halves, \(\repV\) is arbitrary, except that it must continue to be a \(G\)-rep and so needs to obey all the group laws. We call \(\repV\) a blend between \(\repU\) and \(\repW\). Such a blend is illustrated in \autoref{fig:blend_illustration}.

Returning to the Ising example, we can construct a blend from the Ising \(\ZZ_2\)-rep to the trivial \(\ZZ_2\)-rep, where the generator of \(\ZZ_2\) acts as the identity. Indeed, define \(\repV_x[X_j] = X_j\) for all \(j \in \Lambda\), and \(\repV_x[Z_j] = \pm Z_j\), with the minus when \(j\) is in the left half of \(\Lambda\) and the plus sign on the right half.
Then, \(\repV_x\) fixes all operators supported on the right half but implements a spin flip on the left. We still have \(\repV_x^2 = \id\), so \(\repV: \ZZ_2 \to \LAut(A)\) is a blend from the Ising symmetry to the trivial rep. 
The same kind of blend to the trivial rep works for any other on-site symmetry of any group.

However, it is also possible to find non-on-site \(\ZZ_2\)-reps in one spatial dimension which \emph{do not} blend with the trivial rep.
The CZX \(\ZZ_2\)-rep \cite{Chen_1106} is defined on finite systems as conjugation with a circuit \(\prod_{j} \mathrm{CZ}_{j,j+1} \prod_{j} X_j\), where \(\mathrm{CZ}_{j,j+1}\) is a controlled-\(Z\) gate between qubits \(j\) and \(j+1\). In terms of the local action on Pauli operators, we have
\begin{equation}\label{eqn:CZXAction}
    \repU^{\mathrm{CZX}}_x[X_j] = Z_{j-1} X_j Z_{j+1}
    \quad\text{and}\quad
    \repU^{\mathrm{CZX}}_x[Z_j] = -Z_j.
\end{equation}
This \(\ZZ_2\)-rep \emph{cannot} blend with the trivial rep. A naive attempt at constructing such a blend \(\repV_x\) would be to use the formula above for \(j<0\), fix all operators with \(j>0\), and on \(j=0\) define \(\repV_x[X_0] = Z_{-1} X_0\) and \(\repV_x[Z_0] = -Z_0\). This is a local automorphism, but it fails to be a \(\ZZ_2\)-rep. We have that \(\repV_x^2[X_0] = -X_0\).
This obstruction turns out to be impossible to overcome~\cite{Else2014}, and as a consequence this symmetry action has been referred to as ``anomalous.''

We adopt this obstruction to blending as the \emph{definition} of an anomaly on the lattice. Two \(G\)-reps which have no obstruction to blending are in the same blend equivalence class, so we define:
\begin{quote}
    A \emph{(lattice) anomaly} for a group $G$ is an equivalence class of $G$-reps under blend equivalence.
\end{quote}
We write $\LatAnom{d}{G}$ for the set of lattice anomalies of $G$ in $d$ spatial dimensions.
The trivial anomaly is the one that includes on-site symmetries. 
In general if a $G$-rep has trivial anomaly, then we say it is anomaly-free; otherwise, we say it is anomalous.
(In particular, the trivial $G$-rep, which acts as the identity for all $g \in G$, is anomaly-free.)

A 't Hooft anomaly is sometimes defined as an ``obstruction to gauging the symmetry.'' 
Let us remark that if a $G$-rep has a nontrivial lattice anomaly in the sense defined here, then it implies that there is indeed an obstruction to gauging the symmetry on the lattice. 
A necessary condition for a symmetry to be gauged is that there is a way to implement the symmetry action on subregions of space (while preserving the group structure), including the half-volume $(0,\infty) \times \mathbb{R}^{d-1}$. 
If a $G$-rep is not blend equivalent to the trivial rep, this is exactly saying that such an action on the half-volume is impossible. 

To make this mathematically precise, we now give the more complete definition of a blend. We say that a $G$-rep $\repV$ is a \emph{blend} from $\repU$ to $\repW$
if there exists a finite $\xi$ 
such that 
\begin{itemize}
    \item the local sites and local Hilbert spaces associated to $\repV$ are the same as for $\repU$ on $(-\infty,-\xi) \times \mathbb{R}^{d-1}$;
    \item the local sites and local Hilbert spaces associated to $\repV$ are the same as for $\repW$ on $(\xi,\infty) \times \mathbb{R}^{d-1}$;
    \item $\repV_g(a) = \repU_g(a)$
    for any local operator $a$ supported on $(-\infty,-\xi) \times \mathbb{R}^{d-1}$; and
    \item $\repV_g(b) = \repW_g(b)$
    for any local operator $b$ supported on $(\xi,\infty) \times \mathbb{R}^{d-1}$.
\end{itemize}
This is illustrated in \autoref{fig:blend_illustration}.
(More precisely, this is the definition that makes sense for \emph{internal} symmetries; for noninternal symmetries one should allow $\xi$ to depend on $g$.)
Moreover, if there exists a blend from $\repU$ to $\repW$, we say that $\repU$ and $\repW$ are \emph{blend equivalent}.  
This defines an equivalence relation on $G$-reps (see Appendix~\ref{appendix:BlendEquivalenceDef} for some technical details that we are glossing over here). 
Note that blend equivalence does not make sense in $d=0$, so we will restrict to $d \geq 1$ when discussing blend equivalence.

Note that the stacking operation, i.e.,\ tensor product, gives $\LatAnom{d}{G}$ the structure of an Abelian group.
To see that it has inverses, we use a \emph{folding argument}---taking any $G$-rep and applying a spatial reflection to one half of space gives a blend from that $G$-rep tensor its reflection to the trivial rep.

Finally, as an aside, we note that, in addition to blend equivalence, other definitions of equivalence of $G$-reps are possible. One definition that has previously appeared in the literature~\cite{Zhang2023entanglers,Seifnashri2025disentangling} is \emph{conjugacy equivalence.} We say that $G$-reps $\repU$ and $\repW$ are conjugacy equivalent if there exists a QCA $V$ such that $V \repU_g V^{-1} = \repW_g$. (We still say that $G$-reps are conjugacy equivalent if this equation holds up to tensoring with on-site reps.)
One can show that conjugacy equivalence implies blend equivalence. 
In general, blend equivalence does \emph{not} imply conjugacy equivalence because any $G$-rep that is conjugacy equivalent to an on-site rep is internal (recall the definition of ``internal'' above), while for noncompact groups $G$ (e.g.\ $G = \mathbb{Z}$), one can construct noninternal symmetries that are blend equivalent to the trivial rep. 
We conjecture, however, that blend equivalence implies conjugacy equivalence for \emph{finite} groups; in $d=1$, this has recently been proven in Ref.~\cite{Seifnashri2025disentangling}.

\section{Summary of results}
\label{sec:summary}

For convenience, we give a robust summary of our important results in this section.

\subsection{IR-trivial vs IR-nontrivial anomalies}
\label{subsec:ir_trivial}

We define a lattice anomaly to be an obstruction to blending a symmetry with the trivial rep.
Another property often viewed as a manifestation of an anomaly, and often presented as a definition, is an inability to have a topologically trivial symmetric ground state. 
An important point we make in this paper, which has not been well appreciated until now, is that this is actually \emph{not} equivalent to the definition of anomaly given in \autoref{sec:what_is_a_symmetry}.

Specifically, let us define a lattice anomaly of a symmetry $G$ to be \emph{IR trivial} if there exists a $G$-rep with that anomaly such that there is a product state that is invariant under the $G$-rep. 
Clearly, the trivial anomaly is IR trivial. 
However, in the course of this paper, we will see examples of nontrivial anomalies that are also IR trivial. 
In general, IR-trivial anomalies form a subgroup of $\LatAnom{d}{G}$.

\subsection{Relation between lattice anomalies and QFT anomalies}
\label{sec:IntroLatticeToQFT}

The relationship between lattice anomalies and the more familiar 't Hooft anomalies of QFT is captured by the following claim:
There is a group homomorphism
\begin{equation}
\label{eq:anomaly_homomorphism}
    \varphi : \LatAnom{d}{G} \to \QFTAnom{d}{G},
\end{equation}
where $\QFTAnom{d}{G}$ is the classification of QFT anomalies with symmetry $G$ in $d$ spatial dimensions (i.e.\ the \emph{spacetime} dimension is $d+1$).

This homomorphism has the following interpretation. Let $\alpha \in \LatAnom{d}{G}$ be a lattice anomaly. Consider a lattice system with a $G$ symmetry that has the given lattice anomaly, and consider a lattice Hamiltonian that is invariant under the symmetry. Suppose that the low-energy physics of this Hamiltonian can be described by an effective continuum QFT. Then the microscopic $G$ symmetry should act within this QFT. The anomaly for the $G$ symmetry in this effective QFT will be described by $\varphi(\alpha)$.

The QFT anomaly so defined should depend on only the blend class of the $G$-rep---that is, the map from $\LatAnom{d}{G}$ is well defined.
To see this, note that if we have two $G$-reps $\mathcal{U}$ and $\mathcal{V}$ that are in the same blend class, and two lattice Hamiltonians $H$ and $K$ that are symmetric under $\mathcal{U}$ and $\mathcal{V}$, then we should be able to construct a Hamiltonian that blends between $H$ and $K$ and is symmetric under the $G$-rep that blends between $\mathcal{U}$ and $\mathcal{V}$. 
Then, by passing to the IR, we obtain a symmetric interface between two QFTs acted upon by a $G$ symmetry, and it is believed that this implies the two QFTs have the same anomaly.

This presentation of \(\varphi\) outlined above relies on intuition about emergent QFTs in the lattice, essentially because the codomain of \(\varphi\) involves QFTs.
In this paper, we make progress towards a more lattice-based characterization of \(\varphi\) (although, as we will discuss, the precise construction will apply in a slightly different setting).

We can view  \eqnref{eq:anomaly_homomorphism} as describing a kind of ``UV-IR anomaly matching.'' However, an important point is that, unlike for UV-IR anomaly matching between two QFTs, the source and target of $\varphi$ are not the same group, and in fact $\varphi$  will in general be neither injective nor surjective.
 
The lack of surjectivity means that there are some QFT anomalies for a group $G$ which cannot be realized as a $G$-rep on the lattice. 
For example, in $d=1$, there is a $\mathbb{Z}$-classified QFT anomaly for $G = \mathrm{U}(1)$ symmetry [i.e.\ the (1+1)-dimensional chiral anomaly], but it is believed that this cannot be realized as a $\mathrm{U}(1)$-rep on the lattice \cite{Kapustin_2401,Chatterjee_2409}. (However, Ref.~\cite{Chatterjee_2409} constructed a different symmetry on the lattice with a different group structure, i.e.\ it is \emph{not} a $\mathrm{U}(1)$ symmetry on the lattice, which nevertheless acts \emph{in the IR QFT} as a $\mathrm{U}(1)$ symmetry with an anomaly.)
 
The lack of injectivity (i.e.\ there exist nontrivial lattice anomalies $\alpha$ which do not descend to non-trivial QFT anomalies) follows from the existence of IR-trivial lattice anomalies. 
If there is a product state that is invariant under a $G$-rep, then it is usually possible (in particular, we show in Appendix~\ref{appendix:symmetric_gapped_hamiltonian} that it is \emph{always} possible if $G$ is a finite group, or if $G$ is a compact group and the $G$-rep is internal) to construct a Hamiltonian that is symmetric under that $G$-rep, such that the $G$-symmetric product state is the gapped ground state of that Hamiltonian. 
This implies that the IR description of this lattice Hamiltonian is described by the \emph{trivial} QFT, and in particular the anomaly of this QFT is trivial. 
Thus, IR-trivial lattice anomalies are necessarily contained in the kernel of the homomorphism $\varphi$. 
We expect the stronger statement to hold that a lattice anomaly is in the kernel of the homomorphism $\varphi$ \emph{if and only if} it is IR-trivial.
We will give some arguments for this proposition in the course of the paper, but not a full proof.

Thus, if we have a non-trivial, but IR-trivial, lattice anomaly, it may as well be a trivial anomaly from the point of view of the \emph{low-energy} physics. However, it can still have interesting consequences on the physics at higher energy scales. For example, we will discuss applications to MBL physics.

The fact that $\varphi$ is not an isomorphism might seem a bit unsettling, because a conventional wisdom about anomalies is that they are invariant under RG flow.
However, it is important to remember that under a strict conception of what ``RG flow'' means, it is not actually possible to define an RG flow from a lattice model to a QFT, because RG flow in a QFT is continuous, yet degrees of freedom in a lattice model have a fundamentally different character to those of a QFT and there is no way to continuously interpolate between the two. In lattice models, one does consider discrete RG transformations such as block spin decimation, but the crucial point is that the character of the degrees of freedom after the RG transformation is the same as before the RG transformation (there are just less of them). By contrast, passing from a lattice model to an IR description by a QFT involves a fundamental change in the character of the degrees of freedom, under which anomalies need not be preserved. We expect that RG transformations that remain within the setting of lattice systems \emph{will} preserve the lattice anomaly. For example, Ref.~\cite{Vidal2008MERA} defined an RG transformation for quantum many-body lattice systems known as ``entanglement renormalization''. In \autoref{sec:rg_invariance}, we argue that lattice anomalies are indeed preserved under entanglement renormalization.

Finally, let us be a little more specific about what exactly we mean by $\QFTAnom{d}{G}$. It specifically refers to those anomalies that can be canceled by an invertible TQFT in one higher dimension. Thus, we take it as given that $\QFTAnom{d}{G}$ is in one-to-one correspondence with invertible TQFTs in $d+1$ spatial dimensions ($d+2$ space-time dimensions). 
Recall that the definition of TQFT depends on some choice of tangential structure (e.g.\ orientation, spin structure, ...). 
The position we take is that the appropriate tangential structure should be determined by the nature of the lattice system that we started with. 
For example, so far we have been considering bosonic lattice systems with unitary $G$ symmetry; in that case the invertible TQFT should be defined for manifolds equipped with an orientation and a $G$ gauge field.
Such invertible TQFTs are the ones which are believed to classify the SPT phases in bosonic lattice systems. 
(Later we will talk about fermionic lattice systems and antiunitary symmetries, which imply different tangential structures for the invertible TQFT.)
Note that, importantly, these are the properties of the invertible TQFT in one higher dimension that describes the anomaly; it does \emph{not} imply that the actual IR QFT in $d$ spatial dimensions only depends on the orientation and $G$ gauge field. 
In general, this IR QFT will not even be topological, i.e. it depends on the metric.

\subsection{Lattice anomalies, commuting models, and many-body localization (MBL)}
\label{subsec:intro_projector_phases}

QFT anomalies are believed to exhibit a one-to-one correspondence with symmetry-protected topological (SPT) phases in one higher dimension \cite{Senthil_1405}. Since, as we noted in the previous section, lattice anomalies are \emph{not} in one-to-one correspondence with QFT anomalies, it follows that lattice anomalies are \emph{not} in one-to-one correspondence with SPT phases in one higher dimension.

However, it turns out that there is a bulk-boundary correspondence between lattice anomalies and a notion of \emph{phase of matter} in one higher dimension, but the concept of what we mean by a phase of matter needs to be modified. 
First recall that zero-temperature topological phases of matter are defined in terms of the ground state; two ground states of gapped Hamiltonians are said to be in the same phase if the ground states are related by a finite-depth quantum circuit (or some generalization allowing for rapidly decaying tails).

\emph{Commuting models} can also be used to define a phase of matter. 
A commuting model is a collection $\{ h_i : i \in I \}$ of local Hermitian operators indexed by a set $I$, where $[h_i, h_j] = 0$ for all $i,j \in I$. (In particular, we can consider commuting projector models, in which the $h_i$'s are taken to be projectors.)
Specifically, construct the Hamiltonian
\begin{equation}
    H = \sum_{i \in I} h_i,
\end{equation}
whose ground state is necessarily gapped, with no correlations between regions that are further apart than the diameter of the $h_i$'s supports. 
Commuting models have been extensively studied since they allow one to construct exactly solvable models of many topological phases of matter (e.g. Ref.~\cite{Levin2005}).

 We say that commuting models $\{ h_i \}$ and $\{ h_i' \}$ are \emph{equivalent} if $h_i$ and $h_i'$ have the same ground-state subspace for each $i$. Thus, without loss of generality one could just replace each $h_i$ with a projector to get a commuting projector model.

We say that a commuting model is \emph{trivial} if each $h_i$ is supported on a single site, with the ground state subspace of $h_i$ on that site being one-dimensional. 
Trivial commuting models evidently have a non-degenerate ground state which is a product state, but the converse is not true---having a non-degenerate product-state ground state does not automatically imply the commuting model is trivial.
We say two commuting models are in the \emph{same phase} if they are related by a finite-depth quantum circuit, up to equivalence and stacking with trivial commuting models. Note that this is a different condition than the statement that the \emph{ground states} of the two commuting models are in the same phase.

We say that a commuting model $\{ h_i : i \in I \}$ is strongly symmetric under a $G$-rep $\repU_g$ if $\repU_g[h_i] = h_i$ for all $g \in G, i \in I$.
Now suppose we restrict to commuting models which are strongly symmetric, and are in the trivial commuting phase, i.e. the one containing the trivial commuting Hamiltonian if the symmetry is disregarded. 
Then we can say that two such commuting models are in the same \emph{symmetry-protected} commuting phase if there exists a finite-depth quantum circuit relating them in which all the gates can be taken to be individually invariant under the symmetry. 
This is analogous to the concept of SPT phase for ground states.

One of our main results is that, unlike ground state SPT phases, symmetry-protected commuting phases have a concept of ``boundary anomaly'' that is best captured by \emph{lattice anomalies}, not QFT anomalies. 
(This will also hold for commuting models without symmetry that satisfy an \emph{invertibility} property, once we introduce the concept of lattice gravitational anomaly, see \autoref{subsec:anti_unitary_etc} below.)
Physically, this is because the property of being a commuting model on the lattice is not something that is captured by a continuum QFT description.

A key feature of a QFT anomaly 
is that it forbids a symmetric trivial gapped ground state. 
There is an analogous property for lattice anomalies, but it is necessary to phrase it in terms of commuting models. 
Specifically, we will show that in a system with a $G$-rep that carries a non-trivial lattice anomaly (even if it is IR-trivial), it is impossible to find a commuting model which is strongly symmetric under the $G$-rep, \emph{and} is in the trivial phase if one disregards the symmetry. 
(In fact, as for ground states, there is a stronger result that the anomaly also forbids \emph{invertible} commuting phases.)

The concept of a commuting phase might seem somewhat abstract. 
However, there is a related notion of phase of matter that is also related to lattice anomalies and has great physical relevance, specifically the concept of a many-body localized (MBL) phase. 
Many-body localization is a phenomenon that occurs in generic strongly-disordered isolated quantum many-body systems, whereby there are an extensive number of quasi-local integrals of motion. 
An MBL phase is a family of Hamiltonians that can be continuously interpolated into each other while maintaining the property that the dynamics is MBL. 
One can talk about topological phases of MBL systems, which are not in one-to-one correspondence with ground-state MBL phases. 
We will argue that lattice anomalies naturally occur at the boundaries of topological MBL phases, and moreover that a non-trivial lattice anomaly forbids a symmetric MBL phase that lies in a trivial MBL phase (or more generally any MBL phase that satisfies a condition known as \emph{short-range entangled} \cite{Long2024}) when we disregard the symmetry.
Again, the reason that lattice anomalies, rather than QFT anomalies, are related to MBL phases is that MBL is fundamentally a lattice phenomenon that does not have any continuum QFT description.

\subsection{Classification of lattice anomalies: results and conjectures}
Above, we have set up the general concept of a lattice anomaly. The other question is how to classify such anomalies. While we will not give a fully general answer to the question in this paper, we will establish various partial results and conjectures. Specifically:
\begin{itemize}
    \item We argue that there are a sequence of cohomological invariants that partially characterize non-trivial lattice anomalies, based on the classification of QCAs. We rigorously construct the first few of these invariants.
    \item We conjecture a more general classification based on homotopy theory and the idea that QCAs should define a generalized cohomology theory~\cite{Long2024}, and which in particular captures the invariants mentioned above. This can be viewed as analogous to the proposals that the classification of SPT phases of matter can be described by a generalized cohomology theory~\cite{Kitaev2013SRE2,Kitaev2015SRE3,Kitaev2019SRE4,Xiong2018minimalist,Gaiotto2019gencohomology,Kubota2025stablehomotopytheoryinvertible}. 
    \item We also make a separate conjecture that lattice anomalies can be understood in terms of TQFTs, specifically their \emph{Witt classes} (equivalence classes under the relation of supporting gapped interfaces). This generalizes a previous conjecture about the classification of QCAs~\cite{Haah2022,Shirley2022,Haah2023}.
    \item Finally, we establish an equivalence between: lattice anomalies in \(d\) dimensions; QCAs in \(d+1\) dimensions which commute with an on-site symmetry; and invertible commuting models in \(d+1\) dimensions, again with an on-site symmetry.
    This is to be viewed as a bulk-boundary correspondence between lattice anomalies and either of the other two objects.
    (Actually, the precise correspondence involves a different definition of lattice anomaly from the one given above, though it is closely related.)
\end{itemize}

\subsection{Anti-unitary symmetries, fermionic symmetries, and gravitational anomalies}
\label{subsec:anti_unitary_etc}

Here we, describe generalizations of the concept of a $G$-rep that we introduced in \autoref{sec:what_is_a_symmetry}. 
First, there we restricted to \emph{unitary} symmetries, but we can also define an \emph{antiunitary QCA}, which is a map from the algebra of local operators to itself that satisfies all the conditions to be a QCA, except that the linearity condition is replaced by anti-linearity, $U[\alpha a + \beta b] = \alpha^* U[a] + \beta^* U[b]$ where $a,b \in A$ and $\alpha, \beta \in \mathbb{C}$. 
Thus we can talk about symmetry groups containing anti-linear elements. 
In this case when we talk about the symmetry group $G$, we should also keep track of which elements of $G$ are supposed to be antiunitary symmetries (in particular this involves specifying a homomorphism $\sigma : G \to \mathbb{Z}_2$), and then demand that antiunitary elements of $G$ map to antiunitary QCAs. 
As with unitary reps, a $G$-rep involving antiunitary symmetries is said to be \emph{on-site} if $\mathcal{U}_g$ has range zero for all $g \in G$. 

In \autoref{sec:what_is_a_symmetry}, it was implicit that we were talking about so-called bosonic systems (which means that any two operators supported on disjoint sets of lattice sites commute with each other). 
We can also easily discuss fermionic systems as well. 
Then the structure of the algebra of local operators is different, because local operators can have even or odd fermion parity, and if we have two fermion-parity-odd operators supported on disjoint sets, then they \emph{anti-commute} rather than commute.

For fermionic systems one should consider a \emph{fermionic symmetry group} $G_f$ that contains a central $\mathbb{Z}_2$ subgroup which we call $\mathbb{Z}_2^f$. 
For a $G_f$-rep, we demand that the generator of $\mathbb{Z}_2^f$ maps to the \emph{fermion parity automorphism} $\mathcal{P}$, which has the property that $\mathcal{P}(a_+) = a_+$ for fermion-parity-even operators $a_+$, and $\mathcal{P}(a_-) = -a_-$ for fermion-parity-odd operators $a_-$. 
One can distinguish between two different kinds of fermionic symmetries. 
If $G_f$ is just a direct product $G_f = G_b \times \mathbb{Z}_2^f$, then the $G_f$-rep is fully determined by the $G_b$-rep and we won't need to think about fermion parity explicitly. 
On the other hand, $G_f$ could also be a non-trivial extension of $G_b = G_f / \mathbb{Z}_2^f$ by $\mathbb{Z}_2^f$. 
In this case we will say that the symmetry has a non-trivial extension by fermion parity.
For fermionic $G_f$-reps there is also a natural definition of stacking; we defer the precise details to Appendix~\ref{appendix:fermionic_local_algebra}.

Finally, we note that we will also discuss \emph{lattice gravitational anomalies}. 
A lattice gravitational anomaly is a lattice anomaly that does not relate to a symmetry. 
Thus, what is ``anomalous'' in this case ends up being not the action of the symmetry on local operators, but the structure of the algebra of local operators itself. In a system carrying a lattice gravitational anomaly, the algebra of local operators does not simply have the form assumed in \autoref{sec:what_is_a_symmetry} (or the fermionic version thereof), but rather some generalization. 
We defer the precise definitions to later in the paper.

\subsection{Outline of the paper}

We summarize the main results of our work, which also serves as an outline of the paper.

\begin{itemize}
    \item We make a comparison between lattice anomalies and QFT anomalies in one-dimensional systems in \autoref{sec:1DWarmUp}, highlighting broad similarities and specific differences in a simple setting.

    \item Several invariants which diagnose the lattice anomaly of a \(G\)-rep based on the group cohomology of \(G\) are defined in \autoref{sec:CohomologyInvariants}. These invariants are all constructed by considering restrictions of the symmetry action to subsets of the lattice, and can be viewed as generalizations of the invariants of Ref.~\cite{Else2014}.

    \item A detailed example of an anomalous two-dimensional \(\ZZ_2\)-rep based on the radical chiral Floquet circuit of Ref.~\cite{Po2017} is presented in \autoref{sec:em-qubit}. We explain how the anomaly is diagnosed by one of the cohomology invariants.

    \item The consequences of a model possessing a lattice-anomalous symmetry are studied in \autoref{sec:ConsequencesOfAnomalies}. While we find that lattice anomalies forbid symmetric trivial commuting Hamiltonians, they allow for symmetric trivial gapped ground states.
    \item The lattice analogue of gravitational anomalies in QFT (which do not require the presence of a symmetry) is introduced in \autoref{sec:LatticeGravitational}. 
    We associate lattice gravitational anomalies not to a symmetry action, but to the algebra of local operators itself.
    This notion maps homomorphically to QFT gravitational anomalies.

    \item \autoref{sec:BulkBoundaryCorrespondence} describes a bulk-boundary correspondence between (a modified notion of) lattice anomalies and {both: finite-depth circuit equivalence classes of symmetric QCAs (\(G\)-QCAs); and finite-depth circuit equivalence classes of symmetric \emph{invertible} commuting models.}
    A \(G\)-QCA is seen to pump a \(G\)-rep across the system, reminiscent of anomaly in-flow. 
    Similarly, a commuting model supports a lattice anomaly at its boundary.

    \item \autoref{sec:rg_invariance} argues that lattice anomalies are invariant under entanglement renormalization~\cite{Vidal2008MERA}.

    \item Secs.~\ref{sec:GrepFromGTQFT} and~\ref{sec:HomotopyTheory} outline more speculative conjectures regarding the complete classification of lattice anomalies, which we do not completely resolve: 
    \begin{itemize}
        \item In \autoref{sec:GrepFromGTQFT} we propose that there is a subgroup of lattice anomalies which correspond to the Witt classes of symmetry enriched TQFTs, generalizing a similar conjecture regarding QCAs.
    \item In \autoref{sec:HomotopyTheory} we describe how lattice anomalies can be fit into a homotopy-theoretic structure, such that all the cohomology invariants of \autoref{sec:CohomologyInvariants} naturally emerge.
    \end{itemize}
    \item We conclude and discuss future directions in \autoref{sec:Discussion}.
\end{itemize}

Additionally, the main text is accompanied by a series of appendices which contain additional technical or tangential details.

\section{A warm-up: The one-dimensional case}
\label{sec:1DWarmUp}

In this section, we discuss lattice anomalies and their relation to QFT anomalies in one spatial dimension. Most of the results described here are already known in some form in the literature, but we will emphasize the way in which they fit into our general conceptual framework.

Let $G$ be a discrete group. (One can also consider continuous groups fairly easily, but it would require introducing technicalities we do not wish to discuss.)
Let us discuss how to classify lattice anomalies of $G$.

To begin, we discuss how to classify one-dimensional bosonic QCAs.
A blend between QCAs is defined analogously to blends of \(G\)-reps in \autoref{sec:what_is_a_symmetry}, and the complete classification of equivalence classes under blending is known in one-dimension.
Indeed, for any bosonic one-dimensional QCA $U$, one can define an index $\mathcal{I}(U)$, known as the Gross--Nesme--Vogts--Werner (GNVW) index~\cite{Gross2012QCA}, valued in $\mathbb{Q}_\times$ (the group of positive rational numbers, with the group operation being multiplication). 
The GNVW index is a complete invariant for blend classes of one-dimensional QCAs: two QCAs $U$ and $V$ blend if and only if $\mathcal{I}(U) = \mathcal{I}(V)$.
All values of \(\mathcal{I}(U) \in \QQ_\times\) are achieved for some \(U\) on some one-dimensional lattice system.

The GNVW index extends to an invariant for one-dimensional bosonic \(G\)-reps.
Every $G$-rep $\repU$ induces a homomorphism $\mathcal{I}_\repU : G \to \mathbb{Q}_\times$ defined by \(\mathcal{I}_\repU(g) = \mathcal{I}(\repU_g)\).
We will overload terminology and also refer to the homomorphism $\mathcal{I}_\repU$ as the GNVW index of the representation $\repU$. 
It is clear that if two local representations of $G$ have different GNVW index, they are not blend-equivalent and therefore have distinct lattice anomalies.
For concreteness, consider for example the case of $G = \mathbb{Z}$, where $G$ acts as lattice translations. Since lattice translations have non-trivial GNVW index, this symmetry has a non-trivial lattice anomaly.

On the other hand, it is clear that this anomaly is IR-trivial in the sense described in Section \ref{subsec:ir_trivial}, since one can have product states that are invariant under translation symmetry. 
Moreover, recall that
QFT anomalies for bosonic systems for a $G$ symmetry in (1+1)~dimensions are known to be classified by $H^4(\BG, \mathbb{Z})$ \cite{Dijkgraaf1990,Chen2013}, where $\BG$ is the classifying space of $G$, and $H^{\bullet}$ denotes singular cohomology. In particular, if $G = \mathbb{Z}$ then $H^4(\BG, \mathbb{Z})$ is trivial, so there cannot be any QFT anomalies for this symmetry. Therefore, the lattice anomaly of lattice translation symmetry that we mentioned above maps to the \emph{trivial} QFT anomaly under the homomorphism of \eqnref{eq:anomaly_homomorphism}.

The reader may observe that to get a non-trivial GNVW index, it is necessary for the symmetry group $G$ to be infinite and non-internal. However, the reader should not conclude from this that non-trivial lattice anomalies that map to trivial QFT anomalies necessarily involve infinite symmetries like translation. While this is true in one dimension, we will see in higher dimensions that there are examples involving finite groups (\autoref{sec:em-qubit}).

Now, for \emph{fermionic} systems, the QFT anomalies for symmetry $G_f = \mathbb{Z} \times \mathbb{Z}_2^f$ have a $\mathbb{Z}_2$ classification \cite{Cheng_1501}. One can construct a corresponding $G$-rep on the lattice, in which the generator of $\mathbb{Z}$ acts as a ``Majorana translation'' (that is, if the fermionic degrees of freedom are built out a chain of Majorana modes $\gamma_i$, then the generator sends $\gamma_i \mapsto \gamma_{i+1}$). Thus, this is a simple example of a lattice anomaly that does map to a non-trivial QFT anomaly. However, note that lattice anomalies have a 
\(\QQ_\times\)
classification~\cite{Fidkowski2019FloquetFermion}, whereas the QFT anomaly classification collapses to $\mathbb{Z}_2$.

On the other hand, suppose we instead want to set $G_f = \mathbb{Z}_2 \times \mathbb{Z}_2^f$. The QFT anomalies still have a $\mathbb{Z}_2$ classification, but it is now impossible to realize a $\mathbb{Z}_2 \times \mathbb{Z}_2^f$ symmetry on the lattice which maps to the non-trivial anomaly~\footnote{Specifically, this follows from the classification of fermionic QCAs in one dimension~\cite{Fidkowski2019FloquetFermion}. The classifying group contains no elements of finite order, so the generator of a $\mathbb{Z}_2$ symmetry must have trivial QCA index. However, this implies that the $\mathbb{Z}_2$ symmetry cannot permute the trivial and non-trivial one-dimensional fermionic topological phases, which is the key property of the QFT anomaly.}. Thus, this QFT anomaly is not in the image of the map from lattice anomalies.

Returning to the bosonic case, suppose we have a $G$-rep for which the GNVW index vanishes, with $G$ a discrete group. Reference~\cite{Else2014} gave an explicit procedure showing that for any such \(G\)-rep, one can derive a blend invariant index valued in $\mathcal{H}^3[G, \mathrm{U}(1)]$, where $\mathcal{H}^{\bullet}$ denotes group cohomology.  This index gives a partial classification of lattice anomalies for $G$ symmetries. 
For finite groups $G$ it is a mathematical fact that $\mathcal{H}^3[G, \mathrm{U}(1)] \cong H^4(\BG, \mathbb{Z})$. Thus, one expects that lattice anomalies classified by the index of Ref.~\cite{Else2014} valued in $\mathcal{H}^3[G, \mathrm{U}(1)]$ should map directly to the corresponding QFT anomalies. 
Moreover, the construction of Refs.~\cite{Chen2013,Else2014} show that all elements of $\mathcal{H}^3[G,\mathrm{U}(1)]$ can be realized by some $G$-rep on the lattice for finite groups $G$.

In fact, the procedure of Ref.~\cite{Else2014} can also be applied to  representations with non-trivial GNVW index in order to compute their corresponding QFT anomaly \cite{Kapustin_2401}. In this case we invoke the fact that for any GNVW index $\mathcal{I} : G \to \mathbb{Q}_\times$, there exists a local representation $\repV_\mathcal{I}$ whose GNVW index is $\mathcal{I}$, but which maps into a trivial QFT anomaly (which one can verify, for example, by the existence of a product state that is invariant under $\repV_\mathcal{I}$). 
Then for any \(G\)-rep with GNVW index $\mathcal{I}$, we can consider the \(G\)-rep $\repU \otimes \repV_{\mathcal{I}^{-1}}$, which has trivial GNVW index, and compute its QFT anomaly via the procedure of Ref.~\cite{Else2014}.
Since we know that $\repV_{\mathcal{I}^{-1}}$ has trivial QFT anomaly, it follows that $\repU \otimes \repV_{\mathcal{I}^{-1}}$ has the same QFT anomaly as $\repU$ itself.

The above procedure allows one to, for example, compute the QFT anomaly corresponding to $G = \mathrm{SO}(3) \times \mathbb{Z}$, where on the lattice $\mathbb{Z}$ corresponds to lattice translations, and each translation unit cell carries half-integer spin under $\mathrm{SO}(3)$.
Although this is still not a finite group, it remains true in this case that $H^4(\BG,\mathbb{Z}) \cong \grpH{3}{G}{\mathrm{U}(1)} \cong \mathbb{Z}_2$. 
[If one does not want to consider continuous groups, one could just replace $\mathrm{SO}(3)$ with a $\mathbb{Z}_2 \times \mathbb{Z}_2$ group comprising the $\pi$ rotations about the $x$, $y$ and $z$ axes.]
The fact that such a system does not admit a trivial gapped ground state is the famous Lieb-Schultz-Mattis (LSM) theorem~\cite{Lieb1961LSM,Tasaki2022LSMreview}. However, one can make a stronger statement: if the ground state is gapless, the QFT describing the low-energy states must carry the appropriate QFT anomaly~\cite{Cheng2016Tranlsation,Cho2017LSMAnomaly,Jian2018LSMSPT,Cheng2023LatticeAnomalyMatching}.

\section{Cohomology invariants}
\label{sec:CohomologyInvariants}

Many families of QFT anomalies are known to be classified by group cohomology. The same is true for lattice anomalies~\cite{Else2014}, and in this section we define simple blend-equivalence invariants for \(G\)-reps from the group cohomology of \(G\).

These invariants are based on knowing something about the classification of QCAs~\cite{Gross2012QCA,Freedman2020_2DQCA,Haah2022,Shirley2022}. Two QCAs are said to be \emph{equivalent} if they are blend-equivalent, defined analogously to blend equivalence of $G$-reps as in \autoref{sec:what_is_a_symmetry}. Another notion of equivalence that is commonly used for QCAs is \emph{circuit equivalence}---two QCAs \(U,V\) are circuit equivalent if, possibly after stacking with ancillas, we have \(U=XV\) where \(X\) is a finite depth circuit (a product of finitely many layers of local unitary operators, where all operators in the same layer have disjoint support). We prove in \autoref{sec:GQCABlendToCircuit} that circuit equivalence is, in fact, identical to blend equivalence. 
We write $Q_d$ for the set of equivalence classes of QCAs on $\mathbb{R}^d$. The stacking operation endows this with the structure of an Abelian group.

We already mentioned the result that for bosonic systems, $Q_1 \cong \mathbb{Q}_\times$, corresponding to the GNVW index. It has also been shown that $Q_2$ is trivial  for bosonic systems \cite{Freedman2020_2DQCA}. In higher dimensions not very much is known rigorously, but it is believed that for bosonic systems, $Q_3$ at least contains a $\mathbb{Z}_8$ subgroup \cite{Haah2022,Shirley2022}.

The invariants of $G$-reps that we construct are of the form
\begin{equation}
    \grpH{p+1}{G}{\QCAblendclass{d-p}{}}.
\end{equation}
These invariants can be thought of as measuring an obstruction to blending a \(G\)-rep to an on-site \(G\)-rep arising from a symmetry defect of codimension \(p\). They can be defined for either unitary or antiunitary symmetries, and both in bosonic and fermionic systems. However, in this section we will not consider fermionic systems in which the symmetry has a non-trivial extension by fermion parity, as this would introduce additional complications. (For fermionic systems without a non-trivial extension by fermion parity, for which the fermionic symmetry group decomposes as $G_f = G_b \times \mathbb{Z}_2^f$, the ``$G$'' appearing throughout this section means $G_b$.)

Here, we explicitly work through these invariants one-by-one for \(p \leq 2\). In a later section, \autoref{sec:HomotopyTheory}, we will give a more sophisticated (but less rigorous) treatment which captures all invariants at once, at the cost of being substantially more abstract.

\subsection{The \texorpdfstring{$\grpH{1}{G}{\QCAblendclass{d}{}}$}{H1} invariant}

Analogously to the discussion in the one-dimensional case above, in $d$ dimensions a (unitary) $G$-rep $\repU$ induces a homomorphism $\omega : G \to \QCAblendclass{d}{}$. Since $\QCAblendclass{d}{}$ is an Abelian group, we can equivalently write this as an element of $\grpH{1}{G}{\QCAblendclass{d}{}}$. By construction, this \(\omega\) is invariant under blending of \(G\)-reps.

We can directly think of \(\omega\) as an obstruction to blending a \(G\)-rep with the trivial rep. Indeed, a necessary (but not sufficient) condition for a $G$-rep to blend with the trivial rep is that each $\mathcal{U}_g$ blends with the identity as a QCA for each $g \in G$, and $\omega$ is the obstruction to this.

Let us briefly describe how this result gets modified if we include antiunitary symmetries.
The important point is the following. One can prove that if $U$ and $V$ are unitary QCAs, then $U V U^{-1}$ has the same index in $\QCAblendclass{d}{}$ (that is, the same equivalence class) as $V$ itself, which is why $\QCAblendclass{d}{}$ is Abelian~\cite{Haah2022,Freedman2020_2DQCA,Freedman2022group}. However, if $U$ is antiunitary then this is not necessarily the case. Note, however, that since for any two antiunitary QCAs $U,U'$, we have that $U^{-1} U'$ and $U^2$ are unitary, we get a canonical action of $\mathbb{Z}_2$ on $\QCAblendclass{d}{}$ which does not depend on the choice of a specific antiunitary QCA. 
Note that in $d=1$, the $\mathbb{Z}_2$ action on $\QCAblendclass{1}{} \cong \mathbb{Q}_{\times}$ is actually trivial. For example, translations are invariant under complex conjugation in a product basis. 
However, the action can be non-trivial in $d > 1$. For example, Ref.~\cite{Shirley2022} argues that $\QCAblendclass{3}{}$ contains a $\mathbb{Z}_8$ subgroup, and one can verify that these elements of $\QCAblendclass{3}{}$ are odd under complex conjugation.

Now, following the notation that we introduced in \autoref{subsec:anti_unitary_etc}, let $\sigma : G \to \ZZ_2 = \{ 0, 1 \}$ be the homomorphism that keeps track of the unitary/antiunitary elements of $G$. For any $g \in G$, we can write $\mathcal{U}_g =  \mathfrak{U}_g K^{\sigma(g)}$, where $\mathfrak{U}_g$ is a unitary QCA and $K$ denotes complex conjugation in a product basis. Then we define $\omega_g \in \QCAblendclass{1}{}$ to be the index of $\mathfrak{U}_g$. Moreover, since $\mathcal{U}_{g_1} \mathcal{U}_{g_2} = \mathcal{U}_{g_1 g_2}$, we find that
\begin{equation}
    \mathfrak{U}_{g_1} K^{\sigma(g_1)} \mathfrak{U}_{g_2} K^{\sigma(g_2)} =  \mathfrak{U}_{g_1 g_2} K^{\sigma(g_1, g_2)},
\end{equation}
which implies that
\begin{equation}\label{eqn:AntiunitaryOnQCA}
    {}^{\sigma(g_1)} \omega_{g_2} \omega_{g_1} = \omega_{g_1 g_2},
\end{equation}
where $^{n}\omega$, for $n \in \mathbb{Z}_2$ and $\omega \in \QCAblendclass{d}{}$, denotes the action of $\mathbb{Z}_2$ on $\QCAblendclass{d}{}$. We recognize this as defining an element of the twisted cohomology $\grpHtwisted{1}{\sigma}{G}{\QCAblendclass{d}{}}$, where the subscript $\sigma$ indicates that the  group cohomology is defined with respect to an action of $G$ on the coefficients $Q_d$; specifically, antiunitary elements of $G$ act on $\QCAblendclass{d}{}$ through the canonical action of $\mathbb{Z}_2$ on $\QCAblendclass{d}{}$.

\subsection{The \texorpdfstring{$\grpH{2}{G}{\QCAblendclass{d-1}{}}$}{H2} invariant}
\label{sec:H2Invariant}

Suppose that the \(\grpH{1}{G}{\QCAblendclass{d}{}}\) invariant vanishes. For simplicity, we first assume that the $G$-rep is unitary and internal (we will extend to the general case later). Then we can restrict $\mathcal{U}_g$ to any subset $R \subseteq \mathbb{R}^d$.  However, the restriction may fail to remain a \(G\)-rep.

By ``restriction'' we mean that we fix a finite $r \geq 0$ and then for each $g$, we find a QCA $\resU_g$ such that $\resU_g$ acts as the identity on local operators supported on $R^c$, and the same as $\repU_g$ on local operators supported on $\mathrm{int}_r R$. Here we defined the complement $R^c = \mathbb{R}^{d} \setminus R$, and the \(r\)-interior $\mathrm{int}_r R = [(R^c)^{+r}]^c$, where
\begin{multline}
R^{+r} = \{ \mathbf{x} \in \mathbb{R}^d : \mbox{there exists $\mathbf{x}' \in R$} \\ \mbox{such that $|\mathbf{x} - \mathbf{x}'| \leq r$}  \}
\end{multline}
is the \(r\)-neighborhood of \(R\).
To show that such a restriction exists for any subset $R \subseteq \mathbb{R}^d$, we can use the circuit representation of $\mathcal{U}_g$, since the $\grpH{1}{G}{\QCAblendclass{d}{}}$ invariant being zero implies that $\mathcal{U}_g$ is in the trivial circuit equivalence class for all $g \in G$. Specifically, $\resU_g$ can be constructed by removing all the blocks of the circuit not supported within $R$. Note that a choice of restriction involves some arbitrary choice about how the QCA will act near the boundary of $R$, i.e.\ on $R \setminus \mathrm{int}_r R$.

\begin{figure}
    \centering
    \includegraphics{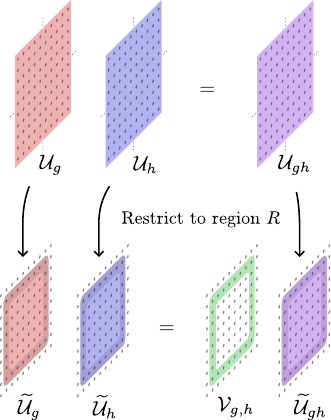}
    \caption{Depiction of the $\grpH{2}{G}{\QCAblendclass{d-1}{}}$ invariant for $d=2$. When the $G$-rep $\repU$ is restricted to a region $R$ with boundary, the homomorphism property only holds up to a one-dimensional QCA $\bdyU_{g,h}$ on $S\subseteq R$ near the boundary of $R$. The value of the cocycle $\omega_{g,h}$ is defined as the QCA index of $\bdyU_{g,h}$.}
    \label{fig:sym-restrict}
\end{figure}

Now, for $g,h\in G$, by the homomorphism property of $G$-reps we must have $\repU_g\repU_h=\repU_{gh}$, and hence the action of $\resU_g\resU_h$ is the same as that of $\resU_{gh}$ on $\mathrm{int}_{2r} R$.
However, this equality may not hold near the boundary.
 The homomorphism property must be modified as (\autoref{fig:sym-restrict})
\begin{equation}
    \label{eq:restricted_homomorphism}
    \resU_g\resU_h=\bdyU_{g,h}\resU_{gh},
\end{equation}
where $\bdyU_{g,h}$ is a QCA defined in the boundary strip $S\subseteq R$ consisting of sites near the boundary (specifically, $S = R \setminus \mathrm{int}_{3r} R$.)

Suppose we take $R$ to be a half-volume $(-\infty,a) \times \mathbb{R}^{d-1}$. Then $S$ is an infinite slab $(a-3r,a) \times \mathbb{R}^{d-1}$. By projecting this slab onto $\{ a \} \times \mathbb{R}^{d-1}$, we can treat $\bdyU_{g,h}$ as a $(d-1)$-dimensional QCA and calculate its index in $\QCAblendclass{d-1}{}$. Denote this index by $\omega_{g,h}$, and regard this as a 2-cochain $\omega$ in the group cohomology of $G$ with coefficients in $\QCAblendclass{d-1}{}$.

Note that $\omega$ cannot be an arbitrary function of $g$ and $h$. For $g,h,k\in G$, we have two ways to write the product of the three restricted actions:
\begin{equation}
    \label{eq:ghk_1}\resU_g\resU_h\resU_k=\bdyU_{g,h}\resU_{gh}\resU_k=\bdyU_{g,h}\bdyU_{gh,k}\resU_{ghk}
\end{equation}
and 
\begin{equation}\label{eq:H2closed}
    \resU_g\resU_h\resU_k=\resU_g\bdyU_{h,k}\resU_{hk}=\left(\resU_g\bdyU_{h,k}\resU_g^{-1}\right)\bdyU_{g,hk}\resU_{ghk}.
\end{equation}
Note that $\resU_g\bdyU_{h,k}\resU_g^{-1}$ is also a QCA on $S$ with index the same as that of $\bdyU_{h,k}$, which follows from the fact that \(\QCAblendclass{d-1}{}\) is Abelian and that composition of QCAs corresponds to multiplication in \(\QCAblendclass{d-1}{}\)~\cite{Haah2022,Freedman2020_2DQCA,Freedman2022group}. 
Therefore, by comparing Eqs.~(\ref{eq:ghk_1},~\ref{eq:H2closed}) we get
\begin{equation}\label{eq:H2closedClass}
    \omega_{g,h}\,\omega_{gh,k}=\omega_{g,hk}\,\omega_{h,k},
\end{equation}
which is the condition for a 2-cochain to be a cocycle.

The construction above seems to depend on a particular choice of restriction $\resU_g$. What happens if we choose a different restriction? If we have two choices of restriction $\resU_g$ and $\resU_g'$, then we must have $\resU_g = \repW_g \resU_g'$ for some QCA $\repW_g$ supported near the boundary of $R$, i.e. on $R \setminus \mathrm{int}_{2r} R$.
Then we have
\begin{equation}\label{eq:H2rechoose}
    \resU'_g\resU'_h\resU'^{-1}_{gh}
    = \diffU_g(\resU_g\diffU_h\resU_g^{-1})(\resU_g\resU_h\resU_{gh}^{-1})\diffU_{gh}^{-1}.
\end{equation}
Passing to blend classes, we have
\begin{equation}
\label{eq:2_coboundary}
    \omega'_{g,h} = \beta_g\,\beta_h\,\omega_{g,h}\,\beta^{-1}_{gh},
\end{equation}
where $\beta_g$ is the blend class of $\diffU_g$ and $\omega$ and $\omega'$ are the 2-cochains defined via $\resU_g$ and $\resU_g'$ respectively.
Therefore, $\omega'$ and $\omega$ differ by a coboundary \(\beta_g \beta_h \beta^{-1}_{gh}\), and are in the same $\grpH{2}{G}{\QCAblendclass{d-1}{}}$ cohomology class.

The calculation above also shows that the cohomology class is independent of the choice of half volume \(R = (-\infty,a)\times \RR^{d-1}\). Indeed, considering another choice of \(R' = (-\infty,a')\times \RR^{d-1}\), we can regard restrictions to \(R\) or \(R'\) as being different choices of restriction to \(R \cup R'\) (now with potentially wider boundary strips \(S^{(\prime)}\)). Then, by the previous paragraph, their associated \(\grpH{2}{G}{\QCAblendclass{d-1}{}}\) invariants are the same.

We emphasize that the vanishing of the $\grpH{2}{G}{\QCAblendclass{d-1}{}}$ class is a \emph{necessary}, but not \emph{sufficient} condition for the $G$-rep to be blend-equivalent to the trivial rep. The vanishing of the $\grpH{2}{G}{\QCAblendclass{d-1}{}}$ class ensures that it is possible to choose the restriction such that $\bdyU_{g,h}$ is a QCA in the trivial equivalence class. Being blend equivalent to the trivial rep would imply that it is possible to choose the restriction such that $\bdyU_{g,h}$ is exactly the identity QCA, which is a stronger condition.

The construction of \(\omega\) can be generalized to the case where $G$ is not an internal symmetry, i.e.\  the range of $\repU_g$ has no uniform upper bound that holds for all $g \in G$. In that case one cannot have a fixed boundary strip $S$ of finite width.
One way of working around this is to choose $R = (-\infty,0)\times \RR^{d-1}$ to be a half volume (corresponding say to the subset of \(\RR^d\) with first coordinate $x< 0$), and for each $g\in G$, we require $\resU_g$ to agree with $\repU_g$ for $x<-r_g$ (where $r_g$ depends on $g$) and act trivially for $x>0$.
Then for each $g,h\in G$, we have the one-dimensional QCA $\bdyU_{g,h}=\resU_g\resU_h\resU_{gh}^{-1}$ defined on the strip $-(r_g+r_h+r_{gh})\leq x \leq 0$.
Passing to blend classes defines an \(\omega_{g,h}\) for each pair \((g,h)\).

There is also no obstacle to extending this invariant to antiunitary symmetries.
Again we want to restrict the action of each $\repU_g$ to some region.
However, since the identity QCA is always unitary, an antiunitary QCA cannot be restricted to act only within a region $R$.
Nevertheless, in order to define the index of a boundary QCA, we do not actually need that the restricted QCA $\resU_g$ acts trivially outside $R$, but only that it has an on-site (i.e.\ range zero) action. In other words, we are computing the obstruction to constructing a blend between $\mathcal{U}$ and an on-site rep.
By using such a $\resU_g$, we can have well-defined $\omega_{g,h}\in\QCAblendclass{d}{}$.
On the other hand, the group cohomology formalism does need to change a little bit. Specifically, because in \eqnref{eq:H2rechoose} we have a conjugation by $\resU_g$, we need to take into account the antiunitary action on $Q_{d-1}$. 
As a consequence, \eqnref{eq:H2closedClass} becomes [compare \eqnref{eqn:AntiunitaryOnQCA}]
\begin{equation}
        \omega_{g,h}\,\omega_{gh,k}={}^{\sigma(g)}\omega_{h,k} \omega_{g,hk},
\end{equation}
where $\sigma : G \to \mathbb{Z}_2 = \{ 0, 1 \}$ is the homomorphism keeping track of anti-unitarity, and we have used the action of $\mathbb{Z}_2$ on $\QCAblendclass{d-1}{}$ as in Eq.~\eqref{eqn:AntiunitaryOnQCA}. \eqnref{eq:2_coboundary} is modified similarly, and thus we
obtain an element of the twisted group cohomology $\grpHtwisted{2}{\sigma}{G}{\QCAblendclass{d-1}{}}$.

\subsection{The \texorpdfstring{$\grpH{3}{G}{\QCAblendclass{d-2}{}}$}{H3} invariant}
\label{sec:H3Invariant}

Suppose that the $\grpH{2}{G}{\QCAblendclass{d-1}{}}$ invariant vanishes. Then, assuming that $d \geq 2$, we can proceed further and derive an invariant valued in $\grpH{3}{G}{\QCAblendclass{d-2}{}}$.

The construction of the $\grpH{3}{G}{\QCAblendclass{d-2}{}}$ invariant is analogous to Ref.~\cite{Else2014}. If the $\grpH{2}{G}{\QCAblendclass{d-1}{}}$ invariant vanishes, then it is possible to define the restrictions $\resU_g$ on a region $R$ such that all of the $\mathcal{V}_{g,h}$s defined by \eqnref{eq:restricted_homomorphism}, which act on the boundary strip $S$, have trivial $(d-1)$-dimensional QCA class. This means that they can themselves be restricted to a subregion $R' \subseteq S$, giving restrictions $\widetilde{\mathcal{V}}_{g,h}$. 
However, the QCAs \(\widetilde{\mathcal{V}}_{g,h}\) themselves must satisfy some consistency conditions descending from the fact that \(\repU_g\) is a \(G\)-rep.
The equality of Eqs.~(\ref{eq:ghk_1},~\ref{eq:H2closed}) must hold  up to a QCA $\mathcal{W}_{g,h,k}$ acting within $\partial R' = (S\setminus R')^{+\xi} \cap R'$ [where \(\xi = O(r)\)]. Subsequently taking the $(d-2)$-dimensional blend class of $\mathcal{W}_{g,h,k}$ defines a 3-cochain $\omega_{g,h,k} \in \QCAblendclass{d-2}{}$. This cochain is, in fact, a cocycle, and its cohomology class is independent of the choice of restriction \(\widetilde{\mathcal{V}}_{g,h}\). Thus, in order to be able to choose all \(\mathcal{W}_{g,h,k}\) to be the identity, we must have that this cohomology class is trivial. If it is nontrivial, then \(\repU_g\) is anomalous.

The proof that \(\omega_{g,h,k}\) is a cocycle and defines an element of $\grpH{3}{G}{\QCAblendclass{d-2}{}}$ which is independent of the choice of restriction is tedious but straightforward, since the algebraic steps are exactly the same as for the $\grpH{3}{G}{\mathrm{U}(1)}$-valued invariant in $d=1$ described in Ref.~\cite{Else2014}. (For some details, see Appendix~\ref{appendix:H3_details}.)

We also remark that the original invariant of Ref.~\cite{Else2014} valued in $\grpH{3}{G}{\mathrm{U}(1)}$ for $d=1$ can be viewed as a special case of the $\grpH{3}{G}{\QCAblendclass{d-2}{}}$ invariant if we formally define ``$\QCAblendclass{-1}{} \cong \mathrm{U}(1)$'' (c.f.~\autoref{sec:HomotopyTheory}).

\subsection{Higher invariants}
\label{subsec:HigherInvariants}

By extrapolating the pattern seen above, one might expect that if the $\grpH{3}{G}{\QCAblendclass{d-2}{}}$ invariant is trivial, then one can define an $\grpH{4}{G}{ \QCAblendclass{d-3}{}}$ invariant, and so on. However, actually defining such an invariant in terms of restrictions of \(G\)-reps is not so easy. (In the case of $d=2$, the $\grpH{4}{G}{\mathrm{U}(1)}$ invariant has recently been constructed in Refs.~\cite{Kapustin_2505,Kawagoe_2507}; it might be possible to apply the same methods to construct the $\grpH{4}{G}{\QCAblendclass{d-3}{}}$ invariant for general $d$.)

In \autoref{sec:HomotopyTheory} we will give a more general perspective on the classification of lattice anomalies through a homotopy-theoretic approach.

\section{A lattice anomaly in two dimensions}\label{sec:em-qubit}

In this section, we describe an example of an anomalous \(\ZZ_2\)-rep in two dimensions with a non-trivial $\grpH{2}{\ZZ_2}{\QCAblendclass{1}{}}$ class.
This example is based on the radical chiral Floquet (CF) cycle introduced in Ref.~\cite{Po2017}. In its original context, this was a finite depth unitary circuit which, when restricted to a finite region, was observed to result in chiral edge dynamics which transport a fractional number of degrees of freedom per cycle. The construction is heavily based around a picture of \(\ZZ_2\) topological order~\cite{Wen2007book,Kitaev1997toric,Kitaev2003FaultTolerant,Kitaev2006BeyondAnyons}, and the radical CF circuit also results in an anyon permutation in the toric code model. Since its introduction, the radical CF circuit has been extended in several ways~\cite{Potter2017DynamicallyEnriched,Sullivan2023TwistNetwork,Roberts2023GeometricRadical,Fidkowski2019FloquetFermion}, and it notably underlies the construction of Floquet codes~\cite{Aasen2022,Hastings2021dynamically,Davydova2024DynamicAutomorphism}.

The picture of \(\ZZ_2\) topological order is not necessary to understand the lattice anomaly structure of the radical CF cycle. However, it does provide an example of a Hamiltonian which is symmetric under the cycle. It is useful to keep this example in mind for \autoref{sec:ConsequencesOfAnomalies}, which considers general consequences for a Hamiltonian which is symmetric under an anomalous symmetry.

In \autoref{sec:EMSwapDef} we define the radical CF cycle, using a presentation entirely in terms of spin degrees of freedom. (Reference~\cite{Po2017} maps to an equivalent formulation in terms of Majoranas.)
We then slightly modify the standard radical CF model so that it is, in fact, a \(\ZZ_2\)-rep in \autoref{sec:EMSwapRep}.  Then, in \autoref{sec:EMSwapAnomaly}, we demonstrate that the radical CF model has a nontrivial lattice anomaly, despite the fact that there are no QFT anomalies for a \(\ZZ_2\) global symmetry in two dimensions.

\subsection{The radical chiral Floquet cycle}
\label{sec:EMSwapDef}

\begin{figure*}
    \centering
    \includegraphics{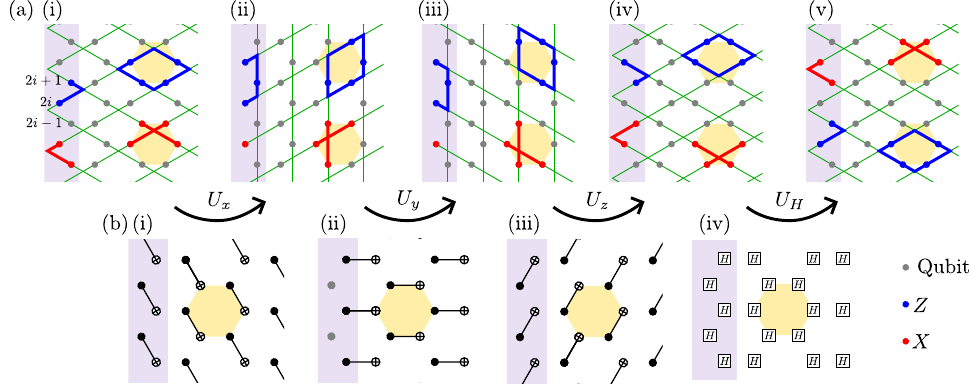}
    \caption{The radical CF circuit~\cite{Po2017} in the qubit picture. Each panel represents the lattice, with the left side being the boundary of the half plane $R$ (with the boundary strip $S$ highlighted in purple). (a)~Operators \(A_h\) (red) and \(B_h\) (blue) are associated to each hexagon, and correspond to the usual toric code plaquette and star stabilizers on a virtual lattice (green). The CNOT gates in (b) result in successive rotations of the star and plaquette operators by \(\pi/3\) around the hexagon. The final Hadamard gate swaps Pauli \(X\) and \(Z\) operators. The action on the boundary operators, which may be viewed as partial star and plaquette operators supported in \(S\), is also shown.
        }
    \label{fig:em-stab}
\end{figure*}

We first describe the radical CF circuit.
The circuit acts on a hexagonal lattice of qubits, with one qubit on each vertex.
The circuit is
\begin{equation}
    U = U_H U_z U_y U_x,
\end{equation}
where \(U_H\) is a tensor product of Hadamard gates on all qubits, and \(U_{x,y,z}\) are edge coverings of the lattice with CNOT gates of different orientations (\autoref{fig:em-stab}).

It is instructive to understand the action of this circuit on the toric code model. For each hexagon in the lattice, define operators \(A_h\) (a product of Pauli \(X\) operators on the upper four qubits) and \(B_h\) (a product of Pauli \(Z\) operators on the lower four qubits), as shown in \autoref{fig:em-stab}(a.i). We can define a virtual square lattice with qubits on edges, also shown in \autoref{fig:em-stab}(a.i), so that it becomes clear that \(A_h\) and \(B_h\) are associated to stars and plaquettes of the virtual lattice, respectively.

The corresponding commuting projector Hamiltonian, 
\begin{equation}
    H_{\mathrm{TC}} = \sum_h \frac{1}{2}(\mathbb{I}-A_h) + \frac{1}{2}(\mathbb{I}-B_h),
\end{equation}
is the well-known toric code model of \(\ZZ_2\) topological order~\cite{Kitaev1997toric,Kitaev2003FaultTolerant,Kitaev2006BeyondAnyons}. It hosts anyonic quasiparticles, \(e\) (corresponding to \(A_h\) excitations) and \(m\) (corresponding to \(B_h\) excitations), which have nontrivial mutual braiding statistics. The composite quasiparticle formed by fusing \(e\) and \(m\) is labeled \(\epsilon\). \(A_h\) and \(B_h\) are called the stabilizers of the model.

The three layers of CNOT gates in \(U\) each rotate the operators \(A_h\) and \(B_h\) around their hexagons by \(\pi/3\) anticlockwise. This is illustrated in \autoref{fig:em-stab}(a.ii-iv). The composition \(U_z U_y U_x\) results in \(A_h\) and \(B_h\) changing places, so that \(U_z U_y U_x[A_h]\) is now positioned in the lower half of the hexagon. Applying a Hadamard gate to every qubit exchanges \(X\) and \(Z\), resulting in
\begin{equation}
    U[A_h] = B_h,\quad U[B_h] = A_h.
\end{equation}
In other words, \(U\) exchanges the $e$ and $m$ excitations for each hexagon.
The Hamiltonian \(H_{\mathrm{TC}}\) is symmetric under \(U\).

\subsection{The \texorpdfstring{\(e\)-\(m\) exchange $\ZZ_2$}{e-m exchange Z2}-rep}
\label{sec:EMSwapRep}

Before proceeding to extract the anomaly of the radical CF circuit by examining a restriction to a half volume, we first slightly modify it so that it becomes a $\ZZ_2$-rep.
Note that $U$ is not a $\ZZ_2$ symmetry by itself.
Even though $U^2$ fixes all bulk toric code operators, it gives each toric code eigenstate a phase factor that depends on the anyon distributions (and so in particular is not the identity).

These phase factors can be intuitively derived by tracking the motion of \(A_h\) and \(B_h\) stabilizers throughout two periods of the evolution. $U^2$ can be decomposed into six steps: $U_x$, $U_y$, $U_z$, $U_HU_xU_H$, $U_HU_yU_H$, $U_HU_zU_H$. Since conjugation by $U_H$ reverses the control and target of CNOTs, the last three evolution steps can be viewed as the continuation of the first three evolution steps. That is, each of the six steps consist of a product of CNOT gates, successively rotated by \(\pi/3\) around the center of a hexagon. All six steps make a full \(2 \pi\) rotation.
Following \(A_h\) and \(B_h\), we see that they rotate around each other, returning to their initial positions after all six steps.
If both \(A_h\) and \(B_h\) are excited, so that the hexagon hosts both an \(e\) and an \(m\) anyon, this results in a braid of these two anyons, and an associated \(-1\) phase.
This shows that for each hexagon, a $-1$ phase is produced if and only if there is an $\epsilon$ anyon on that hexagon.
One can also derive this phase directly (see Appendix~\ref{sec:em-QCA}).

To make the radical CF circuit an actual $\ZZ_2$-rep, we can simply correct this phase by multiplying $U$ by a circuit that produces a phase of $i$ for each hexagon occupied by $\epsilon$ (noting that this circuit commutes with $U$). That is, we define
\begin{subequations}\label{eq:em-rep}
\begin{align}
    \repU_1&=\id \\
    \repU_x&=U\cdot \prod_{h}i^{\frac{1-A_h}{2}\frac{1-B_h}{2}},
\end{align}
\end{subequations}
where we denote \(\ZZ_2 = \{1,x\}\), and the circuit \(\prod_{h}i^{\frac{1-A_h}{2}\frac{1-B_h}{2}}\) acts by conjugation.
We that have $\repU_x^2=\id=\repU_1$ is the identity, and the action of $\repU_x$ on the stabilizers is the same as that of $U$.

We refer to \(\repU: \ZZ_2 \to \QCA_2\) as the \(e\)-\(m\) exchange \(\ZZ_2\)-rep, due to its action on the toric code stabilizers.

The full action of $U$, $U^2$ and $\resU_x^2$ (where \(\resU_x\) is a restriction of \(\repU_x\)) is shown in Appendix~\ref{sec:em-QCA}.

\subsection{Lattice anomaly}
\label{sec:EMSwapAnomaly}

We now examine the restriction of \(\repU\) to a half-volume, and extract its \(\grpH{2}{\ZZ_2}{\QCAblendclass{1}{}}\) invariant as in \autoref{sec:H2Invariant}. We find that this invariant is non-trivial, so that the \(e\)-\(m\) rep is (lattice) anomalous, even though there is no QFT anomaly for a global \(\ZZ_2\) symmetry in two dimensions.

Denote the half plane with \(x>0\) by $R$, and the restriction of \(\repU\) to \(R\) by \(\resU\), for which we include all gates contained completely inside $R$. The boundary strip on which \(\resU\) can differ from \(\repU\) or the trivial \(\ZZ_2\)-rep is highlighted in purple in \autoref{fig:em-stab}, and denoted by \(S\).

The restricted operator $\resU_x^2=\bdyU_{x,x}$ is no longer the identity.
A straightforward calculation shows that $\bdyU_{x,x}$ is nontrivial only on $S$, and the action is (with the labeling of the qubits shown in the upper left of \autoref{fig:em-stab})
\begin{equation}
    \bdyU_{x,x}:
    \begin{aligned}
        Z_{2i}&\mapsto Z_{2i+1}\\
        X_{2i}&\mapsto X_{2i}X_{2i+1}X_{2i+2}\\
        Z_{2i+1}&\mapsto Z_{2i+1}Z_{2i+2}Z_{2i+3}\\
        X_{2i+1}&\mapsto X_{2i+2}.
    \end{aligned}
\end{equation}
This QCA differs from a translation by one site by a finite depth circuit. Specifically, \(\bdyU_{x,x}\) is equal to a translation towards positive \(i\) composed with a circuit of CNOTs on every edge with controls on the even sites and targets on the odd sites.
As a one-dimensional QCA on the boundary strip, it has GNVW index \(\omega_{x,x} = \mathcal{I}(\bdyU_{x,x}) = 2\)~\cite{Gross2012QCA}, and it is straightforward to see \(\omega_{g,h} = \mathcal{I}(\bdyU_{g,h}) = 1\) otherwise (as at least one of \(g\) or \(h\) must be trivial).
This cocycle belongs to a nontrivial cohomology class in \(\grpH{2}{\ZZ_2}{\QCAblendclass{1}{}}\cong \grpH{2}{\ZZ_2}{\QQ_\times}\). Thus, the \(e\)-\(m\) exchange \(\ZZ_2\)-rep has a lattice anomaly.

Indeed, demanding that \(\omega_{g,h} = \beta_g \beta_h/\beta_{gh}\)
implies $\beta_1=1$ and $\beta_x=\sqrt{2}\not\in\QQ_\times$, which is impossible for a QCA with a rational GNVW index. 
This \(\sqrt{2}\) is the origin of the ``radical'' name for the radical CF circuit~\cite{Po2017}.

Once again, it is revealing to consider the action of this restriction on the toric code model.
A key concept in our discussion will be the \emph{boundary algebra}---the algebra of operators on the boundary which commute with all bulk stabilizers \(A_h\) and \(B_h\). Here, we discuss this concept somewhat intuitively, and defer a more general and precise definition to \autoref{sec:ConsequencesOfAnomalies}, where the boundary algebra will continue to play a central role.

The boundary operators---local operators that commute with all the bulk toric code stabilizers---are generated by the \emph{broken stars} and \emph{broken plaquettes} shown in red and blue within the boundary strip in \autoref{fig:em-QCA}. Note that here we are not choosing a gapped boundary for the model, which would involve including a commuting subset of the boundary operators in the Hamiltonian.

Explicitly, the boundary algebra is generated by
\begin{equation}
    Z_{2i} Z_{2i+1}, \quad\text{and}\quad
    X_{2i-1} X_{2i}.
\end{equation}
This is isomorphic to the algebra of symmetric operators for the on-site symmetry \(\overline{X} = \prod_i \sigma^x_i\) in a chain of qubits, where we denote the Pauli operators on this second chain by \(\sigma^\alpha_i\) to distinguish them from the boundary algebra operators. The identification is
\begin{equation}
    Z_{2i} Z_{2i+1} \leftrightarrow \sigma^z_{i}\sigma^z_{i+1}, \quad\text{and}\quad
    X_{2i-1} X_{2i} \leftrightarrow \sigma^x_i.
\end{equation}

As \(\resU_x\) maps the bulk toric code model back to itself, it also must map boundary operators to boundary operators. Indeed, we find (\autoref{fig:em-QCA})
\begin{equation}
    \resU_{x}:
    \begin{aligned}
        Z_{2i} Z_{2i+1}&\mapsto X_{2i+1}X_{2i+2},\\
        X_{2i-1}X_{2i}&\mapsto Z_{2i} Z_{2i+1}.
    \end{aligned}
\end{equation}
This corresponds to the Kramers-Wannier (KW) duality in the \(\sigma^\alpha\) variables,
\begin{equation}
    \resU_{x}:
    \begin{aligned}
        \sigma^z_{i}\sigma^z_{i+1}&\mapsto \sigma^x_{i+1},\\
        \sigma^x_i&\mapsto \sigma^z_{i}\sigma^z_{i+1}.
    \end{aligned}
\end{equation}
Indeed, this KW duality squares to a translation by one site, or to \(\bdyU_{x,x}\) in the original variables. While there is no QCA on the \emph{boundary strip} \(S\) which squares to \(\bdyU_{x,x}\), this KW duality on the \emph{boundary algebra} does square to \(\bdyU_{x,x}\).

It is not coincidence that the toric code is symmetric under \(\repU\), and that the boundary algebra of the toric code admits an automorphism which squares to \(\bdyU_{x,x}\). In \autoref{sec:ConsequencesOfAnomalies}, we will show that a commuting Hamiltonian can only be symmetric under an anomalous symmetry if its boundary algebra is compatible with the anomaly.
In particular, for commuting models in the same MBL phase as the toric code, we will conjecture that being compatible with the \(\ZZ_2\) anomaly of this section is equivalent to having $e$-$m$ exchange in the bulk.

\section{Consequences of lattice anomalies}
\label{sec:ConsequencesOfAnomalies}

In this section, we will address the consequences implied by a lattice Hamiltonian being symmetric under an anomalous $G$-rep.

As we have discussed, in the cases where the anomaly maps to a non-trivial QFT anomaly via the map \eqnref{eq:anomaly_homomorphism}, we have the implication that, if the low-energy physics is governed by a QFT, that QFT must carry the corresponding QFT anomaly. The presence of a nontrivial QFT anomaly has many known consequences~\cite{TongLecturesGauge}.

The situation for an IR-trivial lattice anomaly, which admits a trivially gapped ground state, is less familiar. Nevertheless, we will show that non-trivial lattice anomalies do have sharply defined consequences if one specifically considers Hamiltonians which are the sum of local terms that all commute with each other (such as commuting projector models, for example). Such Hamiltonians have been widely studied in the literature due to their exact solvability. Another practical reason to study this class of Hamiltonians is that any Hamiltonian exhibiting many-body localization (MBL)~\cite{Sierant2024review} is in this class. Thus, the results that we establish will provide strong constraints on localizability in the presence of an anomalous symmetry.

The key concept underlying our results is the \emph{boundary algebra} associated with a boundary truncation of a commuting model. Our definition of the boundary algebra is similar in spirit to that of Ref.~\cite{Jones_2307}, although we will not guarantee that our definitions precisely coincide with theirs in all cases.
It is also closely related to the concept of boundary algebra of a QCA introduced in Refs.~\cite{Freedman2020_2DQCA,Haah2023}.

We will give the precise definition of boundary algebra below. Here we will just mention that given a Hamiltonian that is the sum of commuting terms
\begin{equation}
H = \sum_{i \in I} h_i,
\end{equation}
the boundary algebra is defined with respect to a truncation
\begin{equation}
\widetilde{H} = \sum_{i \in \widetilde{I}} h_i,
\end{equation}
where $\widetilde{I} \subseteq I$. Moreover, there
is a notion of \emph{invertibility} of the boundary algebra.
One of our main results is:
\begin{quote}
    Consider a $G$-rep, and suppose that there is a commuting model with a symmetric boundary truncation  (i.e.\ $\widetilde{H}$ is invariant under the symmetry) such that the boundary algebra is invertible. Then the $G$-rep is blend equivalent to an on-site $G$-rep. That is, the lattice anomaly of the $G$-rep is trivial.
\end{quote}

\newcommand{\TrivCommutingModel}{trivial}
\newcommand{\ATrivCommutingModel}{a trivial}
A sufficient  condition for the boundary algebra to be invertible is that the commuting model is 
\emph{\TrivCommutingModel}, i.e.\ all the local commuting terms that appear in the Hamiltonian act on single sites, such that there is a unique ground state on each site. In fact, it is also sufficient for the commuting model to be invertible: that is, one can stack the commuting model with another one such that the combined commuting model can by transformed by a finite-depth circuit into  \ATrivCommutingModel{} commuting model. One can show (see Appendix~\ref{appendix:InvertibleIffQCAPreparable}) that this is equivalent to saying that there exists a QCA (as opposed to a finite-depth circuit) that transforms the original commuting model to \ATrivCommutingModel{} commuting model.

Thus, symmetric invertible commuting models are forbidden by a non-trivial lattice anomaly. It is pleasing to note the parallel with the result for QFT anomalies, in which case the expected result is that the anomaly is incompatible with symmetric invertible \emph{states}. However, the result for IR-trivial lattice anomalies is weaker---for example, non-invertible commuting models may not necessarily have non-invertible ground states.

These results also have important implications for MBL phases. Recall that an MBL system is characterized by a set of (quasi-)local integrals of motion (LIOMs) which all commute with each other~\cite{Huse2014fullyMBL}. These LIOMs should satisfy a completeness property, which in particular implies that the simultaneous eigenstates of all the LIOMs are non-degenerate. We say that the MBL system is symmetric (which also excludes the case of spontaneous symmetry breaking) if the LIOMs are invariant under the symmetry. Thus, we can apply the theorem described above, and we conclude that if the symmetry has a non-trivial lattice anomaly, then it is impossible to have a symmetric MBL phase that satisfies a condition known as \emph{short-range entangled (SRE)}~\cite{Long2024} (which includes, in particular, any MBL phase that is in the trivial MBL phase if we disregard the symmetry).

Now, consider a $G$-rep with a non-trivial lattice anomaly. The above results show this symmetry is incompatible with a class of commuting models / MBL phases. So which commuting models / MBL phases \emph{are} compatible with the symmetry? In particular, the boundary algebra in this case must be non-invertible. Later in this section, we will take some steps towards answering this question, with a specific focus on the case of the $\grpH{2}{G}{\QCAblendclass{d-1}{}}$ anomaly described in \autoref{sec:H2Invariant}. In particular, we will identify a necessary condition on the boundary algebra in terms of the $\grpH{2}{G}{\QCAblendclass{d-1}{}}$ class.

\subsection{Local algebras}
\label{subsec:local_algebras}
\newcommand{\tl}{\triangleleft}
The boundary algebra will be an example of a general structure we refer to as a \emph{local algebra}. The idea of a local algebras is that they generalize the algebras of local operators we have been discussing until now, while retaining the concept that elements of the algebras are supported on subsets of the lattice sites.

First we recall the concept of a $*$-algebra~\cite{Murphy1990CstarBook}. A $*$-algebra is a vector space over $\mathbb{C}$ equipped with abstract multiplication (not necessarily commutative) and adjoint operations satisfying natural compatibility conditions. It must contain a multiplicative identity, which we refer to as the \emph{unit} and denote by $\unit$.

Then,
we define a \emph{(bosonic) local algebra} over a set $X$ equipped with a distance function $d(\cdot,\cdot)$ (i.e.\ a metric) to be a $*$-algebra $A$ equipped with a relation ``$a \tl S$'', where $a \in A, S \subseteq X$ (in words $a \tl S$ means ``$a$ is supported on $S$'') which satisfies the following properties:
\begin{enumerate}
    \item If $a \tl S$ and $S \subseteq S'$ then $a \tl S'$.
    \item If $a \tl S$ and $b \tl S$ then $a + b \tl S$ and $ab \tl S$.
    \item If $a \tl S$, then $a^{\dagger} \tl S$.
    \item $a$ satisfies $a \tl \emptyset$ if and only if $a = \alpha \unit$ for some $\alpha \in \mathbb{C}$, where $\emptyset$ is the empty set.
    \item $a \tl X$ for all $a \in A$.
    \item For any $a \in A$, and any $S \subseteq X$ with $a \tl S$, there exists a bounded set $S' \subseteq S$ with $a \tl S'$.
    \item \label{item:finiteness_condition} \label{item:last_before_commutator} (``Locally finite-dimensional'') For any bounded set $S \subseteq X$, the subalgebra $A\{ S \}$ is finite-dimensional (as a vector space).
    \item\label{item:commutator_condition}If $a \tl S$ and $b \tl S'$ with $S$ and $S'$ disjoint, then $[a,b] = 0$.
\end{enumerate}
Here we have defined the subalgebra $A\{ S \} := \{ a \in A : a \tl S \}$. We will sometimes use the shorthand notation \(A_S\) in place of \(A\{S\}\).

 Note these conditions are roughly equivalent to the concept of a ``local net of algebras'', as defined in Ref.~\cite{Jones_2307} for example~\footnote{In particular, any local algebra $A$ defines a local net of algebras in the sense of Ref.~\cite{Jones_2307} by assigning $A\{ S \}$ to each bounded set $S$, except that in our definitions, we do not assume a norm on $A$, so in particular we cannot say it is a $C^*$-algebra (although we do indeed expect that in all cases of physical interest, $A$ can be endowed with a norm such that the norm closure of $A$ is a $C^*$-algebra). Conversely, given any local net of algebras in a $C^*$-algebra $A$, we can define a local algebra $A' = \cup_{S \subseteq X \mathrm{bounded}} A\{ S \}$ with $a \tl S \Leftrightarrow$ ($\exists$ bounded $S' \subseteq S$ with $a \in A\{ S' \}$), except that the axioms of Ref.~\cite{Jones_2307} do not necessarily imply the locally finite-dimensional condition.}.
Condition \ref{item:finiteness_condition} encodes the fact that we are working with lattice systems as opposed to, say, quantum field theories.

We call an isomorphism $\varphi : A \to B$ between local algebras on the same set $X$ an \emph{ultra-local isomorphism} if for all $a \in A$, $S \subseteq X$, we have $a \tl S \Leftrightarrow \varphi(a) \tl S$. We also define $\varphi$ to be an \emph{local isomorphism} if there exists some finite $r$ such that for all $a \in A$, $S \subseteq X$ we have that $a
\tl S \implies (\varphi(a) \tl S^{+r} \mbox{ and } \varphi^{-1}(a) \tl S^{+r})$.

If $A$ and $B$ are two local algebras on $X$, we can define the tensor product $A \otimes B$, which we also describe as stacking. An element $c \in A \otimes B$ satisfies $c \tl S$ if and only if $c$ can be written as a finite sum
\begin{equation}
c = \sum_{j=1}^n a_j \otimes b_j,
\end{equation}
with $a_j \tl S$ and $b_j \tl S$.

We can also define the concept of a fermionic local algebra which will apply to fermionic systems (see Appendix~\ref{appendix:fermionic_local_algebra} for the details).
Below, when we refer to a ``local algebra'' without specifying bosonic of fermionic, the statements we make will apply to either one.

Finally, we define the concept of an \emph{on-site algebra}.
First, we say that a $*$-algebra $A$ is a full bosonic matrix algebra if there exists a finite-dimensional Hilbert space  $V$  such that $A$ is isomorphic to the algebra of linear operators on $V$. (See Appendix~\ref{appendix:fermionic_local_algebra} for the definition of a full fermionic matrix algebra.)
Then, we say that a (bosonic/fermionic) local algebra $A$ over a set $X$ is an on-site algebra if there exists $\Lambda \subseteq X$ such that $A$ is ultra-locally isomorphic to a (bosonic/fermionic) tensor product
\begin{equation}
    A \cong \bigotimes_{x \in \Lambda} A_x,
\end{equation}
where $A_x$ is supported on $\{ x \}$ and is a full (bosonic/fermionic) matrix algebra. In order for the locally finite-dimensional condition to be satisfied, we should take the set $\Lambda$ to be locally finite, i.e.\ all bounded subsets of $\Lambda$ are finite. Thus, on-site local algebras correspond to the algebras of local operators that we have previously been considering (\autoref{sec:what_is_a_symmetry}). Moreover, QCAs correspond to local automorphisms of on-site algebras.

\subsection{Commuting models and boundary algebras}
\label{sec:CommutingModelsBoundaryAlgebra}

We consider Hamiltonians of the form
\begin{equation}
\label{eq:commuting_hamiltonian}
    H = \sum_{i \in I} h_i,
\end{equation}
which are built from local Hermitian operators $h_i$ in an on-site algebra over $X$ (for example we could take $X = \mathbb{R}^d$), where $i$ ranges over a countable index set $I$, and $h_i$ and $h_j$ commute for any $i,j \in I$. (In the fermionic case, we require that the $h_i$'s be fermion-parity-even.) Technically, in an infinite system we should view $H$ as a formal linear sum.
Without loss of generality, we can assume that the $h_i$'s are positive (i.e.\ their eigenvalues are non-negative). We will require the model to be frustration-free: that is, there is at least one state which is a zero eigenstate of $h_i$ for all $i \in I$. Finally, we require that there exists a finite $r$ such that $h_i$ is supported on a set $X_i$ of diameter at most $r$ for all $i \in I$. If $\{ h_i \}$ satisfies all of the conditions described above, we will refer to it as a \emph{commuting model}.

To define boundary algebra, we first introduce a precursor notion called \emph{projected algebra}. A boundary algebra will be a projected algebra satisfying a certain condition. Let $A$ be an on-site algebra. The projected algebra $A(H)$ is formally defined for any Hamiltonian $H$ that is the sum of positive commuting terms and is frustration-free.
To give intuition about the construction, let us first restrict to the bosonic case and
suppose that $A$ is the algebra of operators acting on a finite-dimensional Hilbert space $\mathcal{H}$.
Now consider the subspace $\mathcal{H}(H) \leq \mathcal{H}$ comprising those states which are annihilated by all the local terms of $H$, which is just the kernel of \(H\).
[The frustration-free condition ensures that there is at least one such state, so that $\mathcal{H}(H)$ is at least one-dimensional.]
We will refer to states in $\mathcal{H}(H)$ as \emph{projected states}.
Then the projected algebra $A(H)$ is defined to be the algebra of operators on $\mathcal{H}(H)$.

We can give an equivalent construction of the projected algebra as follows. Define $\mathcal{A}(H)$ to be the algebra of operators $a$ on $\mathcal{H}$ which are block diagonal in $\mathcal{H}(H)$ and its orthogonal complement in $\mathcal{H}$. Equivalently, $\mathcal{A}(H)$ is the algebra of operators $a$ on $\mathcal{H}$ such that $a\ket{\psi}, a^{\dagger} \ket{\psi} \in \mathcal{H}(H)$ for all $\ket{\psi} \in \mathcal{H}(H)$. Next, we define $\mathcal{B}(H) \leq \mathcal{A}(H)$ to comprise those operators $b$ which annihilate the projected states, that is, $ b \ket{\psi} = b^{\dagger} \ket{\psi} = 0$ for all $\ket{\psi} \in \mathcal{H}(H)$. Then one can show that $A(H) \cong \mathcal{A}(H) / \mathcal{B}(H)$, where we view $\mathcal{B}(H)$ as a two-sided ideal of $\mathcal{A}(H)$ and take the ring quotient. This construction also allows $A(H)$ to inherit the structure of a local algebra from $A$ if we declare that $a \in A(H)$ is supported on a set $S \subseteq X$ if and only if there exists a representative $a' \in \mathcal{A}(H)$ such that $a'$ is supported on $S$. 
This version of the definition can also easily be generalized to
infinite systems and fermionic systems, as well as to more general local algebras than on-site algebras
(Appendices~\ref{appendix:infinite_system}-\ref{appendix:AlternativeProjectedAlgebra}).

Now, consider a commuting Hamiltonian $H = \sum_{i \in I} h_i$, and choose a restriction
\begin{equation}
\widetilde{H} = \sum_{i \in \widetilde{I}} h_i
\end{equation}
where $\widetilde{I} \subseteq I$, and a set $W$ such that $h_i$ is supported on $W$ for all $i \in I$. Recall that $A_W$ is the subalgebra of $A$ comprising all the operators supported on $W$.
We say that the projected algebra $A_W(\widetilde{H})$ is a \emph{boundary algebra} if there exists a finite $r$ such that all elements of $A_W(\widetilde{H})$ [which recall, are by definition supported on $W$] are supported on $W \cap (W^c)^{+r}$.
We observe that the minimal choice of $W$ is to set $W = X_{\widetilde{H}}$, where $X_{\widetilde{H}}$ is the union of the supports of all \(h_i\) with \(i \in \widetilde{I}\).
For any larger $W$, $A_W(\widetilde{H})$ only differs from $A_{X_{\widetilde{H}}}(\widetilde{H})$ by stacking with the on-site algebra $A_{W \setminus X_{\widetilde{H}}}$.
See \autoref{fig:bdy-alg} for a graphical presentation of the boundary algebra of a simple model.

 In any ``sufficiently nice'' commuting model $\{ h_i \}$ in an on-site algebra over $\mathbb{R}^d$, we expect that after considering all possible choices of restrictions, there should be exactly one boundary algebra (up to some concept of equivalence) for ``nice'' subsets $W \subseteq \mathbb{R}^d$ (for example, if $W$ is a half-volume) that does not depend on the choice of $W$. However, we will not try to identify the precise conditions under which this occurs (but see Ref.~\cite{Jones_2307}). In any case, the general results we prove below do not depend on the uniqueness of the boundary algebra. 

We say that a boundary algebra $A_W(\widetilde{H})$ is \emph{invertible} if there exists a local algebra $B$  supported on $W \cap (W^c)^{+r'}$ for some finite $r'$, such that $A_W(\widetilde{H}) \otimes B$ is locally isomorphic to an on-site algebra.  For example, consider the case where the index set $I$ is just the set of all lattice sites $\Lambda$, and each $h_i$ is supported only on the site $i$ and has a non-degenerate ground state on that site. (We refer to this in the introduction of \autoref{sec:ConsequencesOfAnomalies} as a \TrivCommutingModel~commuting model.) 
Then for any set $W \subseteq \mathbb{R}^d$, we can define 
\begin{equation}
\widetilde{H}_W = \sum_{i \in \Lambda \cap W} h_i
\end{equation}
Then we see that indeed $A_W(\widetilde{H}_W)$ is the trivial algebra, hence it is on-site, hence it is invertible.

Finally, let us note that there is a close connection between the boundary algebra of a commuting model as defined here, and the concept of the boundary algebra of a QCA introduced in Refs.~\cite{Freedman2020_2DQCA,Haah2023} (which is always invertible). Specifically, the connection is that if two commuting models $\{ h_i \}$ and $\{ h_i' \}$ are related through the action of a QCA $U$, i.e. $h_i' = U[h_i]$, then the corresponding boundary algebras differ (after modding out by tensoring with on-site algebras) by tensoring with the boundary algebra of the QCA, as we show in \autoref{sec:BoundaryGRep}.
In particular, consider a commuting model that is invertible (recall from the introduction of \autoref{sec:ConsequencesOfAnomalies} that this is equivalent to preparability from \ATrivCommutingModel{} commuting model by a QCA). Then its boundary algebra is invertible and equivalent to the boundary algebra of the QCA that relates it to \ATrivCommutingModel{} commuting model. It is also interesting to ask a converse question: if a commuting model has invertible boundary algebra, is the commuting model necessarily invertible? We leave the answer for future work.

\begin{figure}
    \centering
    \includegraphics{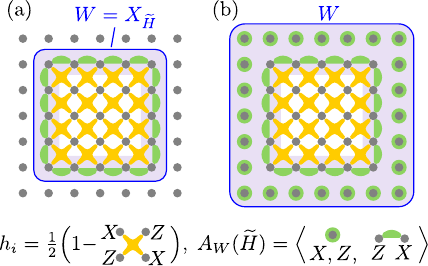}
    \caption{
    A simple example of a boundary algebra which is not on-site is provided by the \emph{Wen plaquette model}---a qubit commuting projector model with terms \(h_i\) as shown~\cite{Wen2003Plaquette}. Here $I$ is the set of plaquettes and $\widetilde{I}$ the 16 plaquettes near the center. The purple strip indicates $W\cap (W^c)^{+r}$. (a) The minimal choice of $W=X_{\widetilde{H}}$ results in a boundary algebra generated by representative Pauli operators \(ZX\) on adjacent qubits. (b) A non-minimal choice of $W$ results in a boundary algebra related to (a) by adding single-site ancillas.
    \label{fig:bdy-alg}}
\end{figure}

\subsection{Constraints from anomalous symmetries on symmetric commuting models}
\label{subsec:symmetric_constraints}

Given a group $G$, then we say that a commuting model $\{ h_i \}$ is \emph{strongly symmetric} under a $G$-rep $\repU$ if each $h_i$ is symmetric, i.e. $\repU_g[h_i] = h_i$ for all \(g \in G\). 
We say that a commuting model is \emph{weakly symmetric} if the Hamiltonian
\begin{equation}
    H = \sum_{i \in I} h_i
\end{equation}
is symmetric.

We say that a commuting model admits a \emph{symmetric boundary algebra} on a set $W \subseteq \Lambda$ if there exists a subset $\widetilde{I} \subseteq I$ such that the restricted Hamiltonian
\begin{equation}
    \widetilde{H} = \sum_{i \in \widetilde{I}} h_i
\end{equation}
is symmetric,
and $A_W(\widetilde{H})$ satisfies the condition to be a boundary algebra. Note that for strongly symmetric models, any choice of restricted Hamiltonian will be symmetric. For weakly symmetric models, it may be possible to find a restricted Hamiltonian that is symmetric (for example if the symmetry just permutes the $h_i$'s over finite orbits), or it may be impossible (for example, in the case of commuting models that are weakly symmetric under translation symmetry). 

It turns out that an anomalous symmetry implies powerful constraints on symmetric boundary algebras. More precisely:

\begin{thm}\label{thm:ObstructionToInvertible}
Consider a $G$-rep $\mathcal{U}$ on an on-site algebra $A$ over $X$, and suppose that there exists an invertible symmetric boundary algebra $A_W(\widetilde{H})$ for some set $W \subseteq X$. Then there exists a $G$-rep $\widetilde{\mathcal{U}}$ acting on an on-site algebra supported on $(W^c)^{+r}$ (for some finite \(r\)) which agrees with $\mathcal{U}$ in $W^c$.
\begin{proof}
The local isomorphism appearing in the definition of invertibility can be extended to a local isomorphism $A(\widetilde{H}) \otimes B \cong M$, where $M$ is an on-site algebra supported on $(W^c)^{+r}$. 
$\widetilde{H}$ being symmetric implies that $\mathcal{U}_g$ defines a local action of $G$ on $A(\widetilde{H})$. 
Tensoring with any \(G\)-rep on  $B$ (for example, the trivial action) and carrying the action through the isomorphism we get an action on $M$. 
This action agrees with $\mathcal{U}_g$ on $W^c$.
\end{proof}
\end{thm}

In particular, this implies that if $X = \mathbb{R}^d$ and we have a symmetric invertible boundary algebra for $W$ a half-volume, then it follows that $\mathcal{U}$ is blend-equivalent to the identity-rep, and so has trivial anomaly.

Note that the provided proof requires for \(A\) to be bosonic and for the \(G\)-rep to be unitary. However, we remark that this result also holds for $G$-reps containing antiunitary elements,  (in which case one finds that $\mathcal{U}$ is blend-equivalent to an on-site rep rather than the identity-rep). The only significant change is that the existence of a $G$-action on $B$ is a somewhat more non-trivial statement. Let us focus on the case that $W$ is a half-volume $(0,\infty) \times \mathbb{R}^{d-1}$. Then it is clear that the spatial reflection of $A_W(H)$ (about some coordinate other than the first one), which we call $\overline{A_W(\widetilde{H})},$ is a blend inverse for $A_W(\widetilde{H})$. Therefore (up to tensoring with on-site algebras, which always admit a $G$-rep by acting with complex conjugation for the antiunitary elements of $G$) we are free to take $B = \overline{A_W(\widetilde{H})}$ as the inverse algebra of $A_W(\widetilde{H})$. Then $B$ clearly admits an action of $G$ simply by taking the spatial reflection of the action on $A_W(\widetilde{H})$.

By using fermionic local algebras, the result also holds for fermionic systems. In the case of non-trivial extension by fermion parity one needs to use the same argument as in the antiunitary case above to obtain a fermionic \(G_f\)-rep on \(B\). In this argument, we also need to be able to construct fermionic $G_f$-reps on on-site algebras. In particular, as long as the fermionic symmetry group $G_f$ is a compact Lie group (which includes finite groups as a special case), we have the following fact (see Appendix \ref{appendix:fermionic_symmetries}), which is sufficient for the proof to work: for any fermionic on-site algebra $A$, there exists a fermionic on-site algebra $B$ with the same support, such that there is a fermionic representation of $G_f$ on $A \otimes B$. (Recall the definition of $G_f$-reps for fermionic symmetries in Section \ref{subsec:anti_unitary_etc}.)

Finally, we note that since we have not given a precise definition of the continuity condition that a $G$-rep should satisfy for continuous groups $G$, the theorem cannot be regarded as having been rigorously proven in the case of continuous groups; however, we expect that it should hold given suitable definitions.

\subsection{Constraints from the \texorpdfstring{\(\grpH{2}{G}{\QCAblendclass{d-1}{}}\)}{H2} anomaly}
\label{subsec:H2_constraints}

The above results show that non-trivial lattice anomalies act as an obstruction to invertible commuting models. In this section we will examine what properties a \emph{non-invertible} commuting model should have in order to be compatible with a non-trivial lattice anomaly.

Our main tool here will continue to be the boundary algebra, and more specifically its group of \emph{local automorphisms}. Recall the definition of local isomorphism from \autoref{subsec:local_algebras}. A local automorphism is a local isomorphism from a local algebra to itself. Thus, QCAs are local automorphisms of on-site algebras, but for non-on-site algebras we also have a group of local automorphisms. (Other authors have also called local automorphisms QCAs even for non-on-site algebras~\cite{Ma2024}.)

The key idea is as follows. Consider a bosonic $G$-rep $\repU$; in this section we will assume that the $\grpH{1}{G}{\QCAblendclass{d}{}}$ anomaly vanishes. Now consider a commuting Hamiltonian $H$ which admits a restriction $\widetilde{H}$ that is symmetric under $\repU$, and such that $A_{W_0}(\widetilde{H})$ is a boundary algebra for some set $W_0$ (and hence $A_{W_0}(\widetilde{H})$ is a symmetric boundary algebra). 
Now choose a set $W \supseteq W_0$ such that there is a restriction $\widetilde{\mathcal{U}}_g$ that acts as the identity on operators supported on $W^c$ and like $\mathcal{U}_g$ on operators supported on $W_0$. We also require that $\mathcal{V}_{g_1,g_2} := \resU_{g_1} \resU_{g_2} \resU_{g_1 g_2}^{-1}$ is supported on $W_0^c$. (For non-internal symmetries one must take $W$ to depend on $g$; for simplicity we will not consider this case, but the arguments below carry through similarly.)
Then $\widetilde{H}$ is also symmetric under $\resU_g$, and therefore $\resU_g$ induces a local automorphism of the boundary algebra $A_W(\widetilde{H})$.

We introduce the group $\mathcal{L}$ of \emph{stable} blend equivalence classes of local automorphisms, in which one is allowed to tensor with arbitrary on-site algebras with the same support as $A_W(\widetilde{H})$ (see Appendix~\ref{appendix:stable_laut} for the precise definition). In particular this ensures that $\mathcal{L}$ does not depend on the choice of $W$.  Note that $\mathcal{L}$ is a group, but unlike the group of blend equivalence classes of QCAs, in general it need not be an \emph{Abelian} group. Moreover, there is a group homomorphism $i : Q_{d-1} \to \mathcal{L}$ (by considering local automorphisms that act only on the on-site part). In fact, one can show that the image of $i$ is contained in the center of $\mathcal{L}$.

Thus, by considering the action of the restricted symmetry $\resU_g$ on the boundary algebra, we obtain a map $\beta : G \to \mathcal{L}$. This map need not be a homomorphism, and depends on the arbitrary choice of restriction, but the induced map $\widetilde{\beta} : G \to \mathcal{L}/i(Q_{d-1})$ is a homomorphism and does not depend on the choice of restriction.

We focus on constraints arising from the \(\grpH{2}{G}{\QCAblendclass{d-1}{}}\) invariant of \(\repU\). We find that there is a constraint that relates this to $\widetilde{\beta}$. Indeed, suppose $\repV_{g_1,g_2}$ has \((d-1)\)-dimensional QCA class $\omega(g_1, g_2) \in Q_{d-1}$. Then it must be the case that
\begin{equation}\label{eq:cocycle_induction}
   i (\omega(g_1, g_2)) = \beta(g_1) \beta(g_2) \beta(g_1 g_2)^{-1}.
\end{equation}
This can be regarded as a relation between the cohomology class $[\omega] \in\grpH{2}{G}{\QCAblendclass{d-1}{}}$ and $\widetilde{\beta}$ which is independent of the choice of cocycle representative for $[\omega]$ and of the choice of lift $\beta$ for $\widetilde{\beta}$.
We will refer to this relation as the \emph{boundary trivialization condition}.

Equation~\eqref{eq:cocycle_induction} reveals features of the boundary algebra, and hence of the commuting model.
First, a nontrivial \([\omega]\) implies that the boundary algebra cannot be locally isomorphic to an on-site algebra. For an on-site algebra we have $\mathcal{L} \cong \QCAblendclass{d-1}{}$, and no non-trivial anomaly classes in $\grpH{2}{G}{\QCAblendclass{d-1}{}}$ can satisfy the boundary trivialization condition. This is also, of course, a corollary of Theorem~\ref{thm:ObstructionToInvertible}.

Second, if $i$ is injective  (which will be the case in the examples we consider below), then \(i(\QCAblendclass{d-1}{}) \cong \QCAblendclass{d-1}{}\), and \eqnref{eq:cocycle_induction} defines a map from homomorphisms \(\widetilde{\beta}: G \to \mathcal{L}/i(\QCAblendclass{d-1}{})\) into $\grpH{2}{G}{\QCAblendclass{d-1}{}}$. 
In particular, if $\widetilde{\beta}$ is trivial then it implies that the $\grpH{2}{G}{\QCAblendclass{d-1}{}}$ anomaly is trivial. Thus, we have a more refined conclusion than the statement that the boundary algebra is not an on-site algebra: if the $\grpH{2}{G}{\QCAblendclass{d-1}{}}$ anomaly is non-trivial, then the quotient \(\mathcal{L}/i(\QCAblendclass{d-1}{})\) must support nontrivial homomorphisms from \(G\).

In order to make these considerations more concrete, let us consider the case $d=2$
and suppose the boundary algebra $A_W(\widetilde{H})$ is locally isomorphic (up to tensoring with on-site algebras) to the subalgebra of $\mathbb{Z}_2$ symmetric operators in an on-site algebra with a 2-dimensional Hilbert space on each site, acted upon by a faithful on-site $\mathbb{Z}_2$-rep. We call this local algebra $A^{\ZZ_2}$.
It follows from Refs.~\cite{Jones_2309,Ma2024} that there is a homomorphism $\mathcal{I} : \mathcal{L} \to \mathcal{L}'$, where $\mathcal{L}' \leq \mathbb{R}_\times$ (it is likely that this is in fact an \emph{isomorphism}, i.e.\ the index fully classifies $\mathcal{L}$, but this has not technically been proven because Ref.~\cite{Ma2024} did not consider stable equivalence classes.) $\mathcal{I}$ is referred to as the ``index'' and agrees with the GNVW index in the sense that if we view $Q_{d-1}$ as $\mathbb{Q}_\times$, then $i' = \mathcal{I} \circ i$ (which is a homomorphism from $Q_{d-1} \to \mathcal{L}'$) corresponds to the inclusion $\mathbb{Q}_\times \leq \mathbb{R}_\times$.
Therefore, a necessary condition for a symmetric boundary algebra with symmetry $G$ is that \eqnref{eq:cocycle_induction} is satisfied, but where we replace $i$ with $i' = \mathcal{I} \circ i$. [If $\mathcal{I}$ is in fact an isomorphism, then this is equivalent to \eqnref{eq:cocycle_induction}. Either way, we at least obtain a \emph{necessary} condition.]

Now suppose that we set $G = \mathbb{Z}_2 = \{ 1, x \}$, and consider a $G$-rep in two dimensions with the property that upon restriction, it squares to a translation on the boundary with GNVW index 2. This corresponds to a non-trivial element of $\grpH{2}{\mathbb{Z}_2}{\QCAblendclass{1}{}}$. Now suppose that $\beta(x)$ is a local automorphism with index $\sqrt{2}$. Then one can show that the boundary trivialization condition (in the version with $i'$) is indeed satisfied. This is precisely what happens in the Hamiltonian discussed in \autoref{sec:em-qubit}, which is symmetric under an anomalous $\mathbb{Z}_2$-rep. It has a boundary algebra that is locally isomorphic to $A^{\ZZ_2}$ and thus supports a ``Kramers-Wannier'' type automorphism, which has index $\sqrt{2}$ and is the ``square-root'' of a translation with index 2. If the boundary algebra instead were an on-site algebra, then only rational indices would be permitted and there would be no way to satisfy the boundary trivialization condition.

More generally, we say that a commuting model whose ground state is in the same phase as the toric code is a \emph{simple realization} of the toric code phase if the following properties hold (in particular, they hold for the toric code itself, viewed as a commuting model):
\begin{itemize}
\item The commuting model admits a boundary algebra which is locally isomorphic (up to tensoring with on-site algebras) to $A^{\ZZ_2}$.
\item The Kramers-Wannier duality automorphisms of the boundary algebra (i.e.\ those with irrational index)  have the property that when considering the induced actions on boundary states~\footnote{In a finite system, by ``boundary states'' we mean the projected states defined in \autoref{sec:CommutingModelsBoundaryAlgebra}, in the case where the projected algebra is a boundary algebra. One can show (also in infinite systems, using the methods of Appendix~\ref{appendix:infinite_system}) that an automorphism of the boundary algebra induces an automorphism of boundary states.}, any one of them turns a state that corresponds to an $e$-condensed boundary for the bulk ground state into an $m$-condensed boundary, and vice versa.
\end{itemize}
(We do not know of any cases where the first condition is satisfied but the second is not, but we do not have a proof.)

Suppose we have a commuting model giving a simple realization of the toric code, and suppose that the boundary algebra postulated in the definition of simple realization is symmetric under a $\mathbb{Z}_2$-rep $\repU$. Since $i' : \QCAblendclass{1}{} \to \mathcal{L}'$ is injective, we obtain a map from homomorphisms $\widetilde{\beta} : \mathcal{L}'/i(\QCAblendclass{d-1}{})$ 
into $\grpH{2}{\mathbb{Z}_2}{\mathcal{L}'}$ as described earlier. In this case the map is easy to construct, as follows.
Let $\gamma$ be the index of $\beta(x)$ (and note that for different choices of lift $\beta$ for $\widetilde{\beta}$, the resulting $\gamma$'s only differ by multiplication by rational numbers). The fact that $x^2 = 1$ implies that $\gamma^2$ must be a rational number. Thus, $\gamma^2$ determines an element of $\mathbb{Q}_\times / \mathbb{Q}_\times^2 \cong \grpH{2}{\mathbb{Z}_2}{\mathbb{Q}_\times}$ which only depends on $\widetilde{\beta}$.

Therefore, \(\repU\) has a nontrivial \(\grpH{2}{\ZZ_2}{\QCAblendclass{1}{}}\) class if and only if $\gamma$ is irrational, i.e.\ the symmetry induces Kramers-Wannier duality on the boundary. If the $\grpH{2}{\ZZ_2}{\QCAblendclass{1}{}}$ class is non-trivial, then invoking the second condition for being a simple realization, the symmetry in the bulk \(\repU_x\) must exchange $e$ and $m$ anyons. (Consider a string operator creating an \(e\) anyon and terminating at an \(e\)-condensing boundary. As the boundary switches to \(m\)-condensing, the bulk anyon must also become an \(m\).)
On the other hand, if $\repU$ has a trivial $\grpH{2}{\mathbb{Z}_2}{\QCAblendclass{d-1}{}}$ anomaly, then $\gamma$ is rational and \(\repU_x\) necessarily does \emph{not} exchange the $e$ and $m$ anyons. Thus, \emph{under the assumption of a simple realization with a symmetric boundary algebra}, we can in some sense view the permutation of anyons as obeying an anomaly matching condition relating it to the anomaly of the microscopic symmetry.

However, we emphasize that this result \emph{cannot} extend to all commuting models which have a ground state in the same phase as the toric code (let alone non-commuting Hamiltonians). 
Indeed, there are commuting projector models with the same ground state topological order as the toric code in which the \(e\)-\(m\) exchange is implemented by an on-site symmetry of the Hamiltonian~\cite{Heinrich2016}, which is certainly not anomalous. 
This is another demonstration that focusing only on the low energy theory is not enough to diagnose consequences of an \(\grpH{2}{\ZZ_2}{\QCAblendclass{1}{}}\) anomaly, reflecting its IR-triviality.

A consequence of the above discussion is that the model of Ref.~\cite{Heinrich2016} is not a simple realization of the toric code.
Indeed, the construction of the model with on-site \(e\)-\(m\) exchange is reminiscent of a string-net construction based on the Ising fusion algebra. 
From this perspective it can be shown that a natural symmetric boundary algebra for the model is locally isomorphic to the fusion algebra for a chain of Ising anyons, which is not locally isomorphic to \(A^{\ZZ_2}\).

Finally, let us mention that the results described above can be readily extended to any symmetry group $G$, and the topological phase of the quantum double of any finite Abelian group~\cite{Kitaev2003FaultTolerant} (Appendix~\ref{appendix:SimpleRealizationDouble}).
One finds that under the condition of a ``simple realization'', the $\grpH{2}{G}{\QCAblendclass{1}{}}$ invariant is completely determined by the action of $G$ on the anyons on top of the ground state. In fact, interestingly it \emph{only} depends on the quantum dimensions of the symmetry defects corresponding to the anyon permutations induced by $G$. (For instance, the results described above for the case of the toric code follow from the fact that the \(e\)-\(m\) exchange defect has quantum dimension $\sqrt{2}$.)

\subsection{Implications for MBL phases}
\label{subsec:MBLImplications}

The results described in Secs.~\ref{subsec:symmetric_constraints} and~\ref{subsec:H2_constraints} apply for any commuting model. However, a particularly physically relevant application is to many-body localized (MBL) phases~\cite{Sierant2024review}. 
This application is important because, unlike general commuting models, MBL phases are actually more-or-less stable to arbitrary local perturbations~\cite{Imbrie2016MBL,Sierant2024review}. 
Strict stability is likely not true for \(d>1\) and with random disorder~\cite{deRoeck2017avalanche}, and some authors have expressed skepticism also for \(d=1\)~\cite{Suntajs2020,Schulz2020,Sels2021obstruction,Sels2023dilute,Morningstar2022avalanche,Sels2021bathavalnce}. Nonetheless, this statement of stability certainly applies to an effective description that describes the evolution under the Hamiltonian for astronomical timescales~\cite{Suntajs2020,Morningstar2022avalanche,Sels2021bathavalnce,Long2023phenomenology,deRoeck2023rigorous,deRoeck2024absence}.

Let us briefly review an algebraic definition of MBL which is useful for our discussion. The phenomenological theory of MBL is based around local integrals of motion (LIOMs)~\cite{Huse2014fullyMBL}, which may emerge in strongly inhomogeneous systems. In the general theory of MBL, LIOMs are typically quasi-local, in the sense that they are not exactly supported on any bounded set in space, but instead have tails that decay exponentially. For mathematical convenience, in this paper we will confine ourselves to systems in which the LIOMs can be taken to be \emph{strictly local}. We expect that most of the results in this paper could be generalized to hold for quasi-local operators, but rigorous proofs may be more difficult.

Thus, we define a LIOM \(\ell\) to be any local operator which commutes with the Hamiltonian, \([H, \ell]=0\). The set of all LIOMs forms a subalgebra $L$ of the algebra $A$ of all local operators. For MBL systems, $L$ is taken to be \emph{commutative},  \emph{complete}, and \emph{locally generated}; if these conditions are satisfied we will refer to $L$ as an \emph{MBL LIOM algebra}. Commutative means that $[\ell,\ell'] = 0$ for all $\ell,\ell' \in L$.
We say that \(L\) is complete if for any $a \in A$ such that \([a,\ell] = 0\) for all \(\ell \in L\), we have $a \in L$. Locally generated means that there exists a finite $r$ and a generating set $\{ l_i \}$ for $L$ such that each $l_i$ is supported on a set of diameter $\leq r$.

MBL phases are characterized by their LIOM algebras. We say that an MBL LIOM algebra is \emph{trivial} if $L$ has a set of generators that are each supported on a single site. We say that two LIOM algebras $L$ and $L'$ are equivalent if they are related by the action of a finite-depth circuit. (Again, in the general theory of MBL phases one is really supposed to consider \emph{quasi-local} transformations of local operators rather than strictly finite-range QCAs, of which finite-depth circuits are a special case.) We say that two systems are in the \emph{same MBL phase} if they have equivalent LIOM algebras. We also treat systems as in the same MBL phase if they are related by stacking with an on-site algebra carrying a trivial MBL LIOM algebra. The \emph{trivial MBL phase} is the one that contains trivial MBL LIOM algebras.

For any commutative local algebra which is a subalgebra of $A$, one can show that it admits a simultaneous eigenstate, i.e. a state which is an eigenstate of $l$ for all $l \in L$. In particular, suppose $L$ is locally generated, with generating set $\{ l_i : i \in I \}$ and choose one such state, which we henceforth will call the \emph{preferred state}. Let $\lambda_i$ denote the eigenvalues of the generators $l_i$ in the preferred state. Then we see that $\{ h_i : i \in I \}$, where
\begin{equation}
    h_i = \Delta_i^\dagger \Delta_i, \quad \Delta_i = l_i - \lambda(l_i)
\end{equation}
satisfies all the conditions to be a commuting model as defined in \autoref{sec:CommutingModelsBoundaryAlgebra}, with a ground state given by the preferred state. 
Thus, we can apply all our results about commuting models in this context. In particular, we can define projected algebras and boundary algebras. 
It is important to remember, however, that the mapping from MBL LIOM algebras to commuting models is not unique, because it depends both on a choice of generators and on a choice of preferred state. This non-uniqueness will pose some complications for some of our arguments.

 We say that an MBL LIOM algebra $L$ corresponds to a \emph{symmetric} MBL phase under a $G$-rep $\mathcal{U}$ if $\mathcal{U}$ fixes every element of $L$. (The case where $\mathcal{U}$ preserves $L$ but acts non-trivially on it physically describes a \emph{spontaneous symmetry breaking} MBL phase, see below.) In that case we see that the commuting model constructed above will always be strongly symmetric.

An immediate application is to \emph{short-range entangled (SRE)} MBL phases.
For a trivial LIOM algebra we can define a commuting model as described above, using the single-site generators and any choice of preferred state, which gives \ATrivCommutingModel{} commuting model.
More generally, we call an MBL LIOM algebra $L$ SRE if there exists a QCA $U$ such that $U[L]$ is a trivial MBL LIOM algebra~\cite{Long2024}. One can show that SRE implies invertibility in the sense that there exists an MBL LIOM algebra $\overline{L}$ such that $L \otimes \overline{L}$ is in the trivial MBL phase. However, unlike for commuting models, we do not have a proof that invertibility implies SRE.

For any SRE MBL LIOM algebra, there exists a choice of generators for $L$ which defines a commuting model with invertible boundary algebra. Therefore, from Theorem~\ref{thm:ObstructionToInvertible}, we find that if we have a $G$-rep $\repU$ such that $L$ describes a \emph{symmetric} MBL phase, i.e.\ $\repU$ fixes $L$, then $\mathcal{U}$ is anomaly free. In other words: if we have an MBL system that is symmetric with respect to an anomalous $G$-rep $\repU$, then $L$ cannot be SRE.

We remark that if the $G$ symmetry is \emph{non-Abelian}, it is generally believed that symmetric MBL is not possible \cite{Potter2016Multiplet,Protopov2017su2,Ware2021nonabelian}. In this case, symmetry under a $G$-rep with $G$ non-Abelian can be viewed as an obstruction to MBL, even in the absence of any anomaly. However, as we will see, if there is an anomaly it can also constrain \emph{spontaneous symmetry-breaking} MBL phases in certain ways. Such phases are \emph{not} obstructed simply by $G$ being non-Abelian (at least if $G$ is finite).

\subsubsection{Spontaneous symmetry breaking}

A lattice anomaly prevents having a symmetric SRE MBL phase. Does it also prevent having a trivial or SRE MBL phase in which the symmetry is spontaneously broken? This is definitely not true for all anomaly invariants. For example, in $d=1$, there is a $\mathbb{Z}_2$-rep on a qubit chain with a non-trivial $\grpH{3}{\mathbb{Z}_2}{\mathrm{U}(1)}$ anomaly, where the generator of the $\mathbb{Z}_2$ is the ``CZX circuit'' (described, for example, in Ref.~\cite{Chen_1106}). In that case one can show that the Pauli $Z_i$ operators are odd under the action of $\mathbb{Z}_2$. Therefore, we can form the MBL Hamiltonian
\begin{equation}
\label{eq:ising_mbl1}
    H = \sum_i J_i Z_i Z_{i+1},
\end{equation}
which is a $\mathbb{Z}_2$-symmetric MBL Hamiltonian in which we view the $\mathbb{Z}_2$ as spontaneously broken (as demonstrated, for example, by the fact that all the eigenstates are two-fold degenerate). In terms of the LIOM algebra, which here is generated by the $Z_i$ operators, this is reflected in the fact that the symmetry preserves the algebra (maps the LIOM algebra to itself) but acts non-trivially on it.

More generally, let us ask which kind of anomalies are compatible with \emph{trivial spontaneous symmetry breaking}.
Specifically, we note that the key feature of Hamiltonians such as Eq.~\eqref{eq:ising_mbl1} is that there are a set of on-site LIOMs which are odd rather than even under the symmetry.  More generally, if we have an MBL LIOM algebra $L$ that is preserved but not fixed by the symmetry, and there is a set of on-site LIOM generators $\{l_i\}$, and for each $g\in G$ and each $l_i$ we have that $\repU_g[l_i]$ is still supported on the same site as $l_i$, then we say that $L$ exhibits trivial spontaneous symmetry breaking.

Generalizing the discussion above for the CZX symmetry, one can verify explicitly that the $\grpH{d+2}{G}{ \mathrm{U}(1)}$ anomaly is compatible with MBL with trivial spontaneous symmetry breaking for any finite group $G$, using the representation of the anomalous symmetry described in Refs.~\cite{Chen2013,Else2014} for example. Moreover, it seems likely that for fermions, the $\grpH{d+1}{G}{\QCAblendclass{0}{}} = \grpH{d+1}{G}{\mathbb{Z}_2}$ anomaly also admits trivial spontaneous symmetry breaking MBL in general (one can see this in $d=0$ for example~\footnote{Of course, there is no real spontaneous symmetry-breaking in $d=0$, but the criterion in terms of the symmetry transformation of the LIOMs can be applied equally well in $d=0$.}). However, we conjecture that the $\grpH{p+1}{G}{ \QCAblendclass{d-p}{}}$ anomalies do \emph{not} admit trivial spontaneous symmetry breaking MBL Hamiltonians for $0 \leq p \leq d-1$, although we will only prove it for $p\in\{0,1\}$.
(We remark that this is a closely analogous result to the classification of SPTs protected by average symmetry \cite{Ma_2209, Ma_2305}, which by a bulk-boundary correspondence should correspond to QFT anomalies of average symmetries. By contrast, the results presented here can be understood as statements about the classification of \emph{lattice} anomalies of average symmetries, in the sense that if we have an MBL LIOM algebra with trivial spontaneous symmetry breaking, then using the on-site LIOMs we can define an \emph{ensemble} of on-site commuting models which are permuted by the symmetry.)

For $p=0$ (and $d \geq 2$) note that if the $\grpH{1}{G}{\QCAblendclass{d}{}}$ anomaly is non-trivial, it means that there exists a $g \in G$ such that $\repU_g$ acts as a non-trivial QCA (i.e.\ a non-trivial element of $Q_d$). We know that any on-site set of generators defines an on-site boundary algebra. Moreover, by assumption the QCA maps an on-site set of generators to an on-site set of generators. On the other hand, the result of \autoref{sec:BoundaryGRep} shows that the corresponding boundary algebras should differ by tensoring with the boundary algebra of the QCA, which is a contradiction. This argument can be adapted to the case of $d=1$ as well.

Now let us consider the argument for $p=1$ and $d \geq 2$ (for simplicity we restrict to unitary symmetries). Fix some set $Y \subseteq \mathbb{R}^d$, and choose $W \supseteq Y$. Given a choice of LIOM generators and preferred state we can construct a commuting model $\{ h_i : i \in I \}$; however in general it will not be strongly or weakly symmetric. If we define the restricted Hamiltonian
\begin{equation}
    \widetilde{H} = \sum_{i \in I : h_i \tl Y} h_i,
\end{equation}
then $\widetilde{H}$ will not in general be invariant under the symmetry. Therefore, rather than an automorphism of the boundary algebra, the restricted symmetry $\resU_g$ induces a local isomorphism $\repV_g^{\widetilde{H}}:A_W(\widetilde{H}) \to A_W(\mathcal{U}_g[\widetilde{H}])$.

On the other hand, for any $\widetilde{H}$ which is the sum of positive \emph{on-site} terms $h_i$ such that there is a unique ground state on each site, we have a canonical ultra-local isomorphism $\eta_{\widetilde{H}} : A_S \to A_W(\widetilde{H})$, where $S = W \setminus Y$.
Hence we can define
\begin{equation}
    \repW_g^{\widetilde{H}} = (\eta_{\repV_g(\widetilde{H})})^{-1} \circ \repV_g^{\widetilde{H}} \circ \eta_{\widetilde{H}},
\end{equation}
which is a local automorphism of $A_S$.  Moreover, if
\begin{equation}
    \resU_{g_1} \resU_{g_2} = \resU_{g_1 g_2}\Omega_{g_1, g_2} 
\end{equation}
where $\Omega_{g_1, g_2}$ is a QCA supported on $S$, then it follows that
\begin{equation}
    \repV_{g_1}^{\repU_{g_2}[\widetilde{H}]} \repV_{g_2}^{\widetilde{H}} =   \mathcal{V}_{g_1 g_2}^{\widetilde{H}} \eta_{\widetilde{H}} \Omega_{g_1, g_2} (\eta_{\widetilde{H}})^{-1}.
\end{equation}
Translating this into a statement about $\repW$, we obtain that
\begin{equation}
\label{eq:W_rep}
    \repW_{g_1}^{\repU_{g_2}[\widetilde{H}]} \repW_{g_2}^{\widetilde{H}} = \repW_{g_1 g_2}^{\widetilde{H}} \Omega_{g_1, g_2},
\end{equation}
The key property that we will need is that although $\repW_{g_1}^{\repU_{g_2}[\widetilde{H}]}$ may not be the same as $\repW_{g_1}^{\widetilde{H}}$, they at least have the same blend class.

In fact, we can show that for any two on-site Hamiltonians $H$ and $H'$, $\repW_g^{\widetilde{H}}$ and $\repW_g^{\widetilde{H}'}$ have the same blend class (where $\widetilde{H}$ and $\widetilde{H}'$ are the restrictions to $X$). Indeed, let us take $W \subseteq \RR^d$ to be the region where the first coordinate $x_1 \geq -\xi$ and $X$ to be the region where $x_1 \geq 0$. Now, consider a third on-site model $H''$ such that the local terms in $H''$ coincide with those of $H'$ for $x_2 < 0$, and coincide with those of $H$ for $x_2 \geq 0$. Then upon defining the restriction $\widetilde{H}''$ to $X$, we see that for some finite $r$, $\repW_g^{\widetilde{H}''}$ will act like $\repW_g^{\widetilde{H}}$ for $x_2 < -r$ and like $\repW_g^{\widetilde{H}'}$ for $x_2 > r$, which shows that they are blend equivalent.

Therefore, if we let $\beta_g$ be the $(d-1)$-dimensional blend class of $\repW_{g_1}^{\widetilde{H}}$, and $\omega_{g_1, g_2}$ be the blend class of $\Omega_{g_1, g_2}$, it follows from \eqnref{eq:W_rep} that
\begin{equation}
    \omega_{g_1, g_2} = \beta_{g_1} \beta_{g_2} \beta_{g_1, g_2}^{-1},
\end{equation}
and hence the $\grpH{2}{G}{\QCAblendclass{d-1}{}}$ invariant must be trivial.

\subsubsection{The toric code MBL phase and anyon dualities}
\label{sec:ToricCodeMBL}

At this point, we have many negative results, and obstructions to MBL phases possessing anomalous symmetries. What classes of phases \emph{can} have an anomalous symmetry?

For the case of non-trivial \(\grpH{2}{G}{\QCAblendclass{d-1}{}}\) invariant, we have seen that it supports neither a symmetric SRE MBL phase, nor a trivial spontaneous symmetry breaking MBL phase. Therefore, we next seek a topologically ordered model which supports an anomalous symmetry. Indeed, the example of \autoref{sec:em-qubit} demonstrates that the MBL LIOM algebra generated by the star and plaquette operators of the toric code has its star and plaquette generators swapped by an anomalous \(\ZZ_2\)-rep with a nontrivial \(\grpH{2}{\ZZ_2}{\QCAblendclass{1}{}}\) class. Thus, topological order is indeed one way to have an MBL system which is compatible with an anomalous symmetry. (Though since the symmetry acts non-trivially on the LIOM algebra, we should view the symmetry as spontaneously broken.)

In this example, the anomalous symmetry leaves the ground state of the toric code invariant, but induces a non-trivial anyon automorphism on anyon excitations above the ground state.
We conjecture that this link between the \(\grpH{2}{\ZZ_2}{\QCAblendclass{1}{}}\) class and anyon automorphisms is intrinsic within the entire toric code MBL phase. That is:
\begin{conjecture}
\label{conj:em_swap}
Consider an MBL LIOM algebra $L$ that lies in the toric code MBL phase (that is, $L$ is equivalent by a finite-depth circuit to the algebra generated by the star and plaquette operators of the toric code). Suppose \(\repU\) is a \(\ZZ_2\)-rep which preserves $L$, i.e.\ \(\repU_g[L] = L\). Suppose furthermore that there is a pure state $\ket{\Psi}$ which is a simultaneous eigenstate of all the elements of $L$ and is left invariant by $\repU$. Then \(\repU\) acts as an $e$-$m$ exchange on isolated anyons on top of $\ket{\Psi}$ if and only if the \(\grpH{2}{\ZZ_2}{\QCAblendclass{1}{}}\) class of \(\repU\) is nontrivial.
\end{conjecture}

In this conjecture, if we in particular have a symmetric MBL Hamiltonian $H$ that is built as the sum of elements of $L$ and has non-degenerate ground state, then we could take $\ket{\Psi}$ to be that ground state.
Furthermore, if we suppose that the low-energy degrees of freedom in \(H\) are described by a topological quantum field theory (TQFT), then the condition that \(\repU\) act as an \(e\)-\(m\) exchange for anyon excitations on top of $\ket{\Psi}$ is equivalent to the statement that \(\repU\) descends to the \(e \leftrightarrow m\) anyon automorphism in this TQFT.

In \autoref{subsec:H2_constraints} we mention an example of a commuting projector model with a ground state in the same phase as the toric code which does not obey the relation between $e$-$m$ exchange and the $\grpH{2}{\mathbb{Z}_2}{ \QCAblendclass{1}{}}$ invariant.
Indeed, that model has an on-site realization of the $e$-$m$ exchange anyon automorphism~\cite{Heinrich2016}.
While that particular counterexample involves a commuting projector Hamiltonian which is not MBL-izable, we suspect that it would also be possible to have an MBL Hamiltonian whose ground state is in the toric code phase, but which does not obey Conjecture \ref{conj:em_swap} by virtue of not being in the same MBL phase as the toric code Hamiltonian. (Recall that there can be distinct MBL phases even when their ground states are in the same phase \cite{Long2024}.)

Furthermore, we can also show that Conjecture \ref{conj:em_swap} need not hold in an MBL Hamiltonian with ground state in the toric code phase if we only demand that $\repU$ fixes the ground state and don't place any requirement on how it acts on the LIOMs. Any two Hamiltonians whose ground states are both in the toric code phase should have their ground states related by a finite-depth circuit $V$ (or at the very least, a generalization with near-exponentially decaying tails~\cite{Hastings2005Quasiadiabatic,Hastings2012Locality}). In particular, let $V$ be a finite-depth circuit that transforms the ground state of the model with an on-site $e$-$m$ exchange symmetry into the ground state of the usual toric code.
Conjugating the on-site \(e\)-\(m\) exchange symmetry by this \(V\) gives a \(\ZZ_2\)-rep which fixes the usual toric code ground state (but not the LIOM algebra) and which still acts on anyon excitations as \(e\)-\(m\) exchange.
However, this symmetry is non-anomalous, as it is conjugate to an on-site \(\ZZ_2\)-rep. 

Let us outline some arguments in favor of Conjecture \ref{conj:em_swap}. If $L$ is in the MBL phase of the toric code, there is a finite depth circuit \(V\) such that \(l_i = V[l^{\mathrm{TC}}_i]\), where \(l^{\mathrm{TC}}_i\) is the ``standard'' generating set for the toric code LIOM algebra, i.e.\ the star and plaquette operators, and $\{ l_i \}$ is a generating set for $L$. Then $\{ l_i\}$ defines a simple realization of the toric code phase (recall the definition of simple realization from \autoref{subsec:H2_constraints}).
If \(\repU\) fixes or permutes the \(l_i\), then the boundary algebra is symmetric, hence the reasoning described in \autoref{subsec:H2_constraints} applies, and the \(\grpH{2}{\ZZ_2}{\QCAblendclass{1}{}}\) class of \(\repU\) determines whether \(\repU\) exchanges \(e\) and \(m\) anyons in \(H\).

However, this is not sufficient to establish Conjecture \ref{conj:em_swap}, because in general, the preservation of the LIOM algebra by the symmetry is not sufficient to show that the specific generating set $\{ l_i \}$ is merely fixed or permuted under the symmetry. We can proceed a bit further using the following argument.
By conjugating with the QCA \(V\) from the previous paragraph, we can always consider \(\repU\) to be acting on the toric code LIOM algebra. It can be shown (Appendix~\ref{appendix:TCBoundaryStructure}) that for any such $\repU$ that preserves the toric code LIOM algebra and the toric code ground state (i.e.\ the $+1$ eigenstate of all the toric code LIOMs), there is a way to construct a set of projectors that \emph{almost} gives a simple realization of the toric code, in the sense that the boundary algebra is the tensor product of the boundary algebra of a simple realization with additional local algebra that can be obtained from an on-site algebra by imposing \emph{local constraints}. We suspect there should be ways to improve this construction so that these additional local constraints vanish. In any case, if there are cases where Conjecture \ref{conj:em_swap} does \emph{not} hold, it must be that these constraints are (a)~unavoidable and (b)~make it possible to implement a local automorphism which squares to a translation, but which is not the Kramers-Wannier duality.

Finally, we note that there is a natural extension of our conjecture to the MBL phases containing other Abelian quantum double models (Appendix~\ref{appendix:TCBoundaryStructure}).

\section{Lattice gravitational anomalies}
\label{sec:LatticeGravitational}

Suppose we have a system without any symmetries at all. In this case, obviously one cannot have a non-trivial $G$-rep. On the other hand, QFTs can have still have 't Hooft anomalies that are not related to any symmetry; these are usually referred to as ``gravitational anomalies''. (In our view this terminology is fairly misleading, but we will still adopt it here in order to follow the previous literature.) The question we now wish to address is: how can we model such anomalies on the lattice?

The answer (previously suggested by works such as Refs.~\cite{Haah2023,Chen2023}) is that a lattice gravitational anomaly should correspond to a system in which the algebra of local operators is \emph{not} an on-site algebra, as we have been considering until now, but rather a more general local algebra (recall the definition of local algebra from \autoref{subsec:local_algebras}). In this section we will restrict ourselves to \emph{invertible} gravitational anomalies, although in principle one can consider ``non-invertible gravitational anomalies'' as well~\cite{Kong_1405,Ji_1905}. (We remark that invertibility is a non-trivial condition that must be imposed for lattice gravitational anomalies, in contrast to symmetry anomalies where invertibility follows from a folding argument, as we discussed previously.)

Therefore, we should require the algebra of local operators to be an invertible local algebra. We previously discussed invertibility in \autoref{sec:CommutingModelsBoundaryAlgebra} in the context of boundary algebras. Here we will state that a local algebra $A$ over $\mathbb{R}^d$ is invertible if there exists a local algebra $\overline{A}$ over $\mathbb{R}^d$ such that $A \otimes \overline{A}$ is locally isomorphic to an on-site algebra.
(Reference~\cite{Haah2023} previously referred to invertible local algebras as ``invertible subalgebras'', but we will not follow this terminology here.) 

Next, we need a definition of when two invertible local algebras define equivalent gravitational anomalies. In particular, in order to have as close a parallel as possible to anomalies of symmetries, we want a concept of ``blend equivalence'' of local algebras. We say that two invertible local algebras $A$ and $B$ blend across a set $R \subseteq \mathbb{R}^d$ if there exists a finite $r$ and an invertible local algebra $C$ over $\mathbb{R}^d$ such that $C \{ \mathrm{int}_r R \}$ is ultra-locally isomorphic to $A \{ \mathrm{int}_r R \}$ and $C \{ \mathrm{int}_r R^c \}$ is ultra-locally isomorphic to $ B \{ \mathrm{int}_r R^c \}$. Recall that $A \{ Y \}$ denotes the subalgebra of $A$ comprising all of the elements of $A$ that are supported on $Y$. (Note that it is crucial that we require the blending algebra $C$ to be invertible---otherwise any local algebra would be blend-equivalent to the trivial algebra simply by taking $C = A\{R\}$. This requirement is also what prevents one from using a folding argument to show that any local algebra is invertible, as we did in the case of symmetry anomalies.)

We prove in \autoref{sec:ConjAnomBlendIffConjugate}
that two invertible local algebras over $\mathbb{R}^d$ are blend-equivalent over $(-\infty,0) \times \mathbb{R}^{d-1}$ if and only if they are \emph{stably equivalent}: that is, they are locally isomorphic up to tensoring with on-site algebras.
Stable equivalence classes of invertible local algebras over $\mathbb{R}^d$ form an Abelian group which we call $\LatGravAnom{d}$, the group of lattice gravitational anomalies in $d$ spatial dimensions.

An important result about lattice gravitational anomalies is that they are in one-to-one correspondence with equivalence classes of QCAs in one higher dimension~\cite{Haah2023}.
(In fact, there is a similar correspondence for lattice anomalies of symmetries as well---we will describe this in \autoref{sec:BulkBoundaryCorrespondence}.) 
Thus, one can translate results about QCAs into results about gravitational anomalies, and vice versa.

As with anomalies of symmetries, by considering Hamiltonians defined inside an invertible local algebra, and then passing to the IR QFT description, one expects there to be a homomorphism
\begin{equation}
    \varphi : \LatGravAnom{d} \to \QFTGravAnom{d},
\end{equation}
where $\QFTGravAnom{d}$ denotes the group of gravitational anomalies of QFTs in $d$ spatial dimensions.
In fact, we can give a much more precise construction of this map, as follows. 
(In \autoref{sec:BulkBoundaryCorrespondence} we provide an analogous construction for symmetry anomalies as well.)

First of all, we already mentioned that there is an isomorphism
\begin{equation}
    \alpha: \LatGravAnom{d} \to \QCAblendclass{d+1}{}.
\end{equation}
Then, we note that there is a homomorphism \(e_*\) from $Q_{d+1}$
to $\invblendclass{d+1}{}$, the group of invertible topological phases without symmetry, defined simply by acting with the QCA on a product state.
Finally, one expects $\invblendclass{d+1}{}$ to be in isomorphic to the group of  QFT gravitational anomalies in $d$ spatial dimensions, by passing to the anomalies describing their edges. 
Thus, the commutative square
\begin{equation}
\begin{tikzcd}
\label{eq:gravitational_square}
 \QCAblendclass{d+1}{} \arrow[r,"e_*"] 
 &\invblendclass{d+1}{} \arrow{d}{\text{edge}} \\
 \LatGravAnom{d} \arrow{u}{\alpha} \arrow[r, dashed,"\varphi"] & \QFTGravAnom{d}
\end{tikzcd}.
\end{equation}
 defines a homomorphism $\varphi :  \LatGravAnom{d} \to \QFTGravAnom{d}$. 

In \eqnref{eq:gravitational_square},
the $\alpha$ map has been proved to be an isomorphism, while the $\mathrm{edge}$ map is at least \emph{expected} to be an isomorphism. However, $e_*$ is in general not an isomorphism. Therefore,
as in the case of anomalies of symmetries, the homomorphism $\varphi$ is in general neither injective nor surjective. For example, Refs.~\cite{Haah2022,Shirley2022} identified QCAs in three spatial dimensions which are believed to be in non-trivial equivalence classes. Thus, these would correspond to non-trivial lattice gravitational anomalies in two spatial dimensions. However, there are no non-trivial QFT gravitational anomalies in two spatial dimensions, so we see that $\varphi$ is not injective in this case.  
On the other hand, in one spatial dimension there exists a QFT gravitational anomaly corresponding to the boundary of the Kitaev $E_8$ invertible topological phase \cite{Kitaev2006BeyondAnyons,Kitaev2011E8,Lan2016E8,Long2024Edge}, but it cannot be in the image of the map $\varphi$ because there are no non-trivial QCAs in two spatial dimensions. 
For an example of a non-trivial QFT gravitational anomaly in three spatial dimensions which \emph{does} descend from an invertible local algebra (equivalently, a non-trivial invertible phase in four spatial dimensions which can be constructed from a product state through a QCA), see Ref.~\cite{Chen2023}.

We then have the following result, which is similar to that of \autoref{subsec:symmetric_constraints}, about the obstruction to invertible commuting Hamiltonians/MBL phases in a system with a lattice gravitational anomaly. Let $A$ be an invertible local algebra over $\mathbb{R}^d$. We can then consider commuting models $\{ h_i : i \in I\}$, where each $h_i$ is a positive~\footnote{In a general local algebra, $a \in A$ is called \emph{positive} if $a = \gamma^\dagger \gamma$ for some $\gamma \in A$. For invertible local algebras this is equivalent to saying that $a \otimes \unit$ maps to a positive operator under the local isomorphism from $A \otimes \overline{A}$ to an on-site algebra.} element of $A$.  Given a formal sum $H = \sum_{i \in I} h_i$,
the projected algebra $A(H)$ can be defined just as well in an invertible local algebra as in an on-site algebra.  Then our main result is:
\begin{thm}
\label{thm:gravitational_thm}
Let $A$ be an invertible local algebra over $\mathbb{R}^d$,  $W \subseteq \mathbb{R}^d$ be a half-volume and suppose there exists a Hamiltonian $\widetilde{H}$ that is a formal sum of commuting elements of $A$ that are supported on $W$, such that $A(\widetilde{H})$ is invertible and supported on $(W^c)^{+r}$ for some finite $r$. Then $A$ has trivial lattice gravitational anomaly.

    \begin{proof}
        In Appendix~\ref{appendix:AlternativeProjectedAlgebra}, we prove that $A(\widetilde{H}) \{ W^c \} $ is ultra-locally isomorphic to $A \{ W^c \}$. Thus, $A(\widetilde{H})$ defines a blend from $A$ to the trivial algebra. Therefore, $A$ has trivial lattice gravitational anomaly.
    \end{proof}
\end{thm}
This result may be easier to interpret if we think about it in terms of commuting models on $A \otimes \overline{A}$, where $\overline{A}$ is the inverse algebra. If we have a commuting model $\{ h_i \}$ on $A \otimes \overline{A}$ which is decoupled between $A$ and $\overline{A}$, i.e.\  each of the local terms acts inside either $A$ or $\overline{A}$, then Theorem~\ref{thm:gravitational_thm} implies that if this commuting model has invertible boundary algebra then the lattice gravitational anomaly of $A$ is trivial. To see this, suppose we have restricted commuting Hamiltonians $\widetilde{H}$ and $\widetilde{H}'$, which are both supported on $W$ and are built out of local terms in $A$ and $\overline{A}$ respectively. Then we prove in Appendix~\ref{appendix:AlternativeProjectedAlgebra} that $(A \otimes \overline{A})(\widetilde{H} \otimes \unit + \unit \otimes \widetilde{H}')$ is ultra-locally isomorphic to $A(\widetilde{H}) \otimes \overline{A}(\widetilde{H}')$. If this is an invertible boundary algebra, then it is invertible and supported on $(W^c)^{+r}$, which implies that $A(\widetilde{H})$ is invertible and supported on $(W^c)^{+r}$, and then we can apply Theorem~\ref{thm:gravitational_thm}.

It is worth noting that, in addition to gravitational anomalies, invertible local algebras also allow us to realize QFT anomalies of symmetries that cannot be realized by a $G$-rep acting on an on-site algebra. As an example, consider the QCA in three spatial dimensions described in Refs.~\cite{Haah2022,Fidkowski2023}. It is invariant under complex conjugation in an appropriate basis; from this one can show that the corresponding invertible local algebra in two spatial dimensions supports an anti-linear automorphism, i.e.\ a local representation of time-reversal symmetry. 
One can construct a symmetry-respecting commuting projector model inside this algebra whose ground state has a topological order that is known to carry a non-trivial QFT anomaly for time-reversal symmetry (but no gravitational anomaly)~\cite{Vishwanath2013ThreeD,Burnell2014ThreeFermion}.
Thus, in this case the lattice gravitational anomaly maps in the IR to trivial QFT gravitational anomaly, but the time-reversal action on the non-trivial invertible local algebra maps in the IR to a QFT anomaly of time-reversal symmetry. It has previously been argued that this QFT anomaly \emph{cannot} be realized through the local action of time-reversal symmetry on an on-site algebra~\cite{Else2014}.

This observation motivates extending the notion of \(G\)-rep to pairs \((A, \repU)\) where \(A\) is an invertible local algebra and \(\repU: G \to \LAut(A)\) is a homomorphism. A \(G\)-rep \((C, \repV)\) is a blend of \((A, \repU)\) and \((B, \repW)\) across \(R\) if \(C\) is a blend from \(A\) to \(B\) and \(\repV\) agrees with \(\repU\) in \(\mathrm{int}_r R\) (using the ultra-local isomorphism of \(C\{\mathrm{int}_r R\}\) and \(A\{\mathrm{int}_r R\}\)) and \(\repV\) similarly agrees with \(\repW\) in \(\mathrm{int}_r R^c\). Lattice anomalies are blend classes of \(G\)-reps (as in Appendix~\ref{appendix:BlendEquivalenceDef}). They form a group under stacking, still denoted by \(\LatAnom{d}{G}\).

The above example involves an antiunitary symmetry. By contrast, one can show that in bosonic systems (and for fermionic systems in which there is no non-trivial extension by fermion parity), if the symmetry $G$ is unitary then there are no additional QFT anomalies of $G$ (with trivial gravitational anomaly) that can be realized through invertible local algebras, beyond those that can realized through $G$-reps acting on on-site algebras. To see this, consider the homomorphism
\begin{equation}
\phi : \LatAnom{d}{G} \to \LatGravAnom{d}
\end{equation}
corresponding to forgetting about the $G$-rep. Under the assumptions mentioned at the start of this paragraph, there is also a homomorphism 
\begin{equation}
    \psi : \mathrm{Anom}_d^{\mathrm{lattice}} \to \LatAnom{d}{G}
\end{equation}
in which we just take the $G$ symmetry to act trivially. This is a right-inverse to $\phi$, i.e. $\phi \circ \psi = \id$, which means that $\phi$ is surjective.
From this it follows that
\begin{multline}
    0 \to \mathrm{SPAnom}_d^{\mathrm{lattice}}(G) \to \LatAnom{d}{G} \\ \xrightarrow{\phi} \LatGravAnom{d} \to 0.
\end{multline}
is a short exact sequence,
where
$\mathrm{SPAnom}_d^{\mathrm{lattice}}(G)$  comprises the anomalies that require the $G$ symmetry to be non-trivial ($\mathrm{SPAnom}$ stands for ``symmetry-protected anomaly''). Moreover, the existence of the right-inverse to $\phi$ also implies that this short exact sequence splits. Hence we obtain the conclusion that
\begin{equation}
\label{eq:lat_anom_decomposition}
   \LatAnom{d}{G}  = \mathrm{SPAnom}_d^{\mathrm{lattice}}(G) \oplus \mathrm{Anom}_d^{\mathrm{lattice}}.
\end{equation}
By a similar argument one also has that 
\begin{equation}
\label{eq:qft_anom_decomposition}
   \QFTAnom{d}{G}  = \mathrm{SPAnom}_d^{\mathrm{QFT}}(G) \oplus \mathrm{Anom}_d^{\mathrm{QFT}}.
\end{equation}
and one can argue that the homomorphism from $\LatAnom{d}{G}$ to $\QFTAnom{d}{G}$ should map the first (second) component of \eqnref{eq:lat_anom_decomposition} into the first (second) component of \eqnref{eq:qft_anom_decomposition}.
From this we see that any $G$ anomaly of a QFT (without gravitational anomaly) that can be obtained from a $G$ action on an invertible local algebra can also be obtained from a $G$ action on an on-site algebra. 

\section{Bulk-boundary correspondence for lattice anomalies}
\label{sec:BulkBoundaryCorrespondence}

There is a close link between QFT anomalies in \(d\) dimensions and symmetry protected states (\(G\)-SPTs) in \(d+1\) dimensions. 
Indeed, there is a one-to-one bulk-boundary correspondence between \(G\)-SPT phases, and more generally invertible phases with \(G\)-symmetry, and the anomalies describing their edges.

In this section, we will show that there is a similar bulk-boundary correspondence that is applicable to \(d\)-dimensional lattice anomalies; however, we must consider invertible phases of \((d+1)\)-dimensional \emph{commuting models}. (Recall the definition of commuting model from \autoref{sec:CommutingModelsBoundaryAlgebra}.) We will also see how invoking bulk-boundary correspondence allows for a more precise formulation of the map from lattice anomalies to QFT anomalies, generalizing the discussion in \autoref{sec:LatticeGravitational}.

\subsection{The case without symmetry: lattice gravitational anomalies on the boundary of commuting models}

Let us first consider the case without symmetry. In that case the relevant lattice anomalies are the lattice gravitational anomalies
of \autoref{sec:LatticeGravitational},
whose classification group we called $\LatGravAnom{d}$. 
We already observed that there is an isomorphism $\LatGravAnom{d} \cong \QCAblendclass{d+1}{}$, where $\QCAblendclass{d+1}{}$ is the group of equivalence classes of QCAs in $d+1$ spatial dimensions. 
In fact, we show that there is a three-way isomorphism
\begin{equation}
\label{eq:three_way_isomorphism}
    \LatGravAnom{d} \cong \InvModelClass{d+1}{} \cong \QCAblendclass{d+1}{}.
\end{equation}
Here $\InvModelClass{d+1}{}$ is the group of invertible phases of commuting models. 
We say that two commuting models are in the same phase if they can be connected by a series of the following moves:
\begin{itemize}
    \item $\{ h_i : i \in I \} \sim \{ U[h_i] : i \in I \}$, where $U$ is any finite-depth circuit.
        \item $\{ h_i : i \in I \} \sim \{ h_i' : i \in I \}$ where each $h_i'$ has the same ground-state subspace as $h_i$; and
    \item Stacking with trivial commuting models (recall the definition of trivial commuting models from \autoref{sec:CommutingModelsBoundaryAlgebra}). The stacking moves will in general change the size of the label set $I$.
\end{itemize}
[Note that we take commuting models to be defined by the (unordered) \emph{set} of local terms; thus two commuting models that differ only by a permutation of the labels in $I$ are identified.]

Phases of commuting models form a commutative monoid; the invertible elements of this monoid are the invertible phases.

To see \eqnref{eq:three_way_isomorphism}, we note that there is a homomorphism $\gamma : \InvModelClass{d+1}{} \to \LatGravAnom{d}$, by evaluating the boundary algebra of the commuting model. (Note that, while we have not established the existence and uniqueness of the boundary algebra for a general commuting model, it is easy to show in the case of \emph{invertible} commuting models.)
Meanwhile, there is a homomorphism $\mu : \QCAblendclass{d+1}{} \to \InvModelClass{d+1}{}$ defined by acting with QCAs on a trivial commuting model.
Finally, we recall from the previous section that there is an isomorphism $\alpha : \LatGravAnom{d} \to \QCAblendclass{d+1}{}$. 
We obtain a commutative diagram
\begin{equation}\label{eqn:Grav_QPAnom_Triangle}
    \begin{tikzcd}
        & \QCAblendclass{d+1}{} \arrow[ld,"\alpha^{-1}"'] \arrow[dd,"\mu"]
        \\ \LatGravAnom{d} & 
        \\ & \InvModelClass{d+1}{} \arrow[lu,"\gamma"]
    \end{tikzcd}.
\end{equation}
(we show commutativity in \autoref{sec:BoundaryGRep}). 
The commutativity of the diagram is sufficient to conclude that $\mu$ must be injective. The surjectivity of $\mu$ follows from Appendix~\ref{appendix:InvertibleIffQCAPreparable}; hence we find that $\mu$ is an isomorphism. This then implies that $\gamma$ is an isomorphism as well.

Instead of commuting models, one can also make similar statements about MBL phases (recall our definition of MBL phases in terms of LIOM algebras from \autoref{subsec:MBLImplications}). 
There is a subtlety, however, because in general the classification of SRE MBL phases in spatial dimension $d+1$ is isomorphic to the quotient $\QCAblendclass{d+1}{} / \QCAblendclass{d+1}{(0)}$, where $\QCAblendclass{d+1}{(0)}$ correspond to the QCAs that can preserve an on-site LIOM algebra~\cite{Long2024} (these were referred to as ``generalized permutations'' in Ref.~\cite{Long2024}). In fact, we conjecture that $\QCAblendclass{d+1}{(0)}$ is always trivial in spatial dimensions $d \geq 1$, but we do not have a proof of this.

In the rest of this section, we will extend the statements of this subsection to the case with symmetry.

\subsection{Bulk-boundary correspondence with symmetry: summary}

In this subsection, we summarize the results presented in the subsequent subsections. We would like to extend \eqnref{eq:three_way_isomorphism} to
\begin{equation}
\label{eq:wrong_G_correspondence}
\quad \LatAnom{d}{G} \cong \InvModelClass{d+1}{G} \cong \QCAblendclass{d+1}{G} \quad \mbox{\emph{(incorrect)}},
\end{equation}
where $Q_{d+1}^G$ is the group of $G$-QCAs in $d+1$ dimensions (i.e. equivalence classes of QCAs which commute with an on-site $G$ symmetry), and $\InvModelClass{d+1}{G}$ is the group of invertible phases that are symmetric under an on-site $G$ symmetry.

However, \eqnref{eq:wrong_G_correspondence} is likely at best correct for lattice anomalies of \emph{internal} \(G\)-reps (even for internal $G$-reps we do not have a proof).
It turns out that the restriction that the QCA should be symmetric under an \emph{on-site} $G$-rep limits the class of \(G\)-reps \(\repU\) which can occur on their boundary to those which have a \emph{conjugacy inverse}---all of which are internal.
(We cannot confirm that any internal \(G\)-rep has a conjugacy inverse.)
The \(G\)-rep \(\repU\) is conjugacy invertible if there is another \(G\)-rep \(\overline{\repU}\) and a local isomorphism \(V\) such that
\begin{equation}
    \repU \otimes \overline{\repU} = V \onsiteV V^{-1},
\end{equation}
where \(\onsiteV\) is an on-site \(G\)-rep. In this setting, it is natural to introduce a notion of \emph{conjugacy anomaly}: the equivalence class of a conjugacy-invertible \(G\)-rep under conjugation by local isomorphisms and stacking with on-site reps. We denote the group of conjugacy anomalies by \(\CLatAnom{d}{G}\).

Similarly, in order for the equivalence to hold, \(\InvModelClass{d+1}{G}\) should be taken as the group of \emph{\(G\)-invertible commuting phases}. A symmetric commuting phase consists of commuting models which are equivalent under the action of finite depth circuits of symmetric gates and stacking with trivial symmetric models.
These form a monoid under stacking, and the \(G\)-invertible phases are the invertible elements of the monoid. Note that it is possible that these are a proper subset of the $G$-symmetric commuting phases which become invertible if we forget about the symmetry. For phases of ground states, it is generally believed that these two notions coincide, but there is no proof for phases of commuting models.

To define $\QCAblendclass{d+1}{G}$, we consider \emph{circuit classes} of \(G\)-QCAs. We say two \(G\)-QCAs are circuit equivalent if they are related by multiplication with a symmetric circuit and stacking ancillas. 
\(\QCAblendclass{d+1}{G}\) is the group of circuit equivalence classes in \(d+1\) dimensions, with the group operation induced by stacking \(G\)-QCAs.

Thus, the version of \eqnref{eq:wrong_G_correspondence} we prove is
\begin{equation}\label{eqn:conj_G_correspondence}
    \CLatAnom{d}{G} \cong \InvModelClass{d+1}{G} \cong \QCAblendclass{d+1}{G}.
\end{equation}

The bulk-boundary correspondence for conjugacy anomalies also allows us to make a more precise characterization of the map from lattice anomalies to QFT anomalies. Expressed in terms of \(G\)-QCAs, we have a commuting square
\begin{equation}
\begin{tikzcd}
 \QCAblendclass{d+1}{G} \arrow[d,swap,"\beta"] \arrow[r,"e_*"] 
 &\invblendclass{d+1}{G} \arrow[d,"\text{edge}"] \\
 \mathrm{CAnom}^{\mathrm{lattice}}_d(G) \arrow[r,"\varphi'"] & \QFTAnom{d}{G}
\end{tikzcd}.
\label{eqn:LatticeToQFTSquare}
\end{equation}
Here, \(\invblendclass{d+1}{G}\) is the group of circuit equivalence classes of invertible states in \(d+1\) dimensions with \(G\)-symmetry. 
The left vertical map 
\(\beta\) is an isomorphism which sends a circuit class of \(G\)-QCAs to the conjugacy anomaly describing its boundary.
Similarly, the right vertical map is the map from a bulk ground state invertible phase to the QFT anomaly describing its edge. 
The lower horizontal map \(\varphi'\) sends a lattice conjugacy anomaly to the QFT anomaly describing the low energy theory of some Hamiltonian symmetric under a \(G\)-rep in the conjugacy class. 
It is similar to \(\varphi\) discussed in \autoref{sec:IntroLatticeToQFT}, except that its domain is conjugacy anomalies. 
Finally the top horizontal map $e_*$ sends a \(G\)-QCA to a state prepared by that QCA. Given the isomorphism $Q_{d+1}^G \cong P_{d+1}^G$, we could equivalently phrase $e_*$ in terms of commuting models, in which case it sends a commuting model to its ground state.

The commutativity of the diagram expresses the compatibility of the lattice and QFT anomaly bulk-boundary correspondences. We obtain the same edge QFT anomaly for an invertible bulk model by either sending the bulk model to its ground state and then the edge QFT anomaly, or first sending the bulk model to its edge conjugacy anomaly and then sending that conjugacy anomaly to its IR QFT anomaly.

The square in Eq.~\eqref{eqn:LatticeToQFTSquare} can also allow us to bypass arguments involving QFTs and \(\varphi'\), which we treat with less rigor than lattice objects. Indeed, we may regard Eq.~\eqref{eqn:LatticeToQFTSquare} as a rigorous definition of \(\varphi'\). 
This allows us to, for instance, precisely identify the IR trivial lattice anomalies (under the assumption that \(G\)-SPT phases form a group under stacking). 
Recall that, in an on-site algebra, we call an anomaly IR trivial if it admits a gapped symmetric Hamiltonian with a unique product ground state. In a more general invertible algebra, we can call an anomaly IR trivial if it admits a gapped symmetric Hamiltonian with a unique \emph{invertible} ground state. 
Intuitively, this is a property of a low-energy effective model which should be captured by QFT, but as stated it is actually a lattice definition. 
Indeed, under the assumption mentioned above, we can prove that a conjugacy anomaly is IR trivial if and only if it is in the kernel of \(e_* \beta^{-1}\). 
Then, provided a precise definition of \(\varphi'\) maintains commutativity in Eq.~\eqref{eqn:LatticeToQFTSquare} and assuming that the ``edge'' map is an isomorphism~\footnote{One subtlety is that it is conceivable that in order for the edge map from $\invblendclass{d+1}{G} \to \QFTAnom{d}{G}$ to be an isomorphism, one should consider equivalences of invertible states under a generalization of finite-depth circuits allowing for rapidly decaying tails, rather than equivalences by strictly local circuits. However in the rest of the paper we do not adopt such a definition of $\invblendclass{d+1}{G}$.}, we can conclude that the IR trivial conjugacy anomalies are exactly those which descend to trivial QFT anomalies.

The primary goal of the remainder of this section is to prove Eqs.~(\ref{eqn:conj_G_correspondence},~\ref{eqn:LatticeToQFTSquare}).
We provide the precise definitions of conjugacy anomaly, \(G\)-QCA, \(G\)-invertible model and the groups appearing in \eqnref{eqn:conj_G_correspondence} in \autoref{sec:CommutingBoundaryAnomalies}.
Sections~\ref{sec:BoundaryGRep}-\ref{sec:InvertibleAlgebraToGQCA} prove \eqnref{eqn:conj_G_correspondence}.
In \autoref{sec:BoundaryGRep}, we establish a generalization of the commuting triangle \eqnref{eqn:Grav_QPAnom_Triangle} that includes symmetry, which reduces the proof of \eqnref{eqn:conj_G_correspondence} to showing the map \(\QCAblendclass{d+1}{G} \to \CLatAnom{d}{G}\) is an isomorphism.
As an independently useful preliminary result, we show two \(G\)-QCAs are circuit equivalent if and only if they blend via a blend with an on-site symmetry in \autoref{sec:GQCABlendToCircuit}, and
\autoref{sec:InvertibleAlgebraToGQCA} then completes the proof of \eqnref{eqn:conj_G_correspondence}.
\autoref{sec:ConjAnomToQFT} completes the derivation of Eq.~\eqref{eqn:LatticeToQFTSquare} and characterizes the IR trivial conjugacy anomalies. 
Finally, \autoref{subsec:ConjugacyToBlend} addresses the relation between conjugacy anomalies and lattice anomalies defined by blends, as considered in the rest of this work. Indeed, similar constructions leading to Eq.~\eqref{eqn:LatticeToQFTSquare} can be repeated for more general lattice anomalies, in which case the \(G\)-QCAs cannot be taken to have on-site symmetries. We leave this to Appendix~\ref{sec:BlendBulkBoundary}.

In the remainder of the section, we will distinguish on-site \(G\)-reps [usually \((d+1)\)-dimensional] from more general \(G\)-reps [usually \(d\)-dimensional] by decorating on-site \(G\)-reps with a dot: \(\onsiteU, \onsiteV, \onsiteW\). Unless stated otherwise, we take \(d\geq 1\).

\subsection{Conjugacy anomalies, \texorpdfstring{\(G\)}{G}-QCAs, and \texorpdfstring{\(G\)}{G}-invertible models}
\label{sec:CommutingBoundaryAnomalies}

In this subsection, we establish our precise definitions of the groups \(\CLatAnom{d}{G}\), \(\QCAblendclass{d+1}{G}\), and \(\InvModelClass{d+1}{G}\) appearing in \eqnref{eqn:conj_G_correspondence}.

\emph{Conjugacy anomalies}.---
We begin with \(\CLatAnom{d}{G}\).
As in \autoref{sec:LatticeGravitational}, we extend the notion of \(G\)-rep to pairs \((A, \repU)\), where \(A\) is an invertible local algebra and \(\repU : G \to \LAut(A)\) is a homomorphism. 
We say that \(G\)-reps \((A, \repU)\) and \((B, \repW)\) are \emph{conjugate} if there are on-site on-site \(G\)-reps \((M_{1,2},\onsiteV_{1,2})\) and a local isomorphism
\begin{equation}
    V : A \otimes M_1 \to B \otimes M_2
\end{equation}
such that
\begin{equation}\label{eqn:ConjugacyEquivalence1}
    \repU \otimes \onsiteV_1 = V ( \repW \otimes \onsiteV_2 ) V^{-1}.
\end{equation}
We say that \((A,\repU)\) is \emph{conjugacy invertible} if there is an \((\overline{A}, \overline{\repU})\) such that \((A \otimes \overline{A}, \repU \otimes \overline{\repU})\) is conjugate to an on-site \(G\)-rep \((M, \onsiteV)\). 
Denote the equivalence class of the \(G\)-rep \((A, \repU)\) under conjugacy by \([A,\repU]\).

The \emph{conjugacy anomaly} of a conjugacy invertible \(G\)-rep \((A,\repU)\) is its equivalence class under conjugacy. This definition mirrors our definition of lattice anomaly, except we have used a finer equivalence relation on \(G\)-reps and restricted to a strict subset of all \(G\)-reps (for example, a translation \(\ZZ\)-rep is not conjugacy invertible, nor is any other non-internal symmetry).
Stacking of \(G\)-reps induces an Abelian group structure on the set of conjugacy anomalies, which we denote \(\CLatAnom{d}{G}\).

\emph{\(G\)-QCAs}.---
Next, we define \(\QCAblendclass{d+1}{G}\). We define a \emph{\(G\)-QCA}~\cite{Gong2020MPU,Zhang_2010,Zhang2023entanglers} to be a pair \((U, \onsiteU)\), where \(U\) is a QCA, \(\onsiteU\) is an on-site \(G\)-rep, and \(\onsiteU_g U = U \onsiteU_g\) for all \(g \in G\). We have suppressed the dependence on the on-site algebra \(M\) over \(\RR^{d+1}\) on which \(U\) and \(\onsiteU_g\) act.
The \(G\)-QCA \((X, \onsiteU)\) is a \emph{symmetric circuit} if \(X\) is equal to a finite depth unitary circuit, and every gate in the circuit decomposition commutes with \(\onsiteU_g\).

Any two \(G\)-QCAs \((U, \onsiteU)\) and \((W, \onsiteW)\) can be considered to act on the same on-site algebra with the same \(G\)-rep by stacking with on-site \(G\)-reps \(\onsiteV_{1,2}\) such that
\begin{equation}\label{eqn:OnSiteCommonGRep}
    \onsiteU \otimes \onsiteV_{1} = \onsiteW \otimes \onsiteV_{2}.
\end{equation}
Then we say \((U, \onsiteU)\) and \((W, \onsiteW)\) are circuit equivalent if there are \(\onsiteV_{1,2}\) such that \eqnref{eqn:OnSiteCommonGRep} holds and \(U\) and \(W\) are related by a symmetric circuit \((X, \onsiteU \otimes \onsiteV_{1})\) as
\begin{equation}\label{eqn:CircuitEquivalence}
    U \otimes \id = X (W \otimes \id),
\end{equation}
where the \(\id\) factors act on the ancillas carrying \(\onsiteV_{1,2}\).
Denote the circuit class of a \(G\)-QCA \((U, \onsiteU)\) by \([U, \onsiteU]\). The set of equivalence classes of \((d+1)\)-dimensional \(G\)-QCAs is \(\QCAblendclass{d+1}{G}\).

Equation~\eqref{eqn:CircuitEquivalence} is, indeed, an equivalence relation, and stacking of \(G\)-QCAs induces a well-defined Abelian group operation on \(\QCAblendclass{d+1}{G}\), with the inverse of \([U, \onsiteU]\) being \([U^{-1}, \onsiteU]\). 
This follows from the observation that \(U \otimes U^{-1} = (U \otimes \id) \mathrm{SWAP} (U^{-1} \otimes \id) \mathrm{SWAP}\) is a symmetric finite depth circuit under \(\onsiteU \otimes \onsiteU\).
Indeed, an individual swap gate between two sites with the same on-site \(G\)-rep is symmetric [the \(\mathrm{SWAP}\) layer], and the conjugation of such a gate by the \(G\)-QCA \(U \otimes \id\) is also symmetric and finitely supported [the \((U \otimes \id)\mathrm{SWAP}(U^{-1} \otimes \id)\) layer].

Similarly, we note that \(G\)-QCAs which differ by multiplication by a symmetric circuit on the right, \((U, \onsiteU)\) and \((UY, \onsiteU)\), are also circuit equivalent. 
This follows from observing that \((UYU^{-1}, \onsiteU)\) is a symmetric circuit, as seen by conjugating gates in \(Y\) by \(U\).

\emph{\(G\)-invertible commuting models}.---
Finally, we define \(\InvModelClass{d}{G}\).
A strongly \(G\)-symmetric commuting model is a pair \((H,\onsiteU)\), where \(H = \sum_{i \in I} h_i\) (properly \(H = \{h_i : i \in I\}\)) is a commuting model in an on-site algebra, \(\onsiteU\) is an on-site \(G\)-rep on said algebra, and \(\onsiteU_g[h_i] = h_i\) for all \(h_i\).
Call \((H,\onsiteU)\) trivial if each \(h_i\) is supported on a single site and has a unique ground state as a single-site operator.

Circuit equivalence can also be defined for commuting models.
Define \((H,\onsiteU)\) and \((K,\onsiteW)\) to be circuit equivalent if there are trivial models \((H', \onsiteV_1)\) and \((K', \onsiteV_2)\) such that \eqnref{eqn:OnSiteCommonGRep} holds and there is a symmetric circuit \((X, \onsiteU\otimes\onsiteV_1)\) such that the set of ground spaces of \(X[h_i] \in X[H \otimes \unit + \unit \otimes H']\) is the same as the set of ground spaces of \(k_{i} \in K \otimes \unit + \unit \otimes K'\).
That is, there is a bijective pairing of terms between the two Hamiltonians, such that the paired terms have the same ground spaces. 
One can check that this is, indeed, an equivalence relation. Denote the equivalence class of \((H, \onsiteU)\) by \([H, \onsiteU]\).

The symmetric model \((H, \onsiteU)\) is \emph{\(G\)-invertible} if there is another symmetric model \((\overline{H},\overline{\repU}^\bullet)\) such that \((H\otimes \unit + \unit\otimes \overline{H}, \onsiteU \otimes \overline{\repU}^\bullet)\) is circuit equivalent to a trivial model. Denote the set of circuit equivalence classes of \(G\)-invertible commuting models in \(d+1\) dimensions by \(\InvModelClass{d+1}{G}\).
Stacking of commuting models induces an Abelian group operation on \(\InvModelClass{d+1}{G}\).

\subsection{Boundary \texorpdfstring{\(G\)}{G}-reps}
\label{sec:BoundaryGRep}

To begin setting up the triple equivalence \eqnref{eqn:conj_G_correspondence}, we construct the commuting triangle corresponding to \eqnref{eqn:Grav_QPAnom_Triangle}:
\begin{equation}\label{eqn:G_QPAnom_Triangle}
    \begin{tikzcd}
        & \QCAblendclass{d+1}{G} \arrow[ld,"\beta"'] \arrow[dd,"\mu"]
        \\ \CLatAnom{d}{G} & 
        \\ & \InvModelClass{d+1}{G} \arrow[lu,"\gamma"]
    \end{tikzcd}.
\end{equation}
The maps \(\beta\), \(\mu\), and \(\gamma\) are all familiar: \(\beta\) sends a QCA to its boundary algebra, and \(\gamma\) does the same for a commuting model, while \(\mu\) sends a QCA to the commuting model it prepares by acting on a trivial model. We only need to enrich these maps with symmetry, and check that the diagram commutes.

\subsubsection{Defining the maps}

Before we can prove the commutativity of \eqnref{eqn:G_QPAnom_Triangle}, we must present more precise definitions of \(\beta\), \(\gamma\), and \(\mu\). 
We will also argue that \(\mu\) is surjective.

\emph{The map \(\beta\)}.--- Beginning with \(\beta\), we define the \emph{boundary \(G\)-rep} for a \(G\)-QCA, generalizing the notion of boundary algebra. Recall that if \(U\) is a \((d+1)\)-dimensional QCA on an on-site algebra \(M\) and \(R = (-\infty,0]\times \RR^d\) is a half-volume, then the \(d\)-dimensional boundary algebra \(A\) is characterized by the ultra-local isomorphism
\begin{equation}\label{eqn:GQCABoundaryAlgebra}
    U[M\{ R\}] \cong M\{\mathrm{int}_r R\} \otimes A,
\end{equation}
where \(r\) is the range of \(U\), and \(A\) is the subalgebra of \(U[M\{ R\}]\) supported in \((\mathrm{int}_r R)^c\). 
If \((U, \onsiteU)\) is a \(G\)-QCA, then \(A\) carries a symmetry \(\repU\) inherited from the on-site \(G\)-rep \(\onsiteU\).
Letting \(\onsiteU_R\) be the restriction of \(\onsiteU\) to \(R\), the symmetry of \(U\) implies 
\begin{equation}
\label{eq:QCA_transformed_symmetry}
    U \onsiteU_R = (\onsiteU_{\mathrm{int}_r R} \otimes \repU) U,
\end{equation}
where both sides of the equation are considered as maps from \(M\{R\}\) to \(M\{\mathrm{int}_r R\} \otimes A\) and \(\repU\) is a homomorphism \(\repU: G \to \LAut(A)\).
That is, \((A, \repU)\) is a \(G\)-rep, identified as the boundary \(G\)-rep of \((U, \onsiteU)\). 
Furthermore, the \(G\)-rep \((A,\repU)\) has a conjugacy inverse obtained via the boundary algebra of \(U\) when restricted to \(R^c\), as we will prove in \autoref{sec:InvertibleAlgebraToGQCA}.
Then \(\beta\) is the induced map between circuit classes of \(G\)-QCAs and conjugacy anomalies. 
We leave the proof that \(\beta\) is well defined to \autoref{sec:InvertibleAlgebraToGQCA}, where we simultaneously prove that it is an isomorphism.
For now, we assume these results.

\emph{The map \(\mu\)}.--- Next, we define \(\mu\).
Given any \(G\)-QCA \((U, \onsiteU)\) and any trivial model \((H_\mathrm{triv}, \onsiteU \otimes {\onsiteU}^*)\), where \({\onsiteU}^*\) is the complex conjugate representation for \(\onsiteU\), 
then we get a \(G\)-invertible model \(((U\otimes \id)[H_\mathrm{triv}], \onsiteU\otimes {\onsiteU}^*)\). 
We stack the on-site \(G\)-rep with its complex conjugate to ensure that there is at least one trivial symmetric model. For instance, the sum of on-site projectors onto the maximally entangled state between the tensor factors.
This construction induces a map $\mu$ on circuit classes.
Indeed, if \((X(U\otimes \id), \onsiteU \otimes \onsiteV)\) is circuit equivalent to \((U, \onsiteU)\), then
\begin{multline}
    \left(X(U\otimes \id)[H_\mathrm{triv} \otimes \unit + \unit \otimes H^{\text{ancilla}}_{\mathrm{triv}}], \right. \\
    \left. \onsiteU \otimes{\onsiteU}^* \otimes \onsiteV \otimes {\onsiteV}^*\right)
\end{multline}
(\(U \otimes \id\) acts as the identity on every tensor factor other than that which carries \(\onsiteU\), and \((H^{\mathrm{ancilla}}_{\mathrm{triv}}, \onsiteV\otimes {\onsiteV}^*)\) is a trivial model)
is circuit equivalent to \(((U\otimes \id)[H_\mathrm{triv}], \onsiteU \otimes{\onsiteU}^*)\), so that \(\mu\) maps circuit classes to circuit classes if $H_\mathrm{triv}$ is fixed.

Furthermore, choosing a different \(H'_{\mathrm{triv}}\) instead of \(H_{\mathrm{triv}}\) also results in a commuting model which is circuit equivalent, as there is an on-site symmetric circuit \(Y\) which swaps the ground states of \(H_{\mathrm{triv}}\) and \(H'_{\mathrm{triv}}\) and fixes the orthogonal complement of these two states in each on-site Hilbert space.
Then, dropping the tensor product with the identity and the \(G\)-rep for brevity, \(U[H_{\mathrm{triv}}]\) and \(UY[H'_{\mathrm{triv}}] = (UYU^{-1})U[H'_{\mathrm{triv}}]\)  have commuting terms with the same ground spaces, so \(U[H_{\mathrm{triv}}]\) and \(U[H'_{\mathrm{triv}}]\) are circuit equivalent. 
Thus, $\mu$ is well-defined. 
It also respects stacking of \(G\)-QCAs, so defines a homomorphism.

The map \(\mu\) is also surjective.
That is, any \(G\)-invertible commuting projector model is related by a \(G\)-QCA to a trivial model.
In the case without symmetry, this is the result of Appendix~\ref{appendix:InvertibleIffQCAPreparable}.
That appendix can be generalized to include an on-site symmetry.
Indeed, the entire calculation as presented carries through identically by replacing commuting models with \(G\)-symmetric models and finite depth circuits with symmetric circuits~\footnote{The final step in Appendix~\ref{appendix:InvertibleIffQCAPreparable} was left implicit, namely showing that if we have a local isomorphism between two on-site algebras, then those algebras can be taken to be the same. The symmetry-enriched analogue is to show that if two on-site \(G\)-reps are conjugate, then they can be taken to be the same after rearranging sites. This follows from a generalization of Ref.~\cite[Appendix~A]{Long2024}, replacing the matching of on-site Hilbert space dimensions with a matching of symmetry representations.}.

\emph{The map \(\gamma\)}.--- Finally, \(\gamma: \InvModelClass{d+1}{G} \to \CLatAnom{d}{G}\) is, similarly to \(\beta\), a symmetry-enrichment of the boundary algebra map.
Let \((H, \onsiteU)\) be a \(G\)-invertible model, and define a restriction \(\widetilde{H}\) which includes all terms with support in \(R^c = (0, \infty) \times \RR^d\).
As the terms in \(H\) have bounded support, this means that \(\widetilde{H}\) is supported in \((\mathrm{int}_r R)^c\) for some \(r\).
Then we have a projected algebra \(M(\widetilde{H})\).
As \(\widetilde{H}\) has no support in \(\mathrm{int}_r R\), we have an ultra-local isomorphism
\begin{equation}\label{eqn:CommutingModelBoundaryAlg}
    M(\widetilde{H}) \cong M\{\mathrm{int}_r R\} \otimes A',
\end{equation}
where \(A'\) is the subalgebra of \(M(\widetilde{H})\) supported in \((\mathrm{int}_r R)^c\).
That is, \(A'\) is the boundary algebra for \(\widetilde{H}\), where we take \(W = (\mathrm{int}_r R)^c\).

Similarly to  \eqnref{eq:QCA_transformed_symmetry}, this algebra inherits a representation of \(G\), and becomes a \(G\)-rep \((A', \repU')\).
This follows from the fact that \(H\) is strongly symmetric, so the boundary algebra of \(\widetilde{H}\) carries a symmetry action, defined by the action of \(\onsiteU_W\) (the restriction of the symmetry action to \(W\)).

We must show that \(\gamma\) is well-defined and a homomorphism.
In particular, that \((A',\repU')\) is conjugacy invertible, and that its conjugacy class is the same if we use a circuit equivalent \((K, \onsiteW)\).

By Appendix~\ref{appendix:AlternativeProjectedAlgebra}, stacking restricted models stacks their boundary algebras, and this extends immediately to the boundary \(G\)-reps. So, if \(\gamma\) is well-defined, it is a homomorphism.

To show that $\gamma$ is well-defined, we need to show that the conjugacy equivalence class of the boundary $G$-rep is invariant under the following operations: (a) stacking the bulk commuting model with a trivial model; (b) replacing the bulk commuting model $\{ h_i \}$ with another one $\{ h_i' \}$ such that $h_i$ and $h_i'$ have the same ground state subspace; (c) acting on the bulk commuting model with a symmetric circuit; and (d) changing the choice of restriction for the bulk commuting model. (a) and (b) are obvious. (d) is obvious for a \emph{trivial} commuting model because different choices of restriction just amount to stacking the boundary $G$-rep with on-site reps. Moreover, once we know (a), (b), and (c), then (d) will follow for an arbitrary $G$-invertible commuting model by stacking with the inverse model. Thus, it only remains to prove (c). Consider a commuting model 
\((H, \onsiteU)\), a symmetric restriction \((\widetilde{H},\onsiteU)\), and let \((X, \onsiteU)\) be a symmetric circuit.  Then \(X\) induces a local isomorphism between the projected algebras \(X: M(\widetilde{H}) \to M(X[\widetilde{H}])\). However, this map does not necessarily preserve the decomposition of \eqnref{eqn:CommutingModelBoundaryAlg}, so does not immediately define a local isomorphism between the boundary algebras. Let \((Y, \onsiteU)\) be a restriction of \((X, \onsiteU)\), such that \(Y[\widetilde{H}] = X[\widetilde{H}]\), and \(Y\) acts as the identity in \(M\{\mathrm{int}_{\xi} R\}\) for some \(\xi\), which exists as \(X\) is a symmetric circuit. Then \(Y\) defines a local isomorphism of boundary algebras:
\begin{equation}
    Y : M_{\mathrm{int}_{\xi} R}(\widetilde{H}) \to M_{\mathrm{int}_{\xi} R}(X[\widetilde{H}]).
\end{equation}
Furthermore, \(Y\) commutes with the \(\onsiteU\) action, and so acts as a conjugacy equivalence for the boundary \(G\)-reps.

As \(\gamma\) is a well-defined homomorphism from \(G\)-invertible models to conjugacy classes of \(G\)-reps, it follows that its image consists of conjugacy-invertible \(G\)-reps, as an inverse for \(\gamma([H, \onsiteU])\) is given by \(\gamma([H, \onsiteU]^{-1})\). Thus, \(\gamma\) takes values in conjugacy anomalies.

\subsubsection{Commutativity}

In this section, we will prove commutativity of \eqnref{eqn:G_QPAnom_Triangle}, taking it as given that the map $\beta : \QCAblendclass{d+1}{G} \to \CLatAnom{d}{G}$ induced by evaluating the boundary $G$-rep of a $G$-QCA is well-defined (which we later prove in \autoref{sec:InvertibleAlgebraToGQCA}).
The commutativity of \eqnref{eqn:G_QPAnom_Triangle} 
follows from showing that the boundary \(G\)-reps \((A, \repU)\) [of \((U, \onsiteU)\)] and \((A', \repU')\) [of \((U[H_{\mathrm{triv}}], \onsiteU)\)] are ultra-locally conjugate. Indeed, the commutativity of \eqnref{eqn:G_QPAnom_Triangle} requires comparing the boundary algebras of \((U, \onsiteU)\) and \(((U\otimes\id)[\widetilde{H}_{\mathrm{triv}}], \onsiteU \otimes {\onsiteU}^*)\), the latter of which we show in this subsection agrees with the boundary algebra for \((U\otimes\id, \onsiteU \otimes {\onsiteU}^*)\), which itself is related by stacking ancillas to \((U, \onsiteU)\).

We use that there is an ultra-local isomorphism (Appendix~\ref{appendix:AlternativeProjectedAlgebra})
\begin{equation}
    \eta : M\{R\} \to M(\widetilde{H}),
\end{equation}
whenever \(\widetilde{H}\) is supported in \(R^c\) and has a nondegenerate ground state on \(M\{R^c\}\) (so any restriction of a trivial model \(\widetilde{H}_{\mathrm{triv}}\), as in the definition of \(\gamma\), works).
The map is defined by sending an operator \(a \in M\{R\}\) to the equivalence class of \(a \otimes \unit_{R^c}\) in \(M(\widetilde{H})\).

Acting with \(U\) defines local isomorphisms
\begin{equation}
    U: M\{R\} \to U[M\{R\}] \cong M\{\mathrm{int}_r R\} \otimes A
\end{equation}
and
\begin{equation}
    U: M(\widetilde{H}) \to  M(U[\widetilde{H}]) \cong M\{\mathrm{int}_r R\} \otimes A'.
\end{equation}
Thus, we get a local isomorphism
\begin{equation}
    U \eta U^{-1}: M\{\mathrm{int}_r R\} \otimes A \to M\{\mathrm{int}_r R\} \otimes A',
\end{equation}
and one can check from the definitions that this isomorphism is actually ultra-local, and hence restricts to an ultra-local isomorphism
\begin{equation}
    \eta' : A \to A'.
\end{equation}

When \((U, \onsiteU)\) is a \(G\)-QCA and \((U[\widetilde{H}],\onsiteU)\) is a strongly symmetric model, then \(\eta'\) is also a conjugacy equivalence for their boundary \(G\)-reps \((A, \repU)\) and \((A',\repU')\).
Indeed, \(\repU\) and \(\repU'\) each descend from the on-site \(G\)-reps \((M\{R\}, \onsiteU_R)\) and \((M(\widetilde{H}), \onsiteU)\), which are ultra-locally conjugate through \(\eta\). 
Applying \(U\) to these \(G\)-reps maps them to \((M\{\mathrm{int}_r R\} \otimes A^{(\prime)}, \onsiteU_{\mathrm{int}_r R} \otimes \repU^{(\prime)})\). 
Then we have that \(\eta'\) conjugates \(\repU\) to \(\repU'\).
That is, the boundary \(G\)-reps for \((U, \onsiteU)\) and \((U[\widetilde{H}], \onsiteU)\) are ultra-locally conjugate.

This proof also applies to gravitational anomalies, as in \eqnref{eqn:Grav_QPAnom_Triangle}, given that lattice gravitational anomalies agree with conjugacy anomalies without symmetry (\autoref{sec:ConjAnomBlendIffConjugate}).

Note that in this section we have described the boundary $G$-rep for a $G$-QCA applied to a trivial commuting model. However, a corollary of this is that for any $G$-symmetric commuting model (not necessarily invertible), acting with a $G$-QCA has the effect of tensoring the boundary $G$-rep with the boundary $G$-rep of the $G$-QCA. Indeed, acting with a $G$-QCA gives the same $G$-symmetric commuting phase as tensoring with a trivial $G$-symmetric commuting model and acting on that with the $G$-QCA instead, as \(U\otimes \id\) is circuit equivalent to \(\id \otimes U\).

\subsubsection{Isomorphism}

We have established that \eqnref{eqn:G_QPAnom_Triangle} commutes and that \(\mu\) is surjective.
It follows from simple diagram chasing and the assumption that \(\beta\) is an isomorphism that \(\mu\) and \(\gamma\) are also isomorphisms.
Indeed, if \(\beta\) is an isomorphism, then \(\mu\) must also be injective, and thus an isomorphism. 
Then \(\gamma = \beta \mu^{-1}\) is also an isomorphism.

Thus, we have reduced the triple equivalence \eqnref{eqn:conj_G_correspondence} to showing that \(\beta\) is an isomorphism. This is established in the following Secs.~\ref{sec:GQCABlendToCircuit} and~\ref{sec:InvertibleAlgebraToGQCA}.

\subsection{Blending of \texorpdfstring{\(G\)}{G}-QCAs implies circuit equivalence}
\label{sec:GQCABlendToCircuit}

Before proving that \(\beta\) is an isomorphism, we establish an independently useful preliminary result.
We prove that the existence of a blend of \(G\)-QCAs implies circuit equivalence of those \(G\)-QCAs.

We demand that the blend in question also uses an on-site symmetry, and will exploit that we do not assume that there is a uniform bound on the local Hilbert space dimension.

We recall the notation for a blend introduced in Appendix~\ref{appendix:BlendEquivalenceDef}. If a \(G\)-rep \(\repV\) in an on-site algebra is a blend between \(\repU\) and \(\repW\) across a fixed half volume \(R\), we write \(\repU \BlendVia{\repV} \repW\). Similarly, we write \(U \BlendVia{V} W\) if \(V\) is a blend between QCAs \(U\) and \(W\).

A \(G\)-QCA \((V,\onsiteV)\) is a blend between \((U, \onsiteU)\) and \((W, \onsiteW)\) across \(R\), written
\begin{equation}
    (U, \onsiteU) \BlendVia{(V,\onsiteV)} (W, \onsiteW),
\end{equation}
if \(\onsiteU \BlendVia{\onsiteV} \onsiteW\) is a blend of \(G\)-reps and \(U \BlendVia{V} W\) is a blend of QCAs.

Suppose we have a blend \((U, \onsiteU) \BlendVia{(V,\onsiteV)} (W, \onsiteW)\) of \((d+1)\)-dimensional \(G\)-QCAs
across \(R = (-\infty,0] \times \RR^d\), where all the \(G\)-reps are on-site. We prove that \((U, \onsiteU)\) and \((W, \onsiteW)\) are circuit equivalent. It is sufficient to show that \((U \otimes W^{-1}, \onsiteU \otimes \onsiteW)\) is circuit equivalent to \((\id, \onsiteU \otimes \onsiteW)\), as circuit equivalence classes form a group under stacking. Absorbing the \(W\) factors into \(U\), we just need to prove that the existence of a blend \((U, \onsiteU) \BlendVia{(V,\onsiteV)} (\id, \onsiteU)\) implies that \((U, \onsiteU)\) is a symmetric circuit. We can even take \(\onsiteV = \onsiteU\) by including additional on-site ancillas in \(\onsiteU\).

First, we recall an argument from Ref.~\cite{Haah2023} which shows that the existence of a blend across \(R\) implies the existence of a blend across any \(R_x = (-\infty, x] \times \RR^d\). For \(x< 0\), just take the \(R = R_0\) blend and compress all sites in \((x, 0] \times \{\vec{y}\}\) to \((x,\vec{y})\). The local Hilbert space dimension grows in doing this, and grows more as \(x\) decreases, but it remains finite for any finite \(x\). 
For \(x>0\), we note that multiplying \((V^{-1},\onsiteU)\) with \((U, \onsiteU)\) provides a blend
\begin{equation}
    (\id, \onsiteU) \BlendVia{(V^{-1} U,\onsiteU)} (U, \onsiteU)
\end{equation}
across \(R\).
The same compression trick can be used to squash this to a blend across \(R_x\) for \(x >0\). Note that the range \(r\) of all these blends can be the same. We add still more ancillas to the system so that all of the blends with \(x \in 3r \ZZ\) are supported in the same \(G\)-rep, which we continue to denote \(\onsiteU\).

Denote the blend across \(R_x\) by \((V_x,\onsiteU)\). We have that
\begin{equation}
    U = (U V_0^{-1}) V_0.
\end{equation}
The \(V_0\) factor is nontrivial only on \(R_{r}\) while the \(U V_0^{-1}\) factor is nontrivial only on \(R_{-2r}^c\). This is the beginning of a circuit decomposition of \(U\)~\cite{Long2024}.

Focusing on negative \(x\), we further split up \(V_0\) as
\begin{equation}
    V_0 = V_{-3r}(V_{-3r}^{-1} V_0) = V_{-3r} \gamma_0,
\end{equation}
where \(V_{-3r}\) has support at most on \(R_{-2r}\) (and so does not touch \(U V_0^{-1}\)), and \(\gamma_0\) has support at most on \((-5r,r] \times \RR^d\). Iterate this splitting, defining \(V_{-3r} = (V_{-3r} V_{-6r}^{-1}) V_{-6r} = \gamma_{-1} V_{-6r}\) and so on (alternating whether the \(\gamma\) factor appears on the left or right) until we have a decomposition
\begin{equation}
    V_0 = \left( \prod_{n=-\infty}^{-1} \gamma_{2n+1} \right) \left( \prod_{n=-\infty}^{0} \gamma_{2n} \right)
\end{equation}
where each \(\gamma_n\) has support at most on \(((3n-5)r,(3n+1)r] \times \RR^d\). Thus, each \(\gamma_n\) has disjoint support from any other \(\gamma_m\) with \(|n-m| \geq 2\), and the expression for \(V_0\) is a finite-depth circuit of \(d\)-dimensional QCAs. Repeating the same exercise for positive \(x\), we get a similar circuit for \(U\),
\begin{equation}\label{eqn:GammaDecomposition}
    U = \left( \prod_{n=-\infty}^{\infty} \gamma_{2n+1} \right) \left( \prod_{n=-\infty}^{\infty} \gamma_{2n} \right).
\end{equation}
We also observe that each of the \(\gamma_n\) are symmetric under \(\onsiteU\), as they are products of symmetric QCAs. The first few steps of this procedure are shown in \autoref{fig:circuit_from_blend}(a).

\begin{figure}
    \centering
    \includegraphics[width=\linewidth]{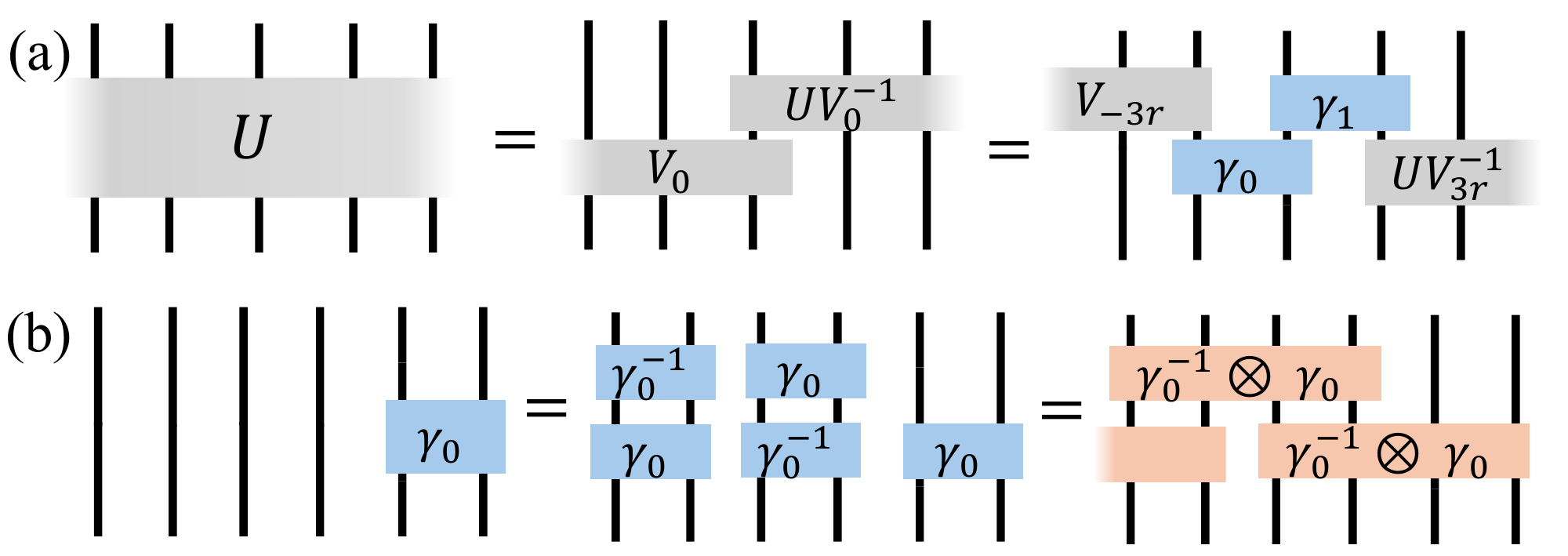}
    \caption{(a)~Given blends \(V_x\) from a \(G\)-QCA \(U\) to the identity across \((-\infty,x]\times \RR^{d}\), a depth-two circuit of \(d\)-dimensional \(G\)-QCAs \(\gamma_x\) can be constructed for \(U\). (b)~Each of the \(\gamma_x\) \(G\)-QCAs can be expressed as a \((d+1)\)-dimensional finite depth circuit composed of symmetric gates, using the fact that \(\gamma_x^{-1} \otimes \gamma_x\) is a finite depth circuit. These gates can be stacked together to form a symmetric circuit decomposition for \(U\).}
    \label{fig:circuit_from_blend}
\end{figure}

The decomposition Eq.~\eqref{eqn:GammaDecomposition} is not yet a symmetric circuit for \(U\), because the gates are actually \(d\)-dimensional QCAs. Unless \(d=0\), this is not what we mean by a symmetric circuit. However, we can exploit the fact that these \(d\)-dimensional QCAs are sitting in an ambient \((d+1)\)-dimensional space to write them as symmetric circuits using an Eilenberg swindle.

Let the restriction of \(\onsiteU\) to the support of \(\gamma_n\) be \(\onsiteU_n\). For \(n >0\), introduce further ancillas which copy \(\onsiteU_n\) on the support of \(\gamma_{n + 2m}\) for integers \(m \geq 1\). Similarly, introduce copies of \(\onsiteU_n\) for \(n\leq 0\) to the support of \(\gamma_{n - 2m}\). Focusing on \(n > 0\), we have
\begin{equation}
    \gamma_n = \gamma_n \prod_{m=1}^\infty \mathbb{T}^{6mr}[\gamma_n^{-1} \gamma_n],
\end{equation}
where \(\mathbb{T}^{6mr}\) is a translation QCA from the support of \(\gamma_n\) to the new tensor factors carrying \(\onsiteU_n\) in the support of \(\gamma_{n + 2m}\). The equality above is simply because \(\gamma_n^{-1} \gamma_n = \id\). However, we can also pair this product as [\autoref{fig:circuit_from_blend}(b)].
\begin{equation}
    \gamma_n = \prod_{m=1}^\infty \mathbb{T}^{6(m-1)r}[ \gamma_n \otimes \mathbb{T}^{6r}[\gamma_n^{-1}]] = \prod_{m=1}^\infty \nu_{nm},
\end{equation}
and \(\gamma_n \otimes \mathbb{T}^{6r}[\gamma_n^{-1}]\) is a symmetric circuit implemented by swap gates conjugated by \(\gamma_n\) The swap gates are individually symmetric because \(\onsiteU_n\) is on-site. The same procedure can be repeated for \(n\leq 0\), this time translating toward negative \(x\). All the factors \(\nu_{nm}\) obtained from this procedure can be arranged into a finite depth circuit.

The result is a finite depth circuit decomposition
\begin{equation}
    U = \left( \prod_{n=-\infty}^{\infty}  \Upsilon_{2n+1} \right) \left( \prod_{n=-\infty}^{\infty} \Upsilon_{2n} \right)
\end{equation}
where each \(\Upsilon_n\) is, for \(n>0\), a finite stack of \(\nu_{n'm}\) circuits with \(0< n'\leq n\). For \(n \leq 0\), \(\Upsilon_n\) is a finite stack of \(\nu_{n'm}\) circuits with \(n \leq n' \leq 0\).
Thus, \(U\) is a symmetric circuit.

Of course, our construction introduces many additional ancillas. Another strategy may remove some of them, but we expect that having no uniform bound on the local Hilbert space dimensions is required for circuit and blend equivalence to agree.

\subsection{Isomorphism between conjugacy anomalies and circuit classes of \texorpdfstring{\(G\)}{G}-QCAs}
\label{sec:InvertibleAlgebraToGQCA}

The proof that the boundary map \(\beta\) is an isomorphism is completed by constructing an inverse isomorphism, \(\alpha\).
The isomorphism \(\alpha\) can be interpreted through a picture of \emph{pumping}: given a conjugacy invertible \(G\)-rep, we construct a \(G\)-QCA which pumps that \(G\)-rep to its boundary. Precisely, that \(G\)-rep is the boundary \(G\)-rep of the constructed \(G\)-QCA---which is precisely the statement that \(\alpha\) is the inverse of \(\beta\).

Our method is essentially a generalization of Ref.~\cite{Haah2023} to include symmetry. A related construction, applying only to anomalies of $G$-reps acting on on-site algebras, appears in Ref.~\cite{Zhang2023entanglers}.

We now define \(\alpha\).
Recall that a \(G\)-rep \((A, \repU)\) is conjugacy invertible if there exists another pair \((\overline{A}, \overline{\repU})\) such that \((A \otimes \overline{A}, \repU \otimes \overline{\repU})\) is conjugate by a local isomorphism \(V\) to an on-site \(G\)-rep \((M, \onsiteV)\),
\begin{subequations}\label{eqn:InvertibleAlgebraTrivialize}
\begin{align}
    V : A \otimes \overline{A} &\to M \\
    V ( \repU \otimes \overline{\repU}) V^{-1} &= \onsiteV.
\end{align}
\end{subequations}

Given a \(d\)-dimensional conjugacy anomaly \([A,\repU]\), choose any representative \((A, \repU)\), any inverse \((\overline{A}, \overline{\repU})\), and a trivializing \(V\) as in Eq.~\eqref{eqn:InvertibleAlgebraTrivialize}. Define an on-site \(G\)-rep in \(\RR^{d+1}\) by
\begin{equation}\label{eqn:OnSitePumpRep}
    \onsiteU = \bigotimes_{z \in \ZZ} \onsiteV_{z} 
    = \bigotimes_{z \in \ZZ} V_z( \repU_z \otimes \overline{\repU}_z )V_z^{-1},
\end{equation}
where the \(z\) subscript indicates the position of the operators along the last dimension in \(\RR^{d+1}\), and the on-site algebra is \(\bigotimes_{z \in \ZZ} M_z\). Define a QCA \(S\) via the circuit shown in \autoref{fig:InvAlgPumpingMap}(a). Explicitly, first apply \(V_z^{-1}\) to split \((M_z,\onsiteV_z)\) into \((A_z \otimes \overline{A}_z, \repU_z \otimes \overline{\repU}_z)\). Then translate all the \((A_z, \repU_z )\) to \(z+1\), and subsequently apply \(V\) to map \((A_{z-1} \otimes \overline{A}_{z}, \repU_{z-1} \otimes \overline{\repU}_{z})\) back to \((M_z,\onsiteV_z)\). The total \(G\)-QCA \((S, \onsiteU)\) so described is indeed symmetric, but cannot necessarily be constructed from symmetric gates. The map \(\alpha\) is defined as
\begin{subequations}
\begin{align}
    \alpha: \CLatAnom{d}{G} &\to \QCAblendclass{d+1}{G} \\
    [A, \repU] &\mapsto [S, \onsiteU]
\end{align}
\end{subequations}
where square brackets denote the appropriate equivalence class.
We show below that this is a well-defined map on equivalence classes.

\begin{figure}
    \centering
    \includegraphics[width=\linewidth]{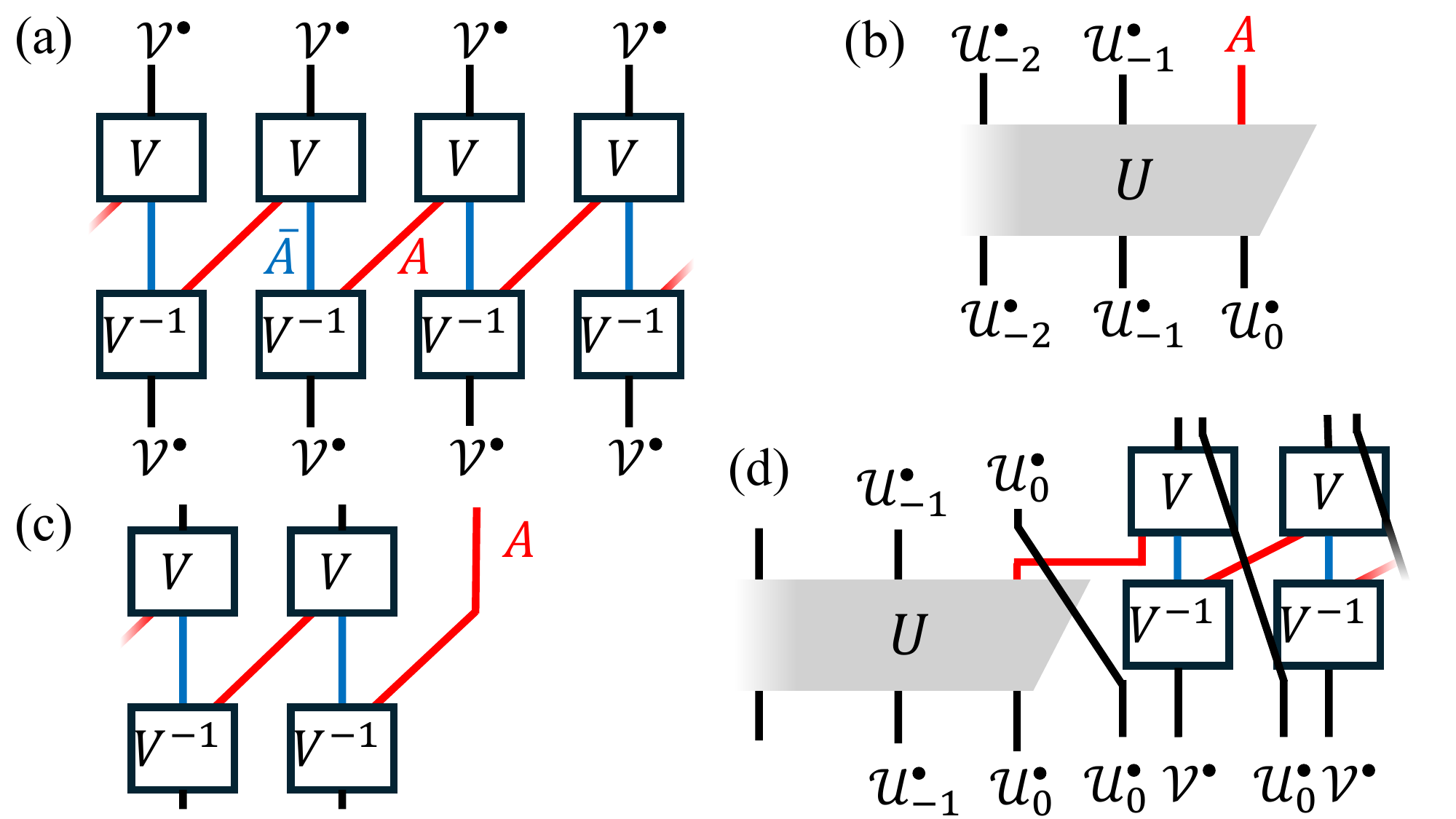}
    \caption{Illustrations of the constructions defining \(\alpha\) and \(\beta\). (a)~\emph{Defining \(\alpha\).} Construct a \(G\)-QCA from an invertible \(G\)-rep by applying the local isomorphism \(V^{-1}\) from \eqnref{eqn:InvertibleAlgebraTrivialize} to split the on-site algebra into \(A\) and \(\overline{A}\), then translate \(A\) in the \(z\) direction and apply \(V\) to recombine \(A\) and \(\overline{A}\) into an on-site algebra. (b)~\emph{Defining \(\beta\).} The boundary algebra map sends a \(G\)-QCA to a conjugacy-invertible \(G\)-rep by applying the QCA to a half volume. The boundary algebra is the part of the image algebra located near the boundary. The boundary algebra carries a \(G\)-symmetry, making it a \(G\)-rep. (c)~\emph{Calculating \(\beta\alpha\).} The boundary \(G\)-rep of (a) is the input \(G\)-rep \((A,\repU)\), up to stacking with an on-site \(G\)-rep. (d)~\emph{Calculating \(\alpha\beta\).} There is a symmetric blend between \((U,\onsiteU)\) and the result of (a) applied to the boundary \(G\)-rep of \((U,\onsiteU)\) stacked with a translation of \(\onsiteU_0\) on-site representations towards negative \(z\). This translation can be blended with the trivial \(G\)-QCA (\autoref{fig:translation_blend}), so there is a symmetric blend between these \(G\)-QCAs.}
    \label{fig:InvAlgPumpingMap}
\end{figure}

Furthermore, we claim that \(\alpha\) is an isomorphism. The inverse map 
for \(\alpha\) will be the map \(\beta\) which sends the circuit class of a \(G\)-QCA \((U, \onsiteU)\) to the conjugacy class of its boundary \(G\)-rep, as in Eqs.~(\ref{eqn:GQCABoundaryAlgebra},~\ref{eq:QCA_transformed_symmetry}). This is illustrated in \autoref{fig:InvAlgPumpingMap}(b), where we coarse-grain so that \(U\) has range \(r=1\) and lattice sites are at integer values of \(z\). Choosing a larger value of \(r\) only stacks the boundary \(G\)-rep with on-site \(G\)-reps.

Supposing for now that \(\alpha\) and \(\beta\) are well defined on equivalence classes, we can identify representatives which show that they are inverses. This is best illustrated graphically. 
\autoref{fig:InvAlgPumpingMap}(c) gives an explicit demonstration that the boundary \(G\)-rep of \((S, \onsiteU)\) is
\((A \otimes M, \repU \otimes \onsiteV)\), so that a representative of \(\beta\alpha([A, \repU])\) is just a stack of \((A, \repU)\) with an on-site \(G\)-rep, the result of which is also in \([A, \repU]\). 
\autoref{fig:InvAlgPumpingMap}(d) shows a symmetric blend between a \(G\)-QCA \((U, \onsiteU)\) and a representative of \(\alpha\beta([U, \onsiteU])\) stacked with a translation of \(\onsiteU_0\), the on-site rep at \(z=0\) in \(\onsiteU\). 
In \(d+1\geq 2\), this translation is blend equivalent with the identity so, \autoref{sec:GQCABlendToCircuit} implies that  \(\alpha\beta([U, \onsiteU])=[U, \onsiteU]\). 
Indeed, the blend of the translation to the identity is accomplished by making the translation turn a corner at the blend interface, as we illustrate in \autoref{fig:translation_blend} and explain now. 
Note first that a translation of all of \(\onsiteU_0\) along \(z\) is a stack of many one-dimensional translations along \(z\). 
Let the first lattice dimension be indexed by \(x\). 
For \(x \geq 0\), fold all one-dimensional translations at \(z=1\) so that they translate their associated on-site \(G\)-rep in the \(+x\) direction, with \(z=1\) constant. 
For \(x<0\), fold them so that they translate in the \(-x\) direction at constant \(z=1\). 
This produces a large density of translations and ancillas at large positive and negative \(x\) for \(z=1\), but the result still has finite local dimension everywhere in \(\RR^{d+1}\). 
The local dimension is not uniformly bounded, but this is still enough to apply \autoref{sec:GQCABlendToCircuit}.

\begin{figure}
    \centering
    \includegraphics[width=\linewidth]{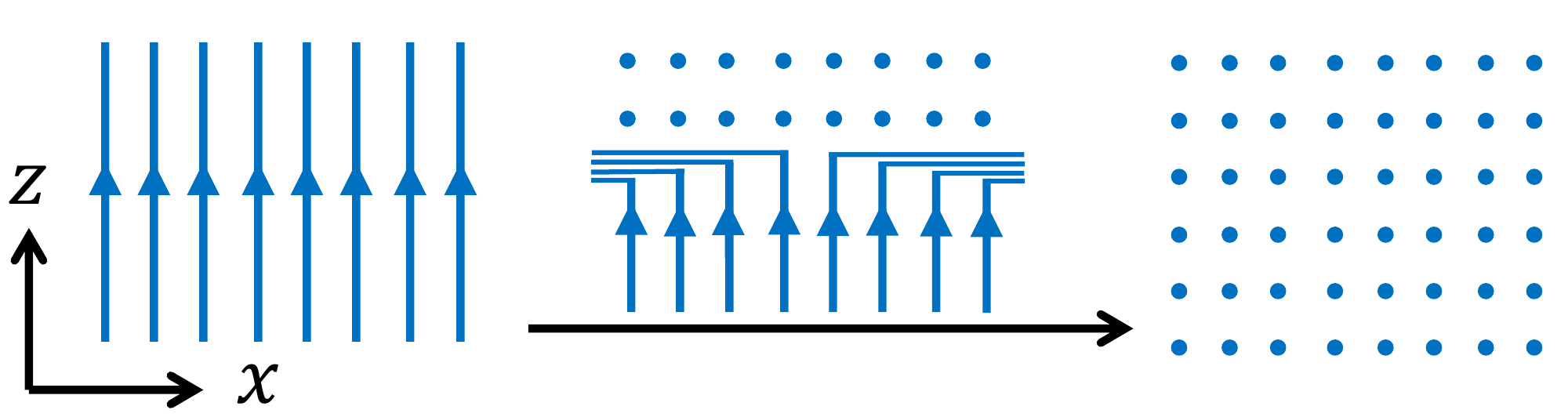}
    \caption{A blend from a translation in the \(z\) direction (left) to the identity (right) along the \(z\) axis in two or more dimensions can be constructed by turning the translations by a right angle at \(z=0\), so that they run along the \(x\) direction (middle). Note that this requires introducing additional ancillas in the \(z=0\) strip, with more ancillas being required for larger \(|x|\).}
    \label{fig:translation_blend}
\end{figure}

Now, if we show that \(\alpha\) and \(\beta\) are well-defined homomorphisms between conjugacy anomalies and circuit classes of \(G\)-QCAs, we have that they are mutual inverses, and thus isomorphisms.

\begin{figure}
    \centering
    \includegraphics[width=\linewidth]{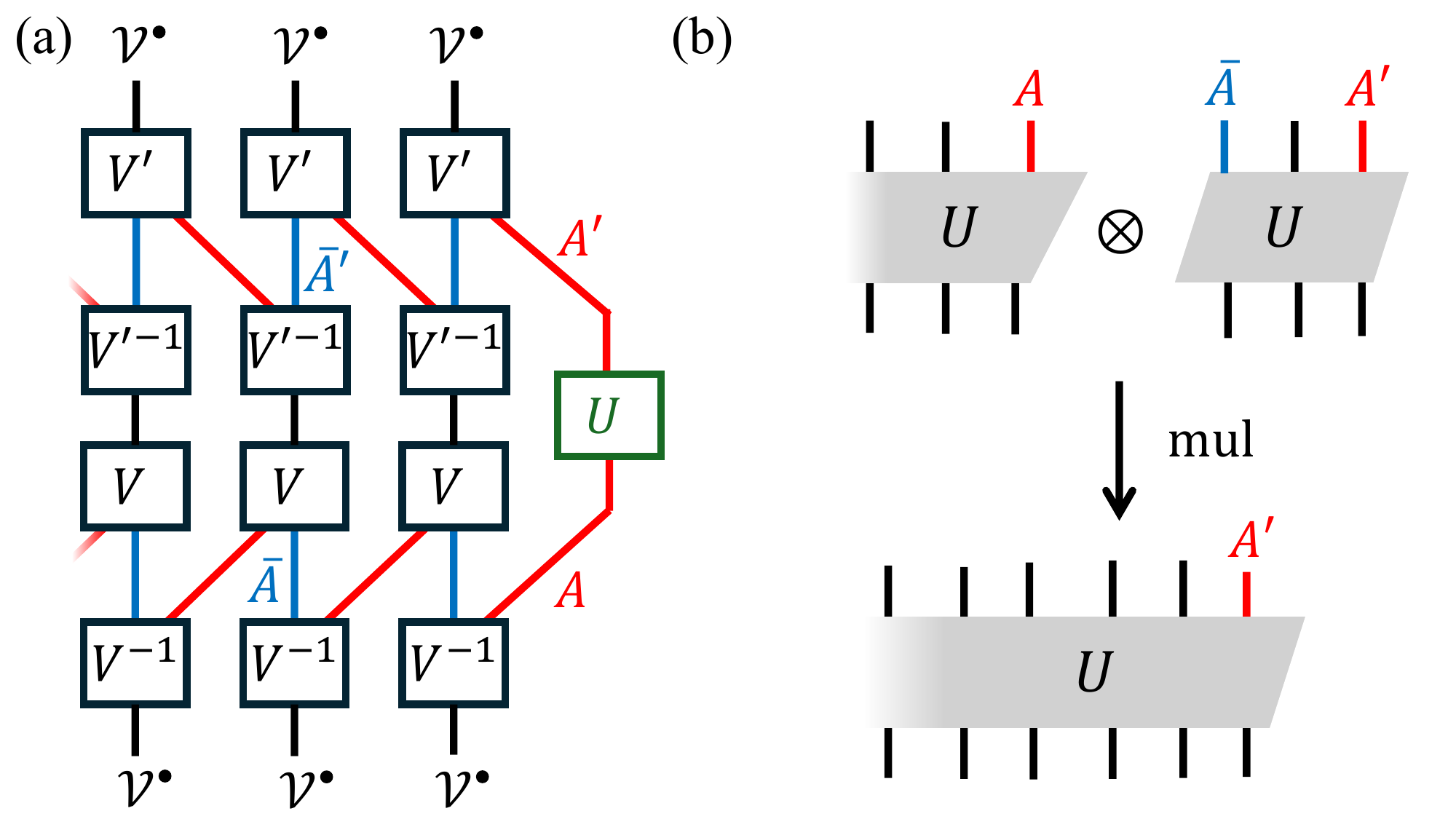}
    \caption{The constructions of \autoref{fig:InvAlgPumpingMap} induce well-defined maps \(\alpha\) and \(\beta\) between conjugacy classes of invertible \(G\)-reps and blend classes of on-site \(G\)-QCAs (and hence circuit classes). (a)~If \(S\) and \(S'\) are \(G\)-QCAs constructed from conjugate \(G\)-reps \((A,\repU)\) and \((A',\repU')\) as in \autoref{fig:InvAlgPumpingMap}(a), there is a symmetric blend from \(S'^{-1}S\) to the identity. \(U: A \to A'\) implements the conjugacy equivalence. (b)~Stacking the boundary algebra \(A\) of \((U, \onsiteU)\) with \(U[M_{(0,z]\times \RR^d}]\) does not change its conjugacy class. However, \(U[M_{(0,z]\times \RR^d}]\) consists of \(\overline{A}\) (an inverse of \(A\)) stacked with \(A'\) (the boundary algebra obtained from a blend at another position), as seen by multiplying the algebras produced from the QCA \(U\) applied to different regions, regarded as subalgebras of \(M\). Thus, the boundary algebra is a blend invariant.}
    \label{fig:InvAlgWellDefined}
\end{figure}

First we show \(\alpha\) is well-defined. Let \((A', \repU')\) be conjugate to \((A,\repU)\), via the local isomorphism
\begin{equation}
    U : (A,\repU) \to (A',\repU').
\end{equation}
Let the associated QCAs in the construction of \(\alpha\) be denoted \(S\) and \(S'\) respectively. We need that \(S'^{-1}S\) is a symmetric circuit. By \autoref{sec:GQCABlendToCircuit}, it suffices to find a symmetric blend with an on-site symmetry between \(S'^{-1}S\) and \(\id\). \autoref{fig:InvAlgWellDefined}(a) demonstrates such a blend. Note that, as a special case, this calculation shows that \(\alpha\) does not depend on the choice of inverse \(\overline{A}\) or isomorphism \(V\), as seen by taking \(A' = A\).

Now we show \(\beta\) is well-defined. That is, if \((X,\onsiteU)\) is a symmetric circuit acting on the on-site algebra \(M\), then the boundary algebra of \((XU, \onsiteU)\) is conjugate to the boundary algebra of \((U, \onsiteU)\). 
Indeed, this is an immediate consequence of the observation that the conjugacy class of the boundary algebra is independent of the coordinate \(z\) defining the blend. 
Observe first that if we define the opposite boundary algebra via
\begin{equation}
    \overline{\eta}_z: U[M\{R_z^c\}] \to \overline{A} \otimes M\{\mathrm{int}_r R_z^c\} ,
\end{equation}
then \(\overline{A}\) (as the notation suggests) is an inverse for the boundary algebra \(A\). Indeed, there is an ultra-local conjugacy equivalence
\begin{equation}
    \mathrm{mul}: A \otimes \overline{A} \to M\{R_{-r}^c \cap R_{z+r}\}
\end{equation}
induced by multiplication when regarding the boundary algebras as subalgebras of \(M\)~\cite{Haah2023}.
The same discussion applies to the boundary \(G\)-reps.
Then, in Eqs.~(\ref{eqn:GQCABoundaryAlgebra},~\ref{eq:QCA_transformed_symmetry}), if we replace \(M\{R_0\}\) with \(M\{R_z\}\) for some \(z >0\), then the new \((A',\repU')\) differs from \((A,\repU)\) only by stacking a \(G\)-rep \(\overline{A} \otimes A' \cong U[M\{R^c_0 \cap R_{z}\}]\) [\autoref{fig:InvAlgWellDefined}(b)], which we explicitly see has trivial conjugacy anomaly. 
Thus, for \((XU, \onsiteU)\), we may blend \(X \BlendVia{\widetilde{X}}\id\) with the identity near \(z> 0\) (so \(\widetilde{X}\) acts trivially for large positive \(z\)) and calculate \(A\) on the opposite side of that blend. The result is manifestly the same as that obtained with \((U, \onsiteU)\), so the boundary algebra of \((XU, \onsiteU)\) is conjugate to the boundary algebra of \((U, \onsiteU)\).

Finally, it is straightforward to see that \(\alpha\) and \(\beta\) induce homomorphisms, as their constructions directly provide that
\begin{subequations}
\begin{align}
    &\alpha([A \otimes B, \repU \otimes \repW]) = [\alpha(A , \repU ) \otimes  \alpha(B, \repW)], \text{ and } \\
    &\beta([U \otimes W, \onsiteU \otimes \onsiteW]) = [\beta(U , \onsiteU ) \otimes  \beta(W, \onsiteW)]
\end{align}
\end{subequations}
for particular choices of \(\overline{A}\) and \(V\).

This concludes the proof that the group of conjugacy anomalies is isomorphic to the group of circuit classes of on-site \(G\)-QCAs. Following on from \autoref{sec:BoundaryGRep}, it also completes the proof of the triple isomorphism \eqnref{eqn:G_QPAnom_Triangle}.

As an aside, we briefly describe the case of \(G\)-reps with \(d=0\). An appropriate notion of equivalence for zero-dimensional \(G\)-reps is simply conjugacy equivalence of finite-dimensional projective representations, with no stacking of ancillas. Equivalence classes of zero-dimensional \(G\)-reps only form a commutative monoid, and not an Abelian group, and as such there can be no isomorphism between conjugacy classes of zero-dimensional \(G\)-reps and circuit classes of one-dimensional \(G\)-QCAs. Instead, the group of circuit classes of one-dimensional \(G\)-QCAs is isomorphic to the Grothendieck group (group of fractions) for the monoid of \(G\)-reps~\cite{Long2024}.

\subsection{Conjugacy invertible \texorpdfstring{\(G\)}{G}-reps blend if and only if they are conjugate}
\label{sec:ConjAnomBlendIffConjugate}

A corollary of the results of Secs.~\ref{sec:GQCABlendToCircuit} and~\ref{sec:InvertibleAlgebraToGQCA} is that conjugacy invertible \(G\)-reps over \(\RR^d\) (\(d \geq 1\)) have a conjugacy invertible blend if and only if they are conjugate.

Suppose we have a blend of conjugacy invertible \(G\)-reps 
\begin{equation}
    (A , \repU ) \BlendVia{(C, \repV)} (B, \repW),
\end{equation}
where \((C, \repV)\) is also conjugacy invertible.
Then applying \(\alpha\) gives a blend of \((d+1)\)-dimensional \(G\)-QCAs
\begin{equation}
    (S_A , \onsiteU ) \BlendVia{(S_C, \onsiteV)} (S_B, \onsiteW)
\end{equation}
[constructed as in \autoref{fig:InvAlgPumpingMap}(a)] along an axis which is not the \(z\)-axis. However, a blend along any axis implies two on-site \(G\)-QCAs are circuit equivalent. Thus, the injectivity of \(\alpha\) implies that \((A , \repU )\) and \((B, \repW)\) are conjugate.

In the other direction, suppose \((A , \repU )\) and \((B, \repW)\) are conjugate. Then there is a blend of \(G\)-QCAs \((S_A , \onsiteU ) \BlendVia{(V, \onsiteV)} (S_B, \onsiteW)\) along any axis (as these \(G\)-QCAs are circuit equivalent), which we take to be anything other than the \(z\)-axis. Then the boundary \(G\)-rep of \((V, \onsiteV)\) along \(z\) is a blend between \((A , \repU )\) and \((B, \repW)\) up to stacking on-site reps, which can subsequently be removed away from the interface of the blend.

In particular, this result shows that lattice gravitational anomalies, defined as blend classes of invertible algebras in \autoref{sec:LatticeGravitational}, are the same as gravitational conjugacy anomalies, defined in terms of local isomorphism classes.

\subsection{Conjugacy anomalies and QFT anomalies}
\label{sec:ConjAnomToQFT}

The isomorphism between conjugacy anomalies and circuit classes of \(G\)-QCAs (or \(G\)-invertible models) is reminiscent of the connection between QFT anomalies and invertible states with \(G\) symmetry---henceforth invertible \(G\)-states. The square Eq.~\eqref{eqn:LatticeToQFTSquare} makes this analogy precise. Furthermore, it provides a mathematically well-defined link between conjugacy anomalies and QFT anomalies, or at least to an invertible \(G\)-state in one higher dimension, which is usually taken to be equivalent to a QFT anomaly. 
Given a conjugacy anomaly, apply \(e_* \alpha\) to obtain an invertible \(G\)-state. 
It is generally accepted that such an invertible \(G\)-state is described by an invertible \(G\)-TQFT, the data of which could in principle be extracted from the fusion and braiding relations (and higher-dimensional variants thereof) of symmetry defects on top of this state. This invertible \(G\)-TQFT is then identified with a QFT anomaly in one lower dimension, which is the anomaly describing the edge of the state.
Thus, we bypass the intuitively defined IR map \(\varphi'\). The only remaining technical step is to construct \(e_*\).

We define an evaluation map \(e\) from \(G\)-QCAs to invertible \(G\)-states by acting with the \(G\)-QCA on a symmetric product state~\cite{Long2024}. Such a symmetric product state always exists (after adding suitable ancillas) when the symmetry action \(\onsiteU\) is on-site. Indeed, for unitary symmetries, we can include additional ancilla degrees of freedom carrying the complex conjugate representation \(\repU^{\bullet*}\), and then the product of maximally entangled states between the on-site reps \(\onsiteU_x\) and \(\repU^{\bullet*}_x\) is symmetric. For antiunitary symmetries, decompose \(\onsiteU_g = \mathfrak{U}^\bullet_g K^{\sigma(g)}\) where \(K\) is on-site and antiunitary (acting as complex conjugation in some basis) and \(\mathfrak{U}^\bullet_g\) is unitary. Then \((\mathfrak{U}^\bullet\otimes \mathfrak{U}^{\bullet*}) K^{\sigma}\) also has a symmetric product state, given by the maximally entangled state between the two factors with real coefficients in the basis fixed by \(K\).

Denote the symmetric product state constructed in the previous paragraph by \(\ket{0}\). The map \(e\) is defined as
\begin{equation}
    e(U, \onsiteU) = (U\ket{0}, \onsiteU ),
\end{equation}
where the pair \((U\ket{0}, \onsiteU )\) is an invertible \(G\)-state.
Indeed, being more precise, an invertible \(G\)-state in an on-site algebra is a pair \((\ket{\psi}, \onsiteU)\) where \(\ket{\psi}\) is a state in an on-site algebra \(M\), \(\onsiteU\) is an on-site \(G\)-rep on \(M\), and we require that \(\onsiteU_g \ket{\psi} = \ket{\psi}\) for all \(g \in G\) and that there exists a \((\ket{\overline{\psi}}, \overline{\repU}^\bullet)\) such that \((\ket{\psi} \otimes \ket{\overline{\psi}}, \onsiteU \otimes \overline{\repU}^\bullet)\) is related by a finite depth circuit (and stacking ancillas) to the state \(\ket{0}\).

The evaluation map \(e\) induces a well-defined map \(e_*\) between circuit classes of \(G\)-QCAs and circuit classes of invertible \(G\)-states. Indeed, \(e(XU, \onsiteU) = (XU\ket{0}, \onsiteU)\) is explicitly circuit equivalent to \((U\ket{0}, \onsiteU)\). It also maps a tensor product of \(G\)-QCAs to a tensor product of states, and so induces a homomorphism between circuit classes,
\begin{equation}
    e_* : \QCAblendclass{d+1}{G} \to \invblendclass{d+1}{G}.
\end{equation}

We now argue that the square Eq.~\eqref{eqn:LatticeToQFTSquare} commutes. We remark that the lower horizontal map \(\varphi'\) has a more intuitive definition in terms of descent of anomalies to the IR QFT defining the low-energy degrees of freedom of a symmetric Hamiltonian \(H\). As such, our argument is similarly intuitive. 

Begin with a conjugacy invertible \(G\)-rep \((A, \repU)\) and let \(H\) be any Hamiltonian (not necessarily commuting) in \(A\) such that \(\repU_g[H] = H\) for all \(g \in G\). We understand \(\varphi'([A,\repU])\) to be the anomaly describing the low energy degrees of freedom of \(H\). We aim to show that \(H\) is an edge Hamiltonian for the invertible \(G\)-state \(e(S_A, \onsiteU)\) (a representative of \(e_*\alpha([A, \repU])\)), and so \(\varphi'([A,\repU])\) should describe the QFT anomaly at the edge of this state. Indeed, letting \(H_{\mathrm{triv}}\) be a sum of on-site projectors each annihilating the single-site states making up \(\ket{0}\), we observe that \(U [H_{\mathrm{triv}}]\) is a symmetric commuting model for \(e(S_A, \onsiteU)\). 
This commuting model is that appearing in the construction of the isomorphism \(\mu\) from \(G\)-QCAs to \(G\)-invertible models.
The boundary \(G\)-rep of \(S_A [H_{\mathrm{triv}}]\) is locally isomorphic to \((A,\repU)\) by \eqnref{eqn:G_QPAnom_Triangle}, and so can support the symmetric Hamiltonian \(H\).

This construction can be used to show that:
\begin{quote}
    A conjugacy anomaly \([A, \repU]\) is in \(\ker e_* \alpha\) if and only if there is a \(G\)-invertible state \((A, \repU, \ket{\psi})\).
\end{quote}

A symmetric state $\ket{\psi}$ in an on-site $G$-rep $(A,\onsiteU)$ is said to be $G$-invertible if there exists an on-site algebra $(B,\onsiteV)$ and a finite-depth circuit $U$ on $A \otimes B$ composed of gates symmetric under $\onsiteU \otimes \onsiteV$ such that $U(\ket{\psi} \otimes \ket{\phi})$ is a product state. (This should be contrasted with ``invertible $G$-state'', in which $U$ is required to be a finite-depth circuit, but is not required to commute with the symmetry.)

We extend this to a general conjugacy-invertible $G$-rep $(A, \repU)$ as follows: a state $\ket{\psi}$ is a \(G\)-invertible state if it is symmetric under $\repU$, and there is another $G$-rep \((\overline{A}, \overline{\repU})\) and a state $\overline{\ket{\psi}}$ symmetric under $\overline{\repU}$ such that \((\overline{A}, \overline{\repU})\) is a conjugacy inverse for \((A, \repU)\) demonstrated via the local isomorphism \(V\) [that is, $V$ applied to $( A \otimes \overline{A}, \repU \otimes\overline{\repU})$ maps it to an on-site $G$-rep], and further that \(V(\ket{\psi} \otimes \ket{\overline{\psi}})\) is a product state. (Note that a representative of a conjugacy anomaly possessing such a \(G\)-invertible state implies all representatives have such a state, constructed by applying the conjugacy equivalence to \(\ket{\psi}\).)
The analogous definition for an invertible $G$-state would only require that \(V(\ket{\psi} \otimes \ket{\overline{\psi}})\) is a product state for some local isomorphism that maps $V$, $A \otimes \overline{A}$ to an on-site algebra, without any requirement on how $V$ maps $\repU \otimes \overline{\repU}$.

We expect that \(G\)-invertible states are in fact the same as invertible \(G\)-states. (Note that it would be sufficient to prove this for on-site $G$-reps, and then the result for general conjugacy-invertible $G$-reps would follow as a corollary.) Indeed, it is inherent in all the standard proposed classifications of $G$-symmetric phases of invertible states that the classification forms a group (for example, this is true under the conjecture that the classification is given by a generalized cohomology theory, see \autoref{sec:HomotopyTheory}), which means that all the phases are $G$-invertible.
(However, we are not aware of any rigorous proof on the lattice.)
The existence of an invertible \(G\)-state for a conjugacy anomaly generalizes the notion of IR trivial anomaly introduced in \autoref{subsec:ir_trivial} to anomalies with a nontrivial gravitational part.
Thus, assuming that \(G\)-invertible states and invertible \(G\)-states are the same, we have the following rephrasing of our stated result:
\begin{quote}
    A conjugacy anomaly \([A, \repU]\) is in \(\ker e_* \alpha\) if and only if it is IR trivial.
\end{quote}

We now proceed to the proof of the original statement. Suppose that \([A,\repU]\) is in \(\ker e_* \alpha\). That is, that \(e(S_A, \onsiteU) = X \ket{0}\), where \((X, \onsiteU)\) is a symmetric circuit (possibly after adding ancillas). Then we have that \(X^{-1} S_A\) fixes \(\ket{0}\),
\begin{equation}
    X^{-1} S_A \ket{0} = \ket{0}.
\end{equation}
Now let \((X^{-1}S_A)|_R: M\{R\} \to M\{\mathrm{int}_\xi R\} \otimes A'\) be a restriction of \(X^{-1}S_A\) to \(R = (-\infty,0]\times \RR^d\). The boundary algebra \(A'\) is locally isomorphic to \(A\), by the fact that \(\beta\) is the inverse of \(\alpha\). Let the isomorphism be \(V: A' \to A\). Then acting with the restriction \((\id \otimes V)(X^{-1}S_A)|_R\) on \(\ket{0}_R\) (the restriction of \(\ket{0}\) to \(R\)) gives
\begin{equation}
    (\id \otimes V)(X^{-1}S_A)|_R \ket{0}_R = \ket{0}_{\mathrm{int}_\xi R} \otimes \ket{\psi},
\end{equation}
where we used that \((X^{-1}S_A)|_R\) fixes \(\ket{0}\) far from the boundary, and \((A, \repU, \ket{\psi})\) is an \(G\)-invertible state. Indeed, an inverse may be constructed by considering a restriction of \(X^{-1}S_A\) to \(R^c\). 

Conversely, suppose that there is a \(G\)-invertible state \((A, \repU, \ket{\psi})\). 
Let an inverse for \(\ket{\psi}\) be \((\overline{A}, \overline{\repU}, \ket{\overline{\psi}})\), and construct \((S_A, \onsiteU)\) as in \autoref{fig:InvAlgPumpingMap}(a) using \((\overline{A}, \overline{\repU})\) and the conjugacy equivalence \(V: A\otimes \overline{A} \to M\) such that 
\begin{equation}
    V (\ket{\psi} \otimes \ket{\overline{\psi}}) = \ket{\emptyset},
\end{equation}
where \(\ket{\emptyset}\) is a symmetric product state in \((M,\onsiteV)\). We see that \(S_A \ket{1} = \ket{1}\) fixes the product state \(\ket{1} = \bigotimes_{z \in \ZZ} \ket{\emptyset}\). 
On each site, there is a symmetric gate in \((M, \onsiteU)\) which swaps the state in \(\ket{0}\) and \(\ket{1}\), and fixes the orthogonal complement to this subspace. Let \(Y\) be the on-site symmetric circuit which is a tensor product of these gates. Then we have
\begin{equation}\label{eqn:OnSiteConjFix0}
    Y S_A Y^{-1} \ket{0} = \ket{0}.
\end{equation}
Define \(X = Y S_A Y^{-1} S_A^{-1}\), which is a symmetric circuit composed of gates from \(Y\) and \(S_A\) conjugating gates from \(Y\). From \eqnref{eqn:OnSiteConjFix0}, we have \(X S_A \ket{0} = \ket{0}\), so that \([A, \repU] \in \ker e_* \alpha\).
This completes the proof.

More basic properties of \(\varphi'\) can also be re-expressed in terms of other paths around the square~\eqref{eqn:LatticeToQFTSquare}. In particular, recall that \(\varphi'\) is, in general, neither surjective nor injective.  The vertical maps in the diagram \eqref{eqn:LatticeToQFTSquare} \emph{are} injective---indeed, they are isomorphisms (or at least \(\beta\) is). However, the map from \(G\)-QCAs to invertible \(G\)-states, \(e_*\), is neither surjective nor injective~\cite{Long2024}, and thus neither is \(\varphi'\).

\subsection{From conjugacy anomalies to blend anomalies}
\label{subsec:ConjugacyToBlend}

So far, this section has been concerned with conjugacy anomalies. Based on the isomorphisms of \eqnref{eqn:conj_G_correspondence}, this is the correct setting to understand the link between \(d\)-dimensional anomalies and \((d+1)\)-dimensional \(G\)-QCAs with on-site symmetries. However, conjugacy anomalies and conjugacy-invertible \(G\)-reps are too restrictive to capture many cases of interest. For instance, the translation \(\ZZ\)-rep in any dimension is not conjugacy invertible, and thus there can be no Lieb-Schultz-Mattis type connections for conjugacy anomalies~\cite{Cheng2016Tranlsation,Cho2017LSMAnomaly,Jian2018LSMSPT,Cheng2023LatticeAnomalyMatching}. Moreover, the invariants constructed in \autoref{sec:CohomologyInvariants} were blend invariants, not only conjugacy invariants, and we can only prove the obstruction results in \autoref{sec:ConsequencesOfAnomalies} for a non-trivial blend class.
It is for this reason that the majority of this paper is concerned with anomalies defined by blend equivalence, rather than conjugacy.

Here, we discuss the extension of a bulk-boundary correspondence to anomalies as used in the rest of this work. For the sake of clarity, we refer to these as blend anomalies to distinguish them from conjugacy anomalies.

As a first observation, we comment that gravitational conjugacy anomalies are the same as gravitational blend anomalies. Taking \(G=1\), a gravitational anomaly in both cases is just an invertible local algebra. \autoref{sec:ConjAnomBlendIffConjugate} shows that blending of such algebras defines the same equivalence relation as local isomorphism stabilized by stacking on-site algebras, so there is no distinction to be made in this case. 

Furthermore, recent results imply that blend anomalies of one-dimensional bosonic \(G\)-reps with finite \(G\) also agree with conjugacy anomalies. 
Reference~\cite{Seifnashri2025disentangling} shows that, for this class of symmetries, two \(G\)-reps in an on-site algebra blend if and only if they are conjugate.
An immediate consequence is that all such \(G\)-reps are conjugacy invertible: this follows from the observation that all \(G\)-reps in an on-site algebra have a blend inverse, given by the reflection of the \(G\)-rep through the blend interface.
Thus, the symmetry-protected conjugacy anomalies (those defined in on-site algebras) of finite bosonic symmetries are the same as symmetry-protected blend anomalies.
The factorization between symmetry-protected and gravitational anomalies in \eqnref{eq:lat_anom_decomposition} then implies that all conjugacy anomalies of finite bosonic symmetries in \(d=1\) agree with blend anomalies, as we already know that gravitational anomalies do not distinguish between blends and conjugacy.
It seems plausible that conjugacy and blend anomalies always coincide for finite groups, even in higher dimensions or for antiunitary symmetries.

More generally, there is a homomorphism
\begin{equation}
    q: \CLatAnom{d}{G} \to \LatAnom{d}{G}
\end{equation}
defined by sending a conjugacy class to the blend class of a representative. As all invertible \(G\)-reps in the same conjugacy class also blend (by \autoref{sec:ConjAnomBlendIffConjugate}), \(q\) is well-defined.

Furthermore, as QFT anomalies are expected to be blend invariants, we get the same map from conjugacy anomalies to QFT anomalies if we first map to blend classes, and only then to the IR QFT anomaly, \(\varphi' = \varphi q\). As a commuting diagram,
\begin{equation}
    \begin{tikzcd}
     \mathrm{CAnom}^{\mathrm{lattice}}_d(G) \arrow[r,"\varphi^\prime"] \arrow[d,swap,"q"] 
     &\QFTAnom{d}{G} \\
     \LatAnom{d}{G} \arrow[ur,"\varphi"] & 
    \end{tikzcd}.
    \label{eqn:CAnomToBlend}
\end{equation}
This can be attached to the bottom of the square in Eq.~\eqref{eqn:LatticeToQFTSquare}, resulting in another commutative square
\begin{equation}
\begin{tikzcd}
 \QCAblendclass{d+1}{G} \arrow[r,"e_*"] \arrow[d,swap,"q \beta"] 
 &\invblendclass{d+1}{G} \arrow[d,"\text{edge}"] \\
 \LatAnom{d}{G} \arrow[r,"\varphi"] & \QFTAnom{d}{G}
\end{tikzcd},
\label{eqn:BlendSquare}
\end{equation}
where the left vertical map \(q \beta\) is no longer an isomorphism for all \(G\). 

Describing \(q\beta\) as defining a bulk-boundary correspondence is problematic. There is no clear way to construct either a left or right inverse for this map, meaning that there is no clear way to send a general blend anomaly to a bulk \(G\)-QCA. This occurs due to the restriction to on-site symmetries for \(G\)-QCAs. While this is natural on the \((d+1)\)-dimensional side, it essentially forces us to consider conjugacy anomalies on the \(d\)-dimensional side. Relaxing the requirement of on-site symmetries, and considering a more general class of \(G\)-QCAs, we can construct an injective map from blend anomalies to blend classes of \(G\)-QCAs which acts as a right inverse for the boundary algebra map. The construction is very similar to that presented here, and so is left to Appendix~\ref{sec:BlendBulkBoundary}. We comment that this alternative construction also produces a commuting square similar to Eq.~\eqref{eqn:LatticeToQFTSquare}.

\section{RG invariance of lattice anomalies}
\label{sec:rg_invariance}
As we have seen, passing from lattice systems to the effective QFT description does not necessarily preserve the anomaly, since lattice anomalies can map to trivial QFT anomaly. However, as we previously argued in \autoref{sec:IntroLatticeToQFT}, going from a lattice model to a QFT cannot strictly be viewed as RG. Instead, one should consider RG transformations that remain within the setting of lattice models.

In particular, in this section, we argue that lattice anomalies are preserved under the entanglement renormalization of Ref.~\cite{Vidal2008MERA}. Entanglement renormalization consists of three stages: a unitary disentangling, a projection to reduce the number of degrees of freedom, and a rescaling. The unitary disentangling manifestly preserves the lattice anomaly, since the $G$-rep simply gets conjugated by a finite-depth circuit. We will now discuss the other two steps.

First, we consider the effect of rescaling
sites in space such that a site located at position $\mathbf{x} \in \mathbb{R}^d$ moves to position $s \mathbf{x}$, where $s > 0$ is a rescaling factor. 
It seems very likely that lattice anomalies are invariant under such a procedure.
Roughly, one expects this to hold since the classification of lattice anomalies is expected to be discrete, so that rescalings that can be implemented continuously cannot lead to a change in the anomaly. 
One can also give more rigorous proofs in certain cases. For example, the result from \autoref{sec:GQCABlendToCircuit} that blend equivalence is the same as circuit equivalence for QCAs implies that the rescaling of a QCA is always blend/circuit equivalent to the original QCA, which implies that all the cohomological invariants discussed in \autoref{sec:CohomologyInvariants} are scale-invariant.
Thus, the lattice anomaly of \autoref{sec:em-qubit} does not flow to zero under entanglement renormalization, despite being IR trivial.
A similar argument for $G$-QCAs, combined with the bulk-boundary correspondence results of \autoref{sec:BulkBoundaryCorrespondence}, shows that the conjugacy anomalies discussed in \autoref{sec:BulkBoundaryCorrespondence} are also scale-invariant. In particular, since conjugacy equivalence implies blend equivalence, this result implies scale invariance of the blend equivalence class for conjugacy-invertible anomalies. This includes all lattice gravitational anomalies.

Finally, we discuss the effect on the lattice anomaly of the projection step. For our purposes it is helpful to phrase the projection step in the language of the projected algebra introduced in \autoref{sec:CommutingModelsBoundaryAlgebra}. Let $A$ be an on-site algebra, and let $\{ h_i : i \in I \}$ be a commuting model in $A$ (we could just take the $h_i$'s to be projectors). Moreover, let us assume this commuting model is strongly symmetric, i.e.\ each $h_i$ is symmetric under the $G$-rep. 
This is necessary in order to ensure that the RG transformation we are constructing is compatible with the symmetry~\footnote{One could also allow the $h_i$s to get permuted amongst themselves by the symmetry. For internal symmetries, this would not affect the conclusions. However, in the case of non-internal symmetries---such as translation symmetry---where the orbits of the permutation are not supported on bounded regions, the lattice anomaly of the permuting symmetry need not be preserved under the projection.}.
Then a way to coarse-grain, i.e. reduce the number of degrees of freedom, is to replace $A$ with the projected algebra $A(H)$, where $H = \sum_i h_i$. Moreover, due to the symmetry condition we imposed, we see that $\repU$ induces a $G$-rep on $A(H)$, which we call \(\repU_{A(H)}\). Thus, the effect of coarse-graining on the $G$-rep consists of replacing $(A,\repU)$ by $(A(H),\repU_{A(H)})$.

Now, we say that a commuting model $\{ h_i \}$ on $\mathbb{R}^d$ defines a \emph{separated projection} if there exists an $r$ such that each $h_i$ is supported on a set $X_i \subseteq \mathbb{R}^d$ of diameter $\leq r$, and moreover $X_i$ and $X_j$ are disjoint for $i \neq j$. One can verify that the projection step in the conventional formulation of entanglement renormalization~\cite{Vidal2008MERA} precisely corresponds to a separated projection. 

Our key result is that
\begin{quote}
    If $\{ h_i \}$ defines a separated projection, then $(A(H), \repU_{A(H)})$ carries the same lattice anomaly as $(A,\repU)$.
\end{quote}
To see this, first note that $A(H)$ is clearly locally isomorphic to an on-site algebra. Moreover, if we consider a half volume $R$ and then define $\widetilde{H} = \sum_{i \in I : X_i \subseteq R} h_i$, then $(A(\widetilde{H}), \repU_{A(H)})$ defines a blend from $(A,\repU)$ to $(A(H), \repU_{A(H)})$.

We can extend this result to include lattice gravitational anomalies if we allow $A$ to be a general invertible local algebra. In this case, the condition of a separated projection needs to be replaced with a condition we call \emph{short-range entangling projection}. A commuting model $\{ h_i \}$ in an invertible local algebra $A$ is said to define a short-range entangling projection if there exists a local algebra $\overline{A}$ and a local isomorphism $\eta : A \otimes \overline{A} \to M$, where $M$ is an on-site algebra, such that $\{ \eta(h_i \otimes \unit) \}$ defines a separated projection. Then one can show that
\begin{quote}
    If $\{ h_i \}$ defines a short-range entangling projection, then $(A(H), \repU_{A(H)})$ carries the same lattice anomaly as $(A,\repU)$.
\end{quote}
To see this note that if we define $\widetilde{H}$ as before, then we have the local isomorphisms
\begin{equation}
A(\widetilde{H}) \otimes \overline{A} \cong (A \otimes \overline{A})(\widetilde{H} \otimes \unit) \cong M(\eta(\widetilde{H} \otimes \unit)),
\end{equation}
which is locally isomorphic to an on-site algebra. Hence we see that $A(\widetilde{H})$ is an invertible local algebra, and $(A(\widetilde{H}), \repU_{A(\widetilde{H})})$ defines a blend from $(A,\repU)$ to $(A(H),\repU_{A(H)})$.

One might ask what is the physical reason to restrict to separated/short-range-entangling projections. The problem is that a non-short-range-entangling projection is too drastic of an operation to deserve to be called an RG transformation. To illustrate this, consider a system in two spatial dimensions with symmetry $G = \mathbb{Z}^2 \times \mathrm{SO}(3)$, where the $\mathbb{Z}^2$ symmetry acts as translation symmetry, and the $\mathrm{SO}(3)$ symmetry acts on-site, but with half-integer spin per unit cell. It is well known that such a system carries a lattice anomaly, and in fact this lattice anomaly maps to a non-trivial QFT anomaly. Nevertheless, it is possible to construct a symmetric commuting projector model whose ground state is in the topological phase of the $\mathbb{Z}_2$ toric code~\footnote{Start with a loop-condensate model for the toric code with the constraint that there must be an odd number of occupied edges incident on each vertex (in order to realize the required translation symmetry fractionalization), and decorate occupied edges with \(\mathrm{SO}(3)\) SPTs.}. The fact that the IR TQFT carries an anomaly of the $G$ symmetry can be attributed to the symmetry fractionalization~\cite{Zaletel_1410}. Now, if we were allowing non-short-range entangling projections, one could simply project into the ground state of the commuting projector model, in which case the coarse-grained algebra $A(H)$ would contain only scalar multiples of the unit---such an algebra cannot carry a non-trivial $G$-rep. Hence, the lattice anomaly upon coarse graining becomes trivial. Because the original lattice anomaly corresponded to a non-trivial QFT anomaly, this shows that not only would such a coarse-graining not preserve the lattice anomaly, it would not even preserve the QFT anomaly. Thus, we should refrain from allowing such coarse-grainings. This is analogous to the situation in QFT, where integrating out gapped TQFTs may not preserve the anomaly and is not a legitimate RG operation.

\section{Towards classification of lattice anomalies: connections with TQFTs}
\label{sec:GrepFromGTQFT}

So far, we have provided particular invariants associated to \(G\)-reps, which can be regarded as partially classifying lattice anomalies. 
The full classification of lattice anomalies---in the sense of a computable presentation of all the groups \(\LatAnom{d}{G}\)---is not yet within reach. 
To move towards a complete classification of lattice anomalies, we require a more systematic framework for identifying obstructions to blending two \(G\)-reps. This section and \autoref{sec:HomotopyTheory} provide partial progress towards this goal through two complementary perspectives, and present conjectures regarding the structure of the full classification.

In this section we present a conjecture for the classification of a large subgroup of lattice anomalies which extends the conjectured classification of QCAs. Namely, we propose that there is an injective map from the \emph{Witt group} of \(G\)-TQFTs to the group of lattice anomalies in the same dimension.

The Witt group of \(G\)-TQFTs is formed by considering \(d\)-dimensional symmetry-enriched topological orders  modulo an equivalence relation which identifies those phases which can admit gapped symmetric interfaces between them, possibly after stacking with an invertible phase (so all SPTs are regarded as trivial). Then our conjecture is formulated as
\begin{equation}
    \LatAnom{d}{G} \geq \mathrm{Witt}(G\text{-TQFT}_d).
    \label{eqn:LatticeAnomalyClassification}
\end{equation}
We remark that this conjecture is manifestly compatible with the statement that there should be a homomorphism from $\LatAnom{d}{G}$ to $\QFTAnom{d}{G}$. Indeed, there is evidently a map from $\mathrm{Witt}(G\text{-TQFT}_d)$ to $\QFTAnom{d}{G}$ since in any $G$-TQFT, one can compute its QFT anomaly, and moreover $G$-TQFTs that support a symmetric gapped interface must have the same $G$-anomaly.

As a concrete example of the conjecture, consider the case $d = 2$, $G = \mathbb{Z}_2$. Then there is a $G$-TQFT which is the IR theory of the toric code, in which the symmetry exchanges $e$ and $m$ excitations. If we forget about the symmetry, it is easy to have a gapped boundary by condensing $e$ or $m$ at the boundary. But it is impossible to make this boundary symmetric (without allowing the symmetry to become spontaneously broken, which one presumably should not allow in defining Witt classes) because the symmetry will always exchange $e$- and $m$-condensed boundaries. Thus, this $G$-TQFT is in the non-trivial Witt class. This would correspond to the lattice anomaly of the $\mathbb{Z}_2$-rep in two dimensions described in \autoref{sec:em-qubit}.

Our conjecture can be viewed as a generalization of a previous proposal that QCAs without any symmetry in \(d+1\) dimensions are classified, up to blend equivalence, by the Witt group of TQFTs in \(d\) dimensions \cite{Haah2022,Haah2023,Shirley2022}. TQFTs form an Abelian monoid under stacking, and the Witt group is the quotient of this monoid formed by identifying those phases which can share a gapped boundary. This is essentially a quotient by blending. The result is a group, because the reflection of a TQFT functions as an inverse element, as seen by a similar folding construction used to find the blend inverse of a \(G\)-rep in Appendix~\ref{appendix:BlendEquivalenceDef}. Now, given any TQFT in \(d\) dimensions, there is a commuting Walker-Wang Hamiltonian \cite{Walker_1104} in \(d+1\) spatial dimensions constructed from the data of that TQFT~\footnote{The original Walker-Wang construction was only in 3 spatial dimensions. However, presumably there should be generalizations in any dimension.}. It is conjectured (and known to be true in some examples) that these Hamiltonians can be prepared from a trivial Hamiltonian (a sum of terms with disjoint support) by a QCA. Thus, we obtain a map from \(d\)-dimensional TQFTs to the QCA which prepares their \((d+1)\)-dimensional Walker-Wang model.

An issue with this argument is that it is not particularly obvious whether any Walker-Wang model can actually be realized from a trivial Hamiltonian by a QCA.
That is, whether the Walker-Wang model, viewed as a commuting model in the sense of previous sections, will be invertible.
Here we give a slightly modified version of the argument in terms of boundary algebras rather than relying on the invertibility of the bulk commuting model. Since the Walker-Wang model is a commuting model, one can define its boundary algebra as defined in \autoref{sec:CommutingModelsBoundaryAlgebra}. This defines a monoid homomorphism from the monoid of TQFTs to the monoid of local algebras. It seems likely (though we will not prove it here) that the boundary algebras of two Walker-Wang models corresponding to TQFTs in the same Witt class are equivalent.
In particular, since Witt classes form a group, which has inverses, it follows that the boundary algebra of any Walker-Wang model is invertible (which may not necessarily imply that the models themselves are invertible).
Moreover, we obtain a group homomorphism
\begin{equation}
    \phi : \mathrm{Witt}(\text{TQFT}_d) \to \LatGravAnom{d}
\end{equation}
from the Witt group into the group of lattice gravitational anomalies---that is, equivalence classes of invertible local algebras in $d$ dimensions, which is isomorphic to \(\QCAblendclass{d+1}{}\).

Meanwhile, we can also construct a map in the other direction as follows. Suppose that any invertible local algebra supports at least one commuting Hamiltonian. The ground state of this Hamiltonian could support topological order which could not occur in a commuting model on a tensor product algebra. For instance, two-dimensional topological orders with nontrivial topological central charge can be realized with commuting Hamiltonians in an invertible local algebra~\cite{Haah2022,Haah2023}. If there are two different topological orders that can occur in the same invertible local algebra, then one postulates that they will necessarily be Witt-equivalent (although it not entirely clear how to prove this). Thus, one obtains a homomorphism
\begin{equation}
    \psi :  \LatGravAnom{d} \to \mathrm{Witt}(\text{TQFT}_d)
\end{equation}
and one can argue that $\phi$ and $\psi$ should be inverses of each other.

We propose a natural extension of this conjecture to \(G\)-anomalies. One expects that there should be a symmetry-enriched version of the Walker-Wang construction that takes the data of the $G$-TQFT in $d$ spatial dimensions and constructs a commuting projector model with on-site $G$ symmetry in $d+1$ spatial dimensions. The bulk symmetry will induce automorphisms of the boundary algebra of this Walker-Wang model.
We hypothesize that the boundary algebra of the $G$-enriched Walker-Wang model should be equivalent to the boundary algebra of the Walker-Wang model built from the underlying TQFT (forgetting about the $G$ symmetry), and in particular is invertible. Hence, we have a $G$ action on an invertible local algebra by local automorphisms, which is the most general definition of $G$-rep~(\autoref{sec:LatticeGravitational}). Moreover, we expect that if we have two $G$-TQFTs in the same Witt class, then the corresponding $G$-reps will be in the same blend equivalence class. Thus we obtain a homomorphism
\begin{equation}
    \phi :  \mathrm{Witt}(G\text{-TQFT}_d) \to \LatAnom{d}{G}.
\end{equation}

Constructing the map going the other way takes a bit more care than in the case without symmetry (actually, it is possible that that case also has the same issues and they have just not been noticed previously). Naively we just want to consider commuting models that are symmetric under the $G$-rep, but this does not uniquely determine a $G$-Witt class. For example, in $d=2$, with toric code topological order and an on-site $\mathbb{Z}_2$ symmetry, the symmetry might or might not exchange $e$ and $m$. In particular Ref.~\cite{Heinrich2016} constructs such a commuting model with the symmetry exchanging $e$ and $m$ (we previously referred to this example in \autoref{subsec:H2_constraints}).

Instead, one can specifically focus on the commuting models whose boundary algebra is the same as the boundary algebra of a Levin-Wen model \cite{Levin2005,Jones_2307}.
We will call these \emph{Levin-Wen-like models}---such models can be viewed as a generalization of the concept of a ``simple realization of the toric code phase'' that we described in \autoref{subsec:H2_constraints}.
Moreover, we require that the boundary algebra can be chosen to be symmetric under the $G$-rep. Then looking at how the $G$-rep acts on the IR topological phase of the commuting model defines a $G$-TQFT, and one postulates that if two Levin-Wen-like models are compatible with the same $G$ anomaly class, then the corresponding $G$-TQFTs must be Witt equivalent. This defines a homomorphism
\begin{equation}
\label{eq:G_psi}
    \psi : \LatAnom{d}{G} \to \mathrm{Witt}(G\text{-TQFT}_d),
\end{equation}
and one expects that $\psi$ and $\phi$ are inverses of each other.

One issue is that there may be elements of $\LatAnom{d}{G}$ which are not compatible with \emph{any} Levin-Wen-like commuting models, or indeed any commuting models at all with ground states in which the symmetry is not spontaneously broken. This is evidently what happens in $d=1$, in which there is no topological order and any anomaly that descends to non-trivial QFT anomaly necessarily implies the ground state of any gapped Hamiltonian (which in particular includes commuting Hamiltonians) spontaneously breaks the symmetry. There may also be examples in $d > 1$~\cite{CZhangPrivateComm}. In these cases, one must replace $\LatAnom{d}{G}$ in \eqnref{eq:G_psi} with a subgroup corresponding to those anomaly classes that support the desired commuting projector models. This why the ``$\leq$'' appears in the statement of the conjecture, \eqnref{eqn:LatticeAnomalyClassification}.

\section{Towards classification of lattice anomalies: homotopy theory}
\label{sec:HomotopyTheory}

In this section, we adapt methods which have been successful in the classification of invertible states to the classification of lattice anomalies. 
In particular, the classification of invertible states (and hence, QFT anomalies) can be treated through a systematic analysis of \emph{symmetry defects}, which has also been described as an analysis of decorated domain walls~\cite{Xiong2018minimalist,Gaiotto2019gencohomology,Wang2020DecoratedDomain,Wang2021DomainWall}. This analysis is most compactly phrased using some technology from homotopy theory and higher categories, specifically 
Kitaev's \(\Omega\)-spectrum proposal~\cite{Kitaev2013SRE2,Kitaev2015SRE3,Kitaev2019SRE4,Xiong2018minimalist,Gaiotto2019gencohomology,Kubota2025stablehomotopytheoryinvertible}. In this section we will describe an analogous approach to understanding lattice anomalies.

The $\Omega$-spectrum proposal for QCAs and the application to the classification of lattice anomalies is described through the language of blends in \autoref{sec:Defects}.
In \autoref{sec:HomotopyEquiv}, we show explicitly how this homotopy-theoretic perspective reproduces the cohomology invariants identified in \autoref{sec:CohomologyInvariants}. In \autoref{sec:SpectrumMap} we examine how the \(\Omega\)-spectrum structure interfaces with the similar structure proposed for QFT anomalies.

Throughout this section, we use the notation \(U \BlendVia{V} W\) to denote that the QCA \(V\) is a blend between \(U\) and \(W\) across a half volume.

\subsection{The homotopy-theoretic classification of lattice anomalies}
\label{sec:Defects}

The homotopy-theoretic classification of lattice anomalies systematically organizes obstructions, arising from different (co)dimensions, to blending two \(G\)-reps.
The \(\grpH{1}{G}{\QCAblendclass{d}{}}\) invariant detects such an obstruction. Indeed, if two \(G\)-reps \(\repU\) and \(\repW\) have different values in \(\grpH{1}{G}{\QCAblendclass{d}{}}\), there is some \(g
\in G\) such that \(\repU_g\) and \(\repW_g\) cannot even blend while remaining a QCA. We regard this as a codimension-0 obstruction: the blend fails due to the impossibility of blending certain \(d\)-dimensional QCAs.
The \(\grpH{2}{G}{\QCAblendclass{d-1}{}}\) invariant then detects a codimension-1 obstruction. It measures when a putative blend fails to be a \(G\)-rep because of the action of a \((d-1)\)-dimensional QCA at the interface between \(\repU\) and \(\repW\).

Continuing down in dimension, we anticipate that there may be obstructions to constructing a \(G\)-rep blending \(\repU\) and \(\repW\) at any codimension \(p\). When all these obstructions vanish, a blending \(G\)-rep can be constructed.

Constructing invariants which witness each of these codimension-\(p\) obstructions becomes increasingly complicated. It is tempting to say that the invariant for the codimension-\(p\) obstruction should be valued in \(\grpH{p+1}{G}{\QCAblendclass{d-p}{}}\). 
This is almost correct, but the general homotopy-theoretic perspective that we describe below implies that for \(p \geq 2\) there may be \emph{higher trivializations}, which render some putative obstructions in \(\grpH{p+1}{G}{\QCAblendclass{d-p}{}}\) trivial; and further that there may be \emph{higher  inconsistencies}, which  prevent putative \(\grpH{p+1}{G}{\QCAblendclass{d-p}{}}\) obstructions from actually being realized. 

The situation is in many ways analogous to the classification of QFT anomalies. In that case one similarly obtains invariants valued in $\grpH{p+1}{G}{\invblendclass{d-p}{}}$ (for the corresponding SPT phases, these have been interpreted, for example, in terms of ``decorated domain wall'' constructions~\cite{Xiong2018minimalist,Gaiotto2019gencohomology,Wang2020DecoratedDomain,Wang2021DomainWall}) and the same increasingly-complicated tower of higher inconsistencies and trivializations emerges. 
These conditions can all be summarized in a succinct formalism by incorporating tools from higher categories and homotopy theory.

In \autoref{subsec:InftyGroupoid} we explain the \(\Omega\)-spectrum proposal for QCAs. 
In \autoref{subsec:OmegaSpectrumClassification} we conjecture that lattice anomalies are classified by a \emph{generalized cohomology} theory constructed from the \(\Omega\)-spectrum of QCA.
In \autoref{subsec:AHSSClassification} we explain how to recover the cohomology invariants of \autoref{sec:CohomologyInvariants} from the \(\Omega\)-spectrum.

\subsubsection{\texorpdfstring{\(\Omega\)}{Omega}-spectra from blends}
\label{subsec:InftyGroupoid}

The core of the homotopy-theoretic approach relies on the proposal that QCAs form an \emph{\(\Omega\)-spectrum}~\cite{Long2024}. 
In this subsection, we make a formulation of this proposal in terms of blends of QCAs. 
However, to state the conjecture in its original topological context, we need a topology on the space of QCAs. For the purpose of our discussion, we will not need to be precise about the details of this topology. We denote the topological space of QCAs in \(d\) dimensions by \(\QCA_d\).

The \(\Omega\)-spectrum proposal states that the space of loops of \((d+1)\)-dimensional QCAs beginning and ending at the identity, written \(\Omega\QCA_{d+1}\), is homotopy equivalent to \(\QCA_d\). In a topological formulation, this is conjectural, but it is much more straightforward to justify in terms of blends. 

To make sense of a ``space of loops'' in terms of blends, we construct a category in which the objects are \((d+1)\)-dimensional QCA, and the arrows (1-morphisms) are blends between \((d+1)\)-dimensional QCAs across a half volume \(R = (-\infty,0]\times \RR^d\) normal to the first axis. But blends are themselves \((d+1)\)-dimensional QCA, and so we can construct a higher category by defining arrows between arrows (2-morphisms) as blends between blends, but this time along the second axis. 
That is, if we have blends \(U \BlendVia{A}W\) and \(U \BlendVia{B} W\), then a 2-morphism between \(A\) and \(B\) is a third blend \(U \BlendVia{C} W\) such that \(C\) is also a blend between \(A\) and \(B\) across a half-volume normal to the second axis.
This can be continued, defining 3-morphisms between 2-morphisms, and so on, until we have gone through all \(d+1\) axes [up to the \((d+1)\)-morphisms]. 
However, we can then go even further by defining \((d+2)\)-morphisms to be homotopies (continuous deformations) between \((d+1)\)-morphisms, \((d+3)\)-morphisms to be homotopies between those, and so on. Thus, we are constructing a so-called $\infty$-category, in which there are $n$-morphisms for all $n$.

This category must also come with a notion of composition of blends. An intuitive definition for the composition of blends $U\BlendVia{A} V$ and $V\BlendVia{B} W$ would be to make a blend from \(U\) to \(V\) at one interface, and subsequently from \(V\) to \(W\) at another interface.
However, our blends are defined along a fixed half volume \(R\), so it is inconvenient to use multiple interfaces [\autoref{fig:qca-blend}(a)].
Instead, we define the composition $U\BlendVia{A} V \BlendVia{B} W$ to be $U\BlendVia{A V^{-1} B} W$.
If the blending interfaces can be deformed, then this definition is equivalent (related by an invertible 2-morphism) to the intuitive one based on using separate interfaces, as sketched in \autoref{fig:qca-blend}(b).
Furthermore, this definition makes all blends invertible: the inverse blend of $U\BlendVia{A} V$, denoted as $V\InverseBlendVia{A}U$, can be realized by $V\BlendVia{V A^{-1} U}U$. Composing with $U\BlendVia{A} V$ on either side will result in a trivial blend $U\BlendVia{U}U$ or $V\BlendVia{V}V$.

Note that in order to use the above definitions, we need to be able to define the compositions \(A V^{-1} B\).
The issue is that, in the previous sections, we allowed a QCA to be defined on any on-site algebra $M$. For a fixed $M$ we have a group of QCAs, denoted $\mathrm{LAut}(M)$, but two QCAs on different on-site algebras cannot be composed. 
In order for all QCAs to be composable, we will instead consider a ``universal'' or ``stabilized'' group $\mathrm{LAut}(*)$, defined such that for any on-site algebra $A$, there exists an injective group homomorphism $\iota_M:\mathrm{LAut}(M) \to \mathrm{LAut}(*)$ with suitable locality properties (for example, it preserves the blend structure), and such that every element of $\mathrm{LAut}(*)$ is in the image of $\iota_M$ for some on-site algebra $M$~\footnote{\(\LAut(*)\) can be constructed as a direct limit of the algebras \(\LAut(M)\), where \(M\) grows by stacking ancillas.}.

\begin{figure}
    \centering
    \includegraphics[]{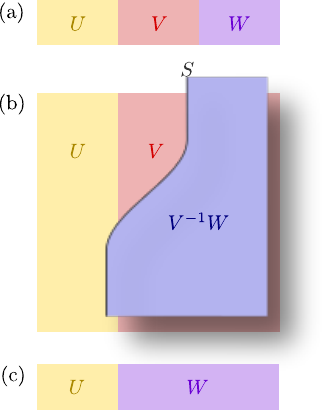}
    \caption{Demonstration that defining the composition of blends $U\BlendVia{A} V \BlendVia{B} W$ as $U\BlendVia{A V^{-1} B} W$ is equivalent to the intuitive picture of compositing blends by concatenation, assuming that the blending interface of $B$ can be deformed. (a)~The intuitive composition of blends. (b)~A blend between (a) and (c) oriented vertically is achieved by multiplying with a curved blend \(\id \BlendVia{V^{-1} B(S)} V^{-1} W\). This demonstrates that (a) and (c) are related by a 2-morphism, and thus equivalent. (In $d=1$, the vertical axis is to be interpreted as homotopy.) Visual stacking indicates multiplication. (c)~The composition realized by $A V^{-1} B$.}
    \label{fig:qca-blend}
\end{figure}

From blends in $\mathrm{LAut}(*)$, we thus obtain a higher category with morphisms of all degrees, all of which are invertible.
Such a category is called an \emph{\(\infty\)-groupoid}, and we denote it \(\Pi(\QCA_{d+1})\). Note that we have not rigorously proven that the composition rules defined above give rise to an $\infty$-groupoid. Formally, this would mean that various coherence conditions for compositions of $n$-morphisms are satisfied, which could be quite non-trivial to verify. We merely state it as a conjecture that $\Pi(\QCA_{d+1})$ forms an $\infty$-groupoid.

For the sake of intuition, the category \(\Pi(\QCA_{d+1})\) can (like any $\infty$-groupoid) essentially be regarded as a topological space (up to weak homotopy equivalence). Indeed, it is natural to conjecture that this construction produces the so-called fundamental \(\infty\)-groupoid of \(\QCA_{d+1}\)---the higher category where objects are QCAs, 1-morphisms are paths between QCAs, 2-morphisms are homotopies between paths, and so on. This object is also usually denoted \(\Pi(X)\) for some space \(X\) [or \(\Pi_\infty(X)\)], and contains all the same information (again up to weak homotopy equivalence) as the full space \(X\) itself. The distinction in our construction is that we have replaced the paths appearing in the fundamental \(\infty\)-groupoid with blends.

The loop space (or loop \(\infty\)-groupoid) of \(\Pi(\QCA_{d+1})\) is obtained by fixing an object, the identity QCA, and constructing a new category with objects given by the \(1\)-morphisms from the identity to the identity, and higher \(n\)-morphisms being the \((n+1)\)-morphisms from \(\Pi(\QCA_{d+1})\). So objects in the loop space are blends from the identity to itself, and arrows in the loop space are vertical (along the second coordinate) blends between these blends.

However, a blend from the identity to itself in \((d+1)\)-dimensions is naturally identified as a QCA in \(d\)-dimensions (with a sea of ancillas on which the QCA acts trivially). Similarly, blends between such blends are naturally blends between \(d\)-dimensional QCA, and so on. We have that the loop space is equivalent to the \(\infty\)-groupoid of QCA in one dimension lower,
\begin{equation}
    \Omega \Pi(\QCA_{d+1}) \simeq \Pi(\QCA_d).
    \label{eqn:OmegaSpecCat}
\end{equation}
This is called the \(\Omega\)-spectrum condition for the sequence of $\infty$-groupoids \((\Pi(\QCA_d))_{d\in \ZZ}\), where we define \(\Pi(\QCA_d)\) for \(d<0\) as \(\Omega^{-d}\Pi(\QCA_0)\). Thus \(\QCAblendclass{d}{} = \pi_{-d}(\QCA_0)\) for \(d<0\), where \(\QCA_0\) is the space of projective unitary operators (stabilized by stacking with ancillas)~\cite{Long2024}.

Any sequence of \(\infty\)-groupoids or topological spaces \((Y_d)_{d\in\ZZ}\) which has \(Y_d \simeq \Omega Y_{d+1}\) (in the case of topological spaces, ``$\simeq$'' means weak homotopy equivalence) is called an \(\Omega\)-spectrum.
Invertible states have also been proposed to form an \(\Omega\)-spectrum, \(\inv_d \simeq \Omega\inv_{d+1}\). As for QCA, this is reasonably straightforward to argue heuristically when using blends~\cite{Gaiotto2019gencohomology}.

Schematically, passing from a topological definition of equivalence (using continuous paths and homotopies) to blend equivalence is accomplished by realizing a path of QCAs \(U(t)\) (\(t\in [0,1]\)) as a texture in space. Suppose we can describe this path locally around a point \((x,\vec{y}) \in \RR \times \RR^d\) as \(U(x,\vec{y},t)\). Then the blend corresponding to this path is
\begin{equation}
    V = \left\{
    \begin{array}{l l}
        U(x,\vec{y},0) \quad&  x<0,\\
        U(x,\vec{y},x/\xi) \quad&  0\leq x \leq \xi, \\
        U(x,\vec{y},1) \quad& \xi < x,
    \end{array}
    \right.
\end{equation}
where \(\xi \gg r\) is some length which is much larger than the range of any QCA in the path \(U(t)\). Of course, this is not rigorous, because we cannot make a continuous interpolation in space on the lattice.

Nonetheless, concepts from topology and homotopy theory can be applied to \(\Pi(\QCA_d)\). Higher homotopy groups of \(\Pi(\QCA_d)\) are defined as isomorphism classes in \(\Omega^n \Pi(\QCA_d)\), and the \(\Omega\)-spectrum condition implies that these are the same as blend equivalence classes in a lower dimension \(\QCAblendclass{d-n}{}\).
We can also make sense of a notion of homotopy of maps \(f\) and \(g\) from some space \(X\) to \(\Pi(\QCA_{d+1})\). Just as in topology, this is just a map \(F\) from \(X \times [0,1]\) [or, if one prefers, an \(\infty\)-functor from \(\Pi(X \times [0,1])\) to \(\Pi(\QCA_{d})\)] such that \(F\) is equal to \(f\) when restricted to \(X\times\{0\}\) and is equal to \(g\) when restricted to \(X\times \{1\}\). Denote the homotopy equivalence classes of maps from \(X\) to \(\Pi(\QCA_d)\) by
\begin{equation}
    \pi_0[X, \QCA_d].
\end{equation}

\subsubsection{Proposed classification for lattice anomalies by generalized cohomology}
\label{subsec:OmegaSpectrumClassification}

Very rich mathematical theory has been built up around \(\Omega\)-spectra.
One of their key applications is that they serve as coefficients for \emph{generalized cohomology theories}.
That is, collections of maps \((h^d)_{d\in\ZZ}\) from topological spaces to Abelian groups which obey most of the axioms of cohomology theories.
(Generalized cohomology is not required to assign trivial groups for \(d \neq 0\) to the topological space consisting of a single point, as ordinary cohomology is.)
In fact, every generalized cohomology theory can be expressed as
\begin{equation}
    h^d(X) \cong \pi_0[X, Y_d]
\end{equation}
for some \(\Omega\)-spectrum \(Y_d\), and conversely every \(\Omega\)-spectrum defines a generalized cohomology theory in this way.

An extension of the \(\Omega\)-spectrum proposal for invertible states is that QFT anomalies---themselves corresponding to invertible \(G\)-states---are classified by~\footnote{The conjecture Eq.~\eqref{eqn:SPT_gencohom} must be modified to include \emph{twisted} generalized cohomology groups for antiunitary symmetries or for fermionic systems when the group \(G\) involves a nontrivial extension of a bosonic symmetry by fermion parity. The same is true for lattice anomalies in Eq.~\eqref{eqn:QCA_gencohom}. We will not describe these cases, beyond noting that they can be dealt with.}
\begin{equation}
    \QFTAnom{d}{G} \cong h_{\inv}^{d+1}(\BG) :=\pi_0[\BG, \inv_{d+1}].
    \label{eqn:SPT_gencohom}
\end{equation}
Here, \(\BG\) is the classifying space of the group \(G\).
For discrete groups, it is defined by the property that its fundamental group is \(G\), \(\pi_1(\BG) \cong G\), and that all its other homotopy groups are trivial.
In this case, it is helpful to describe \(\BG\) as an $\infty$-groupoid with a single object, \(*\), and 1-morphisms for each \(g \in G\) which obey \(* \BlendVia{g} * \BlendVia{h}* = * \BlendVia{gh} *\).

There is a direct generalization of \eqnref{eqn:SPT_gencohom} to the case of lattice anomalies~\cite{Long2024}. We conjecture that there is an isomorphism
\begin{equation}
    \LatAnom{d}{G} \cong h_{\QCA}^{d+1}(\BG) :=\pi_0[\BG, \QCA_{d+1}].
    \label{eqn:QCA_gencohom}
\end{equation}

We motivate this conjecture as follows.
One can show that any \(G\)-rep in an on-site algebra \(\repU : G \to \QCA_d\) induces a map (i.e. a functor of $\infty$-groupoids) $\repU : \BG \to \Pi(\QCA_{d+1})_\id$, where \(\Pi(\QCA_{d+1})_\id\) is the full subcategory of \(\Pi(\QCA_{d+1})\) containing the single object \(\id\).
That is, \(\Pi(\QCA_{d+1})_\id\) contains \(\id\), all 1-morphisms \(\id \BlendVia{U} \id\), all 2-morphisms between those, and so on.
Indeed, each 1-morphism in \(\BG\) maps to some blend $\id \BlendVia{\repU_g} \id$ in $\Pi(\QCA_{d+1})$ (which, recall, is equivalent to a $d$-dimensional QCA), and the composition law for the \(G\)-rep ensures that the functor respects the composition of 1-morphisms:
\begin{equation}
    \left(\id \BlendVia{\repU_g} \id \BlendVia{\repU_h}  \id \right) = \left(\id \BlendVia{\repU_g\repU_h} \id\right) = \left(\id\BlendVia{\repU_{gh}} \id\right).
\end{equation}
In the first equality we used the composition rule in $\Pi(\QCA_{d+1})$.

A blend between \(G\)-reps can be used to construct a homotopy between their associated \(\BG\) maps.
Thus, we have a map from blend classes of \(G\)-reps in on-site algebras---symmetry protected lattice anomalies---to homotopy classes of maps from \(\BG\) to $\Pi(\QCA_{d+1})_\id$.
Also considering maps \(\BG \to \Pi(\QCA_{d+1})\) (not just the identity full subcategory) extends this construction to lattice gravitational anomalies, which are known to be classified by the connected components of \(\Pi(\QCA_{d+1})\), i.e.\ $Q_{d+1} \cong \LatGravAnom{d}$ (see \autoref{sec:LatticeGravitational}).
Thus, we have a map \(\LatAnom{d}{G} \to h^{d+1}_\QCA(\BG)\), which we conjecture to be an isomorphism.

While in this paper we have mainly focused on discrete groups, it seems likely that the conjecture \eqnref{eqn:QCA_gencohom} will still hold for continuous groups as well. The argument goes as follows. For each $d$, $\QCA_d$ is not just a topological space---it comes with a composition operation, turning it into a topological group.
Just as homotopy types of topological spaces give rise to $\infty$-groupoids, topological \emph{groups} give rise to $\infty$-groups (which are $\infty$-groupoids with a multiplication operation which satisfies some coherence conditions). In the blend model of $\Pi(\QCA_d)$, the composition of QCAs gives $\Pi(\QCA_d)$ the structure of an $\infty$-group. Moreover, in this model it is clear that the delooping $\mathrm{B} \Pi(\QCA_{d})$ [the \(\infty\)-groupoid with one object, morphisms given by objects in \(\Pi(\QCA_{d})\) with their multiplication operation giving the composition law, and higher \(n\)-morphisms being the \((n-1)\)-morphisms of \(\Pi(\QCA_{d})\)] satisfies
\begin{equation}
\label{eq:BPiQCA}
    \mathrm{B}\Pi(\QCA_d) \simeq \Pi(\QCA_{d+1})_{\mathrm{id}}.
\end{equation}
Let us assume that the blend model is indeed capturing the homotopy type of $\QCA_d$, so that \eqnref{eq:BPiQCA} holds as a homotopy equivalence of topological spaces:
\begin{equation}
    \mathrm{B} \QCA_d \simeq (\QCA_{d+1})_{\mathrm{id}},
\end{equation}
where the left-hand side now refers to the classifying space of a topological group, and the right hand side is the connected component of the identity in \(\QCA_{d+1}\).
Now consider a continuous group homomorphism $\repU : G \to \QCA_d$. Then it induces a continuous map
\begin{equation}
    f_{\mathrm{\repU}} : \BG \to \mathrm{B} \QCA_d \simeq (\QCA_{d+1})_{\mathrm{id}}.
\end{equation}
Thus, we expect that we can classify $G$-reps by classifying the homotopy classes of such maps, as in the discrete case.

The conjecture \eqnref{eqn:QCA_gencohom} is compatible with the bulk-boundary relation between lattice anomalies in \(d\) dimensions and \(G\)-QCAs in \(d+1\) dimensions discussed in \autoref{sec:BulkBoundaryCorrespondence}, and the proposal that symmetry enriched QCA are similarly classified by generalized cohomology~\cite{Long2024}.
While we have only proved the bulk-boundary correspondence is an isomorphism for \emph{conjugacy anomalies} and \(G\)-QCAs with on-site \(G\)-reps, we suspect there is a more general relationship which includes lattice anomalies defined by blends, and which remains an isomorphism.

\subsubsection{Generalized cohomology invariants}
\label{subsec:AHSSClassification}

Equation~\eqref{eqn:QCA_gencohom} is remarkable, as it implies that the classification of lattice anomalies can be found using only knowledge of QCA without symmetry, via \(h^{d+1}_\QCA\), and the abstract group \(G\), via \(\BG\). It is a direct generalization of the partial classification of \(G\)-reps by ordinary group cohomology, \(\grpH{d+2}{G}{\mathrm{U}(1)}\). 
This is revealed by a computation of \(h_{\QCA}^{d+1}(\BG)\) using the \emph{Atiyah-Hirzebruch spectral sequence} (AHSS). 
The AHSS is a general method for computation of generalized cohomology groups, but in the context of QCAs it can be thought of as organizing obstructions to blending of different codimensions.
The AHSS iteratively constructs a sequence of ``pages'', beginning with the \emph{\(E_2\) page}:
\begin{align}
    E_2^{p, q} &= \singH{p}{\BG}{h^q_\QCA(\mathrm{pt})},
    \label{eqn:InfintePSum}
\end{align}
where the coefficient groups \(h^{q}_\QCA(\mathrm{pt})\) are the values of the generalized cohomology for a point, \(\mathrm{pt}\). 
That is, \(h^{q}_\QCA(\mathrm{pt}) \cong \QCAblendclass{q}{}\).

The idea is that \emph{to a first approximation}, one can define a sequence of invariants that partially classify elements of \(h^{d+1}_\QCA(\BG)\). These invariants are valued in $E_2^{p+1,d-p}$ for $p \in \{-1, 0, 1, \ldots\}$, where the $p$-th invariant is well-defined only if all the earlier invariants vanish. 
The \(p=-1\) invariant encodes the gravitational part of the anomaly, while the $p \in \{0,1,2\}$ invariants correspond to the cohomological invariants discussed in \autoref{sec:CohomologyInvariants}, since the singular cohomology groups \(\singH{p+1}{\BG}{\QCAblendclass{d-p}{}}\) agree with the group cohomology of \(G\), \(\grpH{p+1}{G}{\QCAblendclass{d-p}{}}\). 
Furthermore, the familiar \(\grpH{d+2}{G}{\mathrm{U}(1)}\) invariant for finite groups appears when setting \(p=d+1\), which gives
\begin{multline}
    E_2^{d+2,-1} = \grpH{d+2}{G}{\QCAblendclass{-1}{}} \\
    \cong \grpH{d+2}{G}{\QQ/\ZZ} \cong \grpH{d+2}{G}{\mathrm{U}(1)}.
\end{multline}
We used that \(\QCAblendclass{-1}{} \cong \QQ/\ZZ\)~\cite{Long2024}, and that \(\grpH{p}{G}{\QQ/\ZZ} \cong \grpH{p}{G}{\mathrm{U}(1)}\) for finite \(G\).

An important point, and the reason why we stated ``to a first approximation'' above, is that there are in general 
\emph{higher inconsistencies} which cause certain elements of $\grpH{p+1}{G}{\QCAblendclass{d-p}{}}$ to be physically inconsistent, i.e.\ they cannot be realized, and \emph{higher trivializations} which make some non-trivial elements correspond to a trivial lattice anomaly.
This is handled by subsequent \emph{pages} of the spectral sequence. In particular, there are further differentials
\begin{equation}
    \mathrm{d}^{p,q}_2 : E_2^{p,q} \to E_2^{p+2,q-1}.
\end{equation}
 We are required to keep only the elements that are in the kernel of these differentials, and mod out by the elements that are in the image of a differential.
 This forms an \(E_3\) page
\begin{equation}
    E_3^{p,q} = \frac{\ker \mathrm{d}^{p,q}_2}{\im \mathrm{d}^{p-2,q+1}_2},
\end{equation}
This is iterated with differentials \(\mathrm{d}^{p,q}_r: E_r^{p,q} \to E_{r+1}^{p+r,q-r+1}\) on the \(E_r\) page. At a fixed $p,q$, eventually all further differentials become trivial and $E_{n+1}^{p,q} = E_n^{p,q}$, which defines an $E_\infty$ page in which the actual invariants appear.

These differentials handle the tower of consistency conditions we would encounter in a more direct attempt at the classification. The homotopy theoretic approach organizes them all in a structured way.

\subsection{Equivalence of homotopy and cohomology invariants}
\label{sec:HomotopyEquiv}

The previous section described several invariants taking values in \(\grpH{p+1}{G}{\QCAblendclass{d-p}{}}\). In the homotopy context, the \(\QCAblendclass{d-p}{}\) blend classes are to be interpreted as \(p\)-times iterated \emph{loops} of blends in \(\Pi(\QCA_{d})\) via the \(\Omega\)-spectrum condition~\eqref{eqn:OmegaSpecCat}. 
A cocycle representative \(\omega_{g,h,...,k}\) for an element of \(\grpH{p+1}{G}{\QCAblendclass{d-p}{}}\)
is to be interpreted as the homotopy class of a \(p\)-dimensional surface in \(\Pi(\QCA_{d})\).
This surface is properly a network of (higher) blends in \(\Pi(\QCA_{d})\) in the shape of the surface of a \((p+1)\)-simplex (e.g.\ a triangle for \(p=1\), or the surface of a tetrahedron for \(p=2\)).
Given a \(G\)-rep \(\repU\), the vertices in this network are \(\id\), \(\repU_g\), \(\repU_{gh}\), ..., \(\repU_{gh...k}\).
Edges are blends between these vertices, faces are blends between (compositions of) edges, and so on.
The homotopy class of the surface so constructed measures the obstruction to filling in the interior of the \((p+1)\)-simplex with some codimension-\(p\) blend consistent with all the codimension-\((p-1)\) faces.

As one might expect, these invariants reproduce those presented in \autoref{sec:CohomologyInvariants}.
In Secs.~\ref{sec:HomotopyH2}--\ref{sec:HomotopyH31D}, we explicitly demonstrate that the homotopy invariants generalize our previously defined invariants, as well as other well-known invariants for zero-dimensional \(G\)-reps (projective representations of \(G\)) and one-dimensional \(G\)-reps~\cite{Else2014}. These sections also serve as a more explicit example for the construction of the homotopy invariants.

For the purpose of identifying the \(\grpH{p+1}{G}{\QCAblendclass{d-p}{}}\) invariants in the homotopy language, we will assume that \(G\) is discrete.

\subsubsection{ \texorpdfstring{$\grpH{2}{G}{\QCAblendclass{d-1}{}}$}{H2} class in \texorpdfstring{$d\geq 1$}{higher than one dimension}}
\label{sec:HomotopyH2}

In \autoref{sec:H2Invariant}, the $\grpH{2}{G}{\QCAblendclass{d-1}{}}$ invariant was defined using symmetry restrictions \(\resU_g\) of a \(d\)-dimensional \(G\)-rep \(\repU_g\) to a half-volume \(R\). A cocycle $\omega_{g,h} \in \QCAblendclass{d-1}{}$ is specified by the blend class of \((d-1)\)-dimensional QCAs \(\bdyU_{g,h}\), defined by
\begin{equation}
    \resU_g\resU_h=\bdyU_{g,h}\resU_{gh}.
\end{equation}
Here, we show that the cohomology class defined this way is equivalent to that constructed through the homotopy-theoretic formalism.

In terms of \(\Pi(\QCA_{d})\), a cocycle \(\omega'_{g,h}\) for the \(\grpH{2}{G}{\QCAblendclass{d-1}{}}\) invariant is interpreted as the homotopy class corresponding to a triangle
\begin{equation}\label{eqn:H2BlendLoop}
    \begin{tikzcd}
        & \repU_g \arrow[dr,"\repU_g \resU_h"] \\
        \id \arrow[ur,"\resU_{g}"] \arrow[rr,"\resU_{gh}"] && \repU_{gh}
    \end{tikzcd}.
\end{equation}
Specifically, \(\omega'_{g,h}\) is the homotopy class of a loop around the triangle,
\begin{equation}\label{eq:H2blends}
    \id \BlendVia{\resU_g} \repU_g \BlendVia{\repU_g\resU_h} \repU_{gh} \InverseBlendVia{\resU_{gh}} \id.
\end{equation}
Evaluating the blend associated to this loop using the composition and inversion rules in \(\Pi(\QCA_{d})\), we find that the blending QCA is precisely \(\bdyU_{g,h}\):
\begin{equation}
    (\resU_g)\repU_g^{-1}(\repU_g\resU_h) \repU_{gh}^{-1} (\repU_{gh}\resU_{gh}^{-1}) = \bdyU_{g,h}.
\end{equation}
Thus, the homotopy class \(\omega'_{g,h}\) of this simplex is just the blend class of \(\bdyU_{g,h}\), which was the original definition of the cocycle \(\omega_{g,h}\). One can also check that choosing new restrictions \(\resU_g\) results in cohomologous cocycles, as in \autoref{sec:H2Invariant}.

We can also obtain an intuitive picture for the homotopy invariant if we assume that blend interfaces can be freely deformed (\autoref{fig:qca-blend}).
Then, we can use the intuitive picture of composing blends by concatenating successive blend interfaces, so that the homotopy definition of \(\bdyU_{g,h}\) becomes that shown in \autoref{fig:domain-walls}(a). We can also equivalently think of it in terms of a rule for fusing two domain walls, as shown in Figure \autoref{fig:domain-walls}(b).

\begin{figure}
    \centering
    \includegraphics[]{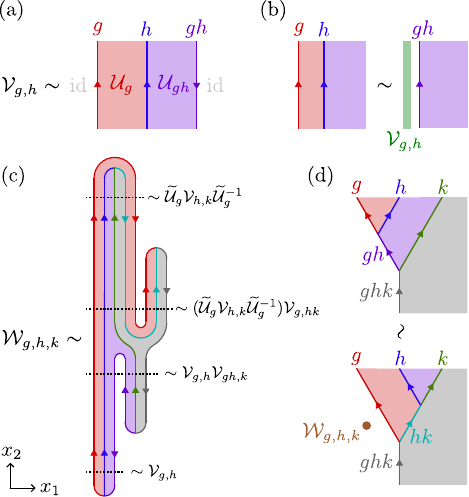}
    \caption{The homotopy picture of the $\grpH{2}{G}{\QCAblendclass{d-1}{}}$ and $\grpH{3}{G}{\QCAblendclass{d-2}{}}$ classes with composition of blends illustrated by concatenation in space, so that blend interfaces are separated as in \autoref{fig:qca-blend}. 
    (a)~Visualization of the loop of blends \eqnref{eq:H2blends} with each blend separated in \(x_1\). 
    (b)~Interpretation of the $\grpH{2}{G}{\QCAblendclass{d-1}{}}$ class as a rule for fusing two domain walls. (c)~Visualization of the network of blends defining the surface appearing in the \(\grpH{3}{G}{\QCAblendclass{d-2}{}}\) invariant.
    Compositions of blends (1-morphisms) are separated along \(x_1\), and compositions of blends-of-blends (2-morphisms) are separated along \(x_2\).
    The indicated cuts correspond to loops of 1-morphisms in the complex, related by the 2-morphisms.
    Note that the middle two cuts merely give two different decompositions of the same QCA.
    Colors are consistent with (d). (d)~Interpretation of $\grpH{3}{G}{\QCAblendclass{d-2}{}}$ as a rule for doing an F-move of domain walls.}
    \label{fig:domain-walls}
\end{figure}

\subsubsection{\texorpdfstring{$\grpH{3}{G}{\QCAblendclass{d-2}{}}$}{H3} class in \texorpdfstring{$d\geq 2$}{higher than two dimensions}}
\label{sec:HomotopyH3}

In \autoref{sec:H3Invariant}, the $\grpH{3}{G}{\QCAblendclass{d-2}{}}$ invariant was defined by finding $\mathcal{W}_{g,h,k}$ such that
\begin{equation}\label{eqn:WSymmetryDef}
    \widetilde{\bdyU}_{g,h}\widetilde{\bdyU}_{gh,k}=\mathcal{W}_{g,h,k}(\resU_g\widetilde{\bdyU}_{h,k}\resU_g^{-1})\widetilde{\bdyU}_{g,hk}
\end{equation}
and defining the cocycle \(\omega_{g,h,k}\) to be the blend class of \(\mathcal{W}_{g,h,k}\).
Here, $\widetilde{\bdyU}_{g,h}$ is a restriction of the boundary QCA $\bdyU_{g,h}$ to a half-volume in the second coordinate \(x_2\).

In the homotopy picture, the cocycle defining the \(\grpH{3}{G}{\QCAblendclass{d-2}{}}\) invariant is associated to the homotopy class of a network of blends in the shape of the surface of a tetrahedron.
The \((d-2)\)-dimensional QCA associated to this surface is evaluated in Appendix~\ref{appendix:3Simplex}.
We find it is precisely \(\mathcal{W}_{g,h,k}\), so that the homotopy theoretic definition of \(\grpH{3}{G}{\QCAblendclass{d-2}{}}\) agrees with the original definition.

This calculation is more simply understood graphically. In \autoref{fig:domain-walls}(c) we illustrate the network associated to the \(\grpH{3}{G}{\QCAblendclass{d-2}{}}\) representative (with all but two dimensions suppressed). 
We illustrate the composition of blends as concatenation (\autoref{fig:qca-blend}) for clarity.
\autoref{fig:domain-walls}(c) is to be understood as a vertical (\(x_2\)) sequence of junctions between horizontal (\(x_1\)) blends from \(\id\) to \(\id\) (marked by dashed lines). 
The horizontal blends themselves arise from the loops of edges defining \(\bdyU_{g,h}\), or compositions thereof.
Thus, the vertical junctions are, in the homotopy picture, faces interpolating between these loops.
The entire network describes a surface in \(\Pi(\QCA_d)\), and the homotopy class of this surface defines the cocycle representative for the \(\grpH{3}{G}{\QCAblendclass{d-2}{}}\) invariant.
Alternatively, this network can be interpreted as a rule for performing an F-move of domain walls---making the F-move is equivalent to multiplying by this network [\autoref{fig:domain-walls}(d)].

\subsubsection{\texorpdfstring{$\grpH{2}{G}{\QCAblendclass{-1}{}}$}{H2} class in \texorpdfstring{$d= 0$}{zero dimensions}}
\label{sec:HomotopyH20D}

In this subsection, we show that the $\grpH{2}{G}{\QCAblendclass{-1}{}}\cong \grpH{2}{G}{\QQ/\ZZ}$ class in zero dimensions captures the usual $\grpH{2}{G}{\mathrm{U}(1)}$ classification of projective representations via the inclusion $\mathbb{Q}/\mathbb{Z} \to \mathbb{R}/\mathbb{Z} \cong \mathrm{U}(1)$. In fact, our calculation also functions as a proof that the only \(\grpH{2}{G}{\mathrm{U}(1)}\) invariants for discrete \(G\) that are achieved in a finite-dimensional Hilbert space are those in the image of the induced map \(\grpH{2}{G}{\QQ/\ZZ} \to \grpH{2}{G}{\mathrm{U}(1)}\). 

Note that in \(d=0\), instead of blends, we need to reinterpret the loop \eqnref{eq:H2blends} as a continuous path in $\QCA_0$, which is just the space of projective unitaries (stabilized by stacking ancillas~\cite{Long2024}).

Let $\repU$ be a zero-dimensional $G$-rep.
As $\QCAblendclass{0}{} \cong \pi_0(\QCA_0)$ is trivial, for each $g\in G$ we can pick a path $\gamma_g:[0,1]\to \QCA_0$ such that $\gamma_g(0)=\id$ and $\gamma_g(1)=\repU_g$.
The path can be lifted to $\tilde\gamma_g:[0,1]\to \mathrm{U}(N)$ for some $N$~\footnote{For infinite groups, one may need to allow for \(N = N_g\) to depend on the group element, with different \(N_g\) all related by stacking or removing ancillas. Any two unitaries \(\tilde\gamma_g(t)\) and \(\tilde\gamma_h(t)\) can still be multiplied by considering them to act on a common Hilbert space with all the ancillas from both \(N_g\) and \(N_h\).}.
The topological counterpart of \eqnref{eq:H2blends} is
\begin{equation}\label{eq:H2loop}
    V(t)=\begin{cases}
        \tilde\gamma_g(3t) & \text{if }0\leq t<\frac{1}{3}, \\
        \tilde\gamma_g(1)\tilde\gamma_h(3t-1) & \text{if }\frac{1}{3}\leq t < \frac{2}{3}, \\
        \tilde\gamma_g(1) \tilde\gamma_h(1) \tilde\gamma_{gh}^{-1}(3t-2) & \text{if }\frac{2}{3}\leq t \leq 1, \\
    \end{cases}
\end{equation}
which is a path in $\mathrm{U}(N)$ from the identity to $\tilde\gamma_g(1) \tilde\gamma_h(1) \tilde\gamma_{gh}^{-1}(1)=\omega_{g,h}^\mathrm{U}$, the value of the $\grpH{2}{G}{\mathrm{U}(1)}$ class of $\repU$.

We claim that the winding number of the composition of paths \eqnref{eq:H2blends} in \(\QCA_0\) is given by \(\omega_{g,h}^\mathrm{U}\), up to a coboundary. 
We calculate the winding number using \(V(t)\). However, \(V(t)\) is not a loop in \(\mathrm{U}(N)\), so we must modify it to a loop such that it defines the same loop of projective unitaries.
This can be done by correcting the phase of \(V(t)\) with $\widetilde{V}(t)=V(t)(\omega_{g,h}^\mathrm{U})^{-t}$.
Then the $\QQ/\ZZ$-valued winding number is
\begin{multline}
    \omega_{g,h}=\frac{i}{2\pi N}\int_0^1 dt \operatorname{Tr}\left[\widetilde V^\dagger(t)\frac{d}{dt}\widetilde V(t)\right]\\
    =\frac{i}{2\pi N}\Big[\int_0^1 ds \operatorname{Tr}\big(\tilde\gamma_g^\dagger(s)\tilde\gamma'_g(s)\big)+\int_0^1 ds \operatorname{Tr}\big(\tilde\gamma_h^\dagger(s)\tilde\gamma'_h(s)\big)\\
    -\int_0^1 ds \operatorname{Tr}\big(\tilde\gamma_{gh}^\dagger(s)\tilde\gamma'_{gh}(s)\big)
    -\operatorname{Tr}\big(\log\omega_{g,h}^\mathrm{U}\big)\Big],
\end{multline}
where \(\gamma'\) is the derivative of \(\gamma\), and the result depends only on \(V\) (not \(\widetilde{V}\)) up to addition by integers.
The first three terms form a $\RR/\ZZ$-valued 2-coboundary (note that its value is independent of $N$).
Thus, regarding \(\omega\) as a cocycle in $\grpH{2}{G}{\mathrm{U}(1)}$ by the inclusion \(\QQ/\ZZ \to \RR/\ZZ\), we have that \(\omega\) and \(\omega^{\mathrm{U}}\) are cohomologous:
\begin{equation}
    \omega \sim \frac{1}{2\pi i}\log \omega^\mathrm{U}.
\end{equation}
Therefore, $\omega\in\grpH{2}{G}{\QQ/\ZZ}$ is equivalent to the usual $\omega^\mathrm{U}\in\grpH{2}{G}{\mathrm{U}(1)}$ invariant for projective representations of finite groups. As a corollary, any value of \(\omega^{\mathrm{U}}\) which is realized in a finite dimensional model must be cohomologous to a \(\QQ/\ZZ\)-valued \(\omega\).

\subsubsection{\texorpdfstring{$\grpH{3}{G}{\QCAblendclass{-1}{}}$}{H3} class in \texorpdfstring{$d= 1$}{one dimension}}
\label{sec:HomotopyH31D}

In this subsection, we show that the $\grpH{3}{G}{\QCAblendclass{-1}{}}\cong \grpH{3}{G}{\QQ/\ZZ}$ class in one dimension (\(d=1\), \(p=2\)) reduces to the $\grpH{3}{G}{\mathrm{U}(1)}$ class of Ref.~\cite{Else2014}, again via the inclusion $\mathbb{Q}/\mathbb{Z} \to \mathbb{R}/\mathbb{Z} \cong \mathrm{U}(1)$.
Again, a consequence of our calculation is that the only values of \(\grpH{3}{G}{\mathrm{U}(1)}\) which are achieved by models with finite local Hilbert space dimension and discrete \(G\) are in the image of \(\grpH{3}{G}{\QQ/\ZZ} \to \grpH{3}{G}{\mathrm{U}(1)}\).

In one dimension, we can define $\bdyU_{g,h}$ using blends [that is, \eqnref{eq:H2blends}] as before, but the blending of such blends along the next axis needs to be reinterpreted as continuous homotopy.
That is, we need to reinterpret \autoref{fig:domain-walls}(c) as a continuous loop in $\QCA_0$.

Let $\repU$ be a one-dimensional $G$-rep with trivial \(\grpH{1}{G}{\QCAblendclass{1}{}}\) class.
Since $\grpH{2}{G}{\QCAblendclass{0}{}}$ is trivial, we can find a path $\gamma_{g,h}:[0,1]\to \QCA_0$ for each \(g,h \in G\) such that $\gamma_{g,h}(0)=\id$ and $\gamma_{g,h}(1)=\bdyU_{g,h}$.
The path can be lifted to $\tilde\gamma_{g,h}:[0,1]\to \mathrm{U}(N)$.

Moreover, as $\tilde\gamma_{g,h}(t)$ acts on a subsystem of the original lattice, we can evaluate the action of the restrictions \(\resU_k\) on the local unitary operator \(\tilde\gamma_{g,h}(t)\). We define 
\begin{equation}
    ^k\tilde\gamma_{g,h}(t)=\resU_k[\tilde\gamma_{g,h}(t)] =\resU_k\tilde\gamma_{g,h}(t)\resU_k^{-1},
\end{equation}
which is to be interpreted a unitary in $\mathrm{U}(N')$ with a possibly larger $N'$ (related to \(N\) by stacking ancillas).

Similar to \eqnref{eq:H2loop}, for $g,h,k\in G$ we can write the continuous counterpart of \autoref{fig:domain-walls}(c) as
\begin{equation}\label{eq:H3loop}
    W(t)=\begin{cases}
        \tilde\gamma_{g,h}(4t) & \text{if }0\leq t<\frac{1}{4}, \\
        \tilde\gamma_{g,h} \tilde\gamma_{gh,k}(4t-1) & \text{if }\frac{1}{4}\leq t<\frac{1}{2}, \\
        \tilde\gamma_{g,h} \tilde\gamma_{gh,k} \tilde\gamma_{g,hk}^{-1}(4t-2) & \text{if }\frac{1}{2}\leq t<\frac{3}{4}, \\
        \tilde\gamma_{g,h} \tilde\gamma_{gh,k} \tilde\gamma_{g,hk}^{-1} {}^g\tilde\gamma_{h,k}^{-1}(4t-3) & \text{if }\frac{3}{4}\leq t\leq 1, \\
    \end{cases}
\end{equation}
where we write $\tilde\gamma_{g,h}(1)=\tilde\gamma_{g,h}$ for brevity.
$V$ is a path from the identity to $\tilde\gamma_{g,h} \tilde\gamma_{gh,k} \tilde\gamma_{g,hk}^{-1} {}^g\tilde\gamma_{h,k}^{-1}$, which, as \(W\) is a lift of a loop in \(\QCA_0\), must be proportional to the identity. Indeed, $\tilde\gamma_{g,h} \tilde\gamma_{gh,k} \tilde\gamma_{g,hk}^{-1} {}^g\tilde\gamma_{h,k}^{-1}$ is a cocycle representative of the \(\mathrm{U}(1)\)-valued 3-cocycle $\omega^{\mathrm{U}}_{g,h,k}$ of Ref.~\cite{Else2014}.
To make \(W\) a loop of unitaries, define $\widetilde{W}(t)=W(t)(\omega^{\mathrm{U}}_{g,h,k})^{-t}$.

We claim the \(\QQ/\ZZ\)-valued winding number of $\widetilde{W}$ produces an equivalent cocycle.
One subtlety is that the expression of $^g\tilde\gamma_{h,k}(t)$ involves $\resU_g$ which is not an operator on $\mathrm{U}(N)$. Nevertheless, we can truncate $\resU_g$ to $\resU'_g$ acting on $\mathrm{U}(N')$ with some larger $N'$ (which corresponds to truncating the QCA to a finite patch) such that $^g\tilde\gamma_{h,k}(t)=\resU'_g\tilde\gamma_{h,k}(t)\resU^{\prime -1}_g$ is still satisfied. Then we can interpret $\widetilde{W}$ as a loop in $\mathrm{U}(N')$ and calculate its $\QQ/\ZZ$-valued winding number as before, resulting in
\begin{multline}
    \omega_{g,h,k}=\frac{i}{2\pi N}\int_0^1 dt \operatorname{Tr}\left[\widetilde{W}^\dagger(t)\frac{d}{dt}\widetilde{W}(t)\right]\\
    =\beta_{g,h}+\beta_{gh,k}-\beta_{g,hk}-\beta_{h,k}-\frac{i}{2\pi N}\operatorname{Tr}\big(\log\omega_{g,h,k}^\mathrm{U}\big),
\end{multline}
where the first four terms form a $\RR/\ZZ$-valued 3-coboundary with (note that the value is independent of the choice of $N$)
\begin{equation}
    \beta_{g,h}=\frac{i}{2\pi N}\int_0^1 ds \operatorname{Tr}\big(\tilde\gamma_{g,h}^\dagger(s)\tilde\gamma'_{g,h}(s)\big).
\end{equation}
As before, we can map \(\omega_{g,h,k}\) to a \(\mathrm{U}(1)\) valued cocycle, and we have that
\begin{equation}
    \omega \sim \frac{1}{2\pi i}\log \omega^\mathrm{U},
\end{equation}
have the same cohomology class.

\subsection{Map to QFT anomalies}
\label{sec:SpectrumMap}

Both \((\QCA_d)_{d\in\ZZ}\) and \((\inv_d)_{d\in\ZZ}\) are conjectured to be \(\Omega\)-spectra. Furthermore, there is a continuous map \(e_d:\QCA_d \to \inv_d\) defined by acting by a QCA on a fixed product state, \(e_d(V) = V \ket{0}\). It is natural to conjecture that \(e_\bullet\) is a map of \(\Omega\)-spectra, in that the square
\begin{equation}
    \begin{tikzcd}
     \Omega \QCA_{d+1} \arrow[r,"e_{d+1}"] 
     & \Omega \inv_{d+1} \\
     \QCA_{d} \arrow[r,"e_d"] \arrow[u] & \inv_{d} \arrow[u]
    \end{tikzcd}
    \label{eqn:QCAtoInvSquare}
\end{equation}
commutes up to homotopy, where the vertical maps are the homotopy equivalences defining the \(\Omega\)-spectrum, and the top horizontal map is defined by applying \(e_{d+1}\) pointwise to a loop \(\gamma\): \(e_{d+1}(\gamma)(t) = \gamma(t)\ket{0}\).

This conjecture is to be interpreted as an extension of \eqnref{eqn:LatticeToQFTSquare}, relating the structure of lattice and QFT anomalies. Indeed, the map of spectra \(e\) induces a homomorphism
\begin{equation}\label{eqn:InducedCohomMap}
    e^{d+1}_{*} : h^{d+1}_{\QCA}(\BG) \to h^{d+1}_{\inv}(\BG),
\end{equation}
which is interpreted as a map of anomalies via the proposal that \(\LatAnom{d}{G}\) is classified by \(h^{d+1}_{\QCA}(\BG)\) for finite \(G\), and that \(\QFTAnom{d}{G}\) is classified by \(h^{d+1}_{\inv}(\BG)\). In fact, a map of spectra \(e_\bullet\) induces a long exact sequence involving the Abelian groups \(h^{d+1}_{\QCA}(\BG)\) and \(h^{d+1}_{\inv}(\BG)\) (for many \(d\)).
This is, in principle, much more constraining than the existence of an induced homomorphism alone, but we leave the identification of specific consequences to future work.

We conclude this section by outlining some expected properties of \(e^{d+1}_{*}\) in \eqnref{eqn:InducedCohomMap}. As we have commented several times in the context of anomalies, \(e^{d+1}_{*}\) is neither injective nor surjective in general. There are nontrivial lattice anomalies which are IR trivial (\autoref{sec:em-qubit}), and there are QFT anomalies which do not descend from any lattice anomaly~\cite{Kapustin_2401}. However, the homotopy point of view for \(e^{d+1}_{*}\) also reveals more interesting behavior beyond what we have already discussed.

Indeed, a particularly interesting example of a nontrivial lattice anomaly which maps to a nontrivial QFT anomaly is provided by Ref.~\cite{Fidkowski2020beyondcohomology}. Here, a three-dimensional \(\ZZ_2\)-rep is considered in which the nontrivial element of \(\ZZ_2\) maps to a nontrivial QCA. That is, this \(\ZZ_2\)-rep is anomalous, with a nontrivial \(\grpH{1}{\ZZ_2}{\QCAblendclass{3}{}}\) invariant. The authors of Ref.~\cite{Fidkowski2020beyondcohomology} then construct a \(\ZZ_2\)-QCA in four dimensions from this \(\ZZ_2\)-rep, similar to our construction from \autoref{sec:BulkBoundaryCorrespondence}. Applying this QCA to a symmetric product state generates a \(\ZZ_2\)-SPT, the edge of which carries a ``beyond cohomology'' (though, not beyond generalized cohomology) anomaly. In our language, this construction is simply applying \(e^4_*\) to the lattice \(\ZZ_2\)-anomaly. The curious feature of this example is that the invariant associated to the three-dimensional QFT anomaly is valued in \(\grpH{2}{\ZZ_2}{\invblendclass{2}{}}\), where \(\invblendclass{2}{} \cong \ZZ\) are blend classes of invertible states in two dimensions. The generator of \(\invblendclass{2}{}\) is called the \(E_8\)-state~\cite{Kitaev2006BeyondAnyons,Kitaev2011E8,Lan2016E8,Long2024Edge}. 
The obstruction to trivializing the QFT anomaly comes from \(E_8\)-states decorating two-dimensional symmetry defects. 
There is a mismatch in the obstruction dimensions between the lattice anomaly (three-dimensional) and the QFT anomaly (two-dimensional). It would be interesting to understand when and why such dimension shifts occur as a result of applying \(e^{d+1}_*\).

\section{Discussion}
\label{sec:Discussion}

We have provided a definition for an anomaly of a global symmetry in lattice systems: a lattice anomaly is a blend equivalence class of a symmetry representation. We have also made significant progress in both the classification of lattice anomalies, and on characterizing the consequences of lattice anomalies for symmetric Hamiltonians.

One interesting extension of our methods would be to lattice systems where one allows the Hilbert space on each site to be a (separable) infinite-dimensional Hilbert space. This would extend our theory to, for instance, systems of quantum harmonic oscillators, or rotors.
We expect that most of our general arguments will carry over to this setting, but importantly, the classification of QCAs will be different---for example, one can argue that the classification of QCAs in one spatial dimension becomes trivial (for example, any translation can be absorbed into a single infinite dimensional site, so translations are blend equivalent to the identity)---and hence the lattice anomaly classification will be different. Moreover, at least some QFT anomalies that cannot be realized as lattice anomalies with finite-dimensional on-site Hilbert spaces can be realized with infinite-dimensional on-site Hilbert spaces  \cite{Chen_1106}. 
Intriguingly, therefore, it appears that allowing infinite-dimensional on-site Hilbert spaces leads the classification of lattice anomalies to become closer to that of QFT anomalies, since certain IR-trivial anomalies (e.g.\ those related to one-dimensional QCAs) that would map to trivial QFT anomalies disappear, while some QFT anomalies that could not be realized as lattice anomalies in the finite-dimensional case reappear. 
Indeed, it is conceivable that lattice anomalies and QFT anomalies might exactly coincide in this case. However, confirming this would require further study. 

Another possible extension would be to QFTs in the continuum. In the \emph{algebraic QFT} approach to axiomatizing QFT~\cite{HaagBook}, one postulates that a QFT possesses an algebra of local operators, which has similar properties to the local algebras we considered in this paper (see \autoref{subsec:local_algebras}) except that it does not need to satisfy the locally finite-dimensional condition. We can then talk about $G$-reps in terms of the local automorphisms of these local algebras.
However, the main difficulty with moving from the lattice to the continuum is that it is not clear what  the appropriate analog of an ``on-site algebra'' should be.

Returning to systems on the lattice with finite-dimensional Hilbert space per site, another extension would be to symmetries which act as ``quasi-local'' automorphisms rather than QCAs of strictly finite range---that is, a local operator can map to an operator that is not exactly supported on any bounded set, but rather has tails that decay rapidly with distance (for example exponentially)~\cite{Ranard2022converseLRbound}. This is likely the physically relevant class when considering anomalies which arise from a bulk-boundary correspondence with an MBL model in one higher dimension, where the commuting Hamiltonian terms are typically only quasi-local. 
Generally, we do not expect that our general results will be substantially different in this case, but it might be more challenging to provide rigorous proofs of the properties we were able to prove rigorously in the case of strictly finite-range QCAs.

One could also consider crystalline symmetries (such as spatial rotation, reflection, etc.), which still act as automorphisms of the algebra of local operators, but do not qualify as QCAs because they do not satisfy the finite range condition. For QFT anomalies/SPT phases, there is the ``crystalline equivalence principle'' \cite{Thorngren_1612}, which holds that the classification of anomalies with crystalline symmetries is actually the same as the classification of anomalies with the same symmetry group acting internally. It would be interesting to determine whether the same correspondence holds at the level of lattice anomalies.

In this paper we have only considered invertible 0-form symmetries. Recently there has been considerable attention devoted to so-called \emph{generalized symmetries}, such as higher-form symmetries~\cite{Gaiotto_1412} and non-invertible symmetries~\cite{Shao_2308}. 
Generalized symmetries and their anomalies appear in many QFTs of interest and can be used to understand properties of those QFTs. Therefore, it would be interesting to develop a theory of lattice anomalies of generalized symmetries.
In the case of higher-form symmetries, in on-site algebras one can only discuss so-called ``non-topological'' higher-form symmetries~\cite{Seiberg_1909,Qi_2010} in which the symmetry operator is not invariant under deformations of the submanifold on which the symmetry acts. These can, however, reduce to topological higher-form symmetries in the low-energy QFT description. We expect that it should be possible to apply methods similar to those of the current paper to classify anomalies of non-topological higher-form symmetries on the lattice. For example, the toric code in two spatial dimensions has two non-topological $\mathbb{Z}_2$ 1-form symmetries corresponding to the $e$ and $m$ string operators. The fact that $e$ and $m$ strings anti-commute when they intersect an odd number of times is the manifestation of the mixed anomaly of these non-topological 1-form symmetries. (The braiding statistics of topological excitations in a gapped topological ground state that is invariant under the higher-form symmetry can give some information about the anomaly \cite{Kobayashi2025higherform}, but this would only be sensitive to the IR-non-trivial part of the lattice anomaly.)

In the case of non-invertible symmetries, they have been defined on the lattice in certain cases, e.g.~\cite{Seiberg_2307,Okada_2403}.
However, it remains unclear what  general definition of a non-invertible symmetry on the lattice  would be analogous to our definition in terms of QCAs for the invertible case. 
Formulating such a definition would be a necessary first step to generalizing our results to non-invertible symmetries.

Another interesting direction for future work would be to develop the classification of phases of commuting models (or relatedly, the classification of MBL phases). We showed in this work that the classification of invertible commuting models corresponds to the classification of lattice anomalies, but for non-invertible commuting models the questions remain wide open. This is also closely connected to the classification of local algebras (which appear at the boundary algebras of commuting models). Given that invertible local algebras can be viewed as (invertible) lattice gravitational anomalies, more general local algebras could be viewed as \emph{non-invertible} lattice gravitational anomalies.

Finally, let us comment on the implications of our work to the concept of ``symmetry topological field theory'' (symTFT) \cite{Apruzzi_2112,Freed_2209} that has recently attracted much attention as a way to organize the symmetries of QFTs. The symmetries of a QFT are described by a symTFT in one higher dimension. It was proposed in Ref.~\cite{Chatterjee_2203} that in lattice systems, the symTFT should correspond to the local isomorphism class of the algebra of the symmetric local operators. We note that this is a somewhat different viewpoint on symmetries than the one in this work, since in our work the symmetry group $G$ is fixed, whereas if one only looks at the local isomorphism class of the symmetric algebra, then different symmetry groups sometimes end up being equivalent (which is also true for symTFTs). 
However, we emphasize that it should always be necessary to take into account the differences between lattice systems and QFTs. 
Thus, one should \emph{not} expect that the symmetric algebras in lattice systems should correspond one-to-one with symTFTs, as proposed in Ref.~\cite{Chatterjee_2203}, because the symTFT is a QFT that does not know about the lattice. Instead, given our results (and those of Ref.~\cite{Jones_2307}) that boundary algebras are obtained specifically from commuting models, we expect that for characterizing symmetries on the lattice, the bulk symTFT should be replaced by a bulk \emph{commuting phase}, and the algebra of symmetric operators then corresponds to the boundary algebra of this bulk commuting phase.

\emph{Note added.---} During the preparation of this manuscript, we became aware of related, independent work in Ref.~\cite{Shirley2025anomalyfree},
which also constructs two-dimensional \(G\)-reps with non-trivial \(\grpH{2}{G}{\QCAblendclass{1}{}}\) class, and demonstrates that they admit trivial gapped ground states for symmetric commuting projector Hamiltonians. 
As a consequence of our Theorem~\ref{thm:ObstructionToInvertible}, these commuting projector models must be non-invertible in the sense defined in \autoref{sec:ConsequencesOfAnomalies}.
Note that Ref.~\cite{Shirley2025anomalyfree} used a terminology for ``anomaly'' that is not compatible with ours---they use the terminology ``anomaly-free'' to refer to $G$-reps that we would identify as carrying an IR-trivial lattice anomaly. We disagree with the assertion of Ref.~\cite{Shirley2025anomalyfree} that lattice anomalies (in our terminology) need not be RG-invariant---see \autoref{sec:rg_invariance}.

Additionally, we are aware of related independent work in Refs.~\cite{Czajka2025anomalieslatticehomotopy,Jones2025unpublished}.

\acknowledgements

The authors thank Maissam Barkeshli, Tyler Ellison, Davide Gaiotto, Jeongwan Haah, Salvatore Pace, Sahand Seifnashri, Wilbur Shirley, Ryan Thorngren, Ruben Verresen, and Carolyn Zhang for useful discussions. 
This work is supported by the Laboratory for Physical Sciences (YTT, DML), 
a Stanford Q-FARM Bloch Fellowship (DML), 
the US Department of Energy, Office of Science Award No.\ DE-SC0019380 (DML),
and a Packard Fellowship in Science and Engineering (DML, PI: Vedika Khemani). Research at Perimeter Institute is supported in part by the Government of Canada through the Department of Innovation, Science and Economic Development; and by the Province of Ontario through the Ministry of Colleges, Universities, Research Excellence and Security (DVE).

\appendix

\section{Blend equivalence}
\label{appendix:BlendEquivalenceDef}

In this appendix, we precisely define blend equivalence for \(G\)-reps.

Recall that \(G\)-reps \(\repU\) and \(\repW\) on the on-site algebra \(M\) are said to blend across \(R \subseteq \RR^d\) (usually taken to be a half-volume) if there is a \(G\)-rep \(\repV\) and distances \(\xi_g\) for \(g \in G\) such that
\begin{subequations}
\begin{align}
    a \in M\{\mathrm{int}_{\xi_g} R\} &\Rightarrow\repV_{g}[a]= \repU_{g}[a], \\
    a \in M\{\mathrm{int}_{\xi_g} R^c\} &\Rightarrow \repV_{g}[a] = \repW_{g}[a].
\end{align}
\end{subequations}
If \(\repV\) is a blend between \(\repU\) and \(\repW\), we write
\begin{equation}
    \repU \BlendVia{\repV} \repW.
\end{equation}
We further write
\begin{equation}
    \repU \TwoWayBlend \repW
\end{equation}
if there exist \(\repV\) and \(\repV'\) such that \(\repU \BlendVia{\repV} \repW\) and \(\repW \BlendVia{\repV'} \repU\).

It is not obvious that the existence of a blend with a fixed \(R\) is an equivalence relation.
We define \emph{blend equivalence} to be the smallest equivalence relation incorporating the relations \(\repU \BlendVia{\repV} \repW\). That is, if there is a sequence of blends like
\begin{equation}
    \repU=\repU_0 \BlendVia{\repV_0} \repU_1 \xleftarrow{\repV_1} \cdots \BlendVia{\repV_{n-1}} \repU_n = \repW,
    \label{eqn:OneWaySequence}
\end{equation}
with intermediate blends going in either direction, then \(\repU\) is blend equivalent to \(\repW\), which we write \(\repU \blend \repW\). The subset \(R\) is implicit in this notation.

We define a lattice anomaly to be an equivalence class of \(\blend\) with \(R = (-\infty,0]\times \RR^{d-1}\) fixed to be a half-volume,
\begin{equation}
    \LatAnom{d}{G} := \{G\text{-reps}\}/\blend.
\end{equation}
Denote the blend equivalence class of \(\repU\) by \([\repU]\).

Competing notions of blend equivalence, and hence lattice anomaly, may require that blends between \(\repU\) and \(\repW\) exist for many different choices of \(R\). These definitions could conceivably produce smaller equivalence classes. In any case, the cohomology invariants defined in \autoref{sec:CohomologyInvariants} are invariant under the coarsest possible definition of blend equivalence, which is the one we have adopted above.

Even with the same fixed \(R\), we could have made an alternative definition of \(\blend\) in which \(\repU \blend \repW\) when there is a finite sequence
\begin{equation}
    \repU=\repU_0 \TwoWayBlend \repU_1\TwoWayBlend \cdots\TwoWayBlend \repU_n = \repW.
\end{equation}
In fact, these two options result in the same definition of \(\blend\). To see this, observe that if \(\repU \TwoWayBlend \repW\), then we have \(\repU \BlendVia{\repV} \repW\), and so any \(G\)-reps which are equivalent under \(\TwoWayBlend\) are also equivalent using \(\BlendVia{\repV}\) for some \(\repV\). On the other hand, if \(\repU \BlendVia{\repV} \repW\), then we have blends \(\repU \BlendVia{\repV} \repV\) and \(\repV \BlendVia{\repU} \repU\), so \(\repU \TwoWayBlend \repV\). Similarly, \(\repV \TwoWayBlend \repW\), and so we have that \(\repU \BlendVia{\repV} \repW\) implies \(\repU \TwoWayBlend\repV\TwoWayBlend \repW\). Thus, any sequence like Eq.~\eqref{eqn:OneWaySequence} can be replaced by a sequence which is twice as long, but every arrow is \(\TwoWayBlend\), and so the two definitions of \(\blend\) are equivalent.

\begin{figure}
    \centering
    \includegraphics[width=\linewidth]{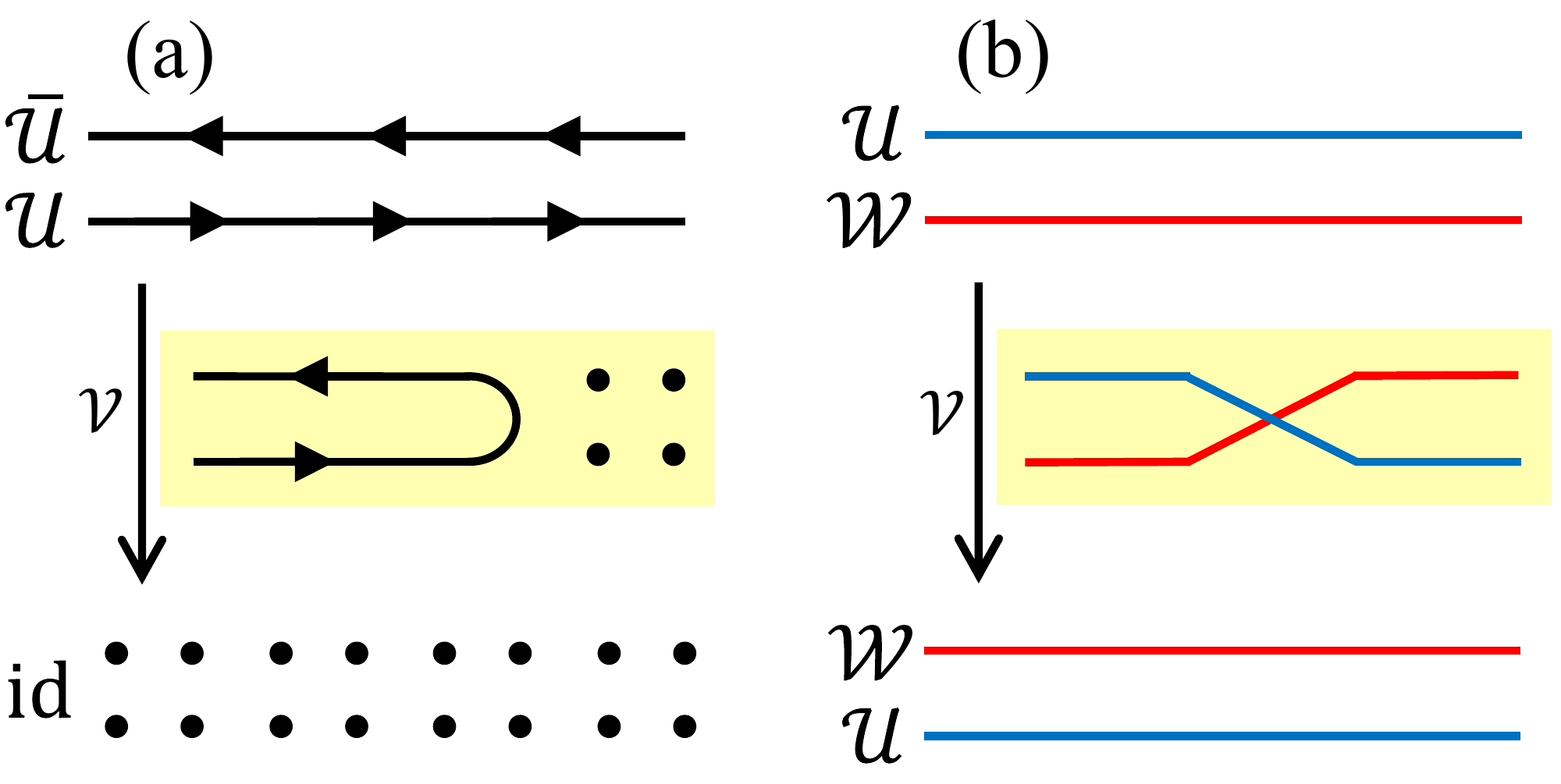}
    \caption{(a)~The reflection of a \(G\)-rep \(\repU\), denoted by \(\overline{\repU}\), acts as an inverse for blend equivalence. Indeed, folding \(\repU\) produces a blend \(\repU \otimes \overline{\repU} \BlendVia{\repV} \id\) (yellow). (b)~Swapping tensor factors in \(R^c\) provides a blend \(\repU \otimes \repW \BlendVia{\repV} \repW \otimes \repU\). In both cases (a,b), the reverse blends \(\id \BlendVia{\repV'} \repU \otimes \overline{\repU}\) and \(\repW \otimes \repU \BlendVia{\repV'}\repU \otimes \repW\) can be constructed similarly.}
    \label{fig:blend_group_laws}
\end{figure}

With \(R\) chosen to be a half-volume, lattice anomalies form an Abelian group under stacking of representations, \(\BlendClass{\repU} \BlendClass{\repW} = \BlendClass{\repU \otimes \repW}\). (Finer choices of blend equivalence typically do not guarantee that inverses exist.) The group operation is well defined: if \(\repU_1 \BlendVia{\repV_1} \repW_1\) and \(\repU_2 \BlendVia{\repV_2} \repW_2\), then the tensor product of blends shows \(\repU_1 \otimes \repU_2 \BlendVia{\repV_1\otimes\repV_2} \repW_1 \otimes \repW_2\); and if \(\repU_1 \blend \repW_1\) and \(\repU_2 \blend \repW_2\), taking a pointwise tensor product of blend sequences \(\repU_{01} \otimes \repU_{02} \TwoWayBlend \cdots \TwoWayBlend \repU_{n1} \otimes \repU_{n2}\) (possibly padding one sequence with constants \(\repW_{1,2}\) until they have the same length) gives a blend \(\repU_1 \otimes \repU_2 \blend \repW_1 \otimes \repW_2\). The identity element is the equivalence class containing on-site \(G\)-reps. An inverse of \(\repU\) is given by the reflection through the plane \(\{0\} \times \RR^{d-1}\), which we write \(\overline{\repU}\). The stack \(\repU \otimes \overline{\repU}\) can be seen to blend with the identity via \emph{folding} of \(\repU\). Finally, the group operation is Abelian, as seen by \emph{crossing} a stack of two \(G\)-reps (\autoref{fig:blend_group_laws}).

\section{IR-trivial implies symmetric gapped Hamiltonian with product-state ground state}
\label{appendix:symmetric_gapped_hamiltonian}
In this appendix we give the proof of the result mentioned in \autoref{sec:IntroLatticeToQFT}: if there is a $G$-rep such that there is a $G$-invariant product state, then that product state can appear as the ground state of a gapped symmetric Hamiltonian (so long as $G$ is finite, or $G$ is a compact  group and the $G$-rep is internal).

Consider a system with lattice sites labelled by a set $\Lambda$. For the moment we will take $\Lambda$ to finite. Let $\mathcal{U}$ be a $G$-rep, and 
suppose we have a $G$-invariant product state $\ket{\Psi} = \bigotimes_{s \in \Lambda} \ket{\phi_s}$. Then at each site $s$ we can introduce the on-site projector $\mathbb{P}_s$ which projects onto $\ket{\phi_s}$. Then the Hamiltonian
\begin{equation}
    H = J \sum_s (\unit - \mathbb{P}_s)
\end{equation}
(with $J > 0$)
has $\ket{\Psi}$ as its unique gapped ground state (with ground-state energy zero and gap $J$). The problem is that even if $\ket{\Psi}$ is $G$-invariant, this does not necessarily imply that $H$ is $G$-invariant.

Therefore, for finite $G$ we introduce the $G$-symmetrized Hamiltonian
\begin{equation}
    \widetilde{H} = \frac{1}{|G|} \sum_{g \in G} \mathcal{U}_g[H].
\end{equation}
Clearly, $\ket{\Psi}$ is still a ground state of $\widetilde{H}$. However, we need to verify that it is the unique ground state, and place a lower bound on the gap.

To show this, first observe that $\mathcal{U}_g[H]$ has the same ground state degeneracy as $H$, therefore $\ket{\Psi}$ is the unique ground state of $\mathcal{U}_g[H]$. Now consider any state $\ket{\Phi}$ that is orthogonal to $\ket{\Psi}$. Then we have
\begin{equation}
\label{eq:H_symmetrization}
    \bra{\Phi} H \ket{\Phi} = \frac{1}{|G|} \sum_{g \in G} \bra{\Phi} \mathcal{U}_g[H] \ket{\Phi}.
\end{equation}
Since $\ket{\Phi}$ is orthogonal to the non-degenerate ground state of $\mathcal{U}_g[H]$, we have
\begin{equation}
    \bra{\Phi} \mathcal{U}_g[H] \ket{\Phi} \geq J. 
\end{equation}
Hence we find that
\begin{equation}
    \bra{\Phi} \widetilde{H} \ket{\Phi} \geq J,
\end{equation}
which shows that the ground state of $\widetilde{H}$ is non-degenerate and it has a gap of at least $J$.

For compact groups, if the $G$-rep is internal then we can just replace the symmetrization of \eqnref{eq:H_symmetrization} by the integral over $g$ with the Haar measure and the argument follows in the same way.

For infinite systems we can make the above arguments rigorous by working in terms of the Gelfand--Naimark--Segal (GNS) Hilbert space constructed from the ground state.

\section{Fermionic local algebras}
\label{appendix:fermionic_local_algebra}

In this appendix, we define fermionic matrix algebras and fermionic local algebras.

First, we define a $\mathbb{Z}_2$-graded $*$-algebra to be a $*$-algebra $A$ with a vector space decomposition $A = A_+ \oplus A_-$, such that the adjoint $\dagger$ preserves $A_+$ and $A_-$, and $A_\sigma A_{\sigma'} \leq A_{\sigma \sigma'}$ where $\sigma, \sigma' \in \{ +1, -1 \}$. The idea is that $A_+$ and $A_-$  represent elements that are odd or even respectively under fermion parity.
The $\mathbb{Z}_2$-grading is equivalent to specifying a \emph{parity automorphism} $\mathcal{P} : A \to A$ such that $\mathcal{P}^2 = 1$,
where we define $A_\pm$ to be the $\pm 1$ eigenspaces of $\mathcal{P}$.

We define a $\mathbb{Z}_2$-graded subalgebra of $A$ to be a subalgebra $B \leq A$ such that $B = (B \cap A_+) \oplus (B \cap A_-)$. This is equivalent to saying that $B$ is preserved under $\mathcal{P}$. Given two $\mathbb{Z}_2$-graded $*$-algebras $A$ and $B$, we say a homomorphism $\eta : A \to B$ is $\mathbb{Z}_2$-graded if $\eta(A_+) \leq B_+$ and $\eta(A_-) \leq B_-$, or in other words $\eta \circ \mathcal{P}_A = \mathcal{P}_B \circ \eta$.

We say a $\mathbb{Z}_2$-graded $*$-algebra $A$ with parity automorphism $\mathcal{P}$ is a full fermionic matrix algebra if there exists a finite-dimensional Hilbert space $V$ and a linear operator $P$ on $V$ with $P^2 = 1$ such that $A$ is the algebra of linear operators on $V$, and $\mathcal{P}[a] = P a P$.

Next, we define a fermionic local algebra over a set $X$ to be a $\mathbb{Z}_2$-graded $*$-algebra equipped with a relation $a \tl S$ that satisfies properties 1--\ref*{item:last_before_commutator} from \autoref{subsec:local_algebras}, while \ref*{item:commutator_condition} is replaced with
\begin{enumerate}
    \item[\ref*{item:commutator_condition}'.] If $a \tl S$ and $b \tl S'$ with $S$ and $S'$ disjoint, then:
    \begin{itemize}
        \item If $a \in A_+$ or $b \in A_+$, then $ab = ba$.
        \item If $a,b \in A_-$, then $ab = -ba$.
    \end{itemize}
\end{enumerate}
and we also add the condition
\begin{enumerate}
    \item[9.] The parity automorphism $\mathcal{P}$ is ultra-local.
\end{enumerate}
In what follows, when we talk about subalgebras of a fermionic local algebra, we will always require them to be $\mathbb{Z}_2$-graded. We will also require (ultra-)local isomorphisms to be $\mathbb{Z}_2$-graded.

We also need to modify the definition of the tensor product for fermionic local algebras.
Let $A$ and $B$ be fermionic local algebras. Then we say that a fermionic local algebra $C$ containing both $A$ and $B$ is the internal tensor product of $A$ and $B$ if:
\begin{itemize}
\item $[A_+, B] = [A, B_+] = 0$.
\item $ab + ba = 0$ for $a \in A_-, b \in B_-$.
\item Treating $A$, $B$ and $C$ as vector spaces, there exists an invertible linear map $l : C \to A \otimes B$, such that $l(a) = a \otimes \unit_B$ for all $a \in A$ and $l(b) = \unit_A \otimes b$ for all $b \in B$.
\end{itemize}
More generally, given two fermionic local algebras $A$ and $B$, we say that another fermionic local algebra $C$ is the (external) fermionic tensor-product of $A$ and $B$ if there are $A',B' \leq C$ such that $C$ is the internal fermionic tensor product of $A'$ and $B'$ and there are ultra-local isomorphisms $A \to A'$ and $B \to B'$. One can show that for any $A$ and $B$, $C$ exists and is uniquely determined up to ultra-local isomorphism.

A local representation on a fermionic local algebra is defined analogously to \autoref{sec:what_is_a_symmetry}. We specify a fermionic symmetry group $G_f$ with a central subgroup $\mathbb{Z}_2^f \leq G_f$. Then a $G_f$-rep on a fermionic local algebra $A$ is a homomorphism from $G_f$ to the group of local automorphisms of $A$
such that the generator of $\mathbb{Z}_2^f$ maps to the parity automorphism of $A$.

Then, we can define stacking of $G_f$-reps. Let $A$ and $B$ be fermionic local algebras and $G_f$-reps $\mathcal{U}$ and $\mathcal{V}$. Then the stacked $G_f$-rep $\mathcal{U} \otimes \mathcal{V}$ is defined on the internal fermionic tensor product $C \geq A,B$ by the property that for all $g \in G_f$, $\mathcal{U}_g \otimes \mathcal{V}_g$ preserves $A$ and $B$, and the automorphisms induced on $A$ and $B$ are equal to $\mathcal{U}_g$ and $\mathcal{V}_g$ respectively.

\section{Technical details of the \texorpdfstring{$\grpH{3}{G}{\QCAblendclass{d-2}{}}$}{H3} invariant}
\label{appendix:H3_details}

We would like to prove that the 3-cochain $\omega_{g,h,k} \in \QCAblendclass{d-2}{}$ defining the \(\grpH{3}{G}{\QCAblendclass{d-2}{}}\) invariant in \autoref{sec:H3Invariant} is indeed a 3-cocycle, and furthermore, that up to 3-coboundaries it is independent of the choice of restrictions. The proof follows precisely the same algebraic steps described in Ref.~\cite[Appendix~B]{Else2014}, so we will not reproduce them here. However, we will mention that in order to for these steps to work, we need to establish the following lemma:
\begin{lemma}
\label{lem:conjugation_lemma}
 Let $U_R$ be a QCA supported on the region $R \subseteq \mathbb{R}^d$, with $\partial R$ being a $d-1$-dimensional hyperplane, let $V$ be a QCA, and let $V_R$ be a QCA supported on $R$ that acts like $V$ in the interior of $R$. Then $V U_R V^{-1} (V_R U_R V_R^{-1})^{-1}$ is a QCA supported in the vicinity of $\partial R$ and has trivial $(d-1)$-dimensional QCA index.
\end{lemma}
This plays the analogous role in the derivation to Eq.~(B9) in Ref.~\cite{Else2014}.

In order to prove Lemma \ref{lem:conjugation_lemma}, we will first give the following Lemma:
\begin{lemma}
\label{lem:simpler_conjugation_lemma}
Let $U$ be a QCA in $d$ dimensions, and let $\Sigma$ be a QCA supported in the vicinity of a $k$-dimensional hyperplane in $\mathbb{R}^d$. Then $\Sigma U \Sigma^{-1} = \Gamma U$, where $\Gamma$ is a QCA supported in the vicinity of the $k$-dimensional hyperplane with the trivial $k$-dimensional QCA class.
\begin{proof}
One can show that $U \Sigma^{-1} U^{-1}$ has the same $k$-dimensional QCA index as $\Sigma^{-1}$. 
[Stack with ancillas and compare to \((U \otimes U^{-1})(\Sigma^{-1}\otimes \id)(U^{-1}\otimes U)\), using that \(U \otimes U^{-1}\) is a finite depth circuit.]
Therefore, we have $U \Sigma^{-1} U^{-1} = \Sigma^{-1} \Gamma$ for some QCA $\Gamma$ supported in the vicinity of the $k$-dimensional hyperplane with trivial $k$-dimensional QCA index. Multiplying this equation on the right by $U$ and on the left by $\Sigma$ gives the desired result.
\end{proof}
\end{lemma}

Now we can give the proof of Lemma \ref{lem:conjugation_lemma}.
We can write $V = \Sigma_{\partial R} V_{R} V_{R^c}$, where $V_{R^c}$ is supported on $R^c = \mathbb{R}^d \setminus R$, and $\Sigma_{\partial R}$ is supported in the vicinity of $\partial R$. Then we have
\begin{equation}
    V U_R V^{-1} = \Sigma_{\partial R} ( V_R U_R V_R^{-1} ) \Sigma_{\partial R}^{-1},
\end{equation}
where we used the fact that $V_{R^c}$ commutes with $U_R$. Then, by Lemma \ref{lem:simpler_conjugation_lemma}, we have that
\begin{equation}
\Sigma_{\partial R} ( V_R U_R V_R^{-1} ) \Sigma_{\partial R}^{-1} = \Gamma V_R U_R V_R^{-1},
\end{equation}
where $\Gamma$ is supported on the vicinity of $\partial R$ and has trivial $(d-1)$-dimensional QCA class. The result follows.

\section{The full actions of the QCAs in the \texorpdfstring{$e$-$m$ exchange $\ZZ_2$}{e-m exchange Z2}-rep}\label{sec:em-QCA}

In this appendix, we state the detailed action of the toric code $e$-$m$ exchange QCA $U$, the corresponding $\ZZ_2$-rep $\repU$, and their restrictions, defined in \autoref{sec:em-qubit}.

In \autoref{fig:em-QCA}, the action of $U$ and $U^2$, and $\resU_x^2=\bdyU_{x,x}$ on all the single-site Pauli $X$ (red) and $Z$ (blue) operators, including those near the boundary, are shown.
For $\bdyU_{x,x}$, the actions on $X$ and $Z$ are simply that of $U^2$ as shown in the figure with the additional plaquette and star operators in the highlighted hexagons removed.

\begin{figure}
    \centering
    \includegraphics[width=\linewidth]{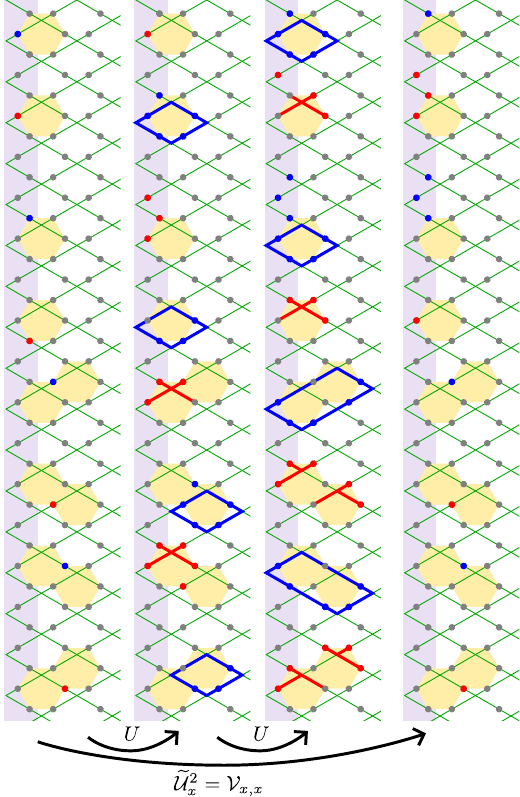}
    \caption{The action of $U$, $U^2$, and $\resU_x^2=\bdyU_{x,x}$ on single-site Pauli operators \(\sigma^x\) (red) and \(\sigma^z\) (blue). Plaquettes highlighted with yellow corresponds to the hexagons where excitations are created by the corresponding single-site Pauli. The left side of the lattice is the boundary. Operators further away from the boundary have the same evolutions as in the last four rows. Note that $\bdyU_{x,x}$ only acts nontrivially on the boundary strip (purple). The meaning of colors follows \autoref{fig:em-stab}.}
    \label{fig:em-QCA}
\end{figure}

\section{Equivalent definitions of invertible commuting model}
\label{appendix:InvertibleIffQCAPreparable}

A commuting model $\{ h_i : i \in I \}$ (with associated Hamiltonian \(H = \sum_{i\in I} h_i\)) in an on-site algebra $A$ is said to be \emph{invertible} if there exists an on-site algebra $\overline{A}$ and a commuting model $\{ \overline{h}_j : j \in J \}$ (with Hamiltonian \(\overline{H}\)) on $\overline{A}$ such the ``stacked'' model
\begin{equation}
    \{ h_i \otimes \mathbb{I} : i \in I \} \cup \{ \mathbb{I} \otimes \overline{h}_j : j \in J \}
\end{equation}
(with Hamiltonian \(H\otimes \unit + \unit \otimes \overline{H}\))
is equivalent by a finite-depth circuit $U_{\mathrm{circ}}$ to a \TrivCommutingModel{} commuting model, i.e.\ one in which each local term acts on a single site, and where the ground state is non-degenerate.

In this appendix we will show that a commuting model is invertible if and only if it is related to an on-site commuting model by a QCA.

First of all, let us suppose $\{ h_i : i \in I \}$ can be related to an on-site commuting model $\{ h_i^{0} : i \in I \}$ by a QCA $U$, that is, $h_i = U[h_i^0]$. Then using the fact that $U \otimes U^{-1}$ is a finite-depth circuit, we immediately find that $U^{-1}[h_i^0]$ is the stacking inverse for $h_i$.

Conversely, suppose we have an invertible commuting model. The non-degeneracy of the ground state for the stacked model implies non-degeneracy of the ground state for $\overline{H}$, and hence the projected algebra $\overline{A}(\overline{H})$ contains only the scalar multiples of the unit. Hence we can define a local isomorphism $U : A \to A' := (A \otimes \overline{A})(U_{\mathrm{circ}}[\mathbb{I} \otimes \overline{H}])$ through the composition
\begin{multline}
    A \cong A \otimes \overline{A}(\overline{H}) \cong (A \otimes \overline{A})(\mathbb{I} \otimes \overline{H}) \\
    \cong (A \otimes \overline{A})(U_{\mathrm{circ}}[\mathbb{I} \otimes \overline{H}]) = A'.
\end{multline}
For the second isomorphism, we used the stacking property of projected algebras (see Appendix \ref{appendix:AlternativeProjectedAlgebra}).
Note that the fact that $U_{\mathrm{circ}}[\mathbb{I} \otimes \overline{H}]$ is a sum of on-site terms implies that $A'$ is an on-site algebra. However it may not be ultra-locally isomorphic to $A$. Moreover, under the isomorphism $U$, $H \otimes \mathbb{I}$ maps to the image of $U_{\mathrm{circ}}[H \otimes \unit] \in \mathcal{A}(U_{\mathrm{circ}}[\unit \otimes \overline{H}])$ under the projection map to $A'$. Since, by assumption, $U_{\mathrm{circ}}[H \otimes \unit]$ is the sum of on-site terms, it follows that the same is true after the projection. Thus, $U$ maps $H \otimes \unit$ to a sum of on-site terms in $A'$. Non-degeneracy of the ground state of $H \otimes \unit + \unit \otimes \overline{H}$ implies that $U[H]$ has non-degenerate ground state.

Note that a local isomorphism from $A$ to $A'$ is not technically a QCA, since the source and target are \emph{different} on-site algebras. However, one can show that it is possible to rearrange sites in \(A'\) so that it becomes ultra-locally isomorphic to \(A\) using techniques similar to Ref.~\cite[Appendix~A]{Long2024}. Then the local isomorphism can be regarded as mapping \(A\) to \(A\).

\section{Projected and boundary algebras in an infinite system}
\label{appendix:infinite_system}

Projected and boundary algebras, as introduced in \autoref{sec:ConsequencesOfAnomalies}, were defined for finite systems. In this appendix, we construct a definition of a projected algebra which also works in infinite systems. For generality, we will also give the definition of projected algebra inside a general local algebra $A$, without requiring it to be an on-site algebra. We caution the reader that for technical reasons, when we refer to the projected algebra in the arguments in main text, we will more precisely mean the formulation of the projected algebra described in Appendix \ref{appendix:AlternativeProjectedAlgebra}. That version is locally isomorphic, but not necessarily \emph{ultra}-locally isomorphic, to the version described in this appendix.

First we recall that for any $*$-algebra $A$, we can define a \emph{state} on $A$ to be a linear map $\omega : A \to \mathbb{C}$ such that $\omega(a^\dagger) = \omega(a)^*$ and $\omega(a^\dagger a) \geq 0$ for all $a \in A$, and $\omega(\unit) = 1$. In what follows we will find it more convenient to work with ``un-normalized'' states for which we do not demand that $\omega(\unit) = 1$; however we note that every un-normalized state can be obtained simply by multiplying a normalized state by a non-negative real number.

Let us first consider the bosonic case. Suppose we have a commuting Hamiltonian $H = \sum_{i \in I} h_i$ (recall that in general we interpret this as a formal linear sum).
We require each $h_i$ to be positive, which in the context of a general local algebra we define to mean that $h_i = \gamma_i^\dagger \gamma_i$ for some $\gamma_i \in A$.
We also require the model to be frustration-free, i.e.\ there exists at least one state $\omega$ such that $\omega(H) = 0$. Then, we define $\Omega(H)$ to be the set of all un-normalized states $\omega$ on $A$ such that $\omega(H) = 0$. [Note that given the positivity conditions on the terms of $H$, $\omega(H) = 0$ if and only if $\omega(h_i) = 0$ for all $i \in I$.] We assume that $H$ is frustration-free so that $\Omega(H)$ is not the empty set.

Now we can define the subalgebra $\mathcal{A}(H) \leq A$ according to
\begin{multline}
    \mathcal{A}(H) = \{ a \in A: \omega(a^{\dagger} H a) = \omega(a H a^{\dagger}) = 0 \\ 
    \mbox{ for all $\omega \in \Omega(H)$} \},
\end{multline}
[note that the positivity of the local terms of $H$ ensures that $\omega(a^\dagger H a) = 0$ if and only if $\omega(a^\dagger h_i a) = 0$ for all $i \in I$.]
We can equivalently write
\begin{equation}\label{eqn:InfiniteProjectedAlgebraPreQuotient}
    \mathcal{A}(H) = \{ a \in A : a \omega, a^\dagger \omega \in \Omega(H) \mbox{ for all $\omega \in \Omega(H)$} \},
\end{equation}
where we defined $a\omega$ to be the un-normalized state $(a\omega)(O) = \omega(a^{\dagger} O a)$.

Equation~\eqref{eqn:InfiniteProjectedAlgebraPreQuotient} makes it clear that \(\mathcal{A}(H)\) is closed under multiplication and the adjoint \(a \mapsto a^\dagger\). To show that \(\mathcal{A}(H)\) is closed under addition, one uses the fact that for any positive $h_i \in A$, and any state $\omega$,
\begin{equation}
    \langle a,b \rangle = \omega(a^{\dagger} h_i b) 
\end{equation}
defines a positive-semidefinite Hermitian form on $A$. Therefore, one has a triangle inequality
\begin{equation}
    \langle a + b , a + b \rangle \leq \left(\sqrt{\langle a, a \rangle} + \sqrt{\langle b, b\rangle}\right)^2.
\end{equation}
Therefore, if $\langle a, a \rangle = \langle b,b \rangle =0$ then $\langle a + b, a + b\rangle = 0$.

Next, we define
\begin{equation}
    \mathcal{B}(H) = \{ a \in A : \omega(a^{\dagger} a) = \omega(a a^\dagger) = 0
    \mbox{ for all $\omega \in \Omega(H)$} \},
\end{equation}
Using a triangle inequality as before, one can show that $\mathcal{B}(H)$ is closed under addition, and any element $a \in \mathcal{B}(H)$ satisfies $a \omega = a^\dagger \omega = 0$ for all $\omega \in \Omega(H)$.
Moreover, it is a subalgebra, in fact a two-sided ideal of $\mathcal{A}(H)$. For example, to show that it is a left-ideal, suppose that $a \in \mathcal{A}(H)$ and $b \in \mathcal{B}(H)$. Then
\begin{equation}
    (ab)\omega = a(b\omega) = 0,
\end{equation}
while
\begin{equation}
    (ab)^\dagger \omega = b^\dagger(a^\dagger \omega) = 0,
\end{equation}
using the fact that $a^\dagger \omega \in \Omega(H)$.

Then the projected algebra corresponding to the commuting model \(H\) is the quotient
\begin{equation}
    A(H) = \mathcal{A}(H)/\mathcal{B}(H).
\end{equation}
where we define the locality structure by stating that if $a \in A(H)$, then $a \tl S$ if and only if there is a representative for $a$ in $\mathcal{A}(H)$ which is supported on $S$.

We can perform a similar procedure in the fermionic case (recall the definitions of fermionic local algebra from Appendix~\ref{appendix:fermionic_local_algebra}). The main difference is that, given a fermionic local algebra $A = A_+ \oplus A_-$, we require states to satisfy the condition that $\omega(a) = 0$ for any $a \in A_-$. Physically, this encodes the condition that states with different fermion parity are in different superselection sectors and cannot be coherently superposed. This condition can equivalently be written in terms of the parity automorphism $\mathcal{P}$ as $\mathcal{P} \omega = \omega$, where we defined $(\mathcal{P} \omega)(O) = \omega(\mathcal{P}[O])$.

We can then form the fermionic local algebras $\mathcal{A}(H)$ and $\mathcal{B}(H)$ as before. (One readily verifies that they are $\mathbb{Z}_2$-graded subalgebras of $A$.) One can then show that $\mathcal{A}(H) / \mathcal{B}(H)$ has the structure of a fermionic local algebra.

The definitions described in this appendix, while simple to state, are often somewhat cumbersome to prove results about.
In Appendix~\ref{appendix:AlternativeProjectedAlgebra} below, we will give an alternative equivalent definition of the projected algebra that is sometimes more convenient.

Here we will just mention some results that follow directly from the definitions in the current appendix. We restrict to the case where the ambient local algebra $A$ has the property that $A \{ S \}$ can be given the structure of a finite-dimensional $C^*$-algebra~\cite{Murphy1990CstarBook} for every bounded set $S$, such that for every pair of bounded $S \subseteq S'$, the inclusion $A \{ S \} \to A \{ S' \}$ is a $C^*$-algebra homomorphism. We remark that, in particular, any on-site or invertible local algebra is a locally-$C^*$ algebra.

First, we note that from the above definition, it is clear that normalized states on $A(H)$ are in one-to-one correspondence with normalized states in $\Omega(H)$.
Second, if $\Omega(H)$ contains only a single normalized state $\omega$, then $A(H)$ is trivial (i.e. it only contains scalar multiples of the unit). To see this, note that from the result we just mentioned, if $\Omega(H)$ contains only a single normalized state $\omega$, it implies that $A(H)$ is a local algebra that admits only a single normalized state. One can show that such a locally-$C^*$ algebra is necessarily trivial.

\section{An alternative representation of the projected algebra}
\label{appendix:AlternativeProjectedAlgebra}

In this appendix, we give an alternative definition of the projected algebra. For simplicity we will focus on the bosonic case, but equivalent definitions and results can be formulated for the fermionic case.
The intuition is based on the following observation: if $A$ is the algebra of operators on a finite-dimensional Hilbert space, then  the definition of projected algebra can be equivalently formulated as $A(H) = \{ \mathbb{P} a \mathbb{P} : a \in A \}$, where $\mathbb{P}$ is the projector onto the ground state subspace of $H$. Below, we develop a version of this definition that is constructed in a more local way, allowing it to be generalized to infinite systems.
We again restrict to the case where the ambient local algebra $A$ is a locally-$C^*$ algebra (recall the definition from Appendix \ref{appendix:infinite_system}).

Locally-$C^*$ algebras are closely related to another concept, which we call \emph{local nets of $C^*$-algebras}. We define a local net of $C^*$-algebras over a metric space $X$ to be an object  $A$ that assigns to each bounded set $S \subseteq X$ a finite-dimensional $C^*$-algebra $A \{ S \}$, and to each pair of bounded sets $S \subseteq S'$ a $C^*$-algebra homomorphism $\eta_{ S, S' }: A_S \to A_{S'}$, such that $\eta_{S,S} = \id$ and the following diagram commutes for any bounded $S_1, S_2, S_3 \subseteq X$ with $S_1 \subseteq S_2 \subseteq S_3$:
\begin{equation}
    \begin{tikzcd}
         & A \{ S_3 \}  \\
    A\{ S_2 \} \arrow[ur,"\eta_{S_2,S_3}"] & \\
            & A \{ S_1 \} \arrow[lu, "\eta_{S_1,S_2}"] 
            \arrow[uu, "\eta_{S_1,S_3}"']
    \end{tikzcd}.
\end{equation}
We also require that if $S$ and $S'$ are disjoint, then any element of the image of $\eta_{S,S \cup S'}$ commutes with any element of the image of $\eta_{S',S \cup S'}$.

If $A$ and $\widetilde{A}$ are local nets of $C^*$-algebras, we say that an ultra-local homomorphism $\varphi : A \to \widetilde{A}$ is a collection of homomorphisms $\varphi_S : A\{ S \} \to \widetilde{A} \{ S \}$ such that the following diagram commutes for all $S \subseteq S'$:
\begin{equation}
\begin{tikzcd}
    A\{ S' \} \arrow[r,"\varphi_{S'}"] & \widetilde{A} \{ S' \} \\ A \{ S \} \arrow[r, "\varphi_{S}"] \arrow[u, "\eta_{S,S'}"] & \widetilde{A} \{ S \} \arrow[u, "\widetilde{\eta}_{S,S'}"']
    \end{tikzcd}
\end{equation}
We say this is an ultra-local isomorphism if $\varphi_S$ is an isomorphism for all $S$.

We say that a local homomorphism $\varphi : A \to \widetilde{A}$ of range $r$ is a collection of homomorphisms $\varphi_S : A \{ S \} \to \widetilde{A} \{ S^{+r} \}$ such that the following diagram commutes for all $S \subseteq S'$:
\begin{equation}
\begin{tikzcd}
    A\{ S' \} \arrow[r,"\varphi_{S'}"] & \widetilde{A} \{ (S')^{+r} \} \\ A \{ S \} \arrow[r, "\varphi_{S}"] \arrow[u, "\eta_{S,S'}"] & \widetilde{A} \{ S^{+r} \} \arrow[u, "\widetilde{\eta}_{S^{+r},(S')^{+r}}"']
    \end{tikzcd}
\end{equation}
We say that a local homomorphism $\varphi_S$ of range $r$ from $A$ to itself is the identity if $\varphi_S = \eta_{S, S^{+r}}$. In general, we say that a local homomorphism $\varphi : A \to \widetilde{A}$ is a local isomorphism if there exists a local homomorphism $\varphi^{-1}$ such that $\varphi^{-1} \circ \varphi$ and $\varphi \circ \varphi^{-1}$ are the identity.

Given a local net of $C^*$-algebras, we define a subalgebra $B \leq A$ to be a collection of subalgebras $B\{ S \} \leq A \{ S \}$ such that $\eta_{S,S'}(B\{S\}) \leq B(S')$. We say that $B$ is a two-sided ideal of $A$ if $B \{ S \}$ is a two-sided ideal of $ A \{ S \}$ for all $S$. If $B$ is a two-sided ideal of $A$, we can define the quotient net $(A/B) \{ S\} = A \{ S \} / B\{ S \}$.

We define a local net of $C^*$-algebras to be \emph{stable} if the homomorphisms $\eta_{S,S'}$ are injective for all $S,S'$. 
For any unstable local net of $C^*$-algebras $A$, we can define a stable net $\overline{A}$ according to
\begin{equation}
\label{eq:stabilization}
\overline{A} = A / A_0,
\end{equation}
where $A_0$ is the two-sided ideal defined by
\begin{equation}
    A_0 \{ S \} = \bigcup_{S' \supseteq S \, \mathrm{bounded}} \mathrm{ker} \, \eta_{S,S'}.
\end{equation}

Note that any locally-$C^*$ algebra $A$ defines a stable local net of $C^*$-algebras with the homomorphism $\eta_{S,S'}$ being the inclusions $A \{ S \} \to A \{ S '\}$. Conversely, one can argue that any stable local net of $C^*$-algebras defines a locally-$C^*$-algebra; thus, they are essentially equivalent concepts. Moreover, the definitions of (ultra-)local isomorphism, quotients, etc. given above for local nets of $C^*$ algebras correspond to the definitions given for local algebras previously.

Now consider a locally-$C^*$ algebra $A$. We say an element $\mathbb{P} \in A$ is a \emph{projector} if it is positive and satisfies $\mathbb{P}^2 = \mathbb{P}$. We then define a \emph{net of projectors} $\mathbb{P}$ to be an object that assigns a projector $\mathbb{P}_S \in A$ supported on $S$ to each bounded set $S \subseteq X$, such that for any two bounded sets $S \subseteq S'$, we have~\footnote{For elements $a,b \in A$, $a \leq b$ means that $b - a$ is positive.} $\mathbb{P}_{S'} \leq \mathbb{P}_S$. We also impose the frustration-free condition: $\mathbb{P}_S \neq 0$ for any bounded set $S \subseteq X$.
Note that for any commuting model $\{ h_i \}$ with $h_i$ positive and supported on $X_i$, we can define the corresponding net of projectors by replacing $h_i$ with a projector $\mathbb{Q}_i$ that annihilates the same set of states (one can show that this always exists in any locally-$C^*$ algebra), and then defining
\begin{equation}
\label{eq:projectors_from_commuting_model}
    \mathbb{P}_S = \prod_{i \in I : X_i \subseteq S} \mathbb{P}_i
\end{equation}
where $\mathbb{P}_i = \unit - \mathbb{Q}_i$. We will also define $\mathbb{Q}_S = \unit - \mathbb{P}_S$.

Given any net of projectors, we can define a projected algebra in an analogous way to Appendix \ref{appendix:infinite_system} by defining $\Omega(\mathbb{P})$ to be the set of states $\omega$ such that $\omega(\mathbb{P}_S) = 1$ for all bounded sets $S$, and then defining
\begin{multline}
    \mathcal{A}(\mathbb{P}) = \{ a \in A : \omega(a \mathbb{Q}_S a^\dagger) = \omega(a^\dagger \mathbb{Q}_S a) = 0 \\ 
    \mbox{ for all bounded $S$ and all $\omega \in \Omega(\mathbb{P})$} \},
\end{multline}
and
\begin{equation}
    \mathcal{B}(\mathbb{P}) = \{ a \in A : \omega(a^\dagger a) = \omega(aa^\dagger) = 0 \mbox{ for all $\omega \in \Omega(\mathbb{P})$} \},
\end{equation}
and
\begin{equation}
    A(\mathbb{P}) = \mathcal{A}(\mathbb{P}) / \mathcal{B}(\mathbb{P}). \label{eq:projector_net_quotient}
\end{equation}

Now we define a \emph{proximity relation} on $X$ to be a symmetric binary relation $\sim$ on $X$. We say that the proximity relation is range-$r$ if $x \sim y$ implies $d(x,y) \leq r$. Then we define $S^{+} = \{ x \in X : x \sim y \mbox{ for some $y \in S$} \}$. With respect to a proximity relation, we say that a net of projectors is \emph{progressive} if for all bounded sets $S \subseteq S'$, we have
\begin{equation}
\mathbb{P}_{S'} = \mathbb{P}_S \mathbb{B},
\end{equation}
for some projector $\mathbb{B}$ that is supported on $(S^c)^{+}$. (Note that this implies that $\mathbb{B}$ and $\mathbb{P}_S$ commute, because the product of two projectors is a projector if and only if they commute.) In particular, given a commuting projector model, \eqnref{eq:projectors_from_commuting_model} defines a progressive net of projectors with respect to the proximity relation
\begin{equation}
    x \sim y \Leftrightarrow \exists i \in I : \{ x, y \} \subseteq X_i.
\end{equation}

Given any progressive net of projectors $\mathbb{P}$ in a locally-$C^*$ algebra $A$, we can define a local net (which might not be stable), which we will call $\mathbb{P} A \mathbb{P}$, according to
\begin{equation}
    (\mathbb{P} A \mathbb{P})\{ S \} = \mathbb{P}_{S} A \{ \mathrm{int} S\} \mathbb{P}_{S},
\end{equation}
with $\mathbb{P}_S$ as the unit of $(\mathbb{P} A \mathbb{P})\{ S \}$, and
with homomorphisms $\eta_{S,S'}(a) = \mathbb{P}_{S'} a \mathbb{P}_{S'}$. Here we have defined $\mathrm{int} S := ((S^c)^{+})^c$. The progressivity ensures that $\eta_{S,S'}$ is a homomorphism, since if $a = \mathbb{P}_S b \mathbb{P}_S$ and $a' = \mathbb{P}_S b' \mathbb{P}_S$ for some $b,b' \in A\{ \mathrm{int} S \}$ then we see that $a$ and $a'$ commute with $\mathbb{B}$, and  $\mathbb{P}_S a' = a'$, and hence we have
\begin{subequations}
\begin{align}
    (\mathbb{P}_{S'} a \mathbb{P}_{S'}) (\mathbb{P}_{S'}a' \mathbb{P}_{S'}) &= \mathbb{P}_{S'} a \mathbb{P}_{S'} a' \mathbb{P}_{S'} \\
    &= \mathbb{P}_{S'} a \mathbb{B} \mathbb{P}_S a' \mathbb{P}_{S'} \\
    &= \mathbb{P}_{S'} \mathbb{B} a b \mathbb{P}_{S'} \\
    &= \mathbb{P}_{S'} ab \mathbb{P}_{S'}
\end{align}
\end{subequations}

We can then form the stabilized net $\overline{\mathbb{P} A \mathbb{P}}$ via \eqnref{eq:stabilization}. Note that this does implicitly depend on a choice of proximity relation, although the choice does not affect the local isomorphism class as long as we consider only finite-range proximity relations. (Thus, different choices will give locally isomorphic, but not necessarily \emph{ultra}-locally isomorphic, results.)

The rest of this appendix will be devoted to proving that $\overline{\mathbb{P} A \mathbb{P}}$ actually agrees (up to local isomorphism) with the projected algebra defined via \eqnref{eq:projector_net_quotient}, and then we will discuss some applications of this result at the end.
Our main result is:
\begin{thm}
\label{thm:equivalent_projected}
    Suppose $\mathbb{P}$ is a net of projectors in a locally-$C^*$ algebra $A$ with respect to a proximity relation of range $r$. Then there exists a local homomorphism (with range $r$) $\varphi : A(\mathbb{P}) \to \overline{\mathbb{P} A \mathbb{P}}$, and an ultra-local homomorphism $\psi : \overline{\mathbb{P} A \mathbb{P}} \to A(\mathbb{P})$ such that $\varphi \circ \psi$ and $\psi \circ \varphi$ are the identity. (More precisely, $\varphi$ maps elements supported on $S$ to elements supported on $S^+$).
\end{thm}
 In particular, this implies that $A(\mathbb{P})$ and $\overline{\mathbb{P} A \mathbb{P}}$ are locally isomorphic.

The proof proceeds by a series of lemmas:

    \begin{lemma}
    \label{lemma:faithful_state}
         We say that a state $\omega$ on a $*$-algebra $A$ is \emph{faithful} if $\omega(a^\dagger a) = 0$ implies $a = 0$. 
         The lemma states that any locally-$C^*$ algebra admits a faithful state.

         \begin{proof}
             The norm closure of a locally-$C^*$ algebra is a separable $C^*$-algebra, and then it is a standard result that any separable $C^*$-algebra admits a faithful state.
         \end{proof}
    \end{lemma}

\begin{lemma}
\label{lemma:projector_characterization_of_B}
    Let $a \in A$. Then $a \in \mathcal{B}(\mathbb{P})$ if and only if there exists a bounded set $S'$ such that $ a \mathbb{P}_{S'} = \mathbb{P}_{S'} a = 0$
    \begin{proof}
        If $a \mathbb{P}_{S'} = 0 = \mathbb{P}_{S'} a$, then we have $a \mathbb{Q}_{S'} = a$, and hence 
        $\omega(a^\dagger a) = \omega(a^\dagger a \mathbb{Q}_{S'}) = 0$. Similarly we find $\omega(a a^\dagger) = 0$.

        For the converse, first let us define for a bounded set $S'$, $\Omega^{S'}(\mathbb{P})$ to be the states $\omega$ such that $\omega(\mathbb{P}_{S'}) = 1$. Then we have
        \begin{equation}
            \Omega(\mathbb{P}) = \bigcap_{S' \text{ bounded}} \Omega^{S'}(\mathbb{P}).
        \end{equation}
        Now, let us consider some bounded set $S$, and define
        \begin{multline}
            \mathcal{B}_S^{S'} = \{ a \in A\{ S \} : \omega(a^\dagger a) = \omega(a a^\dagger) = 0 \\
            \mbox{ for all $\omega \in \Omega^{S'}(\mathbb{P})$} \},
        \end{multline}
        which is a subalgebra of $A \{ S \}$.
        In particular, define $f(l) = \mathrm{dim} \mathcal{B}_S^{S^{+l}}$. This is an increasing function of $l$ that is valued in the non-negative integers and bounded above by $\mathrm{dim} A \{ S \}$, hence there must be some $l_*$ such that $f(l) = f(l_*)$ for all $l \geq l_*$. This implies that $\mathcal{B}_S^{S'} = \mathcal{B}_{S}^{S_*} := \mathcal{B}_*$ for any bounded $S' \supseteq S_*$, where we defined $S_* = S^{+l_*}$. Therefore, $\mathcal{B}(\mathbb{P}) \{ S \} = \mathcal{B}_*$.

        Now, let $\omega$ be a faithful state on $A$, which exists by Lemma~\ref{lemma:faithful_state}. Now define $\omega' = \mathbb{P}_{S_*} \omega$. Then $\omega'(\mathbb{Q}_{S_*}) = 0$, so $\omega' \in \Omega^{S_*}(\mathbb{P})$.  Therefore, if $a \in \mathcal{B}_*$, then $0 = \omega'(a^\dagger a) = \omega(\mathbb{P}_{S_*} a^\dagger a \mathbb{P}_{S_*})$, which implies $ a \mathbb{P}_{S_*} = 0$. Similarly, we find $\mathbb{P}_{S_*} a = 0$.
    \end{proof}
\end{lemma}

\begin{lemma}
\label{lemma:projected_inclusion}
    Let $a \in A \{ \mathrm{int} S \}$. Then $\mathbb{P}_S a \mathbb{P}_S \in \mathcal{A}(\mathbb{P})$.
    \begin{proof}
    For all $\omega \in \Omega(\mathbb{P})$, and all $S' \supseteq S$, we have
    \begin{align}
    \omega([\mathbb{P}_S a^\dagger \mathbb{P}_S] \mathbb{P}_{S'} [\mathbb{P}_S a \mathbb{P}_S]) &= \omega([\mathbb{P}_S a^\dagger \mathbb{P}_S] \mathbb{B} \mathbb{P}_S \mathbb{B} [\mathbb{P}_S a^\dagger \mathbb{P}_S]) \nonumber \\
    &= \omega(\mathbb{B} [\mathbb{P}_S a^\dagger \mathbb{P}_S] [\mathbb{P}_S a^\dagger \mathbb{P}_S] \mathbb{B}) \nonumber \\
    &= (\mathbb{B} \omega)([\mathbb{P}_S a^\dagger \mathbb{P}_S] [\mathbb{P}_S a^\dagger \mathbb{P}_S]), \nonumber \\
    &= \omega([\mathbb{P}_S a^\dagger \mathbb{P}_S] [\mathbb{P}_S a^\dagger \mathbb{P}_S]),
    \end{align}
    which shows that $\mathbb{P}_S a \mathbb{P}_S \in \mathcal{A}(\mathbb{P})$. Here we used the fact that $\mathbb{B}$ commutes with $\mathbb{P}_S$ and $a$, and that $\mathbb{B} \omega = \omega$. To see the latter, note that $\omega = \mathbb{P}_{S'} \omega = \mathbb{P}_S \mathbb{B} \omega$, which can only be true if $\mathbb{B} \omega = \omega$.
    \end{proof}
\end{lemma}

\begin{lemma}
\label{lemma:l_lemma}
    Let $a \in \mathcal{A}(\mathbb{P}) \{ S \}$, and let $S' \supseteq S$. Then the linear map
    \begin{equation}
        l_{S,S'} : a \mapsto \mathbb{P}_{S'} a \mathbb{P}_{S'}
    \end{equation}
    satisfies $l_{S,S'}(a) - a \in \mathcal{B}(\mathbb{P})$ and $l_{S,S'}(a) l_{S,S'}(a') - a a' \in \mathcal{B}(\mathbb{P})$.
    
    \begin{proof}
        First of all, we have $\mathbb{P}_{S'} a \mathbb{P}_{S'} - a = \mathbb{Q}_{S'} a + a \mathbb{Q}_{S'} - \mathbb{Q}_{S'} a \mathbb{Q}_{S'}$. Since $\mathbb{Q}_{S'} \in \mathcal{B}(\mathbb{P})$ and $\mathcal{B}(\mathbb{P})$ is a two-sided ideal in $\mathcal{A}(\mathbb{P})$, we find that $l_{S,S'}(a) - a \in \mathcal{B}(\mathbb{P})$.

        From this we find that $l_{S,S'}(a) l_{S,S'}(a') = (a + b) (a' + b') = aa' + b a' + a b' + bb'$ where $b,b' \in \mathcal{B}(\mathbb{P})$. Therefore, again using the fact that $\mathcal{B}(\mathbb{P})$ is a two-sided ideal in $\mathcal{A}(\mathbb{P})$, we find that $l_{S,S'}(a) l_{S,S'}(a') - aa' \in \mathcal{B}(\mathbb{P})$.
    \end{proof}
\end{lemma}

Using these lemmas, we can give the proof of Theorem~\ref{thm:equivalent_projected}.
\begin{proof}[Proof of Theorem~\ref{thm:equivalent_projected}]
    First, we define $\varphi$ to have components
    \begin{equation}
        \varphi_S : A(\mathbb{P}) \{ S \} \to \overline{(\mathbb{P} A \mathbb{P})} \{ S^{+} \}
    \end{equation}
    that are induced from the linear map $l_{S,S^{+}}$ defined in Lemma~\ref{lemma:l_lemma}. In particular, observe that if $a \in \mathcal{B}(\mathbb{P})\{ S \}$, then by Lemma~\ref{lemma:projector_characterization_of_B}, there exists $S' \supseteq S^{+}$ such that $a \mathbb{P}_{S'} = \mathbb{P}_{S'} a = 0$. This implies that $\mathbb{P}_{S^{+}} a \mathbb{P}_{S^{+}} \in (\mathbb{P} A \mathbb{P})_0 \{ S^{+} \}$. This ensures that $l_{S,S^{+}}$ defines a well-defined linear map $\varphi_S$, and the multiplicative properties of $l$ ensure that $\varphi_S$ is a homomorphism.

    To construct the inverse map, from Lemma~\ref{lemma:projected_inclusion} we have the inclusion
    \begin{equation}
        (\mathbb{P} A \mathbb{P}) \{ S \} \leq \mathcal{A}(\mathbb{P}) \{ S \}
        \label{eq:an_inclusion}
    \end{equation}
    If $a \in (\mathbb{P} a \mathbb{P})_0$, then there exists $S' \supseteq S$ such that $\mathbb{P}_{S'} a \mathbb{P}_{S'} = 0$. From Lemma~\ref{lemma:l_lemma} we have that $\mathbb{P}_{S'} a \mathbb{P}_{S'} - a \in \mathcal{B}(\mathbb{P})$, so we conclude that $a \in \mathcal{B}(\mathbb{P})$. Therefore, the inclusion \eqnref{eq:an_inclusion} induces a homomorphism
    \begin{equation}
        \psi_S :  \overline{\mathbb{P} A \mathbb{P}} \{ S \} \to A(\mathbb{P}) \{ S \},
    \end{equation}
    which defines an ultra-local homomorphism $\psi : \overline{\mathbb{P} A \mathbb{P}} \to A(\mathbb{P})$.
    One then readily verifies (invoking Lemma~\ref{lemma:l_lemma} again) that $\varphi \circ \psi$ and $\psi \circ \varphi$ are the identity. This completes the proof of Theorem~\ref{thm:equivalent_projected}.
\end{proof}

As an example of a result that can be obtained from the new formulation of the projected algebra that would be more challenging to obtain from the original definition, let us prove that, given two nets of commuting projectors $\mathbb{P}$ and $\mathbb{P}'$ on locally-$C_*$ algebras $A$ and $A'$ (we will assume that we fix a proximity relation such that $\mathbb{P}$ and $\mathbb{P}'$ are both progressive), we have an ultra-local isomorphism
\begin{equation}
    \overline{(\mathbb{P} \otimes \mathbb{P}')} \overline{(A \otimes A')} \overline{(\mathbb{P} \otimes \mathbb{P}')} \cong (\overline{\mathbb{P} A \mathbb{P}}) \otimes (\overline{\mathbb{P}' A' \mathbb{P}'}),
\end{equation}
Physically, this encodes the property that the projected/boundary algebra of the stack of two commuting models should be the tensor product of the projected/boundary algebra of the individual commuting models.

Indeed, we have from the definition that
\begin{equation}
    (\mathbb{P} \otimes \mathbb{P}') (A \otimes A') (\mathbb{P} \otimes \mathbb{P}') = (\mathbb{P} A \mathbb{P}) \otimes (\mathbb{P}' A' \mathbb{P}'),
\end{equation}
so the only thing we have to prove is that the stabilization procedure is compatible with the tensor product, which is the content of the following lemma.
\begin{lemma}
    Let $A$ and $\widetilde{A}$ be two local nets of $C^*$ algebras. Then $\overline{A \otimes \widetilde{A}}$ is ultra-locally isomorphic to $\overline{A} \otimes \overline{\widetilde{A}}$.
    \begin{proof}
        First note that
        by a dimension argument similar to the proof of Lemma~\ref{lemma:projector_characterization_of_B}, we can show that for any bounded $S$, there exists a bounded $S' \supseteq S$ such that $A_0 \{ S \} = \mathrm{ker} \, \eta_{S,S'}$ and similarly for $\widetilde{A}_0$ and $(A \otimes \widetilde{A})_0$ (and we can choose $S'$ to be the same for all three). Thus, we have
        \begin{subequations}
        \begin{align}
            \overline{(A \otimes \widetilde{A})} \{ S \} &= (A \otimes \widetilde{A}) / \mathrm{ker} \, (\eta_{S,S'} \otimes \widetilde{\eta}_{S,S'}) \\
            &\cong \mathrm{im} \, (\eta_{S,S'} \otimes \widetilde{\eta}_{S,S'}) \\
            &\cong (\mathrm{im} \, \eta_{S,S'}) \otimes (\mathrm{im} \, \widetilde{\eta}_{S,S'}) \\
            &\cong \overline{A} \{ S \} \otimes \overline{\widetilde{A}} \{ S \}.
        \end{align}
        \end{subequations}
        This defines the ultra-local isomorphism from $\overline{A \otimes \widetilde{A}}$ to $\overline{A} \otimes \overline{\widetilde{A}}$.
    \end{proof}
\end{lemma}

Another result that is easy to show is that if $\mathbb{P}$ is the trivial net, i.e.\ $\mathbb{P}_S = \unit$ for all bounded $S$ (and we choose the trivial proximity relation $x \sim y \Leftrightarrow x = y$, then $\overline{\mathbb{P} A \mathbb{P}}$ is ultra-locally isomorphic to $A$. More generally, if there exists a set $W$ such that $\mathbb{P}_S$ is always supported on $W \cap S$ for all bounded $S$, then we can choose the proximity relation such that if $x \in W^c$ and $y \in X$ then $x \sim y \Leftrightarrow x = y$. In that case one has that $\overline{\mathbb{P} A \mathbb{P}} \{ W^c \}$ is ultra-locally isomorphic to $A \{ W^c \}$.

\section{Existence of fermionic representations}
\label{appendix:fermionic_symmetries}

In this appendix, we will prove the statement about fermionic symmetries mentioned in \autoref{subsec:symmetric_constraints}. Recall that a fermionic symmetry is a representation of \(G_f\) on the local automorphisms of a fermionic local algebra, such that the generator of \(\ZZ_2^f \leq G_f\) maps to a fixed fermion parity automorphism \(\mathcal{P}\), which is a part of the data for the fermionic local algebra. 
We prove that if \(A\) is an on-site fermionic algebra and \(G_f\) is a compact Lie group with \(\ZZ_2^f\) contained in its center, then there is an on-site fermionic algebra \(B\) such that \(A \otimes B\) supports a fermionic representation of \(G_f\).

Let $A$ be a full fermionic matrix algebra on a finite-dimensional Hilbert space $\mathcal{H}_A$ with fermion parity operator $P_A$. Assuming that $\mathrm{dim} \mathcal{H}_A > 0$, it must be the case that either the $+1$ parity subspace of $\mathcal{H}_A$ or the $-1$ parity subspace must have dimension $\geq 1$. Without loss of generality let us assume that the $+1$ parity subspace has dimension $m \geq 1$. Then if we tensor with a full fermionic matrix algebra $B$ on an $D$-dimensional Hilbert space $\mathcal{H}_B$, with $P_B$ acting as $-1$, we see that the dimension of the $-1$ parity subspace of $\mathcal{H}_{A \otimes B} := \mathcal{H}_A \otimes \mathcal{H}_B$ is $Dm$.

Since $G_f$ is a compact Lie group it admits a finite-dimensional faithful unitary representation. Decomposing this representation as a sum of irreps, it must be the case that $\mathbb{Z}_2^f$ acts non-trivially in at least one of these irreps in order for the original representation to be faithful. We denote this irrep $U(g)$, and
we set $D$ equal to its dimension.  Since $\mathbb{Z}_2^f$ is contained in the center, by Schur's Lemma it must act as a scalar in an irreducible representation, hence the generator of $\mathbb{Z}_2^f$ must act as $-1$ since we required it to act non-trivially. Now if we form the tensor-product representation $U'(g) = U(g) \otimes \mathbb{I}_m$, it has dimension $Dm$ and still retains the property that the generator of $\mathbb{Z}_2^f$ acts as $-1$. But now we form the Hilbert space $\mathcal{H}_{A \otimes B}$ as described in the previous paragraph. Then we can define a representation $\widetilde{U}(g)$ of $G_f$ on $\mathcal{H}_{A \otimes B}$ which acts block-diagonally as $U'(g)$ on the $-1$-parity subspace and as the identity on the $+1$ parity subspace. It follows that $\mathbb{Z}_2^f$ acts as fermion parity in this representation. Therefore, we have constructed a fermionic representation of $G_f$ on $\mathcal{H}_{A \otimes B}$. By taking tensor-products, it follows that for any on-site fermionic algebra $A$, there exists an on-site fermionic algebra $B$ with the same support such that there is a fermionic representation of $G_f$ on $A \otimes B$.

\section{Stable equivalence classes of local automorphisms}
\label{appendix:stable_laut}
Let $A$ be a local algebra over $\mathbb{R}^{d}$. We define a set $\mathfrak{L}(A)$ according to
\begin{multline}
    \mathfrak{L}(A) = \{ (B,\varphi) : \mbox{$B$ is an on-site algebra over $\mathbb{R}^d$,} \\ \mbox{$\varphi$ is a local automorphism of $A \otimes B$} \}.
\end{multline}
We furthermore define a monoid operation on this according to
\begin{equation}
\label{eq:laut_monoid}
    (B,\varphi) \cdot (B', \varphi') = (B \otimes B',(\varphi \otimes \mathbb{I}_{B'}) \circ (\varphi' \otimes \mathbb{I}_{B})),
\end{equation}
where $\circ$ denotes composition of local automorphisms on $A \otimes B \otimes B'$.

Then we say that $(B'',\varphi'')$ is a \emph{blend} from $(B,\varphi)$ to $(B',\varphi')$ if there exists a finite $\xi$ such that $B''$ is the same as $B'$ on $(\xi,\infty) \times \mathbb{R}^{d-1}$ and the same as $B$ on $(-\infty,\xi) \times \mathbb{R}^{d-1}$, and $\varphi''$ agrees with $\varphi$ and $\varphi'$ for operators supported on $(-\infty,\xi) \times \mathbb{R}^{d-1}$ and $(\xi,\infty) \times \mathbb{R}^{d-1}$ respectively. Then we define blend equivalence through a chain of blends analogously to Appendix \ref{appendix:BlendEquivalenceDef}.

If we denote blend equivalence by the equivalence relation $\sim$, then the monoid operation~\eqref{eq:laut_monoid} is compatible with this equivalence relation, and therefore $\mathcal{L}(A) := \mathfrak{L}(A) / \sim$ defines a monoid. In fact, one can verify that it is in fact a group, i.e.\ every element has an inverse. To see this, note that $[(B,\varphi^{-1})]$ is the inverse element of $[(B,\varphi)]$ (where $[\cdot]$ denotes equivalence classes), by a similar argument to the one that shows that if $U$ is a QCA then $U^{-1}$ is its stacking inverse (\autoref{sec:CommutingBoundaryAnomalies}).

\section{Simple realizations of Abelian quantum doubles}
\label{appendix:SimpleRealizationDouble}

In this appendix, we generalize the result of \autoref{subsec:H2_constraints} that if a \emph{simple realization} of the toric code is symmetric under a \(\ZZ_2\)-rep, then the \(\ZZ_2\)-rep acts as an \(e\)-\(m\) exchange on anyons if and only if it has nontrivial \(\grpH{2}{\ZZ_2}{\QCAblendclass{1}{}}\) class. Namely, we show that in a simple realization of \emph{any} Abelian quantum double model, the action of a \(G\)-rep on anyons determines its \(\grpH{2}{G}{\QCAblendclass{1}{}}\) class.

Consider a commuting model such that the ground state is in the topological phase of the quantum double of a finite Abelian group $T$. Suppose that the ground state of the commuting model is fixed by a \(G\)-rep \(\repU\). Then \(\repU\) maps isolated excitations with a definite topological charge (anyon type) to other isolated excitations, which must also have definite topological charge. Thus, \(\repU\) determines a representation of \(G\) in the group of anyon automorphisms for the topological order of the ground state. Let the group of anyon automorphisms for the \(T\)-quantum-double be \(\mathcal{T}\). Then we have a homomorphism
\begin{equation}\label{eqn:GrepAnyonAuto}
    \repU_*: G \to \mathcal{T}.
\end{equation}
We will show that if the commuting model qualifies as a ``simple realization'' of the $T$-quantum double, then \(\repU_*\) uniquely determines the \(\grpH{2}{G}{\QCAblendclass{1}{}}\) anomaly of \(\repU\).

We define ``simple realization'' analogously to the case of the $\mathbb{Z}_2$-quantum double discussed in \autoref{subsec:H2_constraints}. First, we require that ``simple realizations'' have a boundary algebra that is locally isomorphic (up to tensoring with on-site algebras) to $A^T$, the subalgebra of $T$-symmetric operators in an on-site algebra in which each site carries a regular representation of $T$. There is a second condition to be a ``simple realization'' which we will state below. We note that, in particular, the standard quantum double Hamiltonian~\cite{Kitaev2003FaultTolerant}, viewed as a commuting model, qualifies as a simple realization.

Now we invoke the theory of local automorphisms of $A^T$ as described in Ref.~\cite{Ma2024}. (We will make the assumption that the classification of Ref.~\cite{Ma2024} applies to the stable equivalence classes, which we have called $\mathcal{L}$ in \autoref{subsec:H2_constraints}, though Ref.~\cite{Ma2024} strictly speaking did not consider stable equivalence classes.) In particular, Ref.~\cite{Ma2024} establishes that there is a surjective homomorphism \(k:\mathcal{L}/\QCAblendclass{1}{} \to \mathcal{T}\). Composing this with the homomorphism $\widetilde{\beta} : G \to \mathcal{L}/\QCAblendclass{1}{}$ discussed in Section \autoref{subsec:H2_constraints}, we obtain a homomorphism $\pi : G \to \mathcal{T}$. We state as one of the assumptions about ``simple realizations'' that $\pi$ actually agrees with the action of the symmetry on the bulk anyons in the ground state, i.e. $\pi = \mathcal{U}_*$.

Moreover, we observe that the index of an element $U \in \mathcal{L}$ can be expressed as
\begin{equation}
    \mathcal{I}(U) = a(U) d_{k(U)},
\end{equation}
where \(d_\mathfrak{t}\) is the quantum dimension of \(\mathfrak{t} \in \mathcal{T}\), regarded as an extrinsic defect~\cite{Barkeshli2019Fractionalization}, and $a(U) \in \mathbb{Q}_\times$.
Thus, combined with the general discussion of the boundary trivialization condition in \autoref{subsec:H2_constraints}, we see that the anomaly class in \(\grpH{2}{G}{\QCAblendclass{1}{}}\) is fully determined by $\pi$, and in particular by $d_{\pi(g)}$.

\section{Anyon automorphisms and lattice anomalies in MBL phases of Abelian quantum double models}
\label{appendix:TCBoundaryStructure}

In \autoref{sec:ToricCodeMBL} we conjectured that a \(\ZZ_2\)-rep which preserves the ground state of a Hamiltonian \(H\) in the MBL phase of the toric code acts as an \(e\)-\(m\) exchange anyon automorphism if and only if its \(\grpH{2}{\ZZ_2}{\QCAblendclass{1}{}}\) anomaly is nontrivial (Conjecture~\ref{conj:em_swap}). 
In this appendix, we extend Conjecture~\ref{conj:em_swap} to MBL phases of more general \(T\)-quantum-double models (with \(T\) finite and Abelian) and more general finite symmetry groups \(G\). Furthermore, we make progress towards a proof of our conjecture by making a characterization of the symmetric boundary algebras in such a phase when the symmetry maps LIOMs to other LIOMs. 

The generalization of Conjecture~\ref{conj:em_swap} is stated as follows.

\begin{conjecture}
\label{conj:em_swap_generalized}
Consider an MBL LIOM algebra $L$ that lies in the MBL phase of the \(T\)-quantum-double model~\cite{Kitaev2003FaultTolerant}. Suppose \(\repU\) is a \(G\)-rep which preserves $L$, i.e.\ \(\repU_g[L] = L\). Suppose furthermore that there is a pure state $\ket{\Psi}$ which is a simultaneous eigenstate of all the elements of $L$ and is left invariant by $\repU$. Then \(\repU_*\), the action of \(\repU\) on anyons on top of $\ket{\Psi}$ as in \eqnref{eqn:GrepAnyonAuto}, determines the \(\grpH{2}{G}{\QCAblendclass{1}{}}\) class of \(\repU\).
\end{conjecture}

As progress towards proving this conjecture, we characterize the boundary algebras of weakly symmetric MBL Hamiltonians in the \(T\)-quantum-double phase. Namely, we show that there is a weakly symmetric commuting model \(H\) with terms in \(L\) and ground state \(\ket{\Psi}\), and a symmetric restriction \(\widetilde{H}\) such that the boundary algebra is
\begin{equation}\label{eqn:TCBounadryStructure}
    A_W(\widetilde{H}) \cong A^T \otimes A_{\mathrm{cons}},
\end{equation}
where \(A^T\) is locally isomorphic to the boundary algebra of the standard quantum-double-model, and \(A_{\mathrm{cons}}\) is an on-site algebra with local constraints. 
That is, $A_{\mathrm{cons}}$ is a projected algebra of an on-site algebra, as defined in \autoref{sec:CommutingModelsBoundaryAlgebra} and Appendices~\ref{appendix:infinite_system}-\ref{appendix:AlternativeProjectedAlgebra}.

In the boundary algebra of a simple realization of the \(T\)-quantum-double, \(\repU_*\) indeed determines the \(\grpH{2}{G}{\QCAblendclass{1}{}}\) invariant (\autoref{subsec:H2_constraints}, Appendix~\ref{appendix:SimpleRealizationDouble}). 
Completing the proof of Conjecture~\ref{conj:em_swap_generalized} would require extending a similar result to boundary algebras of the form \eqnref{eqn:TCBounadryStructure}.

We now proceed to the construction for \eqnref{eqn:TCBounadryStructure}. 
For concision, we will focus on \(T=\ZZ_2\)---the toric code, with the LIOM algebra being generated by star and plaquette operators---and later explain the generalization to all finite Abelian \(T\).
We choose signs for the toric code star, \(A_v\), and plaquette, \(B_p\), operators such that \(\ket{\Psi}\) is the simultaneous \(+1\)-eigenstate of all \(A_v\) and \(B_p\).

Unlike the simple example in \autoref{sec:em-qubit} where $\repU_x$ simply exchanges the star and plaquette LIOMs for each unit cell, in the current setting the stars and plaquettes can be mapped by \(\repU_g\) to a complicated sum of products of other stars and plaquettes.
To construct an MBL Hamiltonian and boundary algebra which is invariant under each $\repU_g$, we cannot simply choose the bulk commuting terms $h_i$ to be the stars and plaquettes and $\widetilde I$ (the subset of terms included in the restricted Hamiltonian \(\widetilde{H}\)) to index projectors supported within some \(W \subseteq\RR^2\).

Instead, we choose an over-complete set of bulk terms: $\mathbb{P}_{v,g}$ and $\mathbb{P}_{p,g}$, which are projectors defined to annihilate the $+1$-eigenspaces of $\repU_g[A_v]$ and $\repU_g[B_p]$, respectively, where $v$ and $p$ run over all stars and plaquettes, and $g$ runs over all (finitely many) elements of $G$.
Since $\repU$ preserves the LIOM algebra, we have that $\mathbb{P}_{v,g}$ and $\mathbb{P}_{p,g}$ are each a sum of products of some $A_v$ and $B_p$. Furthermore, the ground state is still annihilated by all of these projectors, as \(\repU_g\) preserves that state, and each \(\mathbb{P}_{v, g}\) is supported within a finite range of \(A_v\) (similarly for \(\mathbb{P}_{p,g}\)) as \(\repU_g\) is always internal for finite \(G\).

\begin{figure}
    \centering
    \includegraphics[width=\linewidth]{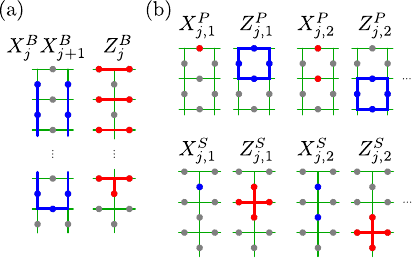}
    \caption{Construction of the boundary algebra for the toric code. Below each operator, a portion of the lattice is shown (with the top being the boundary) showing the definition of the operator on the lattice. Circles indicate qubits and red/blue indicate Pauli $X$/$Z$ operators respectively. (a)~The part of the boundary algebra isomorphic to \(A^{\ZZ_2}\) is generated by anyon string operators terminating on the boundary. (b)~The generators of \(M\), and their correspondence with the generators of an on-site algebra. The $y$ coordinate extends downwards until the region containing all the (original) star and plaquette operators indexed by $\widetilde{I}$.
    }
    \label{fig:tc-bdy-alg}
\end{figure}

To construct the restriction \(\widetilde{H}\), we choose the index set $\widetilde{I}$ to include all projectors \(\mathbb{P}_{p,1}\) and \(\mathbb{P}_{v,1}\) supported in some large region (say a half volume) and all \(\mathbb{P}_{p,g}\) and \(\mathbb{P}_{v,g}\) corresponding to these plaquettes \(p\) and vertices \(v\). Then \(\repU_g\) permutes these projectors, and the associated boundary algebra is symmetric so that the map \(\beta : G \to \mathcal{L}\) can be defined (\autoref{subsec:H2_constraints}). We choose \(W\), the total support of the boundary algebra, to be \(X_{\widetilde{H}}^{+\xi}\), where \(\xi\) is much larger than the range of any $\repU_g$.
Then, the restriction of the $G$-rep, $\widetilde{\repU}_g$, can be chosen so that its action on \(X_{\widetilde{H}}^{+\xi/2}\) is the same as that of $\repU_g$, and that any \(\bdyU_{g,h}\) is supported in \((W^c)^{+\xi/2} \cap W\).

First, we focus on the projectors \(\mathbb{P}_{p,1}\) and \(\mathbb{P}_{v,1}\). \autoref{fig:tc-bdy-alg} illustrates a presentation of the boundary algebra for the projectors with \(g=1\) (the usual toric code boundary) which explicitly demonstrates that it is a tensor product
\begin{equation}
    A_{\bdy} \cong A^{\ZZ_2} \otimes M.
\end{equation}
Here, the subalgebra isomorphic to \(A^{\ZZ_2}\) is represented by anyon loops which begin and end on the boundary of \(W\), which are exchanged by an \(e\)-\(m\) automorphism. 
\autoref{fig:tc-bdy-alg}(a) provides an identification of these strings with the generators \(X^B_j X^B_{j+1}\) and \(Z^B_j\) of a \(\ZZ_2\)-symmetric subalgebra of an on-site algebra with symmetry \(\prod_j Z^B_j\). 
Indeed, the generators for the subalgebra shown in \autoref{fig:tc-bdy-alg}(a) are exactly the same as those for the usual generators of the toric code boundary algebra~\cite{Ho2015Edge,Long2024Edge}.
\autoref{fig:tc-bdy-alg}(b) shows such an isomorphism between the remaining generators of \(A_{\bdy}\) and an on-site algebra \(M\) with Pauli-\(Z\) and \(X\) generators. We choose Pauli-\(Z\)s to be represented by the plaquette and star operators supported in \((W^c)^{+\xi} \cap W\), labeled  \(Z^S_{j,n}\) and \(Z^P_{j,n}\) in \autoref{fig:tc-bdy-alg}(b), where \(j\) indexes the position of the operator parallel to the boundary, and \(n\) indexes the perpendicular position.

Now we restore $\mathbb{P}_{v,g}$ and $\mathbb{P}_{p,g}$ for \(g \neq 1\). By assumption, these are sums of products of $A_v$ and $B_p$, and so are represented in the boundary algebra for $\mathbb{P}_{v,1}$ and $\mathbb{P}_{p,1}$ by finite sums of products of \(Z^S_{j,n}\) and \(Z^P_{j,n}\).
The symmetric boundary algebra \(A_W(\widetilde{H})\) is the projected algebra of \(A_\bdy\), with projectors represented by $\mathbb{P}_{v,g}$ and $\mathbb{P}_{p,g}$. These projectors are all supported in the on-site part of \(A_\bdy\), so Eq.~\eqref{eqn:TCBounadryStructure} follows.

To extend this construction to all finite Abelian groups \(T = \ZZ_{n_1} \times \cdots \times \ZZ_{n_p}\), we note that the \(T\)-quantum-double LIOMs are equivalent by a finite depth circuit to a stack of \(\ZZ_{n_i}\)-quantum-double LIOMs.
The construction of a Hamiltonian with commuting projector terms \(\mathbb{P}_j\) and boundary algebra as in \autoref{fig:tc-bdy-alg} for the \(\ZZ_2\) toric code extends to each of the \(\ZZ_{n_i}\) factors.
However, the additional terms \(\repU_g[\mathbb{P}_j]\) required to make the Hamiltonian symmetric now act on all the \(\ZZ_{n_i}\) boundaries.
Regardless, the resulting boundary algebra is of the form Eq.~\eqref{eqn:TCBounadryStructure}.

\section{Bulk-boundary correspondence for blend anomalies}
\label{sec:BlendBulkBoundary}

In this appendix we adapt the bulk-boundary correspondence between conjugacy anomalies and \(G\)-QCAs in \autoref{sec:BulkBoundaryCorrespondence} to lattice anomalies defined in terms of blends, as used in the rest of the paper.
An analogous construction can also be made for commuting models, but for the sake of brevity we will only discuss \(G\)-QCAs.

As discussed in \autoref{sec:BulkBoundaryCorrespondence}, formulating this correspondence requires an expanded notion of \(G\)-QCA which includes QCAs with symmetries which do not necessarily act on-site. The expanded definition is presented in Appendix~\ref{sec:G-QCA}. In Appendix~\ref{sec:anomaly_to_symmetric} we construct an injective homomorphism
\begin{equation}
    \alpha : \LatAnom{d}{G} \to \QCAblendclass{d+1}{G}
\end{equation}
which encodes the bulk-boundary correspondence. Indeed, \(\alpha\) has a left inverse given by the boundary algebra map \(\beta\) appearing in \autoref{sec:BulkBoundaryCorrespondence}. Finally, in Appendix~\ref{sec:anomaly_to_QFT}, we show that we continue to have a commuting square
\begin{equation}
\begin{tikzcd}
 \QCAblendclass{d+1}{G} \arrow[r,"e_*"] 
 &\invblendclass{d+1}{G} \arrow[d,"\text{edge}"] \\
 \LatAnom{d}{G} \arrow[u,hook,"\alpha"] \arrow[r,"\varphi"] & \QFTAnom{d}{G}
\end{tikzcd},
\label{eqn:AppendixSquare}
\end{equation}
similar to \eqnref{eqn:LatticeToQFTSquare}, except that the left vertical map may not be an isomorphism.

\subsection{Expanded definition of \texorpdfstring{\(G\)}{G}-QCA}
\label{sec:G-QCA}

In this appendix, we use an alternative definition of \(G\)-QCA, presented in this section for \(d \geq 2\).

As in \autoref{sec:BulkBoundaryCorrespondence}, a \(d\)-dimensional \(G\)-QCA is a pair \((U, \onsiteU)\), where \(\onsiteU\) is a \(d\)-dimensional \(G\)-rep on an on-site algebra \(M\), and \(U\) is a QCA on \(M\) such that
\begin{equation}
    \onsiteU_g U = U \onsiteU_g
\end{equation}
for all \(g \in G\). However, we impose different restrictions on the \(G\)-rep \(\onsiteU\), which is no longer demanded to be on-site. Instead, we demand a weaker set of conditions. First, \(\onsiteU\) must factor as a tensor product over the last coordinate in \(\RR^d\). That is,
\begin{equation}\label{eqn:GQCAFactorization}
    \onsiteU = \bigotimes_{z \in \ZZ} \repU_z,
\end{equation}
where \(\repU_z\) is a \(G\)-rep on \(M\{\RR^{d-1} \times \{z\}\}\). 
Second, each \(\repU_z\) must be anomaly free as a \((d-1)\)-dimensional \(G\)-rep.
Finally, \(\onsiteU\) must admit a symmetric product state, which we denote \(\ket{0}\) (suppressing the dependence on \(\onsiteU\)).
(The last condition is only for convenience when constructing \(e_*\). It can always be guaranteed by stacking with a non-anomalous \(G\)-rep on which \(U\) acts trivially.)

Blending of \(G\)-QCAs is defined as usual.
A \(G\)-QCA \((V, \onsiteV)\) is a blend between \((U, \onsiteU)\) and \((W, \onsiteW)\) across \(R \subseteq\RR^d\), written \((U, \onsiteU) \BlendVia{(V,\onsiteV)} (W, \onsiteW)\), if \(U \BlendVia{V} W\) is a blend of QCAs and \(\onsiteU \BlendVia{\onsiteV} \onsiteW\) is a blend of \(G\)-reps.

\begin{figure}
    \centering
    \includegraphics[width=\linewidth]{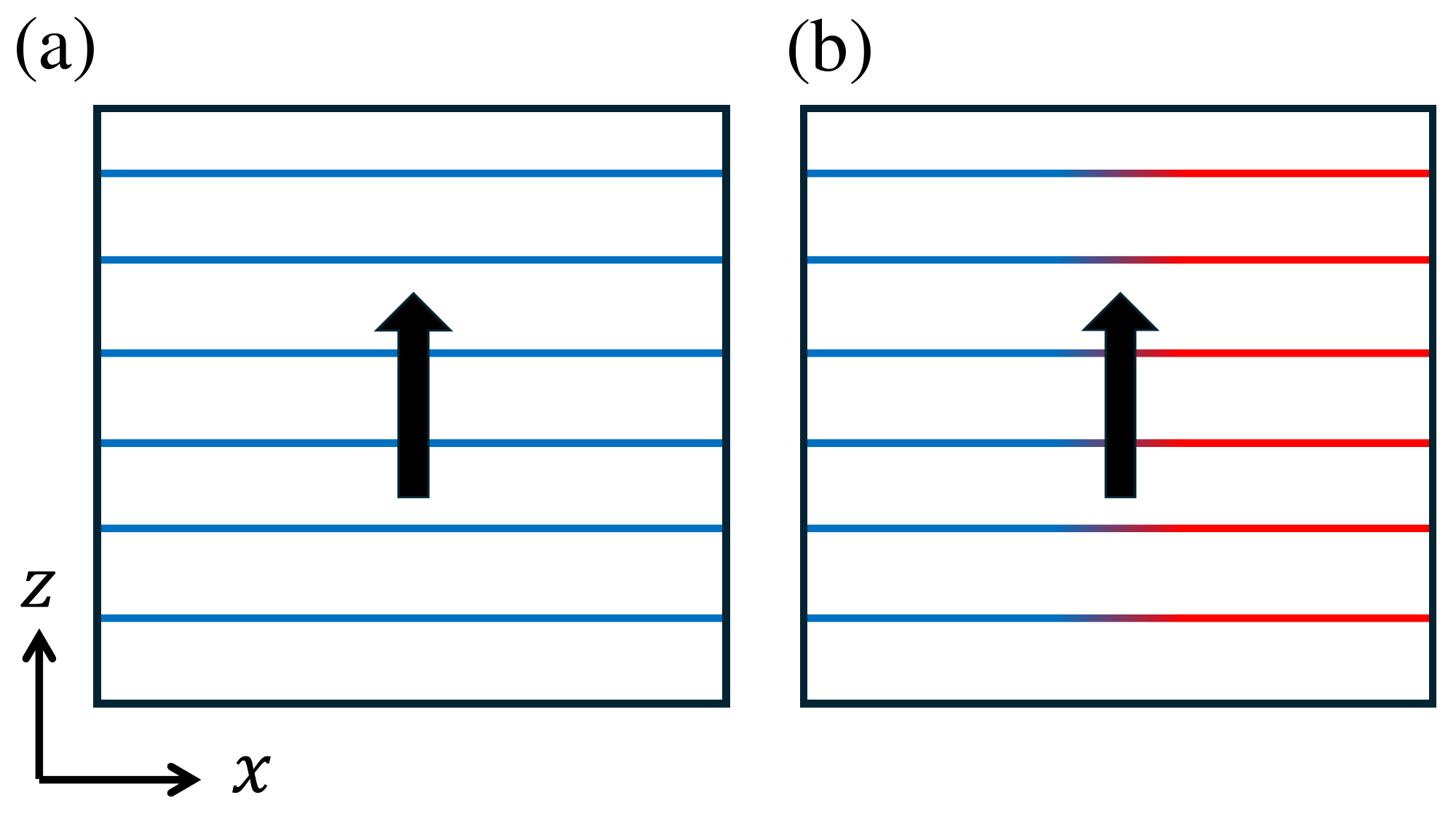}
    \caption{(a)~\(G\)-QCAs are required to be symmetric under a \(G\)-rep which factorizes as a tensor product along the \(z\) axis. For instance, a \(G\)-QCA can act as a translation in the \(z\) direction (arrow) and be invariant under a translationally-invariant stack of \(G\)-reps (blue lines). (b)~However, blends between \(G\)-QCA are defined along another axis (\(x\)). The \(G\)-reps appearing in blends must also factor as tensor products (blue to red lines).}
    \label{fig:GQCA_orientation}
\end{figure}

To define the set of blend equivalence classes of \(G\)-QCAs, which with some abuse of notation we continue to denote \(\QCAblendclass{d}{G}\), we choose \(R = (-\infty, 0]\times \RR^{d-1}\) to be a half volume aligned along the first coordinate in \(\RR^d\). Note that this is a \emph{different} coordinate to that appearing in \eqnref{eqn:GQCAFactorization}. We denote the first coordinate by \(x\) (\autoref{fig:GQCA_orientation}). Equivalence classes are then defined as in Appendix~\ref{appendix:BlendEquivalenceDef}.

\subsection{The bulk and boundary maps}
\label{sec:anomaly_to_symmetric}

The construction of the boundary-to-bulk map \(\alpha\) will be analogous to the similar map in \autoref{sec:BulkBoundaryCorrespondence}. The bulk-to-boundary map \(\beta\) will be identical to \autoref{sec:BulkBoundaryCorrespondence}.

We begin with \(\alpha\). We start with a lattice anomaly \([A,\repU] \in \LatAnom{d}{G}\) with representative \((A, \repU)\), where \(A\) is an invertible algebra and \(\repU\) is a \(G\)-rep on \(A\). We construct a \((d+1)\)-dimensional \(G\)-QCA \((S_A, \onsiteU)\) and define \(\alpha([A,\repU])\) to be the blend class of this \(G\)-QCA.

We begin by constructing the on-site algebra \(M\) on \(\RR^{d+1}\), and the \(G\)-rep \(\onsiteU\) it carries.
Let \((\overline{A}, \overline{\repU})\) be the spatial reflection of \((A, \repU)\), which is an inverse for \((A, \repU)\) by the folding trick of Appendix~\ref{appendix:BlendEquivalenceDef}. That is, their tensor product is anomaly free. Then by \autoref{sec:ConjAnomBlendIffConjugate}, \(\overline{A}\) is also a conjugacy inverse for \(A\), in that there is a local isomorphism
\begin{equation}
    V : A \otimes \overline{A} \to M'_0,
\end{equation}
where \(M'_0\) is on-site. 
Define an anomaly free \(G\)-rep on \(M'_0\) by conjugation of \(\repU \otimes \overline{\repU}\),
\begin{equation}
    \repU'_0 = V (\repU \otimes  \overline{\repU}) V^{-1}.
\end{equation}
In order to ensure that \(\onsiteU\) has a symmetric tensor product state, we include another factor of \(M_0'\) which carries the complex conjugate \(G\)-rep to \(\repU'_0\),
\begin{equation}
    M_0 = M_0' \otimes M_0', \quad
    \repU_0 = \repU_0' \otimes \repU_0'^*.
\end{equation}
To define \(\repU_0'^*\) for a unitary \(G\)-rep, fix any tensor product basis for the local Hilbert spaces of \(M_0'\) and act on \(\repU_{g0}'\) by complex conjugation in this bases. For an antiunitary \(G\)-rep \(\repU_{g0}' = \mathfrak{U}'_{g0} K^{\sigma(g)}\) (where \(K\) is complex conjugation in a fixed basis), define \(\repU_{g0}'^* = \mathfrak{U}'^*_{g0} K^{\sigma(g)}\). With these choices, the maximally entangled state between the \(M_0'\) factors (in a basis fixed by \(K\)) is a symmetric product state.

To complete the construction of the \((d+1)\)-dimensional algebra and \(G\)-rep, we stack these \(M_0\) factors,
\begin{equation}
    M = \bigotimes_{z \in \ZZ} M_z, \quad
    \onsiteU = \bigotimes_{z \in \ZZ} \repU_z,
\end{equation}
where \(M\{\RR^d \times \{z\}\} = M_z \cong M_0\), and \(\repU_z\) is conjugate to \(\repU_0\) under the ultra-local isomorphism identifying \(M_z\) with \(M_0\). The symmetric product state \(\ket{0}\) is again a maximally entangled state between the \(\repU'_0\) and \(\repU_0'^*\) factors.

The definition of the QCA \(S_A\) is the same as in \autoref{sec:BulkBoundaryCorrespondence} (\autoref{fig:InvAlgPumpingMap}), with the complex conjugate factors acting as spectators. That is, act by \(V^{-1} \otimes \id\) on each \(M_z \cong M_0' \otimes M_0'\) factor to transmute these to \(A \otimes \overline{A}\otimes M_0'\), and then translate the \(A\) factors one unit towards positive \(z\) and apply \(V\otimes \id\) to return to the original on-site algebra. Define \(\alpha\) by
\begin{subequations}
  \begin{align}
      \alpha: \LatAnom{d}{G} &\to \QCAblendclass{d+1}{G} \\
      [A,\repU] &\mapsto [S_A, \onsiteU],
  \end{align}  
\end{subequations}
where \([\cdot]\) denotes equivalence classes, similarly to \autoref{sec:BulkBoundaryCorrespondence}.

The bulk-to-boundary map \(\beta\) is the same boundary algebra map from \autoref{sec:BulkBoundaryCorrespondence}. That is, given a \(G\)-QCA \((U, \onsiteU)\), define the invertible boundary algebra, \(A\), across \(W = \RR^d \times (-\infty,0]\) as in \eqnref{eqn:GQCABoundaryAlgebra}, and equip it with a representation of \(G\), \(\repU\), as in \eqnref{eq:QCA_transformed_symmetry}. Note that, to apply \eqnref{eq:QCA_transformed_symmetry}, it is sufficient for \(\onsiteU\) to be a stack of layers of lower-dimensional anomaly free \(G\)-reps as in \eqnref{eqn:GQCAFactorization}, rather than being on-site. Thus, \((A,\repU)\) is a well-defined \(G\)-rep. Define
\begin{subequations}
  \begin{align}
    \beta:  \QCAblendclass{d+1}{G} &\to \LatAnom{d}{G} \\
    [U, \onsiteU] &\mapsto [A,\repU].
  \end{align}  
\end{subequations}

The maps \(\alpha\) and \(\beta\) are well defined. Indeed, if we have a blend of \(G\)-reps \((A, \repU) \BlendVia{(B, \repV)} (C, \repW)\) then we have a blend between \((S_A,\onsiteU)\) and \((S_C, \onsiteW)\) across \(R = (-\infty,0]\times \RR^d\) given by \((S_A, \onsiteU) \BlendVia{(S_B, \onsiteV)} (S_C, \onsiteW)\). It was important here that we defined blends along the \(x\) axis, while we constructed \(S_A\) as a stack along the \(z\) axis. Similarly, if we have a blend \((U, \onsiteU) \BlendVia{(V, \onsiteV)} (Y, \onsiteW)\) (across \(R\)), then the boundary algebra of \((V, \onsiteV)\) (across \(W\)) is a blend between the boundary algebras of \((U, \onsiteU)\) and \((Y, \onsiteW)\).
One can check that \(\alpha\) and \(\beta\) are homomorphisms.

That \(\beta \alpha\) is the identity map on \(\LatAnom{d}{G}\) follows from \autoref{fig:InvAlgPumpingMap}(c), as before. However, it is no longer clear that \(\alpha \beta\) is the identity on \(\QCAblendclass{d}{G}\). One needs to construct a blend across \(R\) between \((U,\onsiteU)\) and a \(G\)-QCA which pumps its boundary algebra across \(W\). Showing that this is possible appears to require stronger results regarding the relation of blends along different axes for \(G\)-QCAs which are not on-site. As such, while we can be sure that \(\alpha\) is injective and that \(\beta\) is surjective, we have not proved that these maps are actually isomorphisms for all \(G\).

\subsection{Lattice anomalies and QFT anomalies}
\label{sec:anomaly_to_QFT}

The square Eq.~\eqref{eqn:AppendixSquare} is completed by constructing the top horizontal map \(e_*\). As in \autoref{sec:BulkBoundaryCorrespondence}, we define a map \(e\) from \(G\)-QCAs to invertible \(G\)-states by applying the QCA to a product state.

Just as we have used an expanded definition of \(G\)-QCA, we must also allow an expanded definition of invertible \(G\)-state. The most general definition of an invertible \(G\)-state is a triple \((A, \repU, \ket{\psi})\) where \((A, \repU)\) is a \(G\)-rep and \(\ket{\psi}\) is a state on \(A\), such that there is a triple \((\overline{A}, \overline{\repU},\ket{\overline{\psi}})\) with \(\overline{A}\) being an inverse for \(A\) and such that \(\ket{\psi}\otimes \ket{\overline{\psi}}\) is related by a finite depth circuit to a product state in \((A\otimes \overline{A}, \repU \otimes \overline{\repU})\). When \(A\) is an on-site algebra, we will usually omit it from the triple. To define the blend classes \(\invblendclass{d}{G}\), we consider the subset of invertible \(G\)-states with \(A = M\) on-site and \(\repU = \onsiteU\) factoring along \(z\) as in \eqnref{eqn:GQCAFactorization}. Then \(\invblendclass{d}{G}\) consists of blend equivalence classes of such states, with the blends being across \(R\) (aligned along \(x\)). The blends must also use \(G\)-reps which factor along \(z\).

This definition of \(\invblendclass{d}{G}\) makes \(e_*\) well-defined. Define the evaluation map \(e\) by applying a \(G\)-QCA \((U, \onsiteU)\) to the symmetric product state \(\ket{0}\).
\begin{equation}
    \repU_g \ket{0} = e^{i \phi_g} \ket{0}
\end{equation}
for all \(g \in G\). Then we simply define \(e\) by the evaluation of \(U\) on this state,
\begin{equation}
    e(U, \onsiteU) = (U\ket{0}, \onsiteU),
\end{equation}
and \(e_*\) is the induced map on blend equivalence classes. Indeed, if \((V, \onsiteV)\) is a blend between \((U, \repU)\) and \((W, \onsiteW)\), then \(e(V, \onsiteV)\) demonstrates a blend between \(e(U, \onsiteU)\) and \(e(W, \onsiteW)\), so \(e_*\) is well-defined.

That \eqnref{eqn:AppendixSquare} commutes follows by showing that any symmetric Hamiltonian \(H\) for a \(G\)-rep \((A,\repU)\) can occur as an edge Hamiltonian for \(e(S_A, \onsiteU)\). This calculation is identical to that for conjugacy anomalies in \autoref{sec:ConjAnomToQFT}.

\section{Homotopy picture for the \texorpdfstring{\(\grpH{3}{G}{\QCAblendclass{d-2}{}}\)}{H3} invariant}
\label{appendix:3Simplex}

In this appendix, we give details for the calculation of the \(\grpH{3}{G}{\QCAblendclass{d-2}{}}\) invariant in the homotopy picture of \autoref{sec:HomotopyTheory}. Namely, we write out the cocycle representative illustrated in \autoref{fig:domain-walls}(c) as a network of blends in \(\Pi(\QCA_{d})\), and show that this network evaluates to the \((d-2)\)-dimensional QCA \(\mathcal{W}_{g,h,k}\), as defined in \eqnref{eqn:WSymmetryDef}.

In the homotopy picture, the cocycle representative for a \(\grpH{3}{G}{\QCAblendclass{d-2}{}}\) class corresponds to the homotopy class of a surface in \(\Pi(\QCA_{d})\),
which is geometrically the surface of a tetrahedron. 
It is constructed as follows. The four \emph{vertices} are
\begin{equation}
    \id,\ \repU_g,\ \repU_{gh}, \text{ and } \repU_{ghk}.
\end{equation}
Assuming all of these are trivial as \(d\)-dimensional QCAs, we have six \emph{edges} connecting these given by blends along the first coordinate \(x_1\):
\begin{widetext}
\begin{equation}
    \id \BlendVia{\resU_g} \repU_g,\quad
    \id \BlendVia{\resU_{gh}} \repU_{gh},\quad
    \id \BlendVia{\resU_{ghk}} \repU_{ghk},\quad
    \repU_g \BlendVia{\repU_g \resU_{h}} \repU_{gh},\quad
    \repU_{gh} \BlendVia{\repU_{gh} \resU_{k}} \repU_{ghk},\quad\text{and }\quad
    \repU_g \BlendVia{\repU_g \resU_{hk}} \repU_{ghk}.
\end{equation}
Assuming that the \((d-1)\)-dimensional QCA defined by
\begin{equation}
    \bdyU_{g,h} 
    = \left(\id \BlendVia{\resU_g} \repU_g \BlendVia{\repU_g\resU_h} \repU_{gh} \InverseBlendVia{\resU_{gh}} \id\right)
    = \resU_g \resU_h \resU^{-1}_{gh}
\end{equation}
is trivial (and similarly for \(\bdyU_{g,hk}\), \(\bdyU_{gh,k}\) and \(\bdyU_{h,k}\)), then we have \((d-1)\)-dimensional blends \(\id \BlendVia{\widetilde{\bdyU}_{g,h}} \bdyU_{g,h}\) along the coordinate \(x_2\). This gives us the following four \emph{faces}:
\begin{subequations}\label{eqn:SimplexFaces}
\begin{equation}
    \begin{tikzcd}
        & \repU_g \arrow[r,"\repU_g \resU_h",""{name=U,below}] & \repU_{gh} \arrow[dr,bend left,"\overline{\resU_{gh}}"] & \\
        \id \arrow[ur,bend left,"\resU_{g}"] \arrow[rrr,"\id"{below},""{name=D,above},xshift=-2pt] &&& \id
        \arrow[Rightarrow,from=D,to=U,"\widetilde{\bdyU}_{g,h}"]
    \end{tikzcd},\qquad
    \begin{tikzcd}
        & \repU_g \arrow[r,"\repU_g \resU_{hk}",""{name=U,below}] & \repU_{ghk} \arrow[dr,bend left,"\overline{\resU_{ghk}}"] & \\
        \id \arrow[ur,bend left,"\resU_{g}"] \arrow[rrr,"\id"{below},""{name=D,above},xshift=-4pt] &&& \id
        \arrow[Rightarrow,from=D,to=U,"\widetilde{\bdyU}_{g,hk}"]
    \end{tikzcd},
\end{equation}%
\begin{equation}
    \begin{tikzcd}
        & \repU_{gh} \arrow[r,"\repU_{gh} \resU_{k}",""{name=U,below}] & \repU_{ghk} \arrow[dr,bend left,"\overline{\resU_{ghk}}"] & \\
        \id \arrow[ur,bend left,"\resU_{gh}"] \arrow[rrr,"\id"{below},""{name=D,above},xshift=-2pt] &&& \id
        \arrow[Rightarrow,from=D,to=U,"\widetilde{\bdyU}_{gh,k}"]
    \end{tikzcd},\qquad\text{and}\qquad
    \begin{tikzcd}
        & \repU_{gh} \arrow[r,"\repU_{gh} \resU_{k}",""{name=U,below}] & \repU_{ghk} \arrow[dr,bend left,"\overline{\repU_g \resU_{hk}}"] & \\
        \repU_g \arrow[ur,bend left,"\repU_g\resU_{h}"] \arrow[rrr,"\repU_g"{below},""{name=D,above},xshift=-2pt] &&& \repU_g
        \arrow[Rightarrow,from=D,to=U," \repU_g \widetilde{\bdyU}_{h,k} "]
    \end{tikzcd}.
\end{equation}
\end{subequations}
The vertical double arrows (\(\Uparrow\)) are blends between the lower arrows and the composition of the upper arrows.

To evaluate the blend associated to this entire complex, we need to rewrite these blends-of-blends such that they are composable. That is, such that the source of one blend is the target of another. This is illustrated in \autoref{fig:domain-walls}(c), which as a composition of blends-of-blends is explicitly
\begin{equation}\label{eqn:FullSimplex}
    \begin{tikzcd}
        \id = \left( \id \BlendVia{\id} \id \right)  \\
        \resU_g \bdyU_{h,k} \resU_g^{-1} = \left( \id \BlendVia{\resU_g} \repU_g \BlendVia{\repU_g \resU_h} \repU_{gh} \BlendVia{\repU_{gh} \resU_k} \repU_{ghk} \InverseBlendVia{\repU_g \resU_{hk}} \repU_g \InverseBlendVia{\resU_g} \id \right) 
        \arrow[u,Rightarrow,"\overline{\resU_g \widetilde{\bdyU}_{h,k} \resU_g^{-1}}",swap]\\
        \bdyU_{g,h}\bdyU_{gh,k} =
        \left( \id \BlendVia{\resU_g} \repU_g \BlendVia{\repU_g\resU_h} \repU_{gh} \BlendVia{\repU_{gh}\resU_k} \repU_{ghk} \InverseBlendVia{\resU_{ghk}} \id\right)
        =(\resU_g \bdyU_{h,k} \resU_{g}^{-1}) \bdyU_{g,hk} 
        \arrow[u,Rightarrow,"\overline{(\resU_g \bdyU_{h,k} \resU_g^{-1})\widetilde{\bdyU}_{g,hk}}",swap] \\
        \bdyU_{g,h} = \left(\id \BlendVia{\resU_g} \repU_g \BlendVia{\repU_g\resU_h} \repU_{gh} \InverseBlendVia{\resU_{gh}} \id\right) \arrow[u,Rightarrow,"\bdyU_{g,h}\widetilde{\bdyU}_{gh,k}",swap] \\
        \id = \left( \id \BlendVia{\id} \id \right) \arrow[u,Rightarrow,"\widetilde{\bdyU}_{g,h}",swap]
    \end{tikzcd}.
\end{equation}
\end{widetext}
Evaluating the vertical compositions using the composition and inversion rules for \(\Pi(\QCA_{d-1})\) gives
\begin{equation}
    \widetilde{\bdyU}_{g,h}  \widetilde{\bdyU}_{gh,k} \widetilde{\bdyU}^{-1}_{g,hk}   (\resU_g \widetilde{\bdyU}^{-1}_{h,k} \resU_g^{-1}) = \mathcal{W}_{g,h,k},
\end{equation}
as claimed. Thus, the homotopy class of this surface is precisely the blend class of \(\mathcal{W}_{g,h,k}\), and the presentations of the \(\grpH{3}{G}{\QCAblendclass{d-2}{}}\) invariant from Secs.~\ref{sec:CohomologyInvariants} and~\ref{sec:HomotopyTheory} agree.

\bibliography{lattice_anomalies}

\end{document}